\pdfoutput=1
\documentclass[reqno]{amsart}

\usepackage[T1]{fontenc}
\usepackage{mathtools}
\usepackage{tikz-cd}
\usepackage{mleftright}
\usepackage{enumitem}
\usepackage{hyperref}
\usepackage{amssymb}
\usepackage{slashed}

\usepackage{microtype}
\usepackage[backend=bibtex, style=trad-abbrv, useprefix=false]{biblatex}
\newtheorem{theorem}{Theorem}[section]
\newtheorem{proposition}[theorem]{Proposition}
\newtheorem{lemma}[theorem]{Lemma}
\newtheorem{corollary}[theorem]{Corollary}

\theoremstyle{definition}
\newtheorem{definition}[theorem]{Definition}
\newtheorem{example}[theorem]{Example}

\newtheorem{theoremdefinition}[theorem]{Theorem-Definition}

\newtheorem{propositiondefinition}[theorem]{Proposition-Definition}

\theoremstyle{remark}
\newtheorem{remark}[theorem]{Remark}

\numberwithin{equation}{section}

\newcommand{\bC}{\mathbb{C}}
\newcommand{\bE}{\mathbb{E}}
\newcommand{\bF}{\mathbf{F}}
\newcommand{\bL}{\mathbb{L}}
\newcommand{\bN}{\mathbb{N}}
\newcommand{\bR}{\mathbb{R}}
\newcommand{\bT}{\mathbb{T}}
\newcommand{\bZ}{\mathbb{Z}}

\newcommand{\cH}{\mathcal{H}}
\newcommand{\cL}{\mathcal{L}}
\newcommand{\cE}{\mathcal{E}}
\newcommand{\cF}{\mathcal{F}}
\newcommand{\cO}{\mathcal{O}}
\newcommand{\cP}{\mathcal{P}}
\newcommand{\cS}{\mathcal{S}}
\newcommand{\cV}{\mathcal{V}}
\newcommand{\sci}{\mathfrak{i}}
\newcommand{\sce}{\mathfrak{e}}
\newcommand{\conj}[1]{\overline{#1}}
\newcommand{\ver}{\mathrm{ver}}
\newcommand{\hor}{\mathrm{hor}}
\newcommand{\bas}{\mathrm{bas}}
\newcommand{\tot}{\mathrm{tot}}
\newcommand{\even}{\mathrm{even}}
\DeclareMathOperator{\ran}{ran}
\newcommand{\CP}{\mathbb{C}\mathrm{P}}
\newcommand{\loc}{\mathrm{loc}}
\newcommand{\alg}{\mathrm{alg}}
\newcommand{\act}{\mathbin{\triangleright}}

\DeclareMathOperator{\Aut}{Aut}
\DeclareMathOperator{\Ad}{Ad}
\DeclareMathOperator{\dAd}{\widehat{Ad}}
\DeclareMathOperator{\Corr}{{\normalfont\textsc{Corr}}}
\DeclareMathOperator{\Diff}{Diff}
\DeclareMathOperator{\tDiff}{\widetilde{Diff}}
\DeclareMathOperator{\coev}{coev}
\DeclareMathOperator{\ev}{ev}
\DeclareMathOperator{\id}{id}
\DeclareMathOperator{\End}{End}
\DeclareMathOperator{\Pic}{{\normalfont\textsc{Pic}}}
\DeclareMathOperator{\dPic}{{\normalfont\textsc{DPic}}}
\DeclareMathOperator{\Circ}{{\normalfont\textsc{Circ}}}
\DeclareMathOperator{\dCirc}{{\normalfont\textsc{DCirc}}}
\DeclareMathOperator{\Hor}{Hor}
\DeclareMathOperator{\pic}{Pic}
\DeclareMathOperator{\dpic}{DPic}
\DeclareMathOperator{\Obj}{Obj}
\DeclareMathOperator{\Unit}{U}
\DeclareMathOperator{\GL}{GL}
\DeclareMathOperator{\bb}{bb}
\DeclareMathOperator{\cl}{cl}
\DeclareMathOperator{\Zent}{Z}
\DeclareMathOperator{\SU}{SU}
\DeclareMathOperator{\vol}{vol}
\DeclareMathOperator{\twogrp}{2{\normalfont Grp}}
\DeclareMathOperator{\CE}{CE}
\DeclareMathOperator{\Tot}{Tot}

\newcommand{\Cstar}{\ensuremath{\mathrm{C}^\ast}}

\newcommand{\hotimes}{\mathbin{\widehat{\otimes}}}

\newcommand{\grp}[1]{{\normalfont\textsc{#1}}}

\newcommand{\tsD}{\widetilde{\slashed{D}}}

\providecommand\given{}
\newcommand\SetSymbol[1][]{%
\nonscript\:#1\vert
\allowbreak
\nonscript\:
\mathopen{}}
\DeclarePairedDelimiterX\Set[1]\{\}{%
\renewcommand\given{\SetSymbol[\delimsize]}
#1
}

\DeclarePairedDelimiterX\hp[2]\lparen\rparen{
\ifblank{#1}{\ifblank{#2}{\cdot,\cdot}{\cdot,#2}}{\ifblank{#2}{#1,\cdot}{#1,#2}}
}

\DeclarePairedDelimiterX\ip[2]\langle\rangle{
\ifblank{#1}{\ifblank{#2}{\cdot,\cdot}{\cdot,#2}}{\ifblank{#2}{#1,\cdot}{#1,#2}}
}

\DeclarePairedDelimiterX\norm[1]\lVert\rVert{
\ifblank{#1}{\:\cdot\:}{#1}
}

\makeatletter
\newcommand{\cconj}[1]{\overline{\dbl@overline{#1}}}
\newcommand{\dbl@overline}[1]{\mathpalette\dbl@@overline{#1}}
\newcommand{\dbl@@overline}[2]{%
  \begingroup
  \sbox\z@{$\m@th#1\overline{#2}$}%
  \ht\z@=\dimexpr\ht\z@-2\dbl@adjust{#1}\relax
  \box\z@
  \ifx#1\scriptstyle\kern-\scriptspace\else
  \ifx#1\scriptscriptstyle\kern-\scriptspace\fi\fi
  \endgroup
}
\newcommand{\dbl@adjust}[1]{%
  \fontdimen8
  \ifx#1\displaystyle\textfont\else
  \ifx#1\textstyle\textfont\else
  \ifx#1\scriptstyle\scriptfont\else
  \scriptscriptfont\fi\fi\fi 3
}
\makeatother

\newcommand{\du}{\mathrm{d}}
\newcommand{\ract}{\triangleleft}
\newcommand{\U}{\operatorname{U}(1)}

\newcommand{\iu}{\mathrm{i}}

\newcommand{\rest}[2]{{#1}\!\!\restriction_{#2}}
\newcommand{\leg}[2]{{#1}_{\langle #2 \rangle}}

\begin{document}

\title[Noncommutative \texorpdfstring{\(\U\)}{U(1)}-gauge theory]{Geometric foundations for classical \texorpdfstring{\(\U\)}{U(1)}-gauge theory on noncommutative manifolds}

\author{Branimir \'{C}a\'{c}i\'{c}}
\address{Department of Mathematics \& Statistics, University of New Brunswick, PO Box 4400, Fredericton, NB\space\space{}E3B 4A3, Canada}
\email{\url{mailto:bcacic@unb.ca}}

\begin{abstract}
	We systematically extend the elementary differential and Riemannian geometry of classical \(\operatorname{U}(1)\)-gauge theory to the noncommutative setting by combining recent advances in noncommutative Riemannian geometry with the theory of coherent \(2\)-groups.
	We show that Hermitian line bimodules with Hermitian bimodule connection over a unital pre-\(\mathrm{C}^\ast\)-algebra with \(\ast\)-exterior algebra form a coherent \(2\)-group, and we prove that weak monoidal functors between coherent \(2\)-groups canonically define bar or involutive monoidal functors in the sense of Beggs--Majid and Egger, respectively.
	Hence, we prove that a suitable Hermitian line bimodule with Hermitian bimodule connection yields an essentially unique differentiable quantum principal \(\U\)-bundle with principal connection and \textit{vice versa}; here, \(\operatorname{U}(1)\) is \(q\)-deformed for \(q\) a numerical invariant of the bimodule connection.
	From there, we formulate and solve the interrelated lifting problems for noncommutative Riemannian structure in terms of abstract Hodge star operators and formal spectral triples, respectively; all the while, we account precisely for emergent modular phenomena of geometric nature.
	In particular, it follows that the spin Dirac spectral triple on quantum \(\operatorname{\mathbb{C}P}^1\) does not lift to a twisted spectral triple on \(3\)-dimensional quantum \(\operatorname{SU}(2)\) but does recover Kaad--Kyed's compact quantum metric space on quantum \(\operatorname{SU}(2)\) for a canonical choice of parameters.
\end{abstract}

\vspace*{-0.5em}
\maketitle
\vfill

\tableofcontents

\newpage

\addtocontents{toc}{\protect\setcounter{tocdepth}{1}}
\section{Introduction}\label{sec:1}

The primordial application of noncommutative (NC) geometry to theoretical physics is the conceptually economical construction of physical models as classical physics on NC manifolds.
For example, in Bellissard--Van Elst--Schulz-Baldes' model of the integer quantum hall effect~\cite{BVESB}, the NC Brouillin zone accounts for both the magnetic field and disorder in the crystal, while in particle physics~\cite{VDDVS} and cosmological models~\cite{Marcolli} using the spectral action principle~\cite{CC}, \(0\)-dimensional NC fibres encode the particle content.
The prototypical such construction is Connes--Rieffel's Yang--Mills gauge theory on irrational NC \(2\)-tori~\cite{CR}, the first of many NC field theories built from a range of seemingly disparate variations on Connes's NC differential geometry~\cite{Connes80,Connes85}.
Indeed, one can approach various aspects or special cases of NC \(\U\)-gauge theory in terms of quantum principal bundles~\cite{BrM,DJ97}, principal \(\U\)-spectral triples~\cite{DS,BMS,CaMe}, or even the spectral action principle~\cite{VDDVS}.

This fragmentary understanding of classical \(\U\)-gauge theory on NC manifolds is an obstacle to physical applications.
For example, in the quantum adiabatic transport approach to the integer quantum Hall effect, one probes the qualitative behaviour of relevant observables by considering the integer quantum Hall effect on general compact Riemann surfaces~\cite{ASZ}.
A satisfactory generalisation to NC Riemann surfaces would require a precise extension of the elementary differential and Riemannian geometry of classical \(\U\)-gauge theory \emph{as a coherent whole} that is compatible with both NC K\"{a}hler geometry~\cite{OB} and the framework of spectral triples~\cite{Connes95}.
Our goal is to effect just such an extension, which would be just as applicable to the study of electromagnetism on NC spacetimes~\cite{MR} and to the refinement of NC \(T\)-duality as applied to the bulk-edge correspondence~\cite{MathaiThiang}.

We construct this extension from the ground up according to the philosophy of quantum Riemannian geometry~\cite{BeMa}.
Thus, an NC Riemannian manifold is an NC manifold---a unital pre-\Cstar-algebra together with a \(\ast\)-exterior calculus---equipped with additional structure, whether an abstract Hodge star operator or a spectral triple.
In stark contrast with other areas of NC geometry and operator algebras, this requires working exclusively `on the nose'---at worst, up to explicit isomorphism.
Fortunately, in our setting, we may obviate any resulting algebraic difficulties through the use of \emph{coherent \(2\)-groups}~\cite{BaezLauda} and \emph{bar categories}~\cite{BeMa09,Egger}.
Moreover, following relevant applications of unbounded \(\mathrm{KK}\)-theory~\cite{BMS,FR,CaMe}, we obviate a wide range of analytic and algebraic difficulties through the systematic use of finite tight Parseval frames on (pre-)Hilbert modules~\cite{FL}.

Our results have several immediate implications that we must leave for future work.
One is that the Gysin sequence for principal \(\U\)-bundles in de Rham cohomology generalises almost \emph{verbatim} to the NC setting; when combined with NC Hodge theory~\cite{OB,OBSVR}, this permits the efficient computation of NC de Rham cohomology for NC principal \(\U\)-bundles as well as a differential-geometric perspective on NC \(T\)-duality.
Another is that, under relatively mild hypotheses, we may constuct moduli spaces of \(\U\)-instantons with fixed topological sector (when non-empty) using NC Hodge theory and a generalised first Chern class in de Rham cohomology, which now fails to be a group homomorphism.
Indeed, the stage is now set for detailed investigation of Chern--Weil theory on NC principal \(\U\)-bundles.

Our results are independent of Schwieger--Wagner's cohomological classification of principal \(\bT^N\)-\Cstar-algebras~\cite{SW1} and Salda\~{n}a's Tannaka--Krein theorem~\cite{Saldana20} for differentiable quantum principal bundles \emph{d'apr\`{e}s} \DJ{}ur\dj{}evi\'{c}~\cite{Dj96}.
However, the former presages the r\^{o}le of coherent \(2\)-groups and their group-cohomological classification in the case of Abelian structure groups, while the latter will be prototypical for any generalisation of our results to non-Abelian or quantum structure groups.

\subsection*{Overview of results}

We begin in \S \ref{sec:2} by developing the elementary theory of NC Hermitian line bundles with unitary connection.
Let \(B\) be a unital pre-\Cstar-algebra with \(\ast\)-exterior algebra \((\Omega_B,\du_B)\).
Building on a proposal of Beggs--Brzezi\'{n}ski~\cite{BeBrz14}, we define \emph{Hermitian line \(B\)-bimodules with connection} to be suitable strong Morita auto-equivalences of \(B\) equipped with suitable extendable bimodule connections~\cite{BeMa18} with respect to \((\Omega_B,\du_B)\).
Then, building on results of Beggs--Majid~\cite{BeMa18}, we prove that Hermitian line \(B\)-bimodules with connection form a coherent \(2\)-group \(\dPic(B)\), the \emph{differential Picard \(2\)-group} of \((B;\Omega_B,\du_B)\).
The isomorphism classes of \(\dPic(B)\) still form a group \(\dpic(B)\), the \emph{differential Picard group}, whose canonical action on the graded centre \(\Zent(\Omega_B)\) of \(\Omega_B\) will appear throughout this work.
By results of Beggs--Majid~\cite{BeMa18}, this \(\dpic(B)\)-action admits a \(1\)-cocycle of supreme importance: the curvature \(2\)-forms of Hermitian line \(B\)-bimodules with connection.

Next, in \S \ref{sec:3}, we develop the elementary theory of NC principal \(\U\)-bundles with principal connection.
Given \(\kappa > 0\), we synthesize a definition of \emph{\(\kappa\)-differentiable quantum principal \(\U\)-bundle with connection} from work of Brzezi\'{n}ski--Majid~\cite{BrM}, Hajac~\cite{Hajac}, \DJ{}ur\dj{}evi\'{c}~\cite{DJ97}, and Beggs--Majid~\cite{BeMa}; here, the differential calculus on \(\U\) is deformed to satisfy \(\du{z} \cdot z = \kappa z \cdot \du{z}\).
One may define a functor that maps a \(\kappa\)-differentiable quantum principal \(\U\)-bundle with connection to its NC associated line bundle with connection of winding number \(-1\); we show that \([E,\nabla_E] \in \dpic(B)\) lies in the essential range of this functor if and only if its curvature \(2\)-form \(\bF_{[E,\nabla_E]}\) satisfies \(\bF_{[E,\nabla_E]} \ract [E,\nabla_E] = \kappa^{-1} \bF_{[E,\nabla_E]}\) with respect to the \(\dpic(B)\)-action on \(\Zent(\Omega_B)\).
Hence, we prove that this functor is an equivalence of categories onto its essential range, generalising the familiar dictionary between Hermitian line bundles with unitary connection and principal \(\U\)-bundles with principal connection.

Our proof depends on two apparently novel technical results on coherent \(2\)-groups.
The first, that \(\bZ\) is the free coherent \(2\)-group on one generator, is a straightforward corollary of Joyal--Street's group-cohomological classification of weak monoidal functors between coherent \(2\)-groups~\cite{JS}.
The second, that every weak monoidal functor between coherent \(2\)-groups is a \emph{bar functor} or \emph{involutive monoidal functor} in the sense of Beggs--Majid~\cite{BeMa18} and Egger~\cite{Egger}, respectively, is a non-trivial application of the coherence theorem for coherent \(2\)-groups of Ulbrich~\cite{Ulbrich} and Laplaza~\cite{Laplaza}.
We view this pair of results as an abstract distillation of Pimsner's construction~\cite{Pimsner}: by applying them to weak monoidal functors from \(\bZ\) to the coherent \(2\)-group \(\Pic(B)\) of Hermitian line \(B\)-bimodules, we may ultimately recover Arici--Kaad--Landi's characterisation~\cite{AKL} of NC topological principal \(\U\)-bundles.

At last, in \S \ref{sec:4}, we turn to the NC Riemannian geometry of NC principal \(\U\)-bundles with principal connection.
The best-known NC \(3\)-manifolds are total spaces of NC principal \(\U\)-bundles with principal connection.
However, \(3\)-dimensional quantum \(\SU(2)\) poses fundamental challenges: for example, it cannot be faithfully represented by a spectral triple~\cite{Schmuedgen}.
We draw on a range of advances in NC Riemannian geometry---unbounded \(\mathrm{KK}\)-theory~\cite{Mesland,KL,BMS}, NC K\"{a}hler geometry~\cite{OB}, and quantum Riemannian geometry~\cite{BeMa}---to \emph{lift} NC Riemannian geometry from well-behaved NC base spaces to NC total spaces.
Our guide is the commutative case: a principal \(\U\)-bundle \(\pi : X \to Y\) with principal connection \(\Pi\) admits a bijection between metrics on \(Y\) and \(\U\)-invariant metrics on \(X\) that make \(\Pi\) orthogonal and the fibres have unit length, which is defined by the constraint that \(\pi\) become a Riemannian submersion~\cite[\S 4]{AB}.

First, in \S\S \ref{sec:4.1} and \ref{sec:4.2}, we consider NC Riemannian geometry \emph{via} abstract Hodge operators: a \emph{Riemannian geometry} on an NC manifold \((B;\Omega_B,\du_B)\) is a pair \((\star,\tau)\), where \(\star\) generalises the Hodge star operator and \(\tau\) is a state generalising integration against the Riemannian volume form.
This suffices to formulate (Euclidean) Maxwell's equations, whose moduli spaces of solutions we study in future work.
We propose an analogous definition of \emph{total Riemannian geometry} for a \(\kappa\)-differentiable quantum principal \(\U\)-bundle with connection \((P;\Omega_P,\du_P;\Pi)\) on \((B;\Omega_B,\du_B)\), where failure of the Hodge operator to be right \(P\)-linear and \(\ast\)-preserving is governed by a commuting pair of modular automorphisms of \(\Omega_P\).
We show that \((\star,\tau)\) lifts to at most one total Riemannian geometry on \((P;\Omega_P,\du_P;\Pi)\), whose existence we characterize in terms of conformality of the corresponding Hermitian line \(B\)-bimodule with connection.
For example, the unique lift of the canonical Riemanian geometry on quantum \(\CP^1\) as an NC K\"{a}hler manifold to the \(q\)-monopole of Brzezi\'{n}ski--Majid~\cite{BrM} recovers a construction of Zampini~\cite{Zampini} together with a canonical choice of parameters.

Next, in \S \ref{sec:4.3}, we consider Connes's familiar NC Riemannian geometry \emph{via} spectral triples~\cite{Connes95}, which, following Schm\"{u}dgen~\cite{Schmuedgen}, we generalise to \emph{bounded commutator representations}.
We propose a definition of \emph{projectable commutator representation}, where represented \(1\)-forms are only locally bounded in a certain sense.
We then use a formal unbounded Kasparov product~\cite{Mesland,KL} to construct an equivalence of categories between faithful bounded commutator representations of \((B;\Omega_B,\du_B)\) and faithful projectable commutator representations of \((P;\Omega_P,\du_P;\Pi)\); isomorphism of the latter is \(\U\)-equivariant unitary equivalence up to perturbation by a suitable \emph{relative remainder}.
If \((B;\Omega_B,\du_B)\) is equipped with a liftable Riemannian geometry and \((P;\Omega_P,\du_P;\Pi)\) is equipped with its unique lift, then the resulting \emph{Hodge--de Rham commutator representation} of \((B;\Omega_B,\du_B)\) lifts to the resulting \emph{total Hodge--de Rham commutator representation} of \((P;\Omega_P,\du_P;\Pi)\).

Finally, in \S \ref{sec:4.4}, we draw on Connes--Moscovici's formalism of \emph{twisted spectral triples}~\cite{CM} to control unboundedness of represented \(1\)-forms.
We consider \emph{modular pairs} \((N,\nu)\), where \(\nu\) is a modular automorphism of \(\Omega_P\) and \(N\) is a suitable unbounded operator satisfying \(\nu = N^{-1}(\cdot)N\); let us say that \((N,\nu)\) \emph{damps} an unbounded operator \(S\) whenever \(N S N\) is bounded.
Hence, we define a \emph{vertical} or \emph{horizontal twist} for a faithful projectable commutator representation to be a modular pair that damps all represented vertical or horizontal \(1\)-forms, respectively.
We demonstrate a universal vertical twist and characterize the existence of horizontal twists using a conformal generalisation of \emph{metric equicontinuity}~\cite{BMR}; in particular, a total Hodge--de Rham representation always admits a canonical horizontal twist.
In the case of \(3\)-dimensional quantum \(\SU(2)\), we show that vertical and horizontal twists are unique but distinct, thereby excluding the existence of non-pathological \(\U\)-equivariant twisted spectral triples.
Nonetheless, we obtain a geometric derivation for the compact quantum metric space on quantum \(\SU(2)\) constructed by Kaad--Kyed~\cite{KK} for a canonical choice of parameters.

In this work, we shall make extensive use of the following running examples:

\begin{enumerate}[leftmargin=*]
	\item the commutative case---\ref{ex:classical0}, \ref{ex:classical1}, \ref{ex:classical2}, \ref{ex:classical3}, \ref{ex:classicaltotal1}, \ref{ex:classicaltotal2}, \ref{ex:classical4}, \ref{ex:classical6};
	\item the \emph{real multiplication instanton}---\ref{ex:heis1}, \ref{ex:heis3}, \ref{ex:heis4}, \ref{ex:heis5}, \ref{ex:heis6}, \ref{ex:heis7}, \ref{ex:heis9}, \ref{ex:heis11}, \ref{ex:heis12}, \ref{ex:heis10}, \ref{ex:heis13};
	\item the \emph{\(q\)-monopole}---\ref{ex:hopf1}, \ref{ex:hopf2}, \ref{ex:hopf3}, \ref{ex:hopf4}, \ref{ex:hopf5}, \ref{ex:hopf6}, \ref{ex:hopf7}, \ref{ex:hopf8}, \ref{ex:hopf9}, \ref{ex:hopf10}, \ref{ex:hopf11}, \ref{ex:hopf12a}, \ref{ex:hopf12}, \ref{ex:hopf13}.
\end{enumerate}

\subsection*{Acknowledgments}

The author thanks E.\ Beggs, C.\ Dunphy, V.\ Husain, A.\ Krutov, M.\ Marcolli, B.\ Mesland, R.\ \'{O} Buachalla, A.\ Rennie, K.\ Strung, N.\ Touikan, and A.\ Zampini for helpful conversations and correspondence, and he especially thanks T.\ V.\ Karthik for numerous technical conversations over the last several years that have indelibly shaped this work.
The author was supported by Natural Sciences and Engineering Research Council of Canada (NSERC) Discovery Grant RGPIN-2017-04249 and a Harrison McCain Foundation Young Scholar Award.

\addtocontents{toc}{\protect\setcounter{tocdepth}{2}}

\section{A coherent \texorpdfstring{\(2\)}{2}-group of NC Hermitian line bundles with connection}\label{sec:2}

In this section, we build on work of Beggs--Brzezi\'{n}ski~\cite{BeBrz14} and Beggs--Majid~\cite{BeMa18} to construct a coherent \(2\)-group of NC Hermitian line bundles with unitary connection over an NC differentiable manifold, the \emph{differential Picard \(2\)-group}, that makes curvature into a canonical group \(1\)-cocycle.
Moreover, we algebraically characterise the fibers of the forgetful functors passing to NC Hermitian line bundles and NC Hermitian vector bundles, respectively.

Let us recall some category-theoretic terminology.
A category is \emph{essentially small} whenever its hom-sets and its class of isomorphism classes are all sets.
A \emph{concrete category} is a category \(\grp{C}\) equipped with a faithful functor \(U : \grp{C} \to \grp{Set}\) to the category \(\grp{Set}\) of sets and functions.
Likewise, we define a \emph{functor category} to be a category \(\grp{C}\) equipped with a faithful functor \(U : \grp{C} \to [\grp{A},\grp{B}]\), where \(\grp{A}\) and \(\grp{B}\) are categories and \([\grp{A},\grp{B}]\) is the usual functor category whose objects are functors \(F : \grp{A} \to \grp{B}\) and whose arrows are natural transformations.
Finally, a subcategory \(\grp{A}\) of a category \(\grp{B}\) is \emph{strictly full} whenever it is full---every arrow in \(\grp{B}\) between objects of \(\grp{A}\) is an arrow of \(\grp{A}\)---and closed under isomorphism.

\subsection{Preliminaries on coherent \texorpdfstring{\(2\)}{2}-groups}

We first review the elementary theory of \emph{coherent \(2\)-groups}, which generalise ordinary groups by permitting the group law, unit, and inversion to satisfy the group axioms up to coherent isomorphisms.
In particular, we show that \(\bZ\) is the free coherent \(2\)-group on one generator. 
We follow the account of Baez--Lauda~\cite{BaezLauda} but with simplications drawn from Laplaza~\cite{Laplaza}.

Recall that a \emph{(weak) monoidal category} is a category \(\grp{C}\) equpped with a bifunctor \(\otimes : \grp{C} \times \grp{C} \to \grp{C}\), the \emph{monoidal product}, a distinguished object \(1\), the \emph{unit}, and natural isomorphisms \((\lambda_a : 1 \otimes a \to a)_{a \in \Obj(\grp{C})}\), the \emph{left unitor}, \((\rho_a : a \otimes 1 \to a)_{a \in \Obj(\grp{C})}\), the \emph{right unitor}, and \(\left(\alpha_{a,b,c} : (a \otimes b) \otimes c \to a \otimes (b \otimes c)\right)_{(a,b,c) \in \Obj(\grp{C})^3}\), the \emph{associator}, that satisfy certain coherence diagrams~\cite[pp.\ 428--9]{BaezLauda}; in particular, it is \emph{strict} whenever its left unitor, right unitor, and associator consist of identity arrows.
Moreover, a \emph{monoidal subcategory} of a monoidal category \(\grp{C}\) is a subcategory \(\grp{D}\) of \(\grp{C}\) that is closed under the monoidal product, contains the unit, and contains all left unitor, right unitor, and associator arrows between its objects.

\begin{example}
	Let \(B\) be a unital associative algebra over \(\bC\).
	The concrete category \(\grp{Bimod}(B)\) of \(B\)-bimodules and \(B\)-bimodule homomorphisms defines a monoidal category with respect to the usual balanced tensor product of \(B\)-bimodules and of \(B\)-bimodule homomorphisms.
	In particular, the associator of \(B\)-bimodules \(E\), \(F\), and \(G\) is \(\alpha_{E,F,G} \coloneqq \left((e \otimes f) \otimes g \mapsto e \otimes (f \otimes g)\right)\),	the unit object is the trivial \(B\)-bimodule \(B\), and the left and right unitors of a \(B\)-bimodule are induced by its left and right \(B\)-module structures, respectively.
\end{example}


\begin{definition}[{S\'{i}nh~\cite{Sinh}; Laplaza~\cite[\S 4]{Laplaza}; Baez--Lauda~\cite{BaezLauda}}]
	A \emph{coherent \(2\)-group} is an essentially small monoidal category \(\grp{G}\) in which every arrow is invertible equipped with a function
	\(
		(g \mapsto \conj{g}) : \Obj(\grp{G}) \to \Obj(\grp{G})
	\)
	called \emph{monoidal inversion} and a family of arrows
	\(
		(\ev_g : \conj{g} \otimes g \to 1)_{g \in \Obj(\grp{G})}
	\)
	in \(\grp{G}\) called \emph{evaluation}.
	Hence, a \emph{sub-\(2\)-group} of a coherent \(2\)-group \(\grp{G}\) is a monoidal subcategory \(\grp{H}\) of \(\grp{G}\) that is closed under monoidal inversion and contains \(\Set{\ev_g \given g \in \Obj(\grp{H})}\).
\end{definition}

A group \(\Gamma\) defines a coherent \(2\)-group: take the discrete category on its underlying set with the strict monoidal structure given by the group law and monoidal inversion given by inversion in the group.
This example admits the following wide-ranging generalisation; for a review of the relevant group cohomology, see~\cite[\S 2.1]{Jones}.

\begin{example}[{see~\cite[\S 2.2]{Jones}}]
	Let \(\Gamma\) be a group, let \(M\) be a \(\Gamma\)-module, and let \(\omega \in Z^3(\Gamma,M)\) be a normalised cocycle. The following defines a coherent \(2\)-group \(\twogrp(\Gamma,M,\omega)\) whose set of objects is \(\Gamma\) and whose arrows are all automorphisms.
	\begin{enumerate}[leftmargin=*]
		\item The automorphism group of an object \(\gamma \in \Gamma\) is \(M \times \Gamma\), where composition of arrows is induced by the group law of \(M\) and the identity of \(\gamma\) is \((1_M,\gamma)\).
		\item The monoidal product on objects is given by the group law of \(\Gamma\), the monoidal product on arrows is given by the group law of \(M \rtimes \Gamma\), the monoidal unit is \(1_\Gamma\), left unitors and right unitors are identity arrows, and the associator of \((\gamma_1,\gamma_2,\gamma_3) \in \Gamma^3\) is \(\alpha_{\gamma_1,\gamma_2,\gamma_3} \coloneqq (\omega(\gamma_1,\gamma_2,\gamma_3),\gamma_1\gamma_2\gamma_3)\).
		\item Monoidal inversion is given by inversion in the group \(\Gamma\), so that evaluation is induced by the group law of \(\Gamma\).
	\end{enumerate}
\end{example}

We now take a closer look at monoidal inversion.
Let \(g\) be an object of a monoidal category \(\grp{G}\).
Recall~\cite[Deff.\ 2.10.1 \& 2.11.1]{EGNO} that an \emph{inverse} for \(g\) is a triple \((h,\sce,\sci)\) consisting of an object \(h\) of \(\grp{G}\) and isomorphisms \(\sce : h \otimes g \to 1\) and \(\sci : 1 \to g \otimes h\) in \(\grp{G}\) that make the following commute for all \(f,g,h \in \Obj(\grp{G})\):

\noindent\begin{minipage}{.49\linewidth}
	\begin{equation}\scalebox{0.95}{
		\begin{tikzcd}[ampersand replacement=\&, column sep=tiny]
			{(g \otimes h)	\otimes g} \&\& {g \otimes (h \otimes g)}\\
			{1 \otimes g} \& {g} \& {g \otimes 1}
			\arrow["{\alpha_{g,h,g}}", from=1-1, to=1-3]
			\arrow["{\lambda_g}"', from=2-1, to=2-2]
			\arrow["{\rho_g}", from=2-3, to=2-2]
			\arrow["{\sci \otimes \id_g}", from=2-1, to=1-1]
			\arrow["{\id_g{} \otimes \sce}", from=1-3, to=2-3]
		\end{tikzcd}
	}\end{equation}
\end{minipage}
\begin{minipage}{.49\linewidth}
	\begin{equation}\scalebox{0.95}{
		\begin{tikzcd}[ampersand replacement=\&, column sep=tiny]
			{(h \otimes g) \otimes h} \&\& {h \otimes (g \otimes h)}\\
			{1 \otimes h} \& {h} \& {h \otimes 1}
			\arrow["{\alpha_{h,g,h}}^{-1}"', from=1-3, to=1-1]
			\arrow["{\lambda_h}"', from=2-1, to=2-2]
			\arrow["{\rho_h}", from=2-3, to=2-2]
			\arrow["{\sce \otimes \id_h}"', from=1-1, to=2-1]
			\arrow["{\id_h{} \otimes \sci}"', from=2-3, to=1-3]
		\end{tikzcd}
	}\end{equation}
\end{minipage}

\noindent Recall, moreover, that an \emph{isomorphism} of inverses \((h_1,\sce_1,\sci_1)\) and \((h_2,\sce_2,\sci_2)\) for the object \(g\) is an isomorphism \(u : h_1 \to h_2\) in \(\grp{G}\) that makes the following diagrams commute for all \(g,h \in \Obj(\grp{G})\):
	
\noindent\begin{minipage}{.49\linewidth}
	\begin{equation}\label{cd:inverseiso1}
		\begin{tikzcd}[ampersand replacement=\&, column sep=tiny, row sep=tiny]
			{h_1 \otimes g} \&\& {h_2 \otimes g}\\
			\& {g} \&
			\arrow["{u \otimes \id_g{}}", from=1-1, to=1-3]
			\arrow["{\sce_1}"', from=1-1, to=2-2]
			\arrow["{\sce_2}", from=1-3, to=2-2]
		\end{tikzcd} 
	\end{equation}
\end{minipage}
\begin{minipage}{.49\linewidth}
	\begin{equation}\label{cd:inverseiso2}
		\begin{tikzcd}[ampersand replacement=\&, column sep=tiny, row sep=tiny]
			{g \otimes h_1} \&\& {g \otimes h_2}\\
			\& {g} \&
			\arrow["{\id_g{} \otimes u}", from=1-1, to=1-3]
			\arrow["{\sci_1}", from=2-2, to=1-1]
			\arrow["{\sci_2}"', from=2-2, to=1-3]
		\end{tikzcd}
	\end{equation}
\end{minipage}

\noindent It is well known that if an object \(g\) of a monoidal category \(\grp{G}\) has an inverse, then it is unique up to unique isomorphism in the above sense~\cite[Prop.\ 2.10.5]{EGNO}.

\begin{theorem}[{Laplaza~\cite[\S 4]{Laplaza}}]\label{thm:laplaza}
	Let \(\grp{G}\) be a coherent \(2\)-group.
	\begin{enumerate}[leftmargin=*]
		\item Monoidal inversion in \(\grp{G}\) uniquely extends to a functor \(\grp{G} \to \grp{G}\) that makes evaluation in \(\grp{G}\) into a natural isomorphism.
		\item There is a unique natural isomorphism \((\coev_g : 1_{\grp{G}} \to g \otimes \conj{g})_{g \in \Obj(\grp{G})}\), such that, for every \(g \in \Obj(\grp{G})\), the triple \((\conj{g},\ev_g,\coev_g)\) is an inverse for \(g\).
		\item There exists a unique natural isomorphism \((\bb_g : g \to \cconj{g})_{g \in \Obj(\grp{G})}\), such that, for every \(g \in \Obj(\grp{G})\), the arrow \(\bb_g : g \to \cconj{g}\) gives an isomorphism of the inverses \((g,\coev_g^{-1},\ev_g^{-1})\) and \((\cconj{g},\ev_{\conj{g}},\coev_{\conj{g}})\) of \(\conj{g}\).
	\end{enumerate}
\end{theorem}

This robust functorial picture of monoidal inversion and evaluation permits a direct statement for general coherent \(2\)-groups of the following result.

\begin{corollary}[{S\'{i}nh~\cite{Sinh}, see~\cite[\S 8.3]{BaezLauda}}]
	Let \(\grp{G}\) be a coherent \(2\)-group.
	Let \(\pi_0(\grp{G})\) be the group of isomorphisms classes in \(\grp{G}\) with group law induced by the monoidal product, and let \(\pi_1(\grp{G})\) be the group of automorphisms of the monoidal unit \(1\) of \(\grp{G}\).
	Then \(\pi_1(\grp{G})\) is Abelian and defines a \(\pi_0(\grp{G})\)-module with respect to the left action \(\act_{\grp{G}}\) given by
	\begin{multline*}
		\forall g \in \Obj(\grp{G}), \, \forall \alpha \in \pi_1(\grp{G}), \\ [g] \act_{\grp{G}} \alpha \coloneqq \coev_g^{-1} \circ (\rho_g \otimes \id_{\conj{g}}) \circ \left((\id_g{} \otimes \alpha) \otimes \id_{\conj{g}}\right) \circ (\rho_g^{-1} \otimes \id_{\conj{g}}) \circ \coev_g.
	\end{multline*}
\end{corollary}

For example, a group \(\Gamma\) \emph{qua} coherent \(2\)-group satisfies \(\pi_0(\Gamma) = \Gamma\) and \(\pi_1(\Gamma) = 1\).
More generally, given a group \(\Gamma\), a \(\Gamma\)-module \(M\), and a normalised cocycle \(\omega \in Z^3(\Gamma,M)\), it follows that \(\pi_0\mleft(\twogrp(\Gamma,M,\omega)\mright) = \Gamma\) and \(\pi_1\mleft(\twogrp(\Gamma,M,\omega)\mright) = M \times \Set{1_\Gamma} \cong M\), where the \(\pi_0\mleft(\twogrp(\Gamma,M,\omega)\mright)\)-module structure on \(\pi_1\mleft(\twogrp(\Gamma,M,\omega)\mright)\) reduces to the given \(\Gamma\)-module structure on \(M\).

We now generalise group homomorphisms to  coherent \(2\)-groups.
Let \(\grp{G}\) and \(\grp{G}^\prime\) be monoidal categories.
A \emph{(weak) monoidal functor} \(F: \grp{G} \to \grp{G}^\prime\) is a functor \(F: \grp{G} \to \grp{G}^\prime\) equipped with an an isomorphism \(F^{(0)} : F(1) \to 1\) and a natural isomorphism \(\left(F^{(2)}_{g, h} : F(g \otimes h) \to F(g) \otimes F(h)\right)_{(g,h) \in \Obj(\grp{G})^2}\) 
satisfying certain coherence diagrams~\cite[pp.\ 429--430]{BaezLauda}.
Given monoidal functors \(P,Q : \grp{G} \to \grp{G}^\prime\), a natural transformation \(\phi : P \Rightarrow Q\) is \emph{monoidal} whenever \(P^{(0)} = Q^{(0)} \circ \phi_1\) and
\(
	\phi_{g \otimes h} \circ P^{(2)}_{g,h} = Q^{(2)}_{g,h} \circ (\phi_g \otimes \phi_h)
\)
for all \(g,h \in \Obj(\grp{G})\).

\begin{definition}[{see~\cite[\S 3]{BaezLauda}}]
	Let \(\grp{G}\) and \(\grp{G}^\prime\) be coherent \(2\)-groups.
	We denote by \(\grp{Hom}(\grp{G},\grp{G}^\prime)\) the essentially small functor category whose objects are monoidal functors \(F : \grp{G} \to \grp{G}^\prime\) and whose arrows are monoidal natural transformations.
	Thus, a \emph{homomorphism} from \(\grp{G}\) to \(\grp{G}^\prime\) is an object of \(\grp{Hom}(\grp{G},\grp{G}^\prime)\), while a \emph{\(2\)-isomorphism} between homomorphisms \(R,S : \grp{G} \to \grp{G}^\prime\) is an arrow \(\eta : R \Rightarrow S\) in \(\grp{Hom}(\grp{G},\grp{G}^\prime)\).
	Moreover, given a homomorphism \(F : \grp{G} \to \grp{G}^\prime\), let \(\pi_0(F) : \pi_0(\grp{G}) \to \pi_0(\grp{G}^\prime)\) and \(\pi_1(F) : \pi_1(\grp{G}) \to \pi_1(\grp{G}^\prime)\) be the respective group homomorphisms induced by \(F\).
\end{definition}

For example, let \(\Gamma_1\) and \(\Gamma_2\) be groups.
A homomorphism of coherent \(2\)-groups \(f : \Gamma_1 \to \Gamma_2\) is simply a group homomorphism  with \(f^{(0)}\) and \(f^{(2)}\) given by identity arrows, so that \(\pi_0(f) = f\) and \(\pi_1(f) = \id_1\).
Moreover, all \(2\)-homomorphisms in \(\grp{Hom}(\Gamma_1,\Gamma_2)\) are simply identity natural isomorphisms.

It turns out that a composition of homomorphisms of coherent \(2\)-groups is again a homomorphism of coherent \(2\)-groups, making the assignments \(\pi_0\) and \(\pi_1\) functorial in the sense of mapping compositions to compositions.
More generally, let \(\grp{G}_1\), \(\grp{G}_2\), and \(\grp{G}_3\) be monoidal categories, and let \(P : \grp{G}_1 \to \grp{G}_2\) and \(Q : \grp{G}_2 \to \grp{G}_3\) be monoidal functors.
Then \(Q \circ P : \grp{G}_1 \to \grp{G}_3\) defines a monoidal functor with respect to \((Q \circ P)^{(0)} \coloneqq Q^{(0)} \circ Q(P^{(0)})\) and \((Q \circ P)^{(2)} \coloneqq \left(Q^{(2)}_{P(g),P(h)} \circ Q(P^{(2)}_{g,h})\right)_{(g,h) \in \Obj(\grp{G}_1)^2}\).

We conclude by using the cohomological classification of coherent \(2\)-groups and their homomorphisms to show that \(\bZ\) is the free coherent \(2\)-group on one generator.
Recall that a \emph{monoidal equivalence} of monoidal categories \(\grp{G}_1\) and \(\grp{G}_2\) is a monoidal functor \(P : \grp{G}_1 \to \grp{G}_2\) for which there exist a monoidal functor \(Q : \grp{G}_2 \to \grp{G}_1\) and monoidal natural isomorphisms \(P \circ Q \Rightarrow \id_{\grp{G}_2}\) and \(Q \circ P \Rightarrow \id_{\grp{G}_1}\).
Coherent \(2\)-groups admit the following classification up to monoidal equivalence.

\begin{theorem}[{S\'{i}nh~\cite{Sinh}, see~\cite[\S 8.3]{BaezLauda}}]\label{thm:sinh}
	Let \(\grp{G}\) be a coherent \(2\)-group.
	There exists a normalised cocycle \(\omega \in Z^3(\pi_0(\grp{G}),\act_{\grp{G}},\pi_1(\grp{G}))\), unique up to cohomology, such that \(\grp{G}\) is monoidally equivalent to \(\twogrp\mleft(\pi_0(\grp{G}),\pi_1(\grp{G}),\omega\mright)\).
	Hence, \(\grp{G}\) is determined up to monoidal equivalence by \((\pi_0(\grp{G}),\pi_1(\grp{G}),\act_{\grp{G}},[\omega])\), where \([\omega] \in H^3(\pi_0(\grp{G}),\pi_1(\grp{G}))\) is the cohomology class of \(\omega\).
\end{theorem}

Thus, the \emph{S\'{i}nh invariant} of a coherent \(2\)-group \(\grp{G}\) is the complete monoidal equivalence invariant \((\pi_0(\grp{G}),\pi_1(\grp{G}),\act_{\grp{G}},[\omega])\) constructed by S\'{i}nh's theorem.
For example, the S\'{i}nh invariant of a group \(\Gamma\) is \((\Gamma,1,\Gamma \times 1 \to 1,1)\).
Given the additional data of a \(\Gamma\)-module \(M\) and a normalised cocycle \(\omega \in Z^3(\Gamma,M)\), the S\'{i}nh invariant of \(\twogrp(\Gamma,M,\omega)\) reduces to \((\Gamma,M,\act,[\omega])\), where \(\act\) is the given \(\Gamma\)-action on \(M\).

Homomorphisms of coherent \(2\)-groups now also admit a cohomological classification.
For simplicity, we give the relevant special case.

\begin{theorem}[{Joyal--Street~\cite[\S 6]{JS}; see~\cite[\S 8.3]{BaezLauda} and~\cite[\S 5.3]{Jacqmin}}]\label{thm:JS}
Let \(G\) and \(\Gamma\) be groups, let \(M\) be a \(\Gamma\)-module (written multiplicatively), and let \(\omega \in Z^3(\Gamma,M)\) be a normalised cocycle.
	Define a category \(\mathcal{H}(G;\Gamma,M,\omega)\) as follows.
	\begin{enumerate}[leftmargin=*]
		\item An object is a pair \((\alpha,\kappa)\), where \(\alpha : G \to \Gamma\) is a group homomorphism and \(\kappa \in B^2(G,M)\) is a normalised \(2\)-cochain with respect to \(\alpha\) with \(\du\kappa = (\alpha^\ast\omega)^{-1}\).
		\item Suppose that \((\alpha_1,\kappa_1)\) and \((\alpha_2,\kappa_2)\) are objects.
			If \(\alpha_1 = \alpha_2\), then an arrow \(\varpi : (\alpha_1,\kappa_1) \to (\alpha_2,\kappa_2)\) is a normalised \(1\)-cochain \(\varpi \in B^1(G,M)\), such that \(\du\mu = \kappa_1 \cdot \kappa_2^{-1}\); else, there are no arrows from \((\alpha_1,\kappa_1)\) to \((\alpha_2,\kappa_2)\).
		\item Composition of composable arrows is given by pointwise multiplication of normalised \(1\)-cochains.
		\item The identity of an object \((\alpha,\kappa)\) is the trivial \(1\)-cochain \(1 : \Gamma \to \pi_1(\grp{G})\).
	\end{enumerate}
	Define a functor \(\Theta : \mathcal{H}(G;\Gamma,M,\omega) \to \grp{Hom}(G,\twogrp(\Gamma,M,\omega))\) as follows.
	\begin{enumerate}[leftmargin=*]
		\item Given an object \((\alpha,\kappa)\), define \(\Theta(\alpha,\kappa) : G \to \twogrp(\Gamma,M,\omega)\) by
			\begin{gather*}
				\forall g \in G, \quad \Theta(\alpha,\kappa)(g) \coloneqq \alpha(g); \quad \Theta(\alpha,\kappa)^{(0)} \coloneqq (1,1);\\
				\forall g,h \in G, \quad \Theta(\alpha,\kappa)^{(2)}_{g,h} \coloneqq \mleft(\kappa(g,h),g h\mright).
			\end{gather*}
		\item Given an arrow \(\mu : (\alpha,\kappa_1) \to (\alpha,\kappa_2)\), define \(\Theta(\mu) : \Theta(\alpha,\kappa_1) \Rightarrow \Theta(\alpha,\kappa_2)\) by setting \(\Theta(\mu)_g \coloneqq (\mu(g),\alpha(g))\) for all \(g \in G\).
	\end{enumerate}
	Then \(\Theta\) is an equivalence of categories.
\end{theorem}

Finally, given coherent \(2\)-groups \(\grp{G}\) and \(\grp{G}^\prime\) and \(g \in \Obj(\grp{G}\), define the evaluation functor \(\epsilon_g : \grp{Hom}(\grp{G},\grp{G}^\prime) \to \grp{G}^\prime\) by
\[
	\forall P \in \Obj(\grp{Hom}(\grp{G},\grp{G}^\prime)), \,\,\, \epsilon_g(P) \coloneqq P(g); \quad \forall \eta \in \operatorname{Hom}(\grp{Hom}(\grp{G},\grp{G}^\prime)), \,\,\, \epsilon_g(\eta) \coloneqq \eta_g.
\] 
We now show that \(\bZ\) is indeed the free coherent \(2\)-group on one generator.

\begin{corollary}\label{cor:abstractpimsner}
	Let \(\grp{G}\) be a coherent \(2\)-group.
	Then \(\epsilon_1 : \grp{Hom}(\bZ,\grp{G}) \to \grp{G}\) is an equivalence of categories.
	Hence, for every object \(g\) of \(\grp{G}\), there exists an essentially unique homomorphism \(F : \bZ \to \grp{G}\) that satisfies \(F(1) \cong g\).
\end{corollary}

\begin{proof}
	By Theorem \ref{thm:sinh}, without loss of generality, there exist a group \(\Gamma\), a \(\Gamma\)-module \(M\), and a normalised cocycle \(\omega \in Z^3(\Gamma,M)\), such that \(\grp{G} = \twogrp(\Gamma,M,\omega)\).
	Hence, let \(\Theta : \cH(\bZ;\Gamma,M,\omega) \to \grp{Hom}(\bZ,\grp{G})\) be the equivalence of categories of Theorem \ref{thm:JS}.
	It suffices to show that \(\epsilon_1 \circ \Theta : \cH(\bZ;\Gamma,M,\omega) \to \grp{G}\) is an equivalence of categories.

	First, we show that \(\epsilon_1 \circ \Theta\) is essentially surjective.
	Let \(\gamma \in \Gamma = \Obj(\grp{G})\), and set \(\alpha_\gamma \coloneqq (k \mapsto \gamma^k)\).
	Since  \(\bZ\) has cohomological dimension \(1\)~\cite[Ex.\ 2.4.(b)]{Brown}, the \(3\)-cocycle 	\(\alpha_\gamma^\ast\omega\) on \(\bZ\) is trivial in cohomology, so that there exists a normalised \(2\)-cochain \(\kappa_\gamma \in B^2(\bZ,M)\) that satisfies \(\du\kappa_\gamma \cdot \alpha_\gamma^\ast\omega = 1\).
	It now follows that \((\alpha_\gamma,\kappa_\gamma)\) is a well-defined object of \(\cH(\bZ;\Gamma,M,\omega)\) satisfying  \(\epsilon_1 \circ \Theta(\alpha_\gamma,\kappa_\gamma) = \gamma\).
	
	Next, we show that \(\epsilon_1 \circ \Theta\) is full.
	Let \((m,\gamma) \in M \times \Gamma = \operatorname{Hom}(\grp{G})\), so that \((m,\gamma)\) is an automorphism of \(\gamma\).
	By the above argument, let \((\alpha_\gamma,\kappa_\gamma)\) be any preimage of \(\gamma\) under \(\epsilon_1 \circ \Theta\), and let \(\beta_{(m,\gamma)} \in Z^1(\bZ,M)\) be the unique normalised \(1\)-cocycle with respect to \(\alpha_\gamma\) satisfying \(\beta_{(m,\gamma)}(1) = m\).
	Then \(\beta_{(m,\gamma)} : (\alpha_\gamma,\kappa_\gamma) \to (\alpha_\gamma,\kappa_\gamma)\) is a well-defined arrow of \(\cH(\bZ;\Gamma,M,\omega)\) satisfying \(\epsilon_1\circ\Theta(\beta_{(m,\gamma)}) = (m,\gamma)\).
	
	Finally, we show that \(\epsilon_1 \circ \Theta\) is faithful.
	Fix a homomorphism \(\alpha : \bZ \to \Gamma\) and normalised \(2\)-cochains \(\kappa,\kappa^\prime \in B^2(\bZ,M)\), such that \(\du\kappa = \du\kappa^\prime = (\alpha^\ast \omega)^{-1}\); suppose that \(\mu_1,\mu_2 : (\alpha,\kappa) \to (\alpha,\kappa^\prime)\) satisfy \(\epsilon_1 \circ \Theta(\mu_1) = \epsilon_1 \circ \Theta(\mu_2)\).
	This means that \(\mu_1,\mu_2 \in B^1(\bZ,M)\) are normalised chains, such that \(\du\mu_1 = \kappa \cdot (\kappa^\prime)^{-1} = \du\mu_2\) and \(\mu_1(1) = \mu_2(1)\).
	It follows that \(\beta \coloneqq \mu_1 \cdot \mu_2^{-1}\) is a normalised \(1\)-cocycle on \(\bZ\) that satisfies \(\beta(1) = 1\), so that \(\beta = (m \mapsto 1)\), hence \(\mu_1 = \mu_2\).
\end{proof}

\subsection{The Picard \texorpdfstring{\(2\)}{2}-group of an NC topological space}\label{sec:2.1}

Let \(B\) be a given unital pre-\Cstar-algebra, which we view as a NC topological space; we define its \emph{positive cone} \(B_+\) to be the set of all elements of \(B\) that are positive in the \Cstar-algebra completion of \(B\).
We now review the theory of NC Hermitian line bundles over \(B\), i.e., strong Morita auto-equivalences~\cite{Rieffel74} passed through the algebraic lens of Beggs--Brzezi\'{n}ski~\cite{BeBrz14}.
This is standard material with adaptations to the setting of pre-\Cstar-algebras; following Kajiwara--Watatani~\cite{KW}, we derive substantial technical simplifications from the systematic use of finite pre-Hilbert module \emph{frames} or \emph{bases}.

Let \(E\) be a right \(B\)-bimodule. 
A \emph{\(B\)-valued inner product} on  \(E\) is a \(\bR\)-bilinear map \(\hp{}{} : E \times E \to B\) that is right \(B\)-linear in the second argument and satisfies
\[
	\forall x,y \in E, \quad \hp{y}{x} = \hp{x}{y}^\ast;
\]
hence, we define a \emph{cobasis} for \(\hp{}{}\) to be finite family \((\epsilon_i)_{i=1}^n\) in \(E\) that satisfies 
\(
	\sum_{i=1}^n \hp{\epsilon_i}{\epsilon_i} = 1
\),
and we say that \(\hp{}{}\) is \emph{strictly full} whenever \(\hp{}{}\) admits a cobasis.
Note that a right \(B\)-module is faithful whenever it admits a strictly full \(B\)-valued inner product~\cite[Lemma 1.5]{KW}.

\begin{definition}[{Rieffel~\cite[\S 6]{Rieffel74}, cf.\ Bass~\cite[\S \textsc{ii}.5]{Ba}}]
	A \emph{Hermitian line \(B\)-bimodule} is a \(B\)-bimodule \(E\) together with strictly full inner products on both \(E\) and \(\conj{E}\), respectively, such that
	\begin{align}
		\forall b \in B, \, \forall x \in E, && \norm{\hp{b  x}{b  x}} &\leq \norm{b}^2 \norm{\hp{x}{x}}, \label{eq:imprimitivity1a}\\
		\forall b \in B, \, \forall x \in E, && \norm{\hp{\conj{x  b}}{\conj{x  b}}} &\leq \norm{b}^2 \norm{\hp{\conj{x}}{\conj{x}}}, \label{eq:imprimitivity1b}\\
		\forall b \in B, \, \forall x,y \in E, && \hp{x}{b  y} &= \hp{b^\ast  x}{y}, \label{eq:imprimitivity2a}\\
		\forall b \in B, \, \forall x,y \in E, && \hp{\conj{x}}{\conj{y  b}} &= \hp{\conj{x  b^\ast}}{\conj{y}}, \label{eq:imprimitivity2b}\\
		\forall x,y,z \in E, && \hp{\conj{x}}{\conj{y}}  z &= x  \hp{y}{z}. \label{eq:imprimitivity3}
	\end{align}
\end{definition}

For example, the \emph{trivial Hermitian line \(B\)-bimodule} is the trivial \(B\)-bimodule \(B\) together with the \(B\)-valued inner products on \(B\) and \(\conj{B}\) defined, respectively, by
\begin{equation}
	\forall b,c \in B, \quad \hp{b}{c} \coloneqq b^\ast c, \quad \hp{\conj{b}}{\conj{c}} \coloneqq bc^\ast. \label{eq:trivialip}\\
\end{equation}
This example admits the following non-trivial generalisation.

\begin{example}\label{ex:triviallines}
	Let \(\phi\) be an isometric \(\ast\)-automorphism of \(B\).
	Let \(B_\phi \coloneqq \Set{b_\phi \given b \in B}\) be \(B\) as a free left \(B\)-module together with the right \(B\)-module structure defined by
	\begin{equation*}
		\forall b,c \in B, \quad b_\phi \cdot c \coloneqq (b  \phi(c))_{\phi},
	\end{equation*}
	and the \(B\)-valued inner products on \(B_\phi\) and \(\conj{B_\phi}\) respectively defined by
	\begin{equation*}
		\forall b,c \in B, \quad \hp{b_\phi}{c_\phi} \coloneqq \phi^{-1}(b^\ast c), \quad \hp{\conj{b_\phi}}{\conj{c_\phi}} \coloneqq bc^\ast.
	\end{equation*}
	Then \(B_\phi\) is a Hermitian line \(B\)-bimodule with cobases \(1_\phi\) for \(B_\phi\) and \(\conj{1_\phi}\) for \(\conj{B_\phi}\).
\end{example}

\begin{example}\label{ex:classical0}
	Let \(X\) be a closed manifold.
	Recall that the commutative unital \(\ast\)-algebra \(C^\infty(X)\) of smooth \(\bC\)-valued functions on \(X\) defines a unital pre-\Cstar-algebra with respect to the supremum norm.
	Given a Hermitian line bundle \(\cE \to X\), the balanced \(C^\infty(X)\)-bimodule \(\Gamma(\cE)\) of smooth global sections of \(\cE\) defines a Hermitian line \(C^\infty(X)\)-bimodule with respect to the \(C^\infty(X)\)-valued inner product on \(\Gamma(\cE)\) induced by the Hermitian metric on \(\cE\) and the \(C^\infty(X)\)-valued inner product on \(\conj{\Gamma(\cE)} \cong \Gamma(\conj{\cE})\) defined by \(\hp{\conj{\sigma_1}}{\conj{\sigma_2}} \coloneqq \hp{\sigma_2}{\sigma_1}\) for \(\sigma_1,\sigma_2 \in \Gamma(\cE)\).
	In particular, cobases for both of these \(C^\infty(X)\)-valued inner products can be constructed using an atlas of local trivialisations for \(\cE \to X\) together with a smooth partition of unity subordinate to the corresponding open cover of \(X\).
\end{example}

Our primary goal for this subsection is the following refinement of standard lore.

\begin{theoremdefinition}[{Rieffel~\cite[\S 6]{Rieffel74}, Brown--Green--Rieffel~\cite{BGR}; cf.\ Bass~\cite[\S \textsc{ii}.5]{Ba}}]\label{thmdef:picard}
	The \emph{Picard \(2\)-group} of \(B\) is the coherent \(2\)-group \(\Pic(B)\) defined as follows.
	\begin{enumerate}[leftmargin=*]
		\item As a category, \(\Pic(B)\) is the concrete category whose objects are Hermitian line \(B\)-bimod\-ules and whose arrows are \(B\)-bimodule isomorphisms \(u : E \to F\) with
		\begin{equation}\label{eq:unitary}
			\forall x,y \in E, \quad \hp{u(x)}{u(y)} = \hp{x}{y}.
		\end{equation}
		\item The monoidal product of objects \(E\) and \(F\) is the balanced tensor product \(E \otimes_B F\) together with the \(B\)-valued inner products on \(E \otimes_B F\) and \(\conj{E \otimes_B F}\) defined by
		\begin{align}
			\forall x_1,y_1,x_2,y_2 \in E, && \hp{x_1 \otimes y_1}{x_2 \otimes y_2} &\coloneqq \hp{y_1}{\hp{x_1}{x_2}  y_2},\label{eq:tensorip}\\
			\forall x_1,y_1,x_2,y_2 \in E, && \hp{\conj{x_1 \otimes y_1}}{\conj{x_2 \otimes y_2}} &\coloneqq \hp{\conj{x_1}}{\hp{\conj{y_1}}{\conj{y_2}}  \conj{x_2}} \label{eq:tensoripconj},
		\end{align}
			respectively; moreover, the monoidal product of arrows is given by their monoidal product in \(\grp{Bimod}(B)\).
		\item The unit object is the trivial Hermitian line \(B\)-bimodule \(B\), and left unitors, right unitors, and associators are given by the corresponding left unitors, right unitors, and associators in \(\grp{Bimod}(B)\), respectively.
		\item The monoidal inverse of a Hermitian line \(B\)-bimodule \(E\) is \(\conj{E}\) with the given \(B\)-valued inner product on \(\conj{E}\) and the \(B\)-valued inner product on \(\cconj{E}\) defined by
		\begin{equation}
			\forall x, y \in E, \quad \hp{\cconj{x}}{\cconj{y}} \coloneqq \hp{x}{y}.
		\end{equation}
		\item The evaluation morphism for an object \(E\) is \(\ev_E : \conj{E} \otimes_B E \to B\) given by
		\begin{equation}
			\forall e_1,e_2 \in E, \quad \ev_E(\conj{e_1}\otimes e_2) \coloneqq \hp{e_1}{e_2}.
		\end{equation}
	\end{enumerate}
	Hence, the \emph{Picard group} of \(B\) is the group \(\pic(B) \coloneqq \pi_0(\Pic(B))\).
\end{theoremdefinition}

\begin{example}[{Bass~\cite[Prop. 5.2]{Ba}}]\label{ex:twist1}
	The following defines a homomorphism of coherent \(2\)-groups \(\tau : \Aut(B) \to \Pic(B)\).
	\begin{enumerate}[leftmargin=*]
		\item Given \(\phi \in \Aut(B)\), let \(\tau(\phi) \coloneqq B_\phi\) be the Hermitian line \(B\)-bimodule of Example~\ref{ex:triviallines}.
		\item Set \(\tau^{(0)} \coloneqq \id_B\); given \(\phi,\psi \in \Aut(B)\), define \(\tau^{(2)}_{\phi,\psi} : \tau(\phi) \otimes_B \tau(\psi) \to \tau(\phi \psi)\) by
		\begin{equation*}
			\forall a,b \in B, \quad \tau^{(2)}_{\phi,\psi}(a_\phi \otimes b_\psi) \coloneqq \left(a\phi(b)\right)_{\phi\psi}.
		\end{equation*}
	\end{enumerate}
\end{example}

Recall~\cite[\S 1]{KW} that a \emph{basis} for a right \(B\)-module \(E\) with respect to a right \(B\)-valued inner product \(\hp{}{}\) is a finite family \((e_i)_{i=1}^n\) in \(E\), such that \(x = \sum_{i=1}^n e_i \hp{e_i}{x}\) for all \(x \in E\).
Thus, we define a \emph{right pre-Hilbert \(B\)-module of finite type} to be a right \(B\)-module \(E\) equipped with a \(B\)-valued inner product \(\ip{}{}\) that admits a basis.
In turn, we denote by \(\grp{Hilb}(B)\) the concrete category whose objects are right pre-Hilbert \(B\)-modules of finite type and whose arrows are isomorphisms of right \(B\)-modules satisfying \eqref{eq:unitary}.

\begin{example}\label{ex:projection}
	Let \(B\) be a unital pre-\Cstar-algebra, let \(n \in \bN\), and let \(\cP \in M_n(B)\) be an \emph{orthogonal projection}, i.e., \(\cP^2 = \cP = P^\ast\).
	Then \(\cP \cdot B^n\) defines a right pre-Hilbert \(B\)-module of finite type with respect to the \(B\)-linear inner product defined by
	\[
		\forall (x_i)_{i=1}^n, (y_i)_{i=1}^n \in \cP \cdot B^n, \quad \hp*{(x_i)_{i=1}^n}{(y_i)_{i=1}^n } \coloneqq \sum_{i=1}^n x_i^\ast y_i.
	\]
\end{example}

Note that if \(E\) is a right pre-Hilbert \(B\)-module of finite type with \(B\)-valued inner product \(\hp{}{}\), then \(E\) is necessarily finitely generated and projective as a right \(B\)-module and \(\hp{}{}\) is necessarily \emph{positive definite} in the sense that
\begin{gather}
	\forall x \in E, \quad \hp{x}{x} \geq 0, \label{eq:posdef1}\\
	\Set{x \in E \given \hp{x}{x} = 0} = \Set{0} \label{eq:posdef2}.
\end{gather}
Thus, every right pre-Hilbert \(B\)-module is isomorphic in \(\grp{Hilb}(B)\) to a right pre-Hilbert \(B\)-module of finite type of the kind constructed in Example \ref{ex:projection}, so that the category \(\grp{Hilb}(B)\) is essentially small.

Now, let \(E\) be a a right pre-Hilbert \(B\)-module of finite type.
By positive-definiteness of the \(B\)-valued inner product \(\hp{}{}\) on \(E\), the norm \(\norm{}{}\) defined by 
\begin{equation*}
	\forall x \in E, \quad \norm{x} \coloneqq \norm{\hp{x}{x}}^{1/2}
\end{equation*}
satisfies the following crucial inequalities:
\begin{align}
	\forall x \in E, \, \forall b \in B, && \norm{x  b} &\leq \norm{x} \cdot \norm{b}, \label{eq:rightnorm}\\
	\forall x,y \in E, && \hp{x}{y}^\ast\hp{x}{y} &\leq \norm{y}^2 \hp{x}{x}. \label{eq:cauchyschwarz}
\end{align}
Hence, one can show that the algebra \(\bL(E)\) of all right \(B\)-linear maps \(E \to E\) defines a unital pre-\Cstar-algebra with respect to the \(\ast\)-operation implicitly defined by
\begin{equation*}
	\forall T \in \bL(E), \, \forall x,y \in E, \quad \hp{x}{T^\ast y} \coloneqq \hp{Tx}{y}
\end{equation*}
and the operator norm induced by the aforementioned norm \(\norm{}\) on \(E\).
At last, given a unital pre-\Cstar-algebra \(A\), we define an \emph{\((A,B)\)-correspondence of finite type} to be a right pre-Hilbert \(B\)-module of finite type \(E\) equipped with a isometric unital \(\ast\)-homomorphism \(A \to \bL(E)\); in particular, when \(A = B\), we call \(E\) a \emph{\(B\)-self-correspondence of finite type}.

\begin{proposition}\label{prop:imprimitivity}
	Let \(E\) be a Hermitian line \(B\)-bimodule equipped with \(B\)-valued inner products \(\hp{}{}_E\) on \(E\) and \(\hp{}{}_{\conj{E}}\) on \(\conj{E}\).
	Then \(E\) and \(\conj{E}\) define \(B\)-self-corresponden\-ces of finite type with respect to \(\hp{}{}_E\) and \(\hp{}{}_{\conj{E}}\), respectively, such that
	\begin{equation}\label{eq:norms}
		\forall x \in E, \quad \norm{\hp{\conj{x}}{\conj{x}}_{\conj{E}}} = \norm{\hp{x}{x}_{E}}.
	\end{equation}	
\end{proposition}

\begin{lemma}[{Rieffel~\cite[Lemma 6.22]{Rieffel74}, Kajiwara--Watatani~\cite[Prop.\ 2.5]{KW}}]\label{lem:coev}
	Let \(B\) be a unital pre-\Cstar-algebra, and let \(E\) be a right pre-Hilbert \(B\)-module of finite type.
	There exists a unique isomorphism of \(\bL(E)\)-bimodules \(\coev_E : \bL(E) \to E \otimes_B \conj{E}\), such that \(\coev_E^{-1}(x \otimes \conj{y}) z = x  \hp{y}{z}\) for all \(x,y,z \in E\).
\end{lemma}

\begin{proof}[Proof of Prop.\ \ref{prop:imprimitivity}]
	Fix cobases \((\epsilon_i)_{i=1}^m\) and \((\conj{e_j})_{j=1}^n\) for \(\hp{}{}_{E}\) and \(\hp{}{}_{\conj{E}}\), respectively.
	Using \eqref{eq:imprimitivity3}, one shows that \((e_j)_{j=1}^n\) is a basis for \(E\) with respect to \(\hp{}{}_{E}\) and that \((\conj{\epsilon_i})_{i=1}^m\) is a basis for \(\conj{E}\) with respect to \(\hp{}{}_{\conj{E}}\), so that \(E\) is a right pre-Hilbert \(B\)-module of finite type with respect to \(\hp{}{}_{E}\), and \(\conj{E}\) is a right pre-Hilbert \(A\)-module of finite type with respect to \(\hp{}{}_{\conj{E}}\).
	
	Next, by \eqref{eq:imprimitivity1a} and \eqref{eq:imprimitivity2a}, the map \(\pi_E : B \to \bL(E)\) defines a bounded \(\ast\)-homomorphism, which is surjective by Lemma~\ref{lem:coev} together with strict fullness of \(\hp{}{}_{E}\) and injective by strict fullness of \(\hp{}{}_{\conj{E}}\).
	By symmetry, this also shows that \(\pi_{\conj{E}} : B \to \bL(\conj{E})\) defines a bounded bijective \(\ast\)-homomorphism \(\pi_{\conj{E}} : B \to \bL(\conj{E})\).
	
	Now, by positive-definiteness of \(\hp{}{}_E\) together with the assumption that \(B\) is a pre-\Cstar-algebra, the data \(\left(E,\hp{}{}_{\conj{E}},\hp{}{}_{E}\right)\) yield a pre-imprimitivity \((B,B)\)-bimodule in the usual sense~\cite[Def.\ 6.10]{Rieffel74} with left \(B\)-valued inner product induced by \(\hp{}{}_{\conj{E}}\).
	Thus, Equation \ref{eq:norms} follows from the corresponding result for pre-imprimitivity bimodules~\cite[Prop.\ 3.11]{RW}.
	We now prove boundedness of \(\pi_E^{-1}\) and \(\pi_{\conj{E}}^{-1}\) as follows.
	Let \(t \in \bL(E)\) be given.
	Using \eqref{eq:imprimitivity3}, one shows that \(\pi_E^{-1}(t) = \sum_{j=1}^n \hp{\conj{t e_j}}{\conj{e_j}}\), so that
	\[
		\norm{\pi_E^{-1}(t)} \leq \sum_{j=1}^n \norm{\hp{\conj{t e_j}}{\conj{e_j}}} \leq \sum_{j=1}^n \norm{\conj{t e_j}}\norm{\conj{e_j}} = \sum_{j=1}^n \norm{t e_j}\norm{e_j} \leq \left(\sum_{j=1}^n \norm{e_j}^2\right) \norm{t},
	\]
	by \eqref{eq:cauchyschwarz} together with \eqref{eq:norms}.
	The same argument also shows that \(\pi_{\conj{E}}^{-1}\) is bounded.
	Thus, \(\pi_E\) and \(\pi_{\conj{E}}\) are bounded bijective \(\ast\)-homomorphisms between unital pre-\Cstar-algebras with bounded inverses, and hence are isometric \(\ast\)-isomorphisms.
\end{proof}

It is easy to check that a \(B\)-self-correspondence of finite type \(E\) admits at most one \(B\)-valued inner product on \(\conj{E}\) making \(E\) into a Hermitian line \(B\)-bimodule.
Indeed, suppose that \(\hp{}{}_1\) and \(\hp{}{}_2\) are two such \(B\)-valued inner products on \(\conj{E}\).
Then
\(
	(\hp{\conj{x}}{\conj{y}}_1-\hp{\conj{x}}{\conj{y}}_2)  z = x \hp{y}{z} - x  \hp{y}{z} = 0 z
\)
for all \(x,y,z \in E\),
so that \(\hp{\conj{x}}{\conj{y}}_1 = \hp{\conj{x}}{\conj{y}}_2\) by strict fullness of either of \(\hp{}{}_1\) or \(\hp{}{}_2\).
Moreover, by Proposition \ref{prop:imprimitivity}, such a \(B\)-valued inner product on \(\conj{E}\) exists only if the left \(B\)-module structure \(B \to \bL(E)\) on \(E\) is an isometric \(\ast\)-isomorphism.
This is not only necessary but sufficient.

\begin{corollary}\label{cor:imprimitivity}
	Let \(E\) be an \(B\)-self-correspondence of finite type with strictly full \(B\)-valued inner product, and let \(\pi_E : B \to \bL(E)\) be the left \(B\)-module structure on \(E\).
	There exists a \(B\)-valued inner product on \(\conj{E}\) making \(E\) into a Hermitian line \(B\)-bimodule if and only if \(\pi_E\) is an isometric \(\ast\)-isomorphism.
\end{corollary}

\begin{proof}
	Suppose that the left \(B\)-module structure \(\pi_E\) is an isometric \(\ast\)-isomorphism.
	Note that \eqref{eq:imprimitivity1a} and \eqref{eq:imprimitivity2a} are already satisfied.
	By Lemma \ref{lem:coev} together with bijectivity of \(\pi_E\), we may define an \(B\)-valued inner product \(\hp{}{}\) on \(\conj{E}\) satsifying \eqref{eq:imprimitivity3} and \eqref{eq:imprimitivity2b} by \(\hp{\conj{x}}{\conj{y}} \coloneqq \pi_E^{-1}(x \otimes \conj{y})\) for \(x,y \in E\);
	indeed, this \(B\)-valued inner product is strictly full since any basis \((e_i)_{i=1}^n\) for \(E\) yields a cobasis \((\conj{e_i})_{i=1}^n\) for \(\conj{E}\).
	Finally, Equation \ref{eq:imprimitivity1b} follows since, for all \(x,y \in E\) and \(b \in B\), by positive definitness of \(\ip{}{}\) on \(E\), isometry of \(\pi_E\), and equations \ref{eq:rightnorm} and \ref{eq:cauchyschwarz},
	\[
		\norm{\hp{\conj{x  b}}{\conj{x  b}}  y} = \norm{x  b \hp{x  b}{y}} = \norm{x  bb^\ast \hp{x}{y}} \leq \norm{x} \norm{b}^2 \norm{\hp{x}{y}} \leq \norm{b}^2 \norm{x}^2 \norm{y}. \qedhere
	\]
\end{proof}

At last, we can prove Theorem-Definition \ref{thmdef:picard} exactly as stated.

\begin{proof}[Proof of Theorem-Definition \ref{thmdef:picard}]
	First, by replacing Hermitian line \(B\)-bimodules with \(B\)-self-correspondences, the proposed definition of \(\Pic(B)\) recovers the familiar essentially small monoidal concrete category \(\grp{Corr}(B)\) whose objects are \(B\)-self-correspondences of finite type~\cite[\S 2.2]{BMZ}; note that essential smallness of \(\grp{Corr}(B)\) follows from essential smallness of the category \(\grp{Hilb}(B)\).
	Corollary \ref{cor:imprimitivity} now implies that the category \(\Pic(B)\) is well-defined as a strictly full subcategory of \(\grp{Corr}(B)\), which clearly contains the monoidal unit \(B\).
	Moreover, Proposition \ref{prop:imprimitivity} and Corollary \ref{cor:imprimitivity} together show that monoidal inversion is well-defined as a function \(\Obj(\Pic(B)) \to \Obj(\Pic(B))\).
	
	Next, let \(E\) and \(F\) be Hermitian \(B\)-line modules, so that their tensor product \(E \otimes_B F\) in \(\grp{Corr}(B)\) is a well-defined \(B\)-self-correspondence of finite type.
	On the one hand, the \(B\)-valued inner product on \(E \otimes_B F\) is strictly full since cobases \((\epsilon_i)_{i=1}^n\) and \((\phi_j)_{j=1}^q\) for \(E\) and \(F\), respectively, yield a cobasis \((\epsilon_i \otimes \phi_j)_{1 \leq i \leq n, 1 \leq j \leq q}\) for \(E \otimes_B F\).
	On the other hand, the \(B\)-valued inner product on the tensor product \(\conj{F} \otimes_B \conj{E}\) in \(\grp{Corr}(B)\) pulls back under the canonical isomorphism of \(B\)-bimodules \((\conj{x \otimes y} \mapsto \conj{y} \otimes \conj{x}) : \conj{E \otimes_B F} \to \conj{F} \otimes_B \conj{E}\) to the \(B\)-valued inner product on \(\conj{E \otimes_B F}\) of \eqref{eq:tensoripconj}, which is strictly full since cobases \((\conj{e_i})_{i=1}^m\) and \((\conj{f_j})_{j=1}^p\) for \(\conj{E}\) and \(\conj{F}\), respectively, yield a cobasis \((\conj{e_i \otimes f_j})_{1 \leq i \leq m, 1\leq j \leq p}\) for  \(\conj{E \otimes_B F}\).
	Equation \eqref{eq:imprimitivity3} for \(E \otimes_B F\) now follows from repeated applications of \eqref{eq:imprimitivity2a}, \eqref{eq:imprimitivity2b}, and \eqref{eq:imprimitivity3}.
	
	Finally, let \(E\) be a Hermitian \(B\)-line module.
	By Lemma~\ref{lem:coev} together with Proposition~\ref{prop:imprimitivity}, the map \(\ev_E\) is an isomorphism of \(B\)-bimodules; that \(\ev_E\) satisfies \eqref{eq:unitary} now follows from observing that for all \(x_1,x_2 \in E\) and \(y_1,y_2\in E\),
	\[
		\hp{\conj{x_1} \otimes y_1}{\conj{x_2} \otimes y_2} = \hp{y_1}{\hp{\conj{x_1}}{\conj{x_2}}  y_2} = \hp{y_1}{x_1  \hp{x_2}{y_2}} = \hp{x_1}{y_1}^\ast \hp{x_2}{y_2}. \qedhere
	\]
\end{proof}

This characterization of the monoidal category \(\Pic(B)\) as a monoidal subcategory of the monoidal category \(\Corr(B)\) of \(B\)-self-correspondences of finite type yields, with superficial changes, a right action of the Picard group \(\pic(B)\) on the \emph{\(K_0\)-monoid} \(\cV(B)\) of isomorphism classes of right pre-Hilbert \(B\)-modules of finite type.
Indeed, given a right pre-Hilbert \(B\)-module of finite type \(E\) and a Hermitian line \(B\)-bimodule \(F\), set
\(
	[E] \ract [F] \coloneqq [E \otimes_B F],
\)
where the balanced tensor product \(E \otimes_B F\) is equipped with the right \(B\)-valued inner product given by \eqref{eq:tensorip}.
We may use this \(\pic(B)\)-action to characterise the fibres of the obvious forgetful map \(\pic(B) \to \cV(B)\); in turn, this helps us understand the information lost when passing from \(\pic(B)\) to the \(K\)-theory of \(B\) or its \Cstar-algebraic completion.

\begin{proposition}[{Bass~\cite[Propp.\ 5.2 \& 5.3]{Ba}}]\label{prop:topstabilizer}
	Let \(\Pi_{\cV(B)} : \pic(B) \to \cV(B)\) denote the set function induced by the forgetful functor \(\Pic(B) \to \grp{Hilb}(B)\).
	Let \(\pic(B)_{[B]}\) denote the stabiliser subgroup of \(\pic(B)\) with respect to \([B] \in \ran \Pi_{\cV(B)}\).
	Then the homomorphism of coherent \(2\)-groups \(\tau : \Aut(B) \to \Pic(B)\) of Example \ref{ex:twist1} yields the exact sequence of groups
	\[
		1 \to \Unit(\Zent(B)) \to \Unit(B) \xrightarrow{u \mapsto \Ad_u} \Aut(B) \xrightarrow{\pi_0(\tau)} \pic(B)_{[B]} \to 1.
	\]
\end{proposition}

Note that this canonically identifies the outer automorphism group of \(B\) with a subgroup of \(\pic(B)\).
What is more surprising is that the \emph{entire} Picard group \(\pic(B)\) acts as isometric \(\ast\)-automorphisms on the centre of \(B\).

\begin{propositiondefinition}[{Fr\"{o}hlich~\cite[Thm.\ 2]{Fr73}, Beggs--Brzezi\'{n}ski~\cite[\S 10]{BeBrz14}}]\label{propdef:topfroh}
	The \emph{Fr\"{o}hlich homomorphism} of the unital pre-\Cstar-algebra \(B\) is the unique group homomorphism \(\Phi : \pic(B) \to \Aut(\Zent(B))\), such that, for every Hermitian line \(B\)-bimodule \(E\), the \emph{Fr\"{o}hlich automorphism} \(\Phi_{[E]}\) of \([E]\) satisfies
	\begin{equation}
		\forall b \in B, \, \forall x \in E, \quad \Phi_{[E]}(b)  x = x  b.
	\end{equation}
	Hence, the canonical left action of \(\pi_0(\Pic(B)) \eqqcolon \pic(B)\) on \(\pi_1(\Pic(B)) = \Unit(\Zent(B))\) is the left action induced by \(\Phi\).
\end{propositiondefinition}

\begin{proof}
	Relative to the references, it remains to show each Fr\"{o}hlich automorphism is isometric.
	Let \(E\) be a Hermitian line \(B\)-bimodule, and let \((e_i)_{i=1}^n\) be a cobasis for \(E\).
	Then,
	\(
		\norm{\Phi_{[E]}^{-1}(b)} = \norm*{\sum_{i=1}^n \hp{e_i}{b  e_i}} \leq \sum_{i=1}^n \norm{e_i} \norm{b  e} \leq \left(\sum_{i=1}^n \norm{e_i}^2\right) \norm{b}
	\) for every \(b \in B\), so that \(\Phi_{[E]}^{-1}\) is bounded. 
	Using \(\conj{E}\) instead now shows that \(\phi_{[E]} = \phi_{[\conj{E}]}^{-1}\) is also bounded.
\end{proof}

\begin{example}\label{ex:classical1}
	We continue from Example \ref{ex:classical0}.
	Let \(\pic(X)\) be the Picard group of isomorphism classes of complex line bundles over \(X\), which admits a right \(\Diff(X)\)-action by pullbacks.
	On the one hand, since any two Hermitian metrics on a line bundle are unitarily equivalent, the map
	\(
		\left([\cE] \mapsto [\Gamma(\cE)]\right) : \pic(X) \to \pic(C^\infty(X))
	\)
	is well-defined.
	On the other hand, \(\left(f \mapsto (f^{-1})^\ast\right) : \Diff(X) \to \Aut(C^\infty(X))\) is an isomorphism~\cite{Mrcun}, so let \(\Psi : \pic(C^\infty(X)) \to \Diff(X)\) be the resulting homomorphism induced by the Fr\"{o}hlich homomorphism of \(C^\infty(X)\).
	Thus, Serre--Swan duality yields a split exact sequence
	\(
		1 \to \pic(X) \xrightarrow{[\cE] \mapsto [\Gamma(\cE)]} \pic(C^\infty(X)) \xrightarrow{\Psi} \Diff(X) \to 1
	\)
	with right splitting \(\phi \mapsto \pi_0(\tau)((\phi^{-1})^\ast)\).
	Moreover, given the resulting isomorphism
	\[
		\left((\phi,[\cE]) \mapsto [\Gamma((\phi^{-1})^\ast\cE)] \cdot \pi_0(\tau)((\phi^{-1})^\ast)\right) : \Diff(X) \ltimes \pic(X) \to \pic(C^\infty(X)),
	\]
	we may identify the Fr\"{o}hlich homomorphism of \(C^\infty(X)\) with the quotient map
	\[
		\left((\phi,[\cE]) \mapsto \phi\right) : \Diff(X) \ltimes \pic(X) \to \Diff(X).
	\]
\end{example}

We conclude by noting certain simplications that arise when \(B\) behaves sufficiently like a \Cstar-algebra.
This will permit us to introduce our first main running example.
In what follows, recall that an element \(a\) of unital pre-\Cstar-algebre \(A\) is \emph{positive} whenever it is positive in the \Cstar-algebra completion of \(A\).

Let \(n \in \bN\), and let \(M_n(B)\) denote the unital \(\ast\)-algebra of \(n \times n\) matrices with entries in \(B\), which is defined by analogy with \(M_n(\bC)\); one calls \(M_n(B)\) a \emph{matrix algebra} over \(B\).
Recall that \(B^n\) defines a right pre-Hilbert \(B\)-module of finite type by Example \ref{ex:projection}; hence observe that matrix multiplication on left defines an injective \(\ast\)-homomorphism \(M_n(\bC) \to \bL(B^n)\).
Thus, the operator norm on \(\bL(B^n)\) pulls back to a \Cstar-norm on \(M_n(B)\).

\begin{definition}
	We say that \(B\) \emph{admits polar decompositions} if, for every \(n \in \bN\) and positive \(b \in M_n(B)\), there exists unique positive \(\sqrt{b} \in M_n(B)\) that satisfies \((\sqrt{b})^2 = b\) and is invertible in \(M_n(B)\) whenever \(b\) is.
	In this case, given \(n \in \bN\), the \emph{polar decomposition} of invertible \(b \in M_n(B)\) is \(b = \operatorname{sgn}(b) \lvert b \rvert\), where \(\lvert b \rvert \coloneqq \sqrt{b^\ast b} \in M_n(B)\) is positive and invertible and \(\operatorname{sgn}(b) \coloneqq b \lvert b \rvert^{-1} \in M_n(B)\) is unitary.
\end{definition}

For example, a unital \Cstar-algebra admits polar decompositions by the holomorphic functional calculus.
More generally, a unital pre-\Cstar-algebra \(B\) admits polar decompositions whenever it and all its matrix algebras are closed under the holomorphic functional calculus in their respective \Cstar-closures.

Finally, recall that a \(B\)-valued inner product on a right \(B\)-module \(E\) is \emph{algebraically full} whenever it satisfies \(
		\operatorname{Span}_{\bC}\Set{\hp{x}{y} \given x,y \in E} = B
	\).
	
\begin{proposition}\label{prop:localcstar}
	Suppose that \(B\) is unital pre-\Cstar-algebra that admits polar decompositions.
	Let \(E\) be a \(B\)-bimodule, let \(\hp{}{}_E\) be a \(B\)-valued inner product on \(E\), and let \(\hp{}{}_{\conj{E}}\) be a \(B\)-valued inner product on \(\conj{E}\).  
	Then \(E\) defines a Hermitian line \(B\)-bimodule with respect to \(\hp{}{}_E\) and \(\hp{}{}_{\conj{E}}\) if and only if the following conditions are all satisfied:
	\begin{enumerate}[leftmargin=*]
		\item \(\hp{}{}_E\) is algebraically full and satisfies \eqref{eq:posdef1}, \eqref{eq:imprimitivity1a} and \eqref{eq:imprimitivity2a};
		\item \(\hp{}{}_{\conj{E}}\) is algebraically full and satisfies \eqref{eq:posdef1}, \eqref{eq:imprimitivity1b} and \eqref{eq:imprimitivity2b};
		\item the \(B\)-valued inner products \(\hp{}{}_E\) and \(\hp{}{}_{\conj{E}}\) respectively satisfy \eqref{eq:imprimitivity3}.
	\end{enumerate}
\end{proposition}

\begin{proof}
	The forward implication is trivial, so we prove the backward implication.
	Suppose that all three conditions are satisfied; it remains to show that both \(\hp{}{}_E\) and \(\hp{}{}_{\conj{E}}\) are strictly full.
	
	Since \(\hp{}{}_E\) is algebraically full, choose finite families \((x_i)_{i=1}^n\) and \((y_i)_{i=1}^n\) in \(E\) that satisfy \(\sum_{i=1}^n \hp{x_i}{y_i}_E = 1\).
	Let \(X \coloneqq \left(\hp{\conj{x_i}}{\conj{x_j}}_{\conj{E}}\right)_{i,j=1}^n \in M_n(B)\), so that
	\begin{multline*}
		1 = \sum_{i,j=1}^n \hp{y_i}{x_i}_E\hp{x_j}{y_j}_E = \sum_{i,j=1}^n \hp{y_i}{x_i  \hp{x_j}{y_j}_E}_E = \sum_{i,j=1}^n \hp{y_i}{\hp{\conj{x_i}}{\conj{x_j}}_{\conj{E}} y_j}_E\\ = \sum_{i,j=1}^n \hp{y_i}{X_{ij}  y_j}_E.
	\end{multline*}
	by \eqref{eq:imprimitivity3}.
	Applying \cite[Cor.\ 2.7]{Rieffel74} to \(X\) as a bounded operator on \(B^n\) with the \(B\)-valued inner product of Example \ref{ex:projection} shows that \(X \geq 0\).
	By our hypothesis on \(B\), there exists \(a = (a_{ij})_{i,j=1}^n \in M_n(B)\), such that \(a^\ast a = X\); hence, \(\left(\sum_{k=1}^n a_{ik}y_k\right)_{i=1}^n\) is a cobasis for \(\hp{}{}_E\).
	An identical argument shows that \(\hp{}{}_{\conj{E}}\) is strictly full.
\end{proof}

We now introduce our first main running example.
Let \(\theta \in \bR\), so that the corresponding (continuous) \emph{NC \(2\)-torus} is the universal \Cstar-algebra \(C_\theta(\bT^2)\) generated by unitaries \(u\) and \(v\) satisfying
\(
	vu = \mathrm{e}^{2\pi\iu{}\theta}uv.
\)
The corresponding \emph{smooth NC \(2\)-torus} \(C^\infty_\theta(\bT^2)\) is the dense unital \(\ast\)-subalgebra of \(C_\theta(\bT^2)\) consisting of Laurent series in \(u\) and \(v\) with rapidly decaying coefficients, which admits polar decompositions since it and all its matrix algebras are closed under the holomorphic functional calculus~\cite{Connes80}.

\begin{example}\label{ex:heis1}
	Let \(\theta \in \bR\) be a quadratic irrationality, so that the subgroup
	\[
		\Gamma_\theta \coloneqq \Set*{g \in \operatorname{SL}(2,\bZ) \given \tfrac{g_{11}\theta+g_{12}}{g_{21}\theta+g_{22}} = \theta, \, g_{21}\theta+g_{22} > 0}
	\]
	of \(\operatorname{SL}(2,\bZ)\) is non-trivial and hence infinite cyclic~\cite[Thm 5.2.10]{HK}.
	Connes's Heisenberg modules~\cite{Connes80} over the unital pre-\Cstar-algebra \(C^\infty_\theta(\bT^2)\) yield, in particular, a homomorphism \(E : \Gamma_\theta \to \Pic(C^\infty_\theta(\bT^2))\) as follows.
	\begin{enumerate}[leftmargin=*]
		\item Given \(g \in \Gamma_\theta\), let \(E(g)\) be the basic Heisenberg module of rank \(g_{21}\theta+g_{22}\) and degree \(g_{21}\)~\cite[\S 1.1]{PolishchukSchwarz}, which defines a Hermitian line \(C^\infty_\theta(\bT^2)\)-bimodule by a result of Rieffel~\cite[Thm 2.15]{Rieffel88} together with Proposition \ref{prop:localcstar}.
		\item Since \(E(1) = C^\infty_\theta(\bT^2)\), set \(E^{(0)} \coloneqq \id_{C^\infty_\theta(\bT^2)}\).
		\item Given \(g,h \in \Gamma_\theta\), let \(E^{(2)}_{g,h} : E(g) \otimes_{A^\infty_\theta} E(h) \to E(gh)\) be the isomorphism of \(C^\infty_\theta(\bT^2)\)-bimodules constructed by Schwarz~\cite[\S 3]{Schwarz} and Dieng--Schwarz~\cite{DiengSchwarz}, which is an isomorphism of Hermitian line \(C^\infty_\theta(\bT^2)\)-bimodules by a result of Vlasenko~\cite[Thm 6.1]{Vlasenko}.
	\end{enumerate}
	In particular, that the functor \(E\) is monodal with respect to \(E^{(0)}\) and \(E^{(2)}\) reduces to a result of Polishchuk--Schwarz~\cite[Prop.\ 1.2]{PolishchukSchwarz}.
\end{example}

\subsection{The differential Picard \texorpdfstring{\(2\)}{2}-group of an NC manifold}

At last, we build on results of Beggs--Majid~\cite{BeMa18} to construct a coherent \(2\)-group of NC Hermitian line bundles with connection over an NC manifold, the \emph{differential Picard \(2\)-group}.

Let us recall some preliminary definitions.
A \emph{graded algebra} is a unital \(\bC\)-algebra \(\Omega\) with vector space decomposition \(\Omega = \bigoplus_{k=0}^\infty \Omega^k\), such that \(1 \in \Omega^0\) and \(\Omega^j \cdot \Omega^{j+k} \subseteq \Omega^{j+k}\) for all \(j,k\in\bN_0\).
Hence, a \emph{graded \(\ast\)-algebra} is a graded algebra \(\Omega\) with a unit- and grading-preserving \(\bC\)-antilinear involution \(\ast : \Omega \to \Omega\), such that
\[
	\forall j,k \in \bN_0, \, \forall \alpha \in \Omega^j, \, \forall \beta \in \Omega^k, \quad (\alpha \beta)^\ast = (-1)^{jk}\beta \alpha.
\]
At last, given a unital pre-\Cstar-algebra \(B\), a \emph{graded \(\ast\)-algebra over \(B\)} is a graded \(\ast\)-algebra \(\Omega\) together with a unital \(\ast\)-isomorphism \(B \to \Omega^0\), which we suppress; in this case, we denote by \(\Aut(\Omega)\) the group of all grading- and \(\ast\)-preserving automorphisms \(\phi\) of \(\Omega\) as a unital \(\bC\)-algebra that restrict to an isometric \(\ast\)-automorphism of \(B\).

Now, suppose that \(B\) is a unital pre-\Cstar-algebra.
We define a \emph{\(\ast\)-quasi-differential graded algebra} or \emph{\(\ast\)-quasi-\textsc{dga}} over \(B\) to be a pair \((\Omega,\du)\), where \(\Omega\) is a graded \(\ast\)-algebra over \(B\) and \(\du : \Omega \to \Omega\) is a \(\ast\)-preserving complex linear map that satisfies \(\du(\Omega^k) \subset \Omega^{k+1}\) for all \(k \in \bN_0\) together with the graded Leibniz rule
\[
	\forall k \in \bN_0, \, \forall \alpha \in \Omega^k, \, \forall \beta \in \Omega, \quad \du_B(\alpha  \beta) = \du_B(\alpha)  \beta + (-1)^k \alpha  \du_B(\beta);
\]
hence, its \emph{graded centre} is the graded \(\ast\)-subalgebra \(\Zent(\Omega)\) of \(\Omega\) defined by
\[
	\forall m \in \bN_0, \quad \Zent(\Omega)^m \coloneqq \Set{\omega \in \Omega^m \given \forall n \in \bN_0, \, \forall \xi \in \Omega^n_B, \, \omega \xi = (-1)^{mn} \xi\omega},
\]
which is closed under \(\du\), and its \emph{dimension} (when it exists) is the largest \(N \in \bN\) such that \(\Omega^N \neq 0\) and \(\Omega^k = 0\) for all \(k > N\).
At last, we call \((\Omega,\du)\) a \emph{\(\ast\)-exterior algebra} over \(B\) whenever \(\Omega\) is generated by \(B\) and \(\du(B)\) and \(\du^2 = 0\).

Finally, we define a concrete category \(\grp{QDGA}\) as follows:
\begin{enumerate}[leftmargin=*]
	\item an object \((B;\Omega,\du)\) consists of a unital pre-\Cstar-algebra \(B\) and \(\ast\)-quasi-\textsc{dga} \((\Omega,\du)\) over \(B\);
	\item an arrow \(f : (B_1,\Omega_1,\du_1) \to (B_2,\Omega_2,\du_2)\) is a grading- and \(\ast\)-preserving homomorphism of unital \(\bC\)-algebras \(f : \Omega_1 \to \Omega_2\) that restricts to a bounded (hence contractive) \(\ast\)-homomorphism \(\rest{f}{B_1} : B_1 \to B_2\) and satisfies \(f \circ \du_1 = \du_2 \circ f\).
\end{enumerate}
In particular, given a \(\ast\)-quasi-\textsc{dga} \((\Omega,\du)\) over a unital pre-\Cstar-algebra \(B\), we denote by \(\Aut(\Omega,\du)\) the automorphism group of \((B;\Omega,\du)\) in this category.

From now on, let \(B\) be a unital pre-\Cstar-algebra with \(\ast\)-exterior calculus \((\Omega_B,\du_B)\), which we view as an NC manifold.
Given a \(B\)-bimodule \(E\), we shall apply the following Sweedler-type notation to elements of \(E \otimes_B \Omega^1_B\) and \(\Omega^1_B \otimes_B E\), respectively:
\[
	\forall \eta \in E \otimes_B \Omega^1_B, \quad \leg{\eta}{0} \otimes \leg{\eta}{1} \coloneqq \eta; \quad \quad \forall \xi \in \Omega^1_B \otimes_B E, \quad \leg{\xi}{-1} \otimes \leg{\xi}{0} \coloneqq \xi.
\]

We now recall the generalization of unitary connection appropriate to our setting.
Let \(E\) be a right pre-Hilbert \(B\)-module of finite type.
Extend the \(B\)-valued inner product \(\hp{}{}\) on \(E\) to a real-bilinear map \(\hp{}{} : E \otimes_B \Omega_B \times E \otimes_B \Omega_B \to \Omega_B\) by
\begin{equation*}
	\forall x,y \in E, \, \forall \alpha,\beta \in \Omega_B, \quad \hp{x \otimes \alpha}{y \otimes \beta} \coloneqq \alpha^\ast \hp{x}{y} \beta.
\end{equation*}
This extension satisfies
\begin{align}
	\forall \xi,\upsilon \in E \otimes_B \Omega_B, \, \forall \beta \in \Omega_B, && \hp{\xi}{\upsilon  \beta} &= \hp{\xi}{\upsilon}  \beta,\label{eq:formip1}\\
	\forall \xi,\upsilon \in E \otimes_B \Omega_B, && \hp{\upsilon}{\xi} &= (-1)^{\deg{\xi}\deg{\upsilon}} \hp{\xi}{\upsilon}^\ast \label{eq:formip2}.
\end{align}
Following Connes~\cite[Def.\ \textsc{ii}.18]{Connes85}, one now defines \emph{right Hermitian connection} on \(E\) to be a complex-linear map \(\nabla : E \to E \otimes_B \Omega^1_B\), such that
\begin{align}
	\forall x \in E, \, \forall b \in B, && \nabla(x  b) &= \nabla x  b + x \otimes \du_B b,\label{eq:rightleibniz}\\
	\forall x,y \in E, && \du_B \hp{x}{y} &= \hp{\nabla x}{y \otimes 1} + \hp{x \otimes 1}{\nabla y}.\label{eq:hermconn}
\end{align}
One can now show that there exists unique complex-linear \(\nabla : E \otimes_B \Omega_B \to E \otimes_B \Omega_B\) extending the right Hermitian connection \(\nabla\), such that
\begin{align*}
	\forall \eta \in E \otimes_B, \, \forall \beta \in \Omega_B, && \nabla(\eta  \beta) &= \nabla(\eta) \beta + (-1)^{\deg{\eta}} \eta  \du \beta,\\
	\forall \xi,\upsilon \in E \otimes_B \Omega_B, && \du_B \hp{\xi}{\upsilon} &= \hp{\nabla \xi}{\eta} + (-1)^{\deg{\xi}}\hp{\xi}{\nabla \eta}.
\end{align*}

\begin{definition}[{Beggs--Majid~\cite[Def. 5.1 \& \S 5.2]{BeMa18}}]
	Let \(E\) be a \(B\)-self-cor\-re\-spon\-dence of finite type.
	A \emph{generalised braiding} for \(E\) is an isomorphism of graded \(B\)-bimodules \(\sigma : \Omega_B \otimes_B E \to E \otimes_B \Omega_B\) that extends \(\rho_E^{-1} \circ \lambda_E : B \otimes_B E \to E \otimes_B B\) and satisfies
	\begin{equation}\label{eq:braidassoc}
		\forall \alpha,\beta \in \Omega_B, \, \forall x \in E, \quad \sigma(\alpha \otimes \leg{\sigma(\beta \otimes x)}{0})  \leg{\sigma(\beta \otimes x)}{1} = \sigma(\alpha\beta \otimes x).
	\end{equation}
	Hence, a \emph{Hermitian bimodule connection} on \(E\) is a pair \((\sigma,\nabla)\), where \(\sigma\) is a Hermitian generalised braiding on \(E\) and \(\nabla\) is a Hermitian right connection on \(E\), such that
	\begin{equation}
		\forall \beta \in \Omega_B, \, \forall \xi \in E \otimes_B \Omega_B, \quad \nabla(\beta  \xi) = \du_B(\beta)  \xi + (-1)^{\deg{\beta}}\beta  \nabla\xi,\label{eq:leftleibniz}
	\end{equation}
	where \(E \otimes_B \Omega_B\) carries the graded \(\Omega_B\)-bimodule structure given by
	\begin{equation}
	\forall \alpha,\beta \in \Omega_B, \, \forall \xi \in E \otimes_B \Omega_B, \quad \alpha  \xi \beta \coloneqq \sigma(\alpha \otimes \leg{\xi}{0})  \leg{\xi}{1}\beta.
	\end{equation}
\end{definition}

\begin{example}\label{ex:trivialconn}
	The pair \((\sigma_B, \nabla_B) \coloneqq (\lambda_{\Omega_B}^{-1} \circ \rho_{\Omega_B},\lambda_{\Omega_B}^{-1} \circ \du)\) defines a Hermitian bimodule connection on the trivial Hermitian line \(B\)-bimodule \(B\); where convenient, we shall abuse notation and identify \((\sigma_B,\nabla_B)\) with \((\id_{\Omega_B},\du)\).
\end{example}

If \(E\) is a \(B\)-self-correspondence of finite type with right Hermitian connection \(\nabla_E\), then there exists at most one Hermitian generalised braiding \(\sigma_E\) on \(E\) that makes \((\sigma_E,\nabla_E)\) into a Hermitian bimodule connection.
Moreover, in this case,
\begin{align}
	\forall k \in \bN_0, \, \forall \xi \in E \otimes_B \Omega^k_B,\, \forall \upsilon \in E\otimes_B \Omega_B, && \du_B\hp{\xi}{\upsilon} &= \hp{\nabla_E\xi}{\upsilon} + (-1)^{k}\hp{\xi}{\nabla_E\upsilon},\label{eq:hermconndiff}\\
	\forall \beta \in \Omega_B, \, \forall \xi,\upsilon \in E \otimes_B \Omega_B, && \hp{\alpha  \xi}{\upsilon} &= \hp{\xi}{\alpha^\ast  \upsilon}; \label{eq:braidherm}
\end{align}
by \eqref{eq:braidassoc}, it suffices to check \eqref{eq:braidherm} when \(\beta \in \du(B)\) and \(\xi,\upsilon \in E \otimes 1\).

We shall use the following characterisation of Hermitian bimodule connections.

\begin{proposition}[{Beggs--Majid~\cite[Lemma 5.2]{BeMa18}}]\label{prop:bemacriterion}
	Let \(E\) be a \(B\)-self-correspondence of finite type, let \(\sigma\) be a generalised braiding on \(E\), and let \(\nabla\) be a Hermitian right connection on \(E\). 
	Then \((\sigma,\nabla)\) defines a Hermitian bimodule connection on \(E\) if and only if the following both hold:
	\begin{align}
		\forall b \in B, \, \forall x \in E, && \nabla(b  \xi) &= \sigma(\du_B b \otimes x) + b  \nabla x,\label{eq:bemacriterion1}\\
		\forall b \in B, \, \forall x \in E, && \nabla^2(b  \xi) &= b  \nabla^2 \xi.\label{eq:bemacriterion2}
	\end{align}
\end{proposition}

We now introduce our first nontrivial family of examples of Hermitian bimodule connections on Hermitian line bimodules.
Let \(\theta \in \bR\).
Recall that the smooth  \(2\)-torus \(C^\infty_\theta(\bT^2)\) admits a canonical \(\ast\)-exterior calculus \((\Omega_\theta(\bT^2),\du)\) due to Connes~\cite{Connes80}.
First, let \(\delta_1\) and \(\delta_2\) be the unique \(\ast\)-derivations on \(C^\infty_\theta(\bT^2)\), such that, respectively
\[
	\forall (m,n) \in \bZ^2, \quad \delta_1(U^m V^n) = 2\pi\iu{}m\,U^m V^n, \quad \delta_2(U^m V^n) = 2\pi\iu{}n\,U^m V^n;
\]
Then let \(\Omega_\theta(\bT^2)\) be the graded \(\ast\)-algebra over \(C^\infty_\theta(\bT^2)\) generated by central self-adjoint elements \(e^1,e^2 \in \Omega^1_\theta(\bT^2)\) satisfying \((e^1)^2 = (e^2)^2 = e^1  e^2 + e^2  e^1 = 0\), and let \(\du\) be the unique \(\ast\)-derivation of degree \(1\) on \(\Omega_\theta(\bT^2)\), such that
\[
	\forall a \in C^\infty_\theta(\bT^2), \,\, \du a \coloneqq \delta_1(a)  e^1 + \delta_2(a)  e^2; \quad \du e^1 \coloneqq 0; \quad \du e^2 \coloneqq 0.
\]
In the case where \(\theta\) is a quadratic irrationality, the basic Heisenberg modules of Example \ref{ex:heis1} admit canonical Hermitian bimodule connections due to Connes.

\begin{example}[{Connes~\cite[Thm 7]{Connes80}, Polishchuk--Schwarz~\cite[Prop.\ 2.1]{PolishchukSchwarz}}]\label{ex:heis3}
	We continue from Example \ref{ex:heis1}.
	Let \(g \in \Gamma_\theta\).
	Connes's maps \(\nabla_{g,1},\nabla_{g,2} : E(g) \to E(g)\) yield a right Hermitian connection \(\nabla_g : E(g) \to E(g) \otimes_B \Omega^1_B\) by setting
	\[
		\forall p \in E(g), \quad \nabla_g(p) \coloneqq \nabla_{g,1}(p) \otimes e^1 + \nabla_{g,2}(p) \otimes e^2;
	\]
	in particular, by~\cite[Thm 7]{Connes80}, the map \(\nabla_g\) satisfies
	\[
		\forall p \in E(g), \quad \nabla_g^2(p) = p \cdot 2\pi\iu{}\frac{g_{21}}{g_{21}\theta+g_{22}}e^1 e^2.
	\]
	Hence, by~\cite[Prop.\ 2.1]{PolishchukSchwarz} and Proposition \ref{prop:bemacriterion}, the map \(\nabla_g\) defines a Hermitian bimodule connection on \(E(g)\) with respect to the generalised braiding \(\sigma_g\) given by
	\[
		\forall i \in \Set{1,2}, \, \forall p \in E(g), \quad \sigma_g(e^i \otimes p) \coloneqq \frac{1}{g_{21}\theta+g_{22}}p \otimes e^i.
	\]
\end{example}

The primary technical advantage of bimodule connections is that they are compatible with balanced tensor products of bimodules.
In fact, they yield a monoidal category of \(B\)-self-correspondence of finite type with Hermitian bimodule connection.

\begin{theorem}[{Beggs--Majid~\cite[Thm 5.3]{BeMa18}, cf.\ Beggs--Brzezi\'{n}ski~\cite[\S 2.4]{BB}}]\label{thm:bema}
	The following defines an essentially small monoidal concrete category \(\grp{DCorr}(B)\).
	\begin{enumerate}[leftmargin=*]
		\item An object of \(\grp{DCorr}(B)\) is \((E,\sigma_E,\nabla_E)\), where \(E\) is a \(B\)-self-correspondence of finite type and \((\sigma_E,\nabla_E)\) is a Hermitian bimodule connection on \(E\);
		\item An arrow \(u : (E,\sigma_E,\nabla_E) \to (F,\sigma_F,\nabla_F)\) is an isomorphism \(f : E \to F\) of \(B\)-bimodules satisfying \eqref{eq:unitary} and
		\begin{equation}\label{eq:intertwineconn}
			\nabla_F \circ u = (u \otimes \id_{\Omega^1_B}) \circ \nabla_E.
		\end{equation}
		\item The tensor product of \((E,\sigma_E,\nabla_E)\) and \((F,\sigma_F,\nabla_F)\) is \((E \otimes_B F,\sigma_{E \otimes_B F},\nabla_{E\otimes_B F})\), where \(E \otimes_B F\) is the balanced tensor product of \(B\)-bimodules equipped with the \(B\)-valued inner product of \eqref{eq:tensorip}, and where
		\begin{align}
			\sigma_{E \otimes_B F} &\coloneqq \alpha^{-1}_{E,F,\Omega_B} \circ (\id_E{} \otimes \sigma_F) \circ \alpha_{E,\Omega_B,F} \circ \sigma_E \otimes \id_F,\\
			\nabla_{E \otimes_B F} &\coloneqq \alpha_{E,F,\Omega_B}^{-1} \circ \left((\id_{E}{} \otimes \sigma_F) \circ \alpha_{E,\Omega_B,F} \circ (\nabla_E \otimes \id) + \id_E{} \otimes \nabla_F\right);
		\end{align}
			moreover, the monoidal product of arrows is given by the balanced tensor product of \(B\)-bimodule homomorphisms, and associators are given by the corresponding associators in \(\grp{Bimod}(B)\).
		\item The unit object of \(\grp{DCorr}(B)\) is \((B,\sigma_B,\nabla_B)\), where \((\sigma_B,\nabla_B)\) is the Hermitian bimodule connection of Example \ref{ex:trivialconn}; moreover, left and right unitors are given by the corresponding left and right unitors of \(\grp{Bimod}(B)\), respectively.
	\end{enumerate}
	In addition, if \(u : (E,\sigma_E,\nabla_E) \to (F,\sigma_F,\nabla_F)\) is an arrow in \(\dPic(B)\), then
	\begin{align}
		\nabla_F \circ (u \otimes \id_{\Omega_B}) &= (u \otimes \id_{\Omega_B}) \circ \nabla_E,\label{eq:intertwine}\\
		\sigma_F \circ (\id_{\Omega_B}{} \otimes u) &= (u \otimes \id_{\Omega_B}) \circ \sigma_E.
	\end{align}
\end{theorem}

\begin{proof}
	Relative to~\cite[Thm 5.7]{BeMa18} and the discussion preceding the proof of Theorem-Definition \ref{thmdef:picard} (with minor changes), it suffices to check that the tensor product is well-defined on objects.
	Let \((E,\sigma_E,\nabla_E)\) and \((F,\sigma_F,\nabla_F)\) be objects of the category \(\grp{DCorr}(B)\).
	A straightforward calculation shows that
	\begin{equation}
		\forall x,v \in E, \, \forall \upsilon,\tau \in F \otimes_B \Omega_B, \quad \hp*{(x \otimes \leg{\upsilon}{0}) \otimes \leg{\upsilon}{1}}{(v \otimes \leg{\tau}{0}) \otimes \leg{\tau}{1}} = \hp*{\upsilon}{\hp{x}{v}\tau}.\label{eq:tensoriplem}
	\end{equation}
	Relative to~\cite[Thm 5.7]{BeMa18}, it remains to check that \(\sigma_{E \otimes_B F}\) and \(\nabla_{E \otimes_B F}\) satisfy \eqref{eq:braidherm} and \eqref{eq:hermconn}, respectively, but this now follows by repeated application of \eqref{eq:tensoriplem}, \eqref{eq:formip1}, and \eqref{eq:formip2}, as appropriate, to \((\sigma_E,\nabla_E)\) and \((\sigma_F,\nabla_F)\).
\end{proof}

We now construct our coherent \(2\)-group of NC Hermitian line bundles with unitary connection.

\begin{theoremdefinition}[{cf.\ Beggs--Majid~\cite[Thm 3.3]{BeMa11}}]\label{thmdef:dpic}
	The \emph{differential Picard \(2\)-group} of \(B\) is the coherent \(2\)-group \(\dPic(B)\) defined as follows.
	\begin{enumerate}[leftmargin=*]
		\item As a monoidal category, \(\dPic(B)\) is the full monoidal subcategory of \(\grp{DCorr}(B)\), whose objects are of the form \((E,\sigma_E,\nabla_E)\), where \(E\) is a Hermitian line \(B\)-bimodule.
		\item The monoidal inverse of an object \((E,\sigma_E,\nabla_E)\) is given by \((\conj{E},\sigma_{\conj{E}},\nabla_{\conj{E}})\), where
			\begin{align}
				\forall \beta \in \Omega_B, \, \forall x \in E, && \sigma_{\conj{E}}(\beta \otimes \conj{x}) &\coloneqq \Upsilon_{\Omega_B,E}\mleft(\conj{\sigma_E^{-1}(x \otimes \beta^\ast)}\mright),\label{eq:dpic1}\\
				\forall x \in E, && \nabla_{\conj{E}} \conj{x} &\coloneqq \Upsilon_{\Omega_B,E}\mleft(\conj{\sigma_E^{-1}(\nabla_E x)}\mright);\label{eq:dpic2}
			\end{align}
			here, by abuse of notation, we let \(\Upsilon_{\Omega_B,E} : \conj{\Omega_B \otimes_B E} \to \conj{E} \otimes_B \Omega_B\) denote the isomorphism of \(B\)-bimodules defined by
			\begin{equation}
				\forall x \in E, \, \forall \beta \in \Omega_B, \quad \Upsilon_{\Omega_B,E}(\conj{\beta \otimes x}) \coloneqq \conj{x} \otimes \beta^\ast.
			\end{equation}
		\item Evaluation arrows are given by the corresponding evaluation arrows in \(\Pic(B)\).
	\end{enumerate}
	Hence, a \emph{Hermitian line \(B\)-bimodule with connection} is an object of \(\dPic(B)\), and an \emph{isomorphism} of Hermitian line \(B\)-bimodules with connection is an arrow of \(\dPic(B)\).
	Finally, the \emph{differential Picard group} of \(B\) is the group \(\dpic(B) \coloneqq \pi_0(\dPic(B))\).
\end{theoremdefinition}

\begin{proof}
	Given Theorem \ref{thm:bema} and Theorem-Definition~\ref{thmdef:picard}, it remains to show that monoidal inversion and evaluation in \(\dPic(B)\) are well-defined.
	Let \(E\) be a Hermitian line \(B\)-bimodule with Hermitian bimodule connection \((\sigma_E,\nabla_E)\).
	
	Let us first show that \(\conj{E}\) admits the Hermitian bimodule connection \((\sigma_{\conj{E}},\nabla_{\conj{E}})\) defined by \eqref{eq:dpic1} and \eqref{eq:dpic2}.
	By a theorem of Beggs--Majid~\cite[Thm 3.3]{BeMa11}, suitably adapted, we know that \(\sigma_{\conj{E}}\) is a \(B\)-bimodule isomorphism, that \(\nabla_{\conj{E}}\) satisfies \eqref{eq:rightleibniz}, and that the pair \((\sigma_{\conj{E}},\nabla_{\conj{E}})\) satisfies \eqref{eq:bemacriterion1}.
	By Proposition \ref{prop:bemacriterion}, it remains to show that \(\sigma_{\conj{E}}\) satisfies \eqref{eq:braidassoc} and that \(\nabla_{\conj{E}}\) satisfies \eqref{eq:hermconn} and \eqref{eq:bemacriterion2}.
	In turn, by construction of \(\sigma_{\conj{E}}\) and \(\nabla_{\conj{E}}\), it therefore suffices to show that, respectively, for all \(\alpha,\beta \in \Omega_B\) and \(x,y,z \in E\),
	\begin{align}
		\sigma_E^{-1}(x \otimes \beta\alpha) &= \leg{\sigma_E^{-1}(x \otimes \beta)}{-1}  \sigma_E^{-1}\mleft(\leg{\sigma_E(x \otimes \beta)}{0} \otimes \alpha\mright), \label{eq:braidassocconj}\\
		\sigma_E\mleft(\du_B\hp{\conj{x}}{\conj{y}} \otimes z\mright) &= \sigma_E\mleft(\hp{\nabla_{\conj{E}}\conj{x}}{\conj{y} \otimes 1} \otimes z\mright) + \sigma_E\mleft(\hp{\conj{x}\otimes 1}{\nabla_{\conj{E}}\conj{y}} \otimes z\mright),\label{eq:hermconnconj}\\
		\sigma_E^{-1}\mleft(\nabla^2_E x\mright) &\begin{multlined}[t]= \leg{\sigma_E^{-1}(\nabla_E x)}{-1}  \sigma_E^{-1}\mleft(\nabla_E(\leg{\sigma_E^{-1}(\nabla_E x)}{0})\mright)\\ + \du_B \leg{\sigma_E^{-1}(\nabla_E x)}{-1} \otimes \leg{\sigma_E^{-1}(\nabla_E x)}{0}.\end{multlined}\label{eq:bemacriterion2conj}
	\end{align}
	
	First, we check \eqref{eq:braidassocconj}. 
	Let \(\alpha,\beta \in \Omega_B\) and \(x \in E\).
	By \eqref{eq:braidassoc} applied to \(\sigma_E\),
	\begin{align*}
		e \otimes \beta\alpha
		&= \sigma_E\mleft(\leg{\sigma_E^{-1}(e \otimes \beta)}{-1} \otimes \leg{\sigma_E^{-1}(e \otimes \beta)}{0}\mright)  \alpha\\
		&= \sigma_E\mleft(\leg{\sigma_E^{-1}(e \otimes \beta)}{-1} \otimes \leg{(\leg{\sigma_E^{-1}(e \otimes \beta)}{0} \otimes \alpha)}{0} \mright)  \leg{(\leg{\sigma_E^{-1}(e \otimes \beta)}{0} \otimes \alpha)}{1}\\
		&=\sigma_E\mleft(\leg{\sigma_E^{-1}(e \otimes \beta)}{-1}  \sigma_E^{-1}\mleft(\leg{\sigma_E^{-1}(e \otimes \beta)}{0} \otimes \alpha\mright)\mright).
	\end{align*}
	
	Next, we check \eqref{eq:hermconnconj}.
	Let \(x,y,z \in E\) be given.
	Since \((\sigma_E,\nabla_E)\) is a Hermitian bimodule connection, it follows that
 	\[
 		\sigma_E\mleft(\du_B\hp{\conj{x}}{\conj{y}} \otimes z\mright) 
 		= \nabla_E\mleft(\hp{\conj{x}}{\conj{y}}  z\mright) - \hp{\conj{x}}{\conj{y}}  \nabla_E z
 		= \nabla_E (x) \hp{y}{z} + x \otimes \hp{\nabla_E y}{z},
 	\]
	On the one hand, by definition of \(\nabla_{\conj{E}}\), we see that
	\[
		\sigma_E\mleft(\hp{\nabla_{\conj{E}}\conj{x}}{\conj{y}} \otimes z\mright) = \sigma_E\mleft(\leg{\sigma_E^{-1}(\nabla_E x)}{-1} \otimes \leg{\sigma_E^{-1}(\nabla_E x)}{0}  \hp{y}{z}\mright) = \nabla_E (x) \hp{y}{z}.
	\]
	On the other hand, by definition of \(\nabla_{\conj{E}}\) together with \eqref{eq:braidherm} applied to \(\sigma_E\),
	\begin{align*}
		x \otimes \hp{\nabla_E y}{z}
		&=x \otimes \hp*{\leg{\sigma_{E}^{-1}(\nabla_E y)}{-1}  (\leg{\sigma_{E}^{-1}(\nabla_E y)}{0} \otimes 1)}{z \otimes 1}\\
		&=x \otimes \hp*{\leg{\sigma_{E}^{-1}(\nabla_E y)}{0} \otimes 1}{\leg{\sigma_{E}^{-1}(\nabla_E y)}{-1}^\ast \otimes z)}\\
		&= \hp*{\conj{x}}{\conj{\leg{\sigma_{E}^{-1}(\nabla_E y)}{0}}}  \sigma_E(\leg{\sigma_{E}^{-1}(\nabla_E y)}{-1}^\ast \otimes z)\\
		&=\sigma_E\mleft(\hp{\conj{x} \otimes 1}{\nabla_{\conj{E}}(\conj{y}))} \otimes z\mright).
	\end{align*}
	
	Finally, we check \eqref{eq:bemacriterion2conj}.
	Let \(x \in E\) be given.
	By \eqref{eq:braidassoc} and \eqref{eq:leftleibniz} applied to \(\sigma_E\) and \((\sigma_E,\nabla_E)\), respectively, 
	\begin{align*}
	&\nabla^2_E x
	=\nabla_E\circ\sigma_E\mleft(\leg{\sigma_E^{-1}(\nabla_E x)}{-1} \otimes \leg{\sigma_E^{-1}(\nabla_E x)}{0}\mright)\\
	&\quad=\sigma_E\mleft(\leg{\sigma_E^{-1}(\nabla_E x)}{-1} \otimes \nabla_E(\leg{\sigma_E^{-1}(\nabla_E x)}{0}) + \du_B \leg{\sigma_E^{-1}(\nabla_E x)}{-1} \otimes \leg{\sigma_E^{-1}(\nabla_E x)}{0}\mright)\\
	&\quad=\sigma_E\mleft(\leg{\sigma_E^{-1}(\nabla_E x)}{-1}  \sigma_E^{-1}(\nabla_E(\leg{\sigma_E^{-1}(\nabla_E x)}{0}))\mright)\\
	&\quad\quad\quad + \sigma_E\mleft( \du_B \leg{\sigma_E^{-1}(\nabla_E x)}{-1} \otimes \leg{\sigma_E^{-1}(\nabla_E x)}{0}\mright).
	\end{align*}
	
	We now show that \(\ev_E : \conj{E} \otimes_B E \to E\) in \(\Pic(B)\) defines a corresponding arrow in \(\dPic(B)\).
	Since \(\nabla_B = \lambda^{-1}_{\Omega_B} \circ \du\), it suffices to show that for all \(x,y \in E\),
	\begin{multline*}
		\du_B \ev_E(\conj{x} \otimes y)\\ = \lambda_{\Omega_B} \circ (\ev_{E}{} \otimes \id_{\Omega_B}) \circ \alpha_{\conj{E},E,\Omega_B}^{-1}\mleft((\id_{\conj{E}}{} \otimes \sigma_E) \circ \alpha_{\conj{E},\Omega_B,E}(\nabla_{\conj{E}}\conj{x} \otimes y) + \conj{x} \otimes \nabla_E y\mright).
	\end{multline*}
	Hence, let \(x,y \in E\) be given.
	By \eqref{eq:braidherm} applied to \(\sigma_E\) and \eqref{eq:hermconn} applied to \(\nabla_E\), 
	\begin{align*}
		&\du_B\ev_E(\conj{x} \otimes y) 
		=\hp{\nabla_E x}{y \otimes 1} + \hp{x \otimes 1}{\nabla_E y}\\
		&\quad=\hp*{\leg{\sigma_E^{-1}(\nabla_E x)}{0} \otimes 1}{\leg{\sigma_E^{-1}(\nabla_E x)}{-1}^\ast  (y \otimes 1)} + \hp{x \otimes 1}{\nabla_E y}\\
		&\quad= \ev_E\mleft(\leg{\nabla_{\conj{E}}(\conj{x})}{0} \otimes \leg{\sigma_E(\leg{\nabla_{\conj{E}}(\conj{x})}{1} \otimes y)}{0}\mright)  \leg{\sigma_E(\leg{\nabla_{\conj{E}}(\conj{x})}{1} \otimes y)}{-1}\\ 
		&\quad\quad\quad+ \ev_E\mleft(\conj{x} \otimes \leg{\nabla_E(y)}{0}\mright)  \leg{\nabla_E(y)}{1}\\
		&\quad=\lambda_{\Omega_B} \circ (\ev_{E}{} \otimes \id_{\Omega_B}) \circ \alpha_{\conj{E},E,\Omega_B}^{-1}\mleft((\id_{\conj{E}}{} \otimes \sigma_E) \circ \alpha_{\conj{E},\Omega_B,E}(\nabla_{\conj{E}}\conj{x} \otimes y) + \conj{x} \otimes \nabla_E y\mright).
	\end{align*}
\end{proof}

\begin{example}[{Connes~\cite[Thm 7]{Connes80}, Polishchuk--Schwarz~\cite[Prop.\ 2.2]{PolishchukSchwarz}}]\label{ex:heis4}
	We continue from Example \ref{ex:heis3}.
	The homomorphism \(E : \Gamma_\theta \to \Pic(C^\infty_\theta(\bT^2))\) of Example \ref{ex:heis1} lifts to \(\hat{E} : \Gamma_\theta \to \dPic(C^\infty_\theta(\bT^2))\) defined as follows.
	\begin{enumerate}[leftmargin=*]
		\item Given \(g \in \Gamma_\theta\), let \(\hat{E}(g) \coloneqq (E(g),\sigma_g,\nabla_g)\), where \((\sigma_g,\nabla_g)\) is the Hermitian bimodule connection on \(E(g)\) of Example \ref{ex:heis3}.
		\item Let \(\hat{E}^{(0)}\) be given by \(\id_{C^\infty_\theta(\bT^2)} \eqqcolon E^{(0)}\).
		\item Given \(g,h \in \Gamma_\theta\), let \(\hat{E}^{(2)}_{g,h}\) be given by \(E^{(2)}_{g,h}\).
	\end{enumerate}
	In particular, that \(\hat{E}^{(0)}\) and \(\hat{E}^{(2)}\) satisfy the required commutative diagrams follows (with superficial changes) from the result of Polishchuk--Schwarz.
\end{example}

\subsection{Canonical actions of the differential Picard group}\label{sec:2.4}

Again, let \(B\) be a unital pre-\Cstar-algebra with \(\ast\)-exterior algebra \((\Omega_B,\du_B)\).
We show that the differential Picard group \(\dpic(B)\) defines a generalised diffeomorphism group for \((B;\Omega_B,\du_B)\), whose action on the \(K_0\)-monoid \(\cV(B)\) characterizes the fibres of the forgetful map \(\dpic(B) \to \cV(B)\) and whose action on the graded centre \(\Zent(\Omega_B)\) admits curvature as a canonical group \(1\)-cocycle.

Let us first consider na\"{i}ve diffeomorphisms of the  manifold \((B;\Omega_B,\du_B)\).
Gel'fand duality initially suggests that a diffeomorphism of \((B;\Omega_B,\du_B)\) should be an automorphism of the \(\ast\)-exterior algebra \((\Omega_B,\du_B)\) over \(B\).
However, as we shall see in Theorem \ref{thm:stabilizer}, inner automorphisms of \(B\), i.e., automorphisms of the form \(b \mapsto u b u^\ast\) for fixed \(u \in \Unit(B)\), will generally only satisfy the following conservative generalisation.

\begin{definition}
	We define the \emph{extended diffeomorphism group} of \(B\) to be the subgroup \(\tDiff(B)\) of \((\Omega^1_B)_{\mathrm{sa}} \rtimes \Aut(\Omega_B)\) consisting of elements \((\omega,\phi)\) satisfying
	\begin{equation}
		\forall \beta \in \Omega_B, \quad \du\beta - \phi \circ \du{} \circ \phi^{-1}(\beta) = \iu{}[\omega,\beta];\label{eq:extdiff}
	\end{equation}
	here, \((\Omega^1_B)_{\mathrm{sa}} \coloneqq \Set{\omega \in \Omega^1_B \given \omega^\ast = \omega}\), while \([\cdot,\cdot]\) denotes the supercommutator in \(\Omega_B\) with respect to parity of degree; hence, an \emph{extended diffeomorphism} of \(B\) with respect to \((\Omega_B,\du)\) is an element of \(\tDiff(B)\).
	Moreover, we say that \((\omega,\phi) \in \tDiff(B)\) is \emph{topologically trivial} whenever \(\rest{\phi}{B} = \id_B\), and we denote by \(\tDiff_0(B)\) the normal subgroup of \(\tDiff(B)\) consisting of topologically trivial elements.
\end{definition}

\begin{example}\label{ex:classical2}
	Continuing from Example \ref{ex:classical1}, equip \(C^\infty(X)\) with the \(\ast\)-exterior algebra \((\Omega(X,\bC),\du)\), and note that \(\Diff(X)\) acts on \(\Omega^1(X,\bR)\) from the right by pullback.
	The map
	\[
		\left( (f,\omega) \mapsto ((f^{-1})^\ast\omega,(f^{-1})^\ast) \right) : \Diff(X) \ltimes \Omega^1(X,\bR) \to \tDiff(C^\infty(X))
	\]
	is an isomorphism that identifies \(\Set{\id_X} \times \Omega^1(X,\bR) \cong \Omega^1(X,\bR)\) with \(\tDiff_0(C^\infty(X))\).
\end{example}

\begin{example}\label{ex:twist2}
	Recall the homomorphism \(\tau : \Aut(B) \to \Pic(B)\) of Example \ref{ex:twist1}.
	The following defines a lift of \(\tau\) to a homomorphism \(\hat{\tau} : \tDiff(B) \to \dPic(B)\).
	\begin{enumerate}[leftmargin=*]
		\item Given \((\omega,\phi) \in \tDiff(B)\), let \(\hat{\tau}_{(\phi,\omega)} \coloneqq (B_\phi,\sigma_\phi,\nabla_{(\omega,\phi)})\), where \(B_\phi \eqqcolon \tau_\phi\) and
		\begin{align}
			\forall \beta \in \Omega_B, \, \forall b \in B, && \sigma_\phi(\beta \otimes b_\phi) &\coloneqq 1_\phi \otimes \phi^{-1}(\beta  b),\\
			\forall b \in B, && \nabla_{(\phi,\omega)}(b_\phi) &\coloneqq 1_\phi \otimes \phi^{-1}(\du b + b  \cdot \iu{}\omega).
		\end{align}
		\item Let \(\hat{\tau}^{(0)}\) be given by \(\id_B \eqqcolon \tau^{(0)}\).
		\item Given \((\omega_1,\phi_1),(\omega_2,\phi_2) \in \tDiff(B)\), let \(\hat{\tau}^{(2)}_{(\omega_1,\phi_1),(\omega_2,\phi_2)}\) be given by \(\tau^{(2)}_{\phi_1,\phi_2}\).
	\end{enumerate}
\end{example}

Note that the homomorphism \(\hat{\tau} : \tDiff(B) \to \dPic(B)\) is faithful on objects, so that it can be viewed as embedding the group \(\tDiff(B)\) in \(\dPic(B)\).

Now, let \(\Pi_{\pic(B)} : \dpic(B) \to \pic(B)\) be the group homomorphism induced by the forgetful functor \(\dPic(B) \to \Pic(B)\), and recall that \(\Pi_{\cV(B)} : \pic(B) \to \cV(B)\) is the set map induced by the forgetful functor \(\Pic(B) \to \grp{Hilb}(B)\).
Hence, note that the right \(\pic(B)\)-action on \(\cV(B)\) of Proposition \ref{prop:topstabilizer} pulls back via \(\Pi_{\pic(B)}\) to a right \(\dpic(B)\)-action on \(\cV(B)\); in turn, this right \(\dpic(B)\)-action correctly pulls back via \(\pi_0(\hat{\tau})\) to the usual pullback action of isometric \(\ast\)-automorphisms on \(\cV(B)\).

Since this \(\dpic(B)\)-action on \(\cV(B)\) is transitive on the range of \(\Pi_{\cV(B)} \circ \Pi_{\pic(B)}\), we may use the resulting stabilizer group \(\dpic(B)_{[B]}\) of \([B] \in \ran(\Pi_{\cV(B)} \circ \Pi_{\pic(B)})\) to characterize the fibres of the forgetful map from the differential Picard group \(\dpic(B)\) to the \(K_0\)-monoid \(\cV(B)\).
Moreover, since \(\Pi_{\pic(B)}\) is a group homomorphism, its kernel yields the fibres of the forgetful map from \(\dpic(B)\) to the (topological) Picard group \(\pic(B)\).
As it turns out, the subgroups \(\dpic(B)_{[B]}\) and \(\ker\Pi_{\pic(B)}\) are completely determined by the group homomorphism \(\pi_0(\hat{\tau}) : \tDiff(B) \to \dpic(B)\).

\begin{theorem}\label{thm:stabilizer}
	Let \(\dpic(B)_{[B]}\) denote the stabilizer subgroup of \(\dpic(B)\) with respect to \([B] \in \ran(\Pi_{\cV(B)} \circ \Pi_{\pic(B)})\),
	and let \(\dAd : \Unit(B) \to \tDiff(B)\) be given by
	\begin{equation}
		\forall u \in \Unit(B), \quad \dAd_u \coloneqq \left(-\iu{}\,\du_B(u) u^\ast,\Ad_u\right).
	\end{equation}
	Then the homomorphisms \(\pi_0(\hat{\tau})\) and \(\dAd\) fit into the exact sequences of groups
	\begin{gather}\label{eq:frohses}
		1 \to \Unit(\Zent(B) \cap \ker\du_B) \to \Unit(B) \xrightarrow{\dAd} \tDiff(B) \xrightarrow{\pi_0(\hat{\tau})} \dpic(B)_{[B]} \to 1,\\
		1 \to \Unit(\Zent(B) \cap \ker\du_B) \to \Unit(\Zent(\Omega_B)^0) \xrightarrow{\dAd} \tDiff_0(B) \xrightarrow{\pi_0(\hat{\tau})} \dpic(B) \xrightarrow{\Pi_{\pic(B)}} \pic(B).
	\end{gather}
\end{theorem}

\begin{proof}
	Before continuing, we must show that \eqref{eq:frohses} is a well-defined diagram of groups.
	A straightforward calculation show that \(\dAd : \Unit(B) \to \tDiff(B)\) is well-defined, so it remains to check that \(\ran \pi_0(\hat{\tau}) \leq \dpic(B)_{[B]}\).
	However, given \((\omega,\phi) \in \tDiff(B)\), the required isomorphism \(U : B \otimes_B B_\phi \to B\) in \(\grp{Hilb}(B)\) is given by \(U \coloneqq \left(b \otimes c_\phi \mapsto \phi^{-1}(bc)\right)\).
	
	We now show that \eqref{eq:frohses} is an exact sequence.
	Exactness at \(\Unit(B)\) is immediate, so we proceed to checking exactness at \(\tDiff(B)\).
	On the one hand, let \((\omega,\phi) \in \ker\pi_0(\hat{\tau})\) be given. 
	Thus, there exists an arrow \(U : (B_\phi,\sigma_\phi,\nabla_{\omega,\phi}) \to (B,\sigma_B,\nabla_B)\) in \(\dPic(B)\).
	Set \(u \coloneqq U(1_\phi)\); we claim that \((\omega,\phi) = \dAd_u\).
	First, observe that \(u \in \Unit(B)\) since the singleton \(\Set{1_\phi}\) defines both a basis and strict cobasis for \(B_\phi\).
	Next, observe that 
	\begin{multline*}
		\beta  u = \lambda_{\Omega_B} \circ \sigma_0 \circ (\id_{\Omega_B}{} \otimes U)(\beta \otimes 1_\phi) = \lambda_{\Omega_B} \circ \sigma_0 \circ (\id_{\Omega_B}{} \otimes U) \circ \sigma_\phi^{-1}(1_\phi \otimes \phi^{-1}(\beta))\\ = u  \phi^{-1}(\beta).
	\end{multline*}
	for all \(\beta \in \Omega_B\), so that \(\phi = \Ad_u\).
	Finally, observe that \(\omega = -\iu{}\,\du(u) u^\ast\) since
	\begin{align*}
		0 &= \lambda_{\Omega_B}\mleft((U \otimes \id_{\Omega_B})(\nabla_{\omega,\phi}1_\phi) - \du_B u U(1_\phi)\mright)\\ 
		&= \lambda_{\Omega_B} \circ (U \otimes \id_{\Omega_B})\mleft(1_\phi \otimes \iu{}\phi^{-1}(\omega)\mright) - \du_B u \\
		&= \iu{}(\omega + \iu{}\du_B(u) u^\ast)  u.
	\end{align*}
	On the other hand, similar calculations show that, for each \(u \in \Unit(B)\), the map \((b_{\Ad_u}  \mapsto b  u) : B_{\Ad_u} \to B\) defines an arrow \(\hat{\tau}_{\dAd_u} \to (B,\sigma_0,\nabla_0)\) in \(\dPic(B)\).
	
	Finally, we check exactness at \(\dpic(B)_{[B]}\).
	Let \((E,\sigma_E,\nabla_E)\) be a Hermitian line \(B\)-bimodule with connection, and suppose that \([(E,\sigma_E,\nabla_E)] \in \dpic(B)_{[B]}\).
	Using \(\lambda_B : B \otimes_B B \to B\), we conclude that there exists an arrow \(U : B \to E\) in \(\grp{Hilb}(B)\).
	Set \(e_0 \coloneqq U(1)\); since \(\Set{1}\) defines both a basis and cobasis for \(B\), it follows that \(\Set{e_0}\) defines both a basis and cobasis for \(E\).
	We shall use \(e_0\) together with \((\sigma_E,\nabla_E)\) to construct \((\omega,\phi) \in \tDiff(B)\) and an arrow \(V : \hat{\tau}_{(\phi,\omega)} \to (E,\sigma_E,\nabla_E)\) in \(\dPic(B)\).
	
	First, define a \(\bC\)-linear map \(\Phi : \Omega_B \to \Omega_B\) by \(\Phi \coloneqq \left(\beta \mapsto \hp{e_0 \otimes 1}{\sigma_E(\beta \otimes e_0)}\right)\);	once we know that the degree-preserving map \(\Phi\) is an element of \(\Aut(\Omega_B)\), we shall set \(\phi \coloneqq \Phi^{-1}\).
	First, note that \(\Phi\) is unit-preserving since \(\Phi(1) = \hp{e_0}{e_0} = 1\).
	Next, note that \(\Phi\) is multiplicative by \eqref{eq:braidassoc} applied to \(\sigma_E\) and \(\ast\)-preserving by \eqref{eq:braidherm} applied to \(\sigma_E\).
	Next, note that \(\Phi\) is bijective since, for all \(\beta \in \Omega_B\) and \(x \in E\),
	\[
		\beta \otimes x = \beta \otimes e_0 \hp{e_0}{x} = \sigma_E^{-1}(e_0 \otimes \hp{e_0 \otimes 1}{\sigma_E^{-1}(\beta \otimes e_0)})  \hp{e_0}{x} = \sigma_E^{-1}(e_0 \otimes \Phi(\beta))  \hp{e_0}{e},
	\]
	so that, in terms of the arrow \(\Upsilon_{\Omega_B,E} : \conj{\Omega_B \otimes_B E} \to \conj{E} \otimes_B \Omega_B\) in \(\grp{Bimod}(B)\),
	\[
		\forall \beta \in \Omega_B, \quad \Phi^{-1}(\beta) = \hp*{\Upsilon_{\Omega_B,E}(\conj{\sigma^{-1}_E(e_0 \otimes \beta)})}{\conj{e_0} \otimes 1}.
	\]
	Finally, note that \(\Phi\) is isometric on \(B\) since, for all \(b \in B\),
	\begin{align*}
		\norm{\Phi(b)} &= \norm{\hp{e_0}{b  e_0}} \leq \norm{b} \cdot \norm{e_0}^2 = \norm{b} \cdot \norm{\hp{e_0}{e_0}} = \norm{b},\\
		\norm{\Phi^{-1}(b)} &= \norm{\hp{\conj{e_0  b}}{\conj{e_0}}} \leq \norm{b} \cdot \norm{\conj{e_0}}^2 = \norm{b} \cdot \norm{\hp{e_0}{e_0}} = \norm{b}.
	\end{align*}

	We now claim that \((\omega,\phi) \coloneqq (-\iu{}\Phi^{-1}(\hp{e_0 \otimes 1}{\nabla_E e_0}),\Phi^{-1})\) defines an element of the group \(\tDiff(B)\).
	Note that \(\omega \in \Omega^1_B\) is self-adjoint since \(\phi \in \Aut(\Omega_B)\) and since
	\[
		0 = \du_B\hp{e_0}{e_0} = \hp{\nabla_E e_0}{e_0 \otimes 1} + \hp{e_0 \otimes 1}{\nabla_E e_0} = \hp{e_0}{\nabla_E e_0}^\ast + \hp{e_0}{\nabla_E e_0}.
	\]
	Thus, it remains to show that \((\phi,\omega)\) satisfies \eqref{eq:extdiff}.
	Let \(n \in \bN_0\) and \(\beta \in \Omega^n_B\).
	Then
	\begin{align*}
		\sigma_E(\du_B\phi(\beta) \otimes e_0)
		&= \nabla_E\mleft(\sigma_E(\phi(\beta) \otimes e_0)\mright)-(-1)^{n}\sigma_E\mleft(\phi(\beta) \otimes \leg{\nabla_E(e_0)}{0}\mright)  \leg{\nabla_E(e_0)}{1}\\
		&= \nabla_E(e_0 \otimes \beta) - (-1)^{n}\sigma_E\mleft(\phi(\beta) \otimes e_0\mright)  \hp{e_0}{\nabla_E e_0}\\
		&= \nabla_E(e_0)  \beta + e_0 \otimes \du\beta - (-1)^{n} e_0 \otimes \omega  \hp{e_0}{\nabla_E e_0}\\
		&= e_0 \otimes \left(\du_B\beta + [\hp{e_0}{\nabla_E e_0},\beta]\right)\\
		&= \sigma_E\mleft(\left(\phi(\du_B\beta)+\iu{}[\omega,\phi(\beta)]\right) \otimes e_0\mright),
	\end{align*}
	so that, indeed, for every \(x \in E\),
	\begin{multline*}
		\du_B\phi(\beta) \otimes x = \du_B\phi(\beta) \otimes e_0  \hp{e_0}{x} = \left(\phi(\du_B\beta)+\iu{}[\omega,\phi(\beta)]\right) \otimes e_0  \hp{e_0}{x}\\ = \left(\phi(\du_B\beta)+\iu{}[\omega,\phi(\beta)]\right) \otimes x.
	\end{multline*}

	Finally, define an arrow \(V : B_\phi \to E\) in \(\Pic(B)\) by \(V(1_\phi) \coloneqq e_0\); we claim that it yields an arrow \(V : \hat{\tau}_{\omega,\phi} \to (E,\sigma_E,\nabla_E)\) in \(\dPic(B)\).
	Indeed, for all \(b \in B\),
	\begin{align*}
		\nabla_E \mleft(V(b_\phi)\mright)
		&= \sigma_E(\du b \otimes e_0) + b  \nabla_E e_0\\
		&= e_0 \otimes \left(\hp{e_0 \otimes 1}{\sigma_E(\du b \otimes e_0)}\right) + b  e_0 \otimes \hp{e_0 \otimes 1}{\nabla_E e_0}\\
		&= e_0 \otimes \phi^{-1}(\du b) + b  e_0 \otimes \iu{}\phi^{-1}(\omega)\\
		&= (V \otimes \id)\mleft(\nabla_{(\omega,\phi)}b_\phi\mright). \qedhere
	\end{align*}
\end{proof}

We have just seen that the generalised diffeomorphism group \(\tDiff(B)\) embeds via \(\hat{\tau}\) in \(\dPic(B)\).
The following refinement of Proposition-Definition \ref{propdef:topfroh} yields a surprising `moral converse': the entire differential Picard group \(\dpic(B)\) acts canonically as automorphisms on the graded centre of \(\Omega_B\) in a manner that will turn out to be explicitly compatible with this embedding by Example \ref{ex:twist3}.

\begin{propositiondefinition}[{Beggs--Majid~\cite[Prop. 5.9]{BeMa18}}]\label{propdef:froh}
	The \emph{Fr\"{o}hlich homomorphism} of \(B\) is the unique group homomorphism \(\hat{\Phi} : \dpic(B) \to \Aut(\Zent(\Omega_B),\du)\), such that, for every Hermitian \(B\)-bimodule with connection \((E,\sigma_E,\nabla_E)\),
	\begin{equation}
		\forall \beta \in \Zent_B(\Omega_B), \, \forall x \in E, \quad \!\!\!\hat{\Phi}_{[E,\nabla_E]}(\beta) \otimes x = \sigma_E^{-1}(x \otimes \beta);
	\end{equation}
	in this case, we call \(\hat{\Phi}_{[E,\nabla_E]}\) the \emph{Fr\"{o}hlich automorphism} of \((E,\sigma_E,\nabla_E)\).
\end{propositiondefinition}

\begin{proof}
	Relative to~\cite[Prop. 5.9]{BeMa18}, it remains to check that for each object \((E,\sigma_E,\nabla_E)\) of \(dPic(B))\), the automorphism \(\hat{\Phi}_{[E,\nabla_E]}\) of the graded algebra \(\Zent_B(\Omega_B)\) is \(\ast\)-preserving and commutes with \(\du_B\) on \(\Zent(\Omega_B)\); observe that \(\rest{\hat{\Phi}_{[E,\nabla_E]}}{\Zent(\Omega_B)^0} = \rest{\Phi_{[E]}}{\Zent(\Omega_B)^0}\)  is isometric by Proposition-Definition \ref{propdef:topfroh}.
	Let \((E,\sigma_E,\nabla_E) \in \Obj(\dPic(B))\) be given, and fix a basis \((e_i)_{i=1}^n\) and a strict cobasis \((\epsilon_j)_{j=1}^m\) for \(E\).
	On the one hand, by the proof of Theorem \ref{thm:stabilizer}, \emph{mutatis mutandis},
	\begin{align}
		\forall \beta \in \Zent(\Omega_B), && \hat{\Phi}_{[E,\nabla_E]}(\beta) &= \sum\nolimits_{i=1}^n \hp*{\Upsilon_{\Omega_B,E}(\conj{\sigma_E^{-1}(e_i \otimes \beta)})}{\conj{e_i} \otimes 1},\label{eq:froh1}\\
		\forall \beta \in \Zent(\Omega_B), && \hat{\Phi}^{-1}_{[E,\nabla_E]}(\beta) &= \sum\nolimits_{j=1}^m \hp*{\epsilon_j \otimes 1}{\sigma_E(\beta \otimes \epsilon_j)};\label{eq:froh2}
	\end{align}
	hence, by \eqref{eq:froh2} and \eqref{eq:braidherm} applied to \(\sigma_E\), the map \(\hat{\Phi}_{[E,\nabla_E]}\) is \(\ast\)-preserving.
	On the other hand, let \(\beta \in \Zent(\Omega_B)\) be given.
	Then, for all \(x \in E\), 
	\begin{align*}
		x \otimes \du_B\hat{\Phi}_{[E,\nabla_E]}^{-1}(\beta)
		&= \nabla_E\mleft(x \otimes \Phi_{[E,\nabla_E]}(\beta)\mright) - \nabla_E (x) \Phi_{[E,\nabla_E]}^{-1}(\beta)\\
		&= \nabla_E\mleft(\beta  (x \otimes 1)\mright) - \beta  \nabla_E x \\
		&= x \otimes \hat{\Phi}_{[E,\nabla_E]}^{-1}(\du_B\beta). \qedhere
	\end{align*}
\end{proof}

\begin{corollary}
	The canonical left action of \(\pi_0(\dPic(B)) \eqqcolon \dpic(B)\) on the Abelian group 
	\(
		\pi_1(\dPic(B)) = \Unit(\Zent(B) \cap \ker\du_B)
	\)
	is the left action induced by \(\hat{\Phi}\).
\end{corollary}

We can now introduce curvature as a \(1\)-cocycle for this \(\dpic(B)\)-action.
For convenience, let us define a \emph{pre-symplectic form} on \(B\) to be self-adjoint \(\beta \in \Zent(\Omega_B)^2\) satisfying \(\du\beta = 0\).
Hence, we denote by \(\cS(B)\) the real subspace of all pre-symplectic forms on \(B\), which we endow with the right \(\dpic(B)\)-action defined by
\begin{equation}
	\forall [E,\nabla_E] \in \dpic(B), \, \forall \omega \in \cS(B), \quad \omega \ract [E,\nabla_E] \coloneqq \hat{\Phi}_{[E,\nabla_E]}^{-1}(\omega).
\end{equation}
Moreover, recall that if \(\Gamma\) is a group and \(M\) is a right \(\Gamma\)-module (written additively), then a map \(c : \Gamma \to M\) is a right \(1\)-cocycle whenever
\[
	\forall \gamma,\eta\in\Gamma, \quad c(\gamma\eta) = c(\gamma) \ract \eta + c(\eta).
\]

\begin{propositiondefinition}[{Beggs--Majid~\cite[Prop.\ 5.9]{BeMa18}}]\label{propdef:curve}
	The \emph{curvature \(1\)-cocycle} of \(B\) is the unique right \(1\)-cocycle \(\bF : \dpic(B) \to \cS(B)\), such that, for every Hermitian \(B\)-bimodule with connection \((E,\sigma_E,\nabla_E)\),
	\begin{equation}\label{eq:curve}
		\forall \xi \in E \otimes_B \Omega_B, \quad \nabla_E^2 \xi = \xi \cdot \iu{}\bF_{[E,\nabla_E]};
	\end{equation}
	in this case, we call \(\bF_{[E,\nabla_E]} \in \cS(B)\) the \emph{curvature \(2\)-form} of \([E,\nabla_E]\).
\end{propositiondefinition}

\begin{proof}
	First, let \((E,\sigma_E,\nabla_E) \in \Obj(\dPic(B))\); fix a basis \((e_i)_{i=1}^m\) and a cobasis \((\epsilon_j)_{j=1}^n\) for \(E\).
	The map \(\nabla_E^2 : E \to E \otimes_B \Omega_B\) is right \(B\)-linear by repeated applications of \eqref{eq:rightleibniz} and left \(B\)-linear by \eqref{eq:bemacriterion2}.
	Thus, 
	\[
		\nabla_E^2 x = \nabla_E^2\mleft(\sum\nolimits_{j=1}^n \hp{\conj{x}}{\conj{\epsilon_j}}  \epsilon_j\mright) = \sum\nolimits_{j=1}^m \hp{\conj{x}}{\conj{\epsilon_j}}  \nabla^2_E \epsilon_j = x \otimes \sum\nolimits_{j=1}^m \hp{\epsilon_j \otimes 1}{\nabla^2_E\epsilon_j}
	\]
	for every \(x \in E\),
	so that the \(2\)-form \(\bF_{[E,\nabla_E]} \coloneqq -\iu{}\sum_{j=1}^m \hp{\epsilon_j \otimes 1}{\nabla^2_E\epsilon_j} \in \Omega^2_B\) satisfies
	\begin{equation}\label{eq:curv2form}
		\forall x \in E, \quad \nabla_E^2 x = x \otimes \iu{}\bF_{[E,\nabla_E]}.
	\end{equation}
	On the one hand, if \(\varpi\) is any \(2\)-form  satisfying \eqref{eq:curv2form}, then
	\[
		\varpi = \sum\nolimits_{j=1}^n \hp{\epsilon_j}{\epsilon_j}  \varpi = \sum\nolimits_{j=1}^n \hp{\epsilon_j \otimes 1}{\epsilon_j \otimes \varpi} = -\iu{}\sum\nolimits_{j=1}^n \hp{\epsilon_j \otimes 1}{\nabla_E^2\epsilon_j} = \bF_{[E,\nabla_E]};
	\]
	in fact, this uniqueness implies that \(\bF_{[E,\nabla_E]}\) depends only on \([E,\nabla_E] \in \dpic(B)\).
	On the other, \(\bF_{[E,\nabla_E]}\) is self-adjoint by construction from \((\epsilon_j)_{j=1}^n\) and repeated applications of \eqref{eq:hermconndiff}.

	We now show that \(\bF_{[E,\nabla_E]}\) is central, is closed, and satisfies \eqref{eq:curve}.
	First, by repeated applications of \eqref{eq:leftleibniz}, it follows that \(\nabla_E^2\sigma_E(\beta \otimes x) = \sigma_E(\beta \otimes x) \cdot \iu{}\bF_{[E,\nabla_E]}\) for all \(\beta \in \Omega_B\) and \(x \in E\), so that by invertibility of the map \(\sigma_E\), the \(2\)-form \(\bF_{[E,\nabla_E]}\) satisfies \eqref{eq:curve}.
	Next, by \eqref{eq:curve} and repeated applications of \eqref{eq:rightleibniz}, it follows that
	\(
		x \otimes \bF_{[E,\nabla_E]}  \beta = -\iu{}\nabla_E^2(x \otimes \beta) = x \otimes \beta  \bF_{[E,\nabla_E]}
	\)
	for every \(\beta \in \Omega_B\) and \(x \in E\),
	so that \(\bF_{[E,\nabla_E]}\) is central. 
	Finally, \(\bF_{[E,\nabla_E]}\) is closed since for every \(x \in E\), by \eqref{eq:curve},
	\[
		x \otimes \iu{}\,\du_B\bF_{[E,\nabla_E]} = \nabla_E(x \otimes \iu{}\bF_{[E,\nabla_E]}) - \nabla_E(x) \cdot \iu{}\bF_{[E,\nabla_E]} = \nabla_E(\nabla_E^2 x) - \nabla_E^2(\nabla_E x) = 0.
	\]
	
	Finally, by~\cite[Prop.\ 5.9]{BeMa18}, \emph{mutatis mutandis}, the map \([E,\nabla_E] \mapsto \bF_{[E,\nabla_E]}\) satisfies
	\(
		\bF_{[E \otimes_B F,\nabla_{E \otimes_B F}]} = \hat{\Phi}_{[F,\nabla_F]}^{-1}(\bF_{[E,\nabla_E]}) + \bF_{[F,\nabla_F]}
	\)
	for all objects \((E,\sigma_E,\nabla_E)\) and \((F,\sigma_F,\nabla_F)\) of \(\dPic(B)\),
	which is the required cocycle identity.
\end{proof}

\begin{example}\label{ex:classical3}
	We continue from Example \ref{ex:classical2}.
	Let \(\Omega^2(X,\bR)_{\cl}\) denote the \(\Diff(X)\)-invariant \(\bR\)-subspace of closed real \(2\)-forms on \(X\).
	On the one hand, let \(\Psi : \dpic(C^\infty(X)) \to \Diff(X)\) be the homomorphism induced by the Fr\"{o}hlich homomorphism of \(C^\infty(X)\).
	On the other, recall that the ordinary differential cohomology group \(\check{H}^2(X)\) is the group of isomorphism classes of Hermitian line bundles on \(X\) with unitary connection~\cite[Ex.\ 2.7]{HS}.
	Then, by Serre--Swan duality, 
	\[
		1 \to \check{H}^2(X) \xrightarrow{[\mathcal{E},\nabla_{\mathcal{E}}] \mapsto [\Gamma(\mathcal{E}),\nabla_{\mathcal{E}}]} \dpic(C^\infty(X)) \xrightarrow{\Psi} \Diff(X) \to 1
	\]
	defines a split exact sequence with canonical right splitting \(\phi \mapsto [\hat{\tau}(0,(\phi^{-1})^\ast)]\).
	Given the resulting isomorphism \(\Diff(X) \ltimes \check{H}^2(X)  \to \dpic(C^\infty(X))\) defined by
	\[
		(\phi,[\cE,\nabla_{\cE}]) \mapsto [\Gamma((\phi^{-1})^\ast\cE),(\phi^{-1})^\ast\nabla_{\cE}]  [\hat{\tau}(0,(\phi^{-1})^\ast)]
	\]
	we may identify the Fr\"{o}hlich homomorphism \(\Phi\) with the quotient map
	\[
		\left((\phi,[\cE,\nabla_{\cE}]) \mapsto \phi \right) : \Diff(X)  \ltimes \check{H}^2(X) \to \Diff(X)
	\]
	and the curvature \(1\)-cocycle \(\bF : \dpic(C^\infty(X)) \to \Omega^2(X,\bR)_{\cl}\) with the map
	\[
		\left((\phi,[\cE,\nabla_{\cE}]) \mapsto \phi^\ast \operatorname{tr}(\nabla^2_{\cE})\right) : \Diff(X) \ltimes \check{H}^2(X) \to \Omega^2(X,\bR)_{\cl}.
	\]
\end{example}

\begin{example}\label{ex:twist3}
	The homomorphism \(\hat{\tau}\) of Example \ref{ex:twist2} satisfies
	\[
		\forall (\omega,\phi) \in \tDiff(B), \quad \hat{\Phi} \circ \pi_0(\hat{\tau})(\omega,\phi) = \rest{\phi}{\Zent(\Omega_B)}, \quad
		 \bF \circ \pi_0(\hat{\tau})(\omega,\phi) = \phi^{-1}\mleft(\du\omega - \iu{}\omega^2\mright).
	\]
\end{example}

\begin{example}[{Connes~\cite[Thm 7]{Connes80}}]\label{ex:heis5}
	Continuing from Example \ref{ex:heis4}, observe that \(\Zent(\Omega_\theta(\bT^2))\) is the complex Gra\ss{}mann algebra in the self-adjoint generators \(e^1\) and \(e^2\) of degree \(1\).
	Hence, the homomorphism \(\hat{E} : \Gamma_\theta \to \dPic(C^\infty_\theta(\bT^2))\) satisfies
	\[
		\forall g \in \Gamma_\theta, \quad \hat{\Phi} \circ \pi_0(\hat{E})(g) = \bigoplus_{k=0}^2 (g_{21}\theta+g_{22})^{k} \id_{\Zent(\Omega_\theta)^k}, \,\, \bF \circ \pi_0(\hat{E})(g) = \frac{2\pi g_{21}}{g_{21}\theta+g_{22}}e^1  e^2.
	\]
\end{example}

\section{Reconstruction of NC principal \texorpdfstring{\(\U\)}{U(1)}-bundles with connection}\label{sec:3}

We now generalise the familiar correspondence between Hermitian line bundles with unitary connection and principal \(\U\)-bundles with principal connection to the NC setting.
This takes the form of an explicit equivalence of categories that can be viewed as an adaptation of Pimsner's construction~\cite{Pimsner} from the \Cstar-algebraic literature to the setting of NC differential geometry.

In what follows, we say that a representation \(U : \U \to \GL(V)\) is \emph{of finite type} whenever \(V = \bigoplus_{k \in \bZ}^\alg V_k\), where \(V_k \coloneqq \Set{v \in V \given \forall z \in \U, \, U_z v = z^k v}\) for all \(k \in \bZ\).

\subsection{Monoidal inversion and homomorphisms of coherent \texorpdfstring{\(2\)}{2}-groups}\label{sec:3.1}

First we leverage the coherence theorem for coherent \(2\)-groups of Ulbrich~\cite{Ulbrich} and Laplaza~\cite{Laplaza} to show that every homomorphism of coherent \(2\)-groups canonically defines a \emph{bar functor} or \emph{involutive monoidal functor} in the sense of Beggs--Majid and Egger~\cite{Egger}, respectively.
This will obviate any difficulties related to reconstructing \(\ast\)-structures on \textsc{nc} principal \(\U\)-bundle with principal connection.

We first recall the additional categorical structure that will fully capture the behaviour of monoidal inversion in a coherent \(2\)-group.

\begin{definition}[{Beggs--Majid~\cite{BeMa09}, Egger~\cite{Egger}}]
	A \emph{strong bar category} is a monoidal category \(\grp{G}\) equipped with a functor \(\conj{\cdot} : \grp{G} \to \grp{G}\), an isomorphism \(\star : 1 \to \conj{1}\), and natural isomorphisms
	\(\left(\bb_g : g \to \cconj{g}\right)_{g \in \Obj(\grp{G})}\) and \(\left(\Upsilon_{g,h} : \conj{g \otimes h} \to \conj{h}\otimes\conj{g}\right)_{(g,h) \in \Obj(\grp{G})^2}\), such that \(\bb_{\conj{g}} = \conj{\bb_g}\) for every \(g \in \Obj(\grp{G})\) and the following coherence diagrams commute for all \(g,h,k \in \Obj(\grp{G})\):
	
		\noindent\begin{minipage}[]{.33\linewidth}
		\begin{equation*}
			\begin{tikzcd}[ampersand replacement=\&, column sep=small]
				{1} \&\& {\conj{1}}\\
				\& {\cconj{1}} \&
				\arrow["{\star}", from=1-1, to=1-3]
				\arrow["{\bb_1}"', from=1-1, to=2-2]
				\arrow["{\conj{\star}}", from=1-3, to=2-2]
			\end{tikzcd}
		\end{equation*}
	\end{minipage}
	\begin{minipage}[]{.33\linewidth}
		\begin{equation*}
		\begin{tikzcd}[ampersand replacement=\&, column sep=small]
			{\conj{g \otimes 1}} \& {\conj{1} \otimes \conj{g}}\\
			{\conj{g}} \& {1 \otimes \conj{g}}
			\arrow["{\Upsilon_{g,1}}", from=1-1, to=1-2]
			\arrow["{\rho_{\conj{g}}^{-1}}", from=2-1, to=2-2]
			\arrow["{\conj{\rho_g}}"', from=1-1, to=2-1]
			\arrow["{\star \otimes \id_{\conj{g}}}", from=1-2, to=2-2]
		\end{tikzcd}
		\end{equation*}
	\end{minipage}
	\begin{minipage}[]{.33\linewidth}
	\begin{equation*}
		\begin{tikzcd}[ampersand replacement=\&, column sep=small]
			{\conj{1 \otimes g}} \& {\conj{g} \otimes \conj{1}}\\
			{\conj{g}} \& {\conj{g} \otimes 1}
			\arrow["{\Upsilon_{1,g}}", from=1-1, to=1-2]
			\arrow["{\lambda_{\conj{g}}^{-1}}", from=2-1, to=2-2]
			\arrow["{\conj{\lambda_g}}"', from=1-1, to=2-1]
			\arrow["{\id_{\conj{g}}{} \otimes \star}", from=1-2, to=2-2]
		\end{tikzcd}
	\end{equation*}
	\end{minipage}
	
	\noindent
	\begin{minipage}[]{.65\linewidth}
		\begin{equation*}
			\begin{tikzcd}[ampersand replacement=\&]
				{\conj{(g \otimes h) \otimes k}} \& {\conj{k} \otimes \conj{g \otimes h}} \& {\conj{k} \otimes (\conj{h} \otimes \conj{g})}\\
				{\conj{g \otimes (h \otimes k)}} \& {\conj{h \otimes k} \otimes \conj{g}} \& {(\conj{k} \otimes \conj{h}) \otimes \conj{g}}
				\arrow["{\Upsilon_{g \otimes h,k}}", from=1-1, to=1-2]
				\arrow["{\id{} \otimes \Upsilon_{g,h}}", from=1-2, to=1-3]
				\arrow["{\Upsilon_{g, h \otimes k}}", from=2-1, to=2-2]
				\arrow["{\Upsilon_{h,k} \otimes \id_g}", from=2-2, to=2-3]
				\arrow["{\conj{\alpha_{g,h,k}}}"', from=1-1, to=2-1]
				\arrow["{\alpha_{\conj{k},\conj{h},\conj{g}}}"', from=2-3, to=1-3]
			\end{tikzcd}
		\end{equation*}
	\end{minipage}
	\begin{minipage}[]{.33\linewidth}
		\begin{equation*}
			\begin{tikzcd}[ampersand replacement=\&, column sep=large]
				{g \otimes h} \& {\cconj{g} \otimes \cconj{h}}\\
				{\cconj{g \otimes h}} \& {\conj{\conj{h} \otimes \conj{g}}}
				\arrow["{\bb_g \otimes \bb_h}", from=1-1, to=1-2]
				\arrow["{\conj{\Upsilon_{g,h}}}", from=2-1, to=2-2]
				\arrow["{\bb_{g \otimes h}}"', from=1-1, to=2-1]
				\arrow["{\Upsilon_{\conj{h},\conj{g}}}"', from=2-2, to=1-2]
			\end{tikzcd}
		\end{equation*}	
	\end{minipage}
\end{definition}

For example, given a unital pre-\Cstar-algebra \(B\), the monoidal category \(\grp{Bimod}(B)\) defines a strong bar category with \(\star : B \to \conj{B}\), \(\bb\) and \(\Upsilon\) defined as follows:
	\begin{align*}
		\forall b \in B, && \star(b) &\coloneqq \conj{b^\ast},\\
		\forall E \in \Obj(\grp{Bimod}(B)), \, \forall x \in E, && \bb_E(x) &\coloneqq \cconj{x},\\
		\forall E,F \in \Obj(\grp{Bimod}(B)), \, \forall x \in E, \, \forall y \in F, && \Upsilon_{E,F}(\conj{x \otimes y}) &\coloneqq \conj{y} \otimes \conj{x}.
	\end{align*}
	
We now recall the coherence theorem for coherent \(2\)-groups.
Call an arrow of a coherent \(2\)-group \(\grp{G}\) \emph{structural} if it lies in the smallest subclass of \(\operatorname{Hom}(\grp{G})\) that:
\begin{enumerate}[leftmargin=*]
	\item contains the identity arrows, associators, left unitors, and right unitors of \(\grp{G}\) as a monoidal category and the evaluation arrows of \(\grp{G}\) as a coherent \(2\)-group;
	\item is closed under composition, inversion, and the monoidal product in \(\grp{G}\) as a monoidal category and monoidal inversion in \(\grp{G}\) as a coherent \(2\)-group.
\end{enumerate}
Hence, given endofunctors \(P,Q : \grp{G} \to \grp{G}\) of a coherent \(2\)-group \(\grp{G}\), we say that a natural transformation \(\eta : P \Rightarrow Q\) is \emph{structural} whenever \(\eta_g\) is structural for every object \(g\) of \(\grp{G}\).
For example, given a coherent \(2\)-group \(\grp{G}\), the natural isomorphisms \(\coev\) and \(\bb\) of Theorem~\ref{thm:laplaza} are both structural~\cite[Lemm.\ 4.4 \& 4.5]{Laplaza}.

\begin{theorem}[{Ulbrich~\cite{Ulbrich},~Laplaza~\cite[\S 2]{Laplaza}}]\label{thm:coherence}
	Let \(\grp{G}\) be a coherent \(2\)-group.
	For every pair \((g,h)\) of objects of \(\grp{G}\), there is at most one structural arrow \(g \to h\) in \(\grp{G}\).
\end{theorem}

Our first application of the coherence theorem is that a coherent \(2\)-group canonically defines a strong bar category with respect to monoidal inversion.

\begin{corollary}[{cf.\ Laplaza~\cite[p.\ 310]{Laplaza}}]\label{cor:laplaza}
	Let \(\grp{G}\) be a coherent \(2\)-group.
	There exist a unique structural isomorphism \(\star : 1 \to \conj{1}\) and a unique structural natural isomorphism
	\(
		 \left(\Upsilon_{g,h} : \conj{g \otimes h} \to \conj{h}\otimes\conj{g}\right)_{g,h \in \Obj(\grp{G})}
	\)
	making \(\grp{G}\) into a strong bar category with respect to monoidal inversion and the natural isomorphism \(\bb\) of Theorem~\ref{thm:laplaza}.
\end{corollary}

\begin{proof}
	First, construct a structural arrow \(\star : 1 \to \conj{1}\) by setting
	\(
		\star \coloneqq \lambda_{\conj{1}} \circ \coev_1
	\).
	Next, given objects \(g\) and \(h\) of \(\grp{G}\), construct a structural arrow \(\Upsilon_{g,h} : \conj{g \otimes h} \to \conj{h} \otimes \conj{g}\) as follows:
	first, construct a structural arrow \(\widetilde{\coev}_{g \otimes h} : 1 \to (g \otimes h) \otimes (\conj{h} \otimes \conj{g})\) by setting
	\[
		\widetilde{\coev}_{g \otimes h} \coloneqq \alpha_{g \otimes h,\conj{h},\conj{g}} \circ \left(\alpha^{-1}_{g,h,\conj{h}} \otimes \id_{\conj{g}}{}\right) \circ \left((\id_g{} \otimes \coev_h) \otimes \id_{\conj{g}}{}\right) \circ \left(\rho_g^{-1} \otimes \id_{\conj{g}}{}\right) \circ \coev_g,
	\]
	and then set
	\(
		\Upsilon_{g,h} \coloneqq \lambda_{\conj{h}\otimes\conj{g}} \circ (\ev_{g\otimes h}{} \otimes \id_{\conj{h}\otimes\conj{g}}) \circ \alpha^{-1}_{\conj{g\otimes h},g\otimes h,\conj{h}\otimes\conj{g}} \circ (\id_{\conj{g\otimes h}}{} \otimes \widetilde{\coev}_{g \otimes h}{}) \circ \rho^{-1}_{\conj{g \otimes h}}
	\).
	The claim now follows by Theorem \ref{thm:coherence}.
\end{proof}

\begin{remark}
	Let \(\grp{G}\) be a coherent \(2\)-group. 
	The structural isomorphism \(\star : 1 \to \conj{1}\) is the unique isomorphism of the inverses \((1,\lambda_1,\lambda_1^{-1})\) and \((\conj{1},\ev_1,\coev_1)\) of \(1\).
	Likewise, given \(g,h \in \Obj(\grp{G})\), the structural isomorphism \(\Upsilon_{g,h}\) is the unique isomorphism of the inverses \((\conj{h} \otimes \conj{g},\widetilde{\ev}_{g \otimes h}, \widetilde{\coev}_{g \otimes h})\) and \((\conj{g \otimes h},\ev_{g \otimes h},\coev_{g \otimes h})\) of \(g \otimes h\), where \(\widetilde{\ev}_{g \otimes h} : (\conj{h} \otimes \conj{g}) \otimes (g \otimes h) \to 1\) \and \(\widetilde{\coev}_{g \otimes h} : 1 \to (g \otimes h) \otimes (\conj{g} \otimes \conj{h})\) are the unique such structural arrows.
\end{remark}

For example, let \(B\) be a unital pre-\Cstar-algebra.
Then the canonical strong bar category structure on \(\Pic(B)\) of Corollary \ref{cor:laplaza} is that induced by the aforementioned strong bar category structure on \(\grp{Bimod}(B)\).

We now come to the main definition of this subsection.

\begin{definition}[{Beggs--Majid~\cite{BeMa09}, Egger~\cite{Egger}}]
	Let \(\grp{G}\) and \(\grp{G}^\prime\) be strong bar categories. 
	\begin{enumerate}[leftmargin=*]
		\item A \emph{bar functor} \(F: \grp{G} \to \grp{G}^\prime\) consists of a monoidal functor \(F: \grp{G} \to \grp{G}^\prime\) together with a natural isomorphism
		\(
			\left(F^{(-1)}_g : \conj{F(g)} \to F(\conj{g})\right)_{g \in \Obj(\grp{G})}
		\)
		making the following diagrams commute for all \(g,h \in \Obj(\grp{G})\):
		
		\noindent\begin{minipage}{0.5\linewidth}
			\begin{equation}\label{eq:bar1}\scalebox{0.9}{
				\begin{tikzcd}[ampersand replacement=\&]
						{\conj{F(1)}} \& {F(\conj{1})} \& {F(1)}\\
						{\conj{1}} \&\& 1
						\arrow["{F^{(-1)}_1}", from=1-1, to=1-2]
						\arrow["{F(\star^{-1})}", from=1-2, to=1-3]
						\arrow["{\star^{-1}}", from=2-1, to=2-3]
						\arrow["{\conj{F^{(0)}}}"', from=1-1, to=2-1]
						\arrow["{F^{(0)}}", from=1-3, to=2-3]
				\end{tikzcd}
			}\end{equation}	
		\end{minipage}
		\begin{minipage}{0.49\linewidth}
			\begin{equation}\label{eq:bar2}\scalebox{0.9}{
				\begin{tikzcd}[ampersand replacement=\&]
					{F(g)} \& {F(\cconj{g})}\\
					{\cconj{F(g)}} \& {\conj{F(\conj{g})}}
					\arrow["{F(\bb_g)}", from=1-1, to=1-2]
					\arrow["{\conj{F^{(-1)}_g}}", from=2-1, to=2-2]
					\arrow["{\bb_{F(g)}}"', from=1-1, to=2-1]
					\arrow["{F^{(-1)}_{\conj{g}}}"', from=2-2, to=1-2]
				\end{tikzcd}
			}\end{equation}
		\end{minipage}

		\begin{equation}\label{eq:bar3}
			\begin{tikzcd}[ampersand replacement=\&, column sep=large]
				{\conj{F(g) \otimes F(h)}} \& {\conj{F(h)} \otimes \conj{F(g)}} \& {F(\conj{h}) \otimes F(\conj{g})}\\
				{\conj{F(g \otimes h)}} \& {F(\conj{g \otimes h})} \& {F(\conj{h} \otimes \conj{g})}
				\arrow["{\Upsilon_{F(g),F(h)}}", from=1-1, to=1-2]
				\arrow["{F^{(-1)}_g \otimes F^{(-1)}_h}", from=1-2, to=1-3]
				\arrow["{F^{(-1)}_{g \otimes h}}", from=2-1, to=2-2]
				\arrow["{F(\Upsilon_{g,h})}", from=2-2, to=2-3]
				\arrow["{\conj{F^{(2)}_{g,h}}}"', from=1-1, to=2-1]
				\arrow["{F^{(2)}_{\conj{h},\conj{g}}}", from=1-3, to=2-3]
			\end{tikzcd}
		\end{equation}
		\item Given bar functors \(P,Q : \grp{G} \to \grp{G}^\prime\), a monoidal natural transformation \(\phi : P \Rightarrow Q\) is a \emph{bar natural transformation} if \(Q_g^{(-1)} \circ \conj{\phi_g} = \phi_{\conj{g}} \circ P_g^{(-1)}\) for all \(g \in \Obj(\grp{G})\).
	\end{enumerate}
\end{definition}

Given a homomorphism of coherent \(2\)-groups \(F\), the following will supply \(F^{(-1)}\).

\begin{proposition}[{Baez--Lauda~\cite[Thm 6.1]{BaezLauda}}]\label{prop:baezlauda}
	Let \(F : \grp{G} \to \grp{G}^\prime\) be a homomorphism of coherent \(2\)-groups.
	There exists a unique natural transformation \(\left(F^{(-1)}_g : \conj{F(g)} \to F(\conj{g})\right)_{g \in \Obj(\grp{G})}\) making the following commute for all \(g \in \Obj(\grp{G})\):	
	
\noindent\begin{minipage}{0.5\linewidth}
	\begin{equation}\label{eq:barfunctor1}\scalebox{0.8}{
		\begin{tikzcd}[ampersand replacement=\&,column sep=tiny]
		{\conj{F(g)} \otimes F(g)} \&\& {F(\conj{g}) \otimes F(g)}\\
		{1} \& {F(1)} \& {F(\conj{g} \otimes g)}
		\arrow["{F^{(-1)}_g \otimes \id_g}", from=1-1, to=1-3]
		\arrow["{F^{(0)}}", from=2-1, to=2-2]
		\arrow["{F(\ev_g)}"', from=2-3, to=2-2]
		\arrow["{\ev_g}"', from=1-1, to=2-1]
		\arrow["{F^{(2)}_{\conj{g},g}}", from=1-3, to=2-3]
		\end{tikzcd}
	}\end{equation}
\end{minipage}
\begin{minipage}{0.49\linewidth}
	\begin{equation}\label{eq:barfunctor2}\scalebox{0.8}{
		\begin{tikzcd}[ampersand replacement=\&, column sep=tiny]
		{F(g) \otimes \conj{F(g)}} \&\& {F(g) \otimes F(\conj{g})}\\
		{1} \& {F(1)}\& {F(g \otimes \conj{g})}
		\arrow["{\id_g{} \otimes F^{(-1)}_g}", from=1-1, to=1-3]
		\arrow["{F^{(0)}}", from=2-1, to=2-2]
		\arrow["{F(\coev_g)}", from=2-2, to=2-3]
		\arrow["{\coev_g}", from=2-1, to=1-1]
		\arrow["{F^{(2)}_{g,\conj{g}}}", from=1-3, to=2-3]
		\end{tikzcd}
	}\end{equation}
\end{minipage}
\end{proposition}

At last, we come to our main technical result.

\begin{theorem}\label{thm:barfunctor}
	Let \(\grp{G}\) and \(\grp{G}^\prime\) be coherent \(2\)-groups.
	\begin{enumerate}[leftmargin=*]
		\item Let \(F : \grp{G} \to \grp{G}^\prime\) be a homomorphism.
		Then \(F\) defines a bar functor with respect to the canonical natural isomorphism \(F^{(-1)}\) of Proposition \ref{prop:baezlauda}.
		\item Let \(P,Q : \grp{G} \to \grp{G}^\prime\) be homomorphisms, so that \(P\) and \(Q\) uniquely define bar functors satisfying \eqref{eq:barfunctor1} and \eqref{eq:barfunctor2}.
		Then every \(2\)-isomorphism \(\eta : P \Rightarrow Q\) is a bar natural transformation.
	\end{enumerate}	
\end{theorem}

\begin{lemma}[{Baez--Lauda~\cite[Proof of Thm 6.1]{BaezLauda}}]\label{lem:baezlauda}
	Let \(\grp{G}\) be a coherent \(2\)-group.
	For every object \(g\) of \(G\) and every inverse \((h,\sce,\sci)\) of \(g\), there exists a unique isomorphism of the inverses \((\conj{g},\ev_g,\coev_g)\) and \((h,\sce,\sci)\) of \(g\).
	Moreover, for every inverse \((h,\sce,\sci)\) of \(g\) and \(u : \conj{g} \to h\), the arrow \(u\) satisfies \eqref{cd:inverseiso1} with respect to the inverses \((\conj{g},\ev_g,\coev_g)\) and \((h,\sce,\sci)\) of \(g\) if and only if \(u\) is the unique isomorphism of the inverses \((\conj{g},\ev_g,\coev_g)\) and \((h,\sce,\sci)\) of \(g\).
\end{lemma}

\begin{lemma}\label{lem:bar}
	Let \(F : \grp{G} \to \grp{G}^\prime\) be a homomorphism of coherent \(2\)-groups, and let \(F^{(-1)}\) be the natural isomorphism of Proposition \ref{prop:baezlauda}.
	Let \(g,h \in \Obj(\grp{G})\), and let 
	\(
		\sce_{g, h} : (\conj{h} \otimes \conj{g}) \otimes (g \otimes h) \to 1\) and \(\sce_{F(g),F(h)} : (\conj{F(h)} \otimes \conj{F(g)}) \otimes (F(g) \otimes F(h)) \to 1
	\)
	be the unique such structural arrows.
	The following diagram commutes:
	\[\resizebox{\linewidth}{!}{
		\begin{tikzcd}[ampersand replacement=\&]
			{(\conj{F(h)} \otimes \conj{F(g)}) \otimes (F(g) \otimes F(h))} \& {(F(\conj{h}) \otimes F(\conj{g})) \otimes (F(g) \otimes F(h))} \& {F(\conj{h} \otimes \conj{g}) \otimes F(g \otimes h)} \\
			{1} \& {F(1)} \& {F((\conj{h} \otimes \conj{g}) \otimes (g \otimes h))}
			\arrow["{\sce_{F(g),F(h)}}"', from=1-1, to=2-1]
			\arrow["\substack{(F^{(-1)}_h \otimes F^{(-1)}_g) \otimes \id\\{}}", from=1-1, to=1-2]
			\arrow["\substack{F^{(2)}_{\conj{h},\conj{g}} \otimes F^{(2)}_{g,h}\\{}}", from=1-2, to=1-3]
			\arrow["{F^{(2)}_{\conj{h}\otimes\conj{g},g\otimes h}}", from=1-3, to=2-3]
			\arrow["{F(\sce_{g,h})}", from=2-3, to=2-2]
			\arrow["{F^{(0)}}", from=2-2, to=2-1]
		\end{tikzcd}
	}\]
\end{lemma}

\begin{proof}
	Let \(g,h \in \Obj(\grp{G})\).
	Note that we can construct \(\sce_{g, h}\) and \(\sce_{F(g),F(h)}\) as
	\begin{align*}
		\sce_{g \otimes h} &= \ev_h{} \circ \rho_{\conj{h}} \circ (\id{} \otimes \ev_g) \otimes \id{} \circ \alpha_{h,\conj{g},g}^{-1},\\
		\sce_{F(g) \otimes F(h)} &= \ev_{F(h)}{} \circ \rho_{\conj{F(h)}} \circ (\id{} \otimes \ev_{F(g)}) \otimes \id{} \circ \alpha_{F(h),\conj{F(g)},F(g)}^{-1}.
	\end{align*}
	Hence, our claim will follows from commutativity of the following diagram, where, for visual clarity, we replace the symbol \(\otimes\) with \(\cdot\).
		\[\resizebox{\linewidth}{!}{
		\begin{tikzcd}[ampersand replacement=\&, column sep = small]
			{(\conj{F(h)} \cdot \conj{F(g)}) \cdot (F(g) \cdot F(h))} \& {(F(\conj{h}) \cdot F(\conj{g})) \cdot (F(g) \cdot F(h))} \& {F(\conj{h}\cdot\conj{g}) \cdot (F(g) \cdot F(h))} \& {F(\conj{h} \cdot \conj{g}) \cdot F(g \cdot h)} \& {F((\conj{h} \cdot \conj{g}) \cdot (g \cdot h))}\\
			{((\conj{F(h)} \cdot \conj{F(g)}) \cdot F(g)) \cdot F(h)} \& {((F(\conj{h}) \cdot F(\conj{g})) \cdot F(g)) \cdot F(h)} \& {(F(\conj{h}\cdot\conj{g}) \cdot F(g)) \cdot F(h)} \& \&\\
			{(\conj{F(h)} \cdot (\conj{F(g)} \cdot F(g))) \cdot F(h)} \& {(F(\conj{h}) \cdot (F(\conj{g}) \cdot F(g))) \cdot F(h)} \& {(F((\conj{h}\cdot\conj{g}) \cdot g) \cdot F(h)} \& \& {F(((\conj{h} \cdot \conj{g}) \cdot g) \cdot h)}\\
			\& {(F(\conj{h}) \cdot F(\conj{g} \cdot g)) \cdot F(h)} \& {F(\conj{h}\cdot(\conj{g}\cdot g)) \cdot F(g)} \& \& {F((\conj{h} \cdot (\conj{g} \cdot g)) \cdot h)}\\
			\& {(F(\conj{h}) \cdot F(1)) \cdot F(h)} \& {F(\conj{h}\cdot 1) \cdot F(g)} \& \& {F((\conj{h} \cdot 1) \cdot h)}\\
			{(\conj{F(h)} \cdot 1) \cdot F(h)} \& {(F(\conj{h}) \cdot 1) \cdot F(h)} \& \& \&\\
			{\conj{F(h)} \cdot F(h)} \& {F(\conj{h}) \cdot F(h)} \& \& \& {F(\conj{h} \cdot h)}\\
			{1} \& \& \& \& {F(1)}
			\arrow["\substack{(F^{(-1)}_h \cdot F^{(-1)}_g) \cdot \id\\{}}", from=1-1, to=1-2]
			\arrow["\substack{F^{(2)}_{\conj{h},\conj{g}} \cdot \id\\{}}", from=1-2, to=1-3]
			\arrow["\substack{\id{} \cdot F^{(2)}_{g,h}\\{}}", from=1-3, to=1-4]
			\arrow["\substack{F^{(2)}_{\conj{h}\cdot\conj{g},g \cdot h}\\{}}", from=1-4, to=1-5]
			\arrow["\substack{(F^{(2)}_{\conj{h},\conj{g}} \cdot \id) \cdot \id\\{}}", from=2-2, to=2-3]
			\arrow["\substack{(F^{(-1)}_h \cdot (F^{(-1)}_g \cdot \id)) \cdot \id\\{}}", from=3-1, to=3-2]
			\arrow["\substack{F^{(2)}_{(\conj{h} \cdot \conj{g}) \cdot g,h}\\{}}", from=3-3, to=3-5]
			\arrow["\substack{F^{(2)}_{\conj{h},\conj{g}\cdot g} \cdot \id \\{}}", from=4-2, to=4-3]
			\arrow["\substack{F^{(2)}_{\conj{h} \cdot (\conj{g}\cdot g),h} \cdot \id \\{}}", from=4-3, to=4-5]
			\arrow["{F^{(2)}_{\conj{h},1} \cdot \id}", from=5-2, to=5-3]
			\arrow["{F^{(2)}_{\conj{h} \cdot 1,h}}", from=5-3, to=5-5]
			\arrow["{(F^{(-1)}_h \cdot \id) \cdot \id}", from=6-1, to=6-2]
			\arrow["{F^{(-1)}_h \cdot \id}", from=7-1, to=7-2]
			\arrow["{F^{(2)}_{\conj{h},h}}", from=7-2, to=7-5]
			\arrow["{F^{(0)}}"', from=8-5, to=8-1]
			\arrow["{\alpha^{-1}_{\conj{F(h)}\cdot\conj{F(g)},F(g),F(h)}}"', from=1-1, to=2-1]
			\arrow["{\alpha^{-1}_{F(\conj{h}) \cdot F(\conj{g}),F(g),F(h)}}", from=1-2, to=2-2]
			\arrow["{\alpha^{-1}_{F(\conj{h}\cdot\conj{g}),F(g),F(h)}}", from=1-3, to=2-3]
			\arrow["{F(\alpha^{-1}_{\conj{h}\cdot\conj{g},g,h})}", from=1-5, to=3-5]
			\arrow["{\alpha_{\conj{F(h)},\conj{F(g)},F(g)} \cdot \id}"', from=2-1, to=3-1]
			\arrow["{\alpha_{F(\conj{h}),F(\conj{g}),F(g)} \cdot \id}", from=2-2, to=3-2]
			\arrow["{F^{(2)}_{\conj{h}\cdot\conj{g},g} \cdot \id}", from=2-3, to=3-3]
			\arrow["{(\id{} \cdot \ev_{F(g)}) \cdot \id}"', from=3-1, to=6-1]
			\arrow["{(\id{} \cdot F^{(2)}_{\conj{g},g}) \cdot \id}", from=3-2, to=4-2]
			\arrow["{F(\alpha_{\conj{h},\conj{g},g}) \cdot \id}", from=3-3, to=4-3]
			\arrow["{F(\alpha_{\conj{h},\conj{g},g}\cdot\id)}", from=3-5, to=4-5]
			\arrow["{(\id \cdot F(\ev_g)) \cdot \id}", from=4-2, to=5-2]
			\arrow["{F(\id \cdot \ev_g) \cdot \id}", from=4-3, to=5-3]
			\arrow["{F((\id \cdot \ev_g) \cdot \id)}", from=4-5, to=5-5]
			\arrow["{(\id{} \cdot F^{(0)}) \cdot \id}", from=5-2, to=6-2]
			\arrow["{F(\rho_{\conj{h}}) \cdot \id}", from=5-3, to=7-2]
			\arrow["{F(\rho_{\conj{h}} \cdot \id)}", from=5-5, to=7-5]
			\arrow["{\rho_{\conj{F(h)}}}"', from=6-1, to=7-1]
			\arrow["{\rho_{F(\conj{h})} \cdot \id}", from=6-2, to=7-2]
			\arrow["{\ev_{F(h)}}"', from=7-1, to=8-1]
			\arrow["{F(\ev_h)}", from=7-5, to=8-5]
		\end{tikzcd}
	}\]
	However, this now follows from applying naturality of \(\alpha\), monoidality of \(F\), naturality of \(F^{(2)}\), and commutativity of \eqref{eq:barfunctor1} as appropriate to each sub-diagram.
\end{proof}

\begin{proof}[Proof of Theorem \ref{thm:barfunctor}]
	First, let \(F : \grp{G} \to \grp{G}^\prime\) be a monoidal functor; let \(F^{(-1)}\) be the resulting natural isomorphism of Proposition \ref{prop:baezlauda}.
	Given \(g \in \Obj(\grp{G})\), set \(\widetilde{\ev}_{F(g)} \coloneqq (F^{(0)})^{-1} \circ F(\ev_g) \circ F^{(2)}_{\conj{g},g}\) and \(\widetilde{\coev}_{F(g)} \coloneqq (F^{(2)}_{g,\conj{g}})^{-1} \circ F(\coev_g) \circ F^{(0)}\), so that \((F(\conj{g}),\widetilde{\ev}_{F(g)},\widetilde{\coev}_{F(g)})\) is an inverse for \(F(g)\).
	By Lemma~\ref{lem:baezlauda}, \(F^{(-1)}\) is defined as follows: given \(g \in \Obj(\grp{G})\), let \(F^{(-1)}_g : \conj{F(g)} \to F(\conj{g})\) be the unique isomorphism of the inverses \((\conj{F(g)},\ev_{F(g)},\coev_{F(g)})\) and \((F(\conj{g}),\widetilde{\ev}_{F(g)},\widetilde{\coev}_{F(g)})\) of \(F(g)\).
	
	We first show that \eqref{eq:bar1} commutes.
	By Lemma \ref{lem:baezlauda}, it suffices to show that the arrow \(F(\star) \circ (F^{(0)})^{-1} \circ \star^{-1} \circ \conj{F^{(0)}}\) satisfies \eqref{cd:inverseiso1} with respect to the inverses \((\conj{F(1)},\ev_{F(1)},\coev_{F(1)})\) and \((F(\conj{1}),\widetilde{\ev}_{F(1)},\widetilde{\coev}_{F(1)})\) of \(F(1)\).
	This is now follows by applying bifunctoriality of \(\otimes\), coherence in \(\grp{G}\), coherence in \(\grp{G}^\prime\), monoidality of \(F\), or naturality of \(F^{(2)}\) as appropriate to each sub-diagram:
	\[\resizebox{\linewidth}{!}{
		\begin{tikzcd}[ampersand replacement=\&, row sep=small]
			{\conj{F(1)} \otimes F(1)} \& {\conj{1} \otimes F(1)} \& \& {1 \otimes F(1)} \& {F(1) \otimes F(1)} \& {F(\conj{1}) \otimes F(1)}\\
			\& {\conj{1} \otimes 1} \& {1 \otimes 1}  \&\& F(1 \otimes 1) \& \\
			{1} \& \& \& {F(1)} \& \& {F(\conj{1} \otimes 1)}
			\arrow["{F^{(0)} \otimes \id}", from=1-1, to=1-2]
			\arrow["{\star^{-1} \otimes \id}", from=1-2, to=1-4]
			\arrow["{F^{(0)} \otimes \id}", from=1-4, to=1-5]
			\arrow["{F(\star) \otimes \id}", from=1-5, to=1-6]
			\arrow["{\star^{-1} \otimes \id}", from=2-2, to=2-3]
			\arrow["{F^{(0)}}"', from=3-4, to=3-1]
			\arrow["{F(\ev_1)}"', from=3-6, to=3-4]
			\arrow["{\ev_{F(1)}}"', from=1-1, to=3-1]
			\arrow["{\lambda_{F(1)}}", from=1-4, to=3-4]
			\arrow["{F^{(2)}_{1,1}}", from=1-5, to=2-5]
			\arrow["{F^{(2)}_{\conj{1},1}}", from=1-6, to=3-6]
			\arrow["{\conj{F^{(0)}} \otimes F^{(0)}}", from=1-1, to=2-2]
			\arrow["{\id{} \otimes F^{(0)}}"', from=1-4, to=2-3]
			\arrow["{\ev_1}"', from=2-2, to=3-1]
			\arrow["{\lambda_1}", from=2-3, to=3-4]
			\arrow["{F(\lambda_1)}"', from=2-5, to=3-4]
			\arrow["{F(\star \otimes \id)}", from=2-5, to=3-6]
		\end{tikzcd}
	}\]
	
	Next, we show that \eqref{eq:bar2} commutes.
	Let \(g \in \Obj(\grp{G})\).
	By Lemma~\ref{lem:baezlauda}, it suffices to show that
	\(
		 F(\bb_g) \circ \bb_{F(g)}^{-1} \circ (\conj{F^{(-1)}_g})^{-1}
	\)
	satisfies \eqref{cd:inverseiso1} with respect to the inverses \((\conj{F(\conj{g})},\ev_{F(\conj{g})},\coev_{F(\conj{g})})\) and \((F(\cconj{g}),\widetilde{\ev}_{F(\conj{g})},\widetilde{\coev}_{F(\conj{g})})\) of \(F(\conj{g})\).
	This now follows by applying bifunctoriality of \(\otimes\), coherence in \(\grp{G}\), coherence in \(\grp{G}^\prime\), commutativity of \eqref{eq:barfunctor2}, naturality of \(\ev\), and naturality of \(F^{(2)}\) as appropriate to each sub-diagram:
	\[\resizebox{\linewidth}{!}{
		\begin{tikzcd}[ampersand replacement=\&, column sep=large, row sep = small]
			{\conj{F(\conj{g})} \otimes F(\conj{g})} \& {\cconj{F(g)} \otimes F(\conj{g})} \&\& {F(g) \otimes F(\conj{g})} \& {F(\cconj{g}) \otimes F(\conj{g})}\\
			\& {\cconj{F(g)} \otimes \conj{F(g)}} \& {F(g) \otimes \conj{F(g)}} \& {F(g \otimes \conj{g})} \& \\
			{1} \& \& \& {F(1)} \& {F(\cconj{g}\otimes \conj{g})}
			\arrow["\substack{(\conj{F^{(-1)}_g})^{-1} \otimes \id\\{}}", from=1-1, to=1-2]
			\arrow["{\bb_{F(g)}^{-1} \otimes \id}", from=1-2, to=1-4]
			\arrow["{F(\bb_g) \otimes \id}", from=1-4, to=1-5]
			\arrow["{\bb_{F(g)} \otimes \id}", from=2-2, to=2-3]
			\arrow["{F^{(0)}}"', from=3-4, to=3-1]
			\arrow["{F(\ev_{\conj{g}})}"', from=3-5, to=3-4]
			\arrow["{\ev_{F(\conj{g})}}"', from=1-1, to=3-1]
			\arrow["{\conj{F^{(-1)}_g} \otimes F^{(-1)}_g}", from=2-2, to=1-1]
			\arrow["{\id{} \otimes F^{(-1)}_g}", from=2-3, to=1-4]
			\arrow["{F^{(2)}_{g,\conj{g}}}"', from=1-4, to=2-4]
			\arrow["{F^{(2)}_{\cconj{g},\conj{g}}}", from=1-5, to=3-5]
			\arrow["{F(\coev_g)}", from=3-4, to=2-4]
			\arrow["{F(\bb_g \otimes \id)}", from=2-4, to=3-5]
			\arrow["{\ev_{\conj{F(g)}}}"', from=2-2, to=3-1]
			\arrow["{\coev_{F(g)}}"', from=3-1, to=2-3]
		\end{tikzcd}
	}\]
	
	Finally, let us show that \eqref{eq:bar3} commutes.
	Let \(g,h \in \Obj(\grp{G})\) be given; for convenience, let \(\sce_{g, h}\) and \(\sce_{F(g), F(h)}\) be defined as in Lemma~\ref{lem:bar}.
	By Lemma~\ref{lem:baezlauda}, it suffices to show that
	\(
		F(\Upsilon_{g,h}^{-1}) \circ F^{(2)}_{\conj{h},\conj{g}} \circ F^{(-1)}_h \otimes F^{(-1)}_g \circ \Upsilon_{F(g),F(h)} \circ (\conj{F^{(2)}_{g,h}})^{-1}
	\)
	satisfies \eqref{cd:inverseiso1} with respect to the inverses \((\conj{F(g \otimes h)},\ev_{F(g\otimes h)},\coev_{F(g\otimes h)})\) and \((F(\conj{g\otimes h}),\widetilde{\ev}_{F(g\otimes h)},\widetilde{\coev}_{F(g\otimes h)})\)
	of \(F(g \otimes h)\).
	This now follows by applying coherence in \(\grp{G}\), coherence in \(\grp{G}^\prime\), bifunctoriality of \(\otimes\), naturality of \(\ev\), naturality of \(F^{(2)}\) and Lemma~\ref{lem:bar} as appropriate to each sub-diagram:
	\[\resizebox{\linewidth}{!}{
		\begin{tikzcd}[ampersand replacement=\&, row sep=small]
			{\conj{F(g\otimes h)} \otimes F(g\otimes h)} \& \& \& {1}\\
			{\conj{F(g) \otimes F(h)} \otimes F(g\otimes h)} \& {\conj{F(g) \otimes F(h)} \otimes (F(g) \otimes F(h))} \& \& \\
			{(\conj{F(h)} \otimes \conj{F(g)}) \otimes F(g\otimes h)} \& {(\conj{F(h)} \otimes \conj{F(g)}) \otimes (F(g) \otimes F(h))} \&\&\\
			{(F(\conj{h}) \otimes F(\conj{g})) \otimes F(g\otimes h)} \& {(F(\conj{h}) \otimes F(\conj{g})) \otimes (F(g) \otimes F(h))} \&\&\\
			{F(\conj{h}\otimes\conj{g}) \otimes F(g \otimes h)} \& {F(\conj{h}\otimes\conj{g})\otimes (F(g)\otimes F(h))} \& \&\\
			\& {F((\conj{h}\otimes\conj{g}) \otimes (g\otimes h))} \& \& {F(1)}\\
			{F(\conj{g \otimes h}) \otimes F(g \otimes h)} \& \& \& F(\conj{g\otimes h} \otimes (g \otimes h))
			\arrow["{\ev_{F(g \otimes h)}}", from=1-1, to=1-4]
			\arrow["{\id{}\otimes F^{(2)}_{g,h}}"', from=2-2, to=2-1]
			\arrow["{\ev_{F(g) \otimes F(h)} \otimes \id}", from=2-2, to=1-4]
			\arrow["{\sce_{F(g), F(h)}}"', from=3-2, to=1-4]
			\arrow["{\id{}\otimes F^{(2)}_{g,h}}"', from=5-2, to=5-1]
			\arrow["{F(\sce_{g,h})}", from=6-2, to=6-4]
			\arrow["{F^{(2)}_{\conj{g\otimes h},g\otimes h}}", from=7-1, to=7-4]
			\arrow["{(\conj{F^{(2)}_{g,h}})^{-1} \otimes \id}"', from=1-1, to=2-1]
			\arrow["{F^{(0)}}"', from=6-4, to=1-4]
			\arrow["{\Upsilon_{F(g),F(h)} \otimes\id}"', from=2-1, to=3-1]
			\arrow["{\Upsilon_{F(g),F(h)} \otimes\id}"', from=2-2, to=3-2]
			\arrow["{(F^{(-1)}_h \otimes F^{(-1)}_g) \otimes \id}"', from=3-1, to=4-1]
			\arrow["{(F^{(-1)}_h \otimes F^{(-1)}_g) \otimes \id}"', from=3-2, to=4-2]
			\arrow["{F^{(2)}_{\conj{h},\conj{g}} \otimes \id}"', from=4-1, to=5-1]
			\arrow["{F^{(2)}_{\conj{h},\conj{g}} \otimes \id}"', from=4-2, to=5-2]
			\arrow["{F(\Upsilon_{g,h}^{-1}) \otimes \id}"', from=5-1, to=7-1]
			\arrow["{F(\ev_{g\otimes h})}"', from=7-4, to=6-4]
			\arrow["{F(\Upsilon_{g,h}^{-1} \otimes \id})", from=6-2, to=7-4]
			\arrow["{F^{(2)}_{\conj{h} \otimes \conj{g},g\otimes h}}"', from=5-1, to=6-2]
		\end{tikzcd}
	}\]
	
	Now, let \(P, Q : \grp{G} \to \grp{G}^\prime\) be monoidal functors, and let \(\eta : P \Rightarrow Q\) be a monoidal natural transformation.
	Let \(g \in \Obj(\grp{G})\).
	To show that \eqref{eq:bar3} commutes, it suffices to show that \(\phi_{\conj{g}}^{-1} \circ Q_g^{(-1)} \circ \conj{\phi_g}\) satisfies \eqref{cd:inverseiso1} with respect to the inverses \((\conj{P(g)},\ev_{P(g)},\coev_{P(g)})\) and \((P(\conj{g}),\widetilde{\ev}_{P(g)},\widetilde{\coev}_{P(g)})\) of \(P(g)\). In turn, it suffices to show that the following diagram commutes:
	\[
		\begin{tikzcd}[ampersand replacement=\&, row sep=small]
			{\conj{P(g)} \otimes P(g)} \& {\conj{Q(g)} \otimes Q(g)} \& {Q(\conj{g}) \otimes Q(g)} \& {P(\conj{g}) \otimes P(g)}\\
			\& {Q(1)} \& {Q(\conj{g} \otimes g)} \&\\
			{1} \& {P(1)} \& \& {P(\conj{g}\otimes g)}
			\arrow["{\conj{\phi_g} \otimes \phi_g}", from=1-1, to=1-2]
			\arrow["{Q^{(-1)}_g \otimes \id}", from=1-2, to=1-3]
			\arrow["{\phi_{\conj{g}} \otimes \phi_g}", from=1-3, to=1-4]
			\arrow["{Q(\ev_g)}", from=2-3, to=2-2]
			\arrow["{P(\ev_g)}", from=3-4, to=3-2]
			\arrow["{P^{(0)}}", from=3-2, to=3-1]
			\arrow["{\ev_{P(g)}}"', from=1-1, to=3-1]
			\arrow["{\ev_{Q(g)}}"', from=1-2, to=3-1]
			\arrow["{Q^{(2)}_{\conj{g},g}}", from=1-3, to=2-3]
			\arrow["{P^{(2)}_{\conj{g},g}}", from=1-4, to=3-4]
			\arrow["{Q^{(0)}}", from=2-2, to=3-1]
			\arrow["{\phi_1}", from=3-2, to=2-2]
			\arrow["{\phi_{\conj{g}\otimes g}}", from=3-4, to=2-3]
		\end{tikzcd}
	\]
	However, this diagram commutes by applying naturality of \(\ev\), commutativity of \eqref{eq:barfunctor1} for \(Q\), naturality of \(\phi\), and monoidality of \(\phi\) as appropriate to each sub-diagram.
\end{proof}

\begin{example}[{Buss--Meyer--Zhu~\cite[Thm 3.3]{BMZ}}]\label{ex:fell}
	Let \(B\) be a unital pre-\Cstar-algebra, let \(\Gamma\) be a group, and let \(F : \Gamma \to \Pic(B)\) be a homomorphism.
	The disjoint union \(\mathcal{F} \coloneqq \bigsqcup_{\gamma \in \Gamma} F(\gamma)\) defines a pre-Fell bundle over \(\Gamma\) in the sense of Exel~\cite[Def.\ 24.2]{exelbook} with respect to the fibrewise multiplication \(\mathcal{F} \times\mathcal{F} \to \mathcal{F}\) and the fibrewise \(\ast\)-operation \(\mathcal{F} \to \mathcal{F}\) defined, respectively, by
	\begin{align*}
		\forall \gamma,\eta\in \Gamma, \, \forall p \in F(\gamma), \, \forall q \in F(\eta), && p  q &\coloneqq F^{(2)}_{\gamma,\eta}(\gamma \otimes \eta),\\
		\forall \gamma \in \Gamma, \, \forall p \in F(\gamma), && p^\ast &\coloneqq F^{(-1)}_\gamma(\conj{p}),
	\end{align*}
	where \(F^{(-1)}\) is the natural transformation of Theorem~\ref{thm:barfunctor}.
	Note that Theorem~\ref{thm:barfunctor} as applied to \(\grp{Hom}(\Gamma,\Pic(B))\) recovers Buss--Meyer--Zhu's construction~\cite[Proof of Thm 3.3]{BMZ} of the fibrewise \(\ast\)-operation on \(\mathcal{F}\).
\end{example}

\subsection{Generalised crossed products via homomorphisms of coherent \texorpdfstring{\(2\)}{2}-groups}

We now revisit the well-known theory of \textsc{nc} topological principal \(\U\)-bundles~\cite{ABE,BeBrz14,AKL} from the perspective of coherent \(2\)-groups.
This will let us generalise Pimsner's construction~\cite{Pimsner} as adapted by Abadie--Eilers--Exel~\cite{ABE} to NC differential geometry by replacing the Picard \(2\)-group with the differentiable Picard \(2\)-group.
In what follows, let \(B\) be a unital pre-\Cstar-algebra; again, recall that its \emph{positive cone} \(B_+\) is the set of elements that are positive in the \Cstar-completion of \(B\).

Let us define a \emph{\(\Unit(1)\)-pre-\Cstar-algebra of finite type} is a unital pre-\Cstar-algebra \(P\) equipped with a strongly continuous \(\Unit(1)\)-action of finite type by isometric \(\ast\)-automorphisms.
In this case, the spectral subspace \(P^{\Unit(1)} = P_0\) is a unital \(\ast\)-subalgebra of \(P\), and the decomposition of complex vector spaces \(P = \oplus_{k\in\bZ}P_k\) defines a \(\bZ\)-grading of the unital \(\ast\)-algebra \(P\) in the sense that \(P_m \cdot P_n \subseteq P_{m+n}\) for all \(m,n \in \bZ\) and \(\ast(P_m) \subseteq P_{-m}\) for all \(m \in \bZ\).
This permits the following minimalistic definition of topological quantum principal \(\U\)-bundle.

\begin{definition}[{cf.\ Arici--Kaad--Landi~\cite[\S 4.2]{AKL}}]
	A \emph{topological quantum principal \(\U\)-bundle} is a pre-\Cstar-algebra \((P,\alpha)\) of finite type, such that there exist finite families \((e_i)_{i=1}^m\) and \((\epsilon_j)_{j=1}^n\) in \(P_1\) satisfying \(\sum_{i=1}^m e_ie_i^\ast = 1\) and \(\sum_{j=1}^n \epsilon_j^\ast \epsilon_j = 1\).
\end{definition}

This definition is slightly unconventional but relates to more familiar definitions as follows.
Let \((P,\alpha)\) be a topological quantum principal \(\U\)-bundle.
On the one hand, by an observation of N\u{a}st\u{a}sescu--Van Ostaeyen~\cite[Lemma I.3.2]{NVO}, the \(\Unit(1)\)-action \(\alpha\) is \emph{principal} in the sense that 
\(
	\operatorname{Span}_{\bC}\Set{z^k \otimes pq \given k \in \bZ, \, p \in P_k, \,q\in P} = \mathcal{O}(\Unit(1)) \otimes_{\bC} P.
\)
On the other hand, by an observation of Ulbrich~\cite[Lemma 2.1]{Ulbrichgraded}, it follows that the \(\bZ\)-grading \(P = \bigoplus_{k \in \bZ}P_k\) of \(P\) is \emph{strong} in the sense that \(P_m \cdot P_n = P_{m+n}\) for all \(m,n\in\bZ\).
The familiar fact that \(\alpha\) is principal if and only if the \(\bZ\)-grading of \(P\) is strong~\cite[Lemma I.3.2]{NVO} yields the familiar algebraic definition of (topological) quantum principal \(\U\)-bundle in the literature.

\begin{example}\label{ex:classicaltotal1}
Let \(\pi : X \to Y\) be a compact differentiable principal \(\Unit(1)\)-bundle with principal right \(\Unit(1)\)-action \(\sigma : \Unit(1) \to \Diff(X)\).
Then
\begin{equation*}
		C^\infty_{\mathrm{alg}}(X) \coloneqq \bigoplus_{k \in \bZ}^{\mathrm{alg}} \Set{\omega \in C^\infty(X) \given \forall z \in \Unit(1), \, (\sigma_z)^\ast \omega = z^{k} \omega}
\end{equation*}
defines a topological quantum principal \(\U\)-bundle with respect to the \(\U\)-action \(\alpha \coloneqq (z \mapsto (\sigma_{z^{-1}})^\ast)\), and the pullback map \(\pi^\ast : C^\infty(Y) \to C^\infty_{\mathrm{alg}}(X)^{\Unit(1)}\) is an isometric \(\ast\)-isomorphism.
In particular, one can use an atlas of local bundle trivialisations for \(\pi : X \to Y\) together with a subordinate smooth partition of unity on \(Y\) to construct a finite family \((e_i)_{i=1}^m\) in \(C^\infty_{\mathrm{alg}}(X)_1\) satisfying \(\sum_{i=1}^n e_i e_i^\ast = \sum_{i=1}^n e_i^\ast e_i = 1\).
\end{example}

The following introduces our second main running example, the first genuinely NC example of a topological quantum principal \(\U\)-bundle in the literature.

\begin{example}[{Brzezi\'{n}ski--Majid~\cite[\S 5.2]{BrM}}]\label{ex:hopf1}
	Let \(q \in (0,\infty) \setminus \Set{1}\), so that the corresponding \emph{quantum special unitary group} \'{a} la Woronowicz~\cite{Woronowicz87} is the universal \Cstar-algebra \(C_q(\SU(2))\) generated by elements \(a\) and \(c\) satisfying
	\[
		ac = qca, \quad ac^\ast = q c^\ast a, \quad c^\ast c = c c^\ast, \quad a^\ast a + c^\ast c = 1, \quad aa^\ast + q^2 cc^\ast = 1;
	\]
	the corresponding unital pre-\Cstar-algebra \(\cO_q(\SU_2)\) is the dense unital \(\ast\)-subalgebra of \(C_q(\SU_2)\) consisting of complex polynomials in \(a\), \(a^\ast\), \(c\), and \Cstar.
	Then \(\cO_q(\SU(2))\) defines a topological quantum principal \(\U\)-bundle with respect to the unique \(\U\)-action of finite type \(\alpha\) satisfying \(\alpha_z(a) = za\) and \(\alpha_z(c) = zc\) for all \(z \in \U\); in particular, the families \((a,qc)\) and \((a,c)\) in \(\cO_q(\SU(2))_1\) respectively satisfy \(aa^\ast + (qc)(qc)^\ast = 1\) and \(a^\ast a + c^\ast c = 1\).
	Moreover, the \(\Unit(1)\)-action \(\alpha\) satisfies \(\cO_q(\SU(2))^{\Unit(1)} = \cO_q(\CP^1)\), where \(\cO_q(\CP^1)\), the algebraic \emph{standard Podle\'{s} sphere}~\cite{Podles87}, is the unital \(\ast\)-subalgebra of \(\cO_q(\SU(2))\) consisting of complex polynomials in the elements \(c^\ast c\), \(ac^\ast\), and \(ca^\ast\).
\end{example}

We note that our rather strict definition of topological quantum principal \(\U\)-bundle reduces to a simpler definition whenever the \(\ast\)-subalgebra of \(\U\)-invariant elements is sufficiently like a \Cstar-algebra.

\begin{proposition}
	Let \(P\) be a unital pre-\Cstar-algebra with \(\Unit(1)\)-action of finite type \(\alpha\), such that \(P^{\Unit(1)}\) admits polar decompositions.
	Then \((P,\alpha)\) is a topological quantum principal \(\U\)-bundle if and only if \(P_1 \cdot P_{-1} = P^{\Unit(1)}\) and \(P_{-1} \cdot P_1 = P^{\Unit(1)}\).
\end{proposition}

\begin{proof}
	For each \(k \in \bZ\), the spectral subspace \(P_k\) defines a \(P^{\U}\)-bimodule with positive definite \(P^{\U}\)-valued inner product \(\hp{}{}_k \coloneqq \left((p,q) \mapsto p^\ast q\right)\).
	Moreover, for each \(k \in \bZ\), the \(P^{\U}\)-valued inner products \(\hp{}{}_k\) and \(\hp{}{}_{-k}\) satisfy \(p \cdot \hp{q}{r}_k = \hp{p^\ast}{q^\ast}_{-k} \cdot r\) for all elements \(p,q,r\in P_k\).
	Hence, we may apply the proof of Proposition~\ref{prop:localcstar}, \emph{mutatis mutandis}, to \(P_1\) and \(P_{-1}\), where \(P_{-1}\) admits the isomorphism of \(B\)-bimodules \((p \mapsto \conj{p^\ast}) : P_{-1} \to \conj{P_1}\).
\end{proof}

Recall that \(B\) is a fixed unital pre-\Cstar-algebra.
We define the concrete category \(\Circ(B)\) of \emph{topological quantum principal \(\U\)-bundles over \(B\)} as follows:
\begin{enumerate}[leftmargin=*]
	\item an object of \(\Circ(B)\) is a topological quantum principal \(\U\)-bundle together with an isometric \(\ast\)-isomorphism \(\iota_P : B \to P^{\Unit(1)}\);
	\item an arrow \(f : P \to Q\) in \(\Circ(B)\) is a \(\Unit(1)\)-equivariant isometric \(\ast\)-isomorphism, such that \(f \circ \iota_P = \iota_Q\).
\end{enumerate}	
We can now make precise sense of associated line bundles in the NC setting.

\begin{proposition}[{Exel~\cite[\S 2]{Exel}, Schwieger--Wagner~\cite[\S 4.1]{SW1}}]\label{prop:assocline}
	The following defines a functor \(\cL : \Circ(B) \to \grp{Hom}(\bZ,\Pic(B))\).
	\begin{enumerate}[leftmargin=*]
		\item Let \(P\) be a topological quantum principal \(\U\)-bundle over \(B\).
			Define a homomorphism \(\cL(P) : \bZ \to \Pic(B)\) as follows:
			\begin{enumerate}[leftmargin=*]
				\item given \(k \in \bZ\), let \(\cL(P)(k) \coloneqq P_k\) as a vector space with \(B\)-bimodule structure
				\begin{equation}
					\forall a,b \in B, \, \forall p \in P_k, \quad a  p  b \coloneqq \iota_P(a)p\iota_P(b)
				\end{equation}
				and the \(B\)-valued inner products on \(P_k\) and \(\conj{P_k}\) defined, respectively, by
				\begin{equation}
					\forall p,q \in P_k, \quad \hp{p}{q} \coloneqq \iota_P^{-1}(p^\ast q), \quad \hp{\conj{p}}{\conj{q}} \coloneqq \iota_P^{-1}(pq^\ast);\\
				\end{equation}
				\item set \(\cL(P)^{(0)} \coloneqq \iota_P^{-1}\);
				\item given \(m,n \in \bZ\), let \(\cL(P)^{(2)}_{m,n} : \cL(P)(m) \otimes \cL(P)(n) \to \cL(P)(m+n)\) be induced by multiplication in \(P\).
			\end{enumerate}
		\item Let \(f : P \to Q\) be an isomorphism of topological quantum principal \(\U\)-bundles over \(B\).
			Define the corresponding \(2\)-isomorphism \(\cL(f) : \cL(P) \Rightarrow \cL(Q)\) by
			\begin{equation}
				\forall k \in \bZ, \quad \cL(f)_k\coloneqq \rest{f}{P_k}.
			\end{equation}
	\end{enumerate}
\end{proposition}

\begin{proof}
	This is mostly a straightforward exercise in checking definitions.
	Let \(P \in \Obj(\Pic(B))\) be given.
	When checking that the functor \(\cL(P) : \bZ \to \Pic(B)\) is well defined, the only non-trivial point is strict fullness of all \(B\)-valued inner product.
	Let \((e_i)_{i=1}^m\) and \((\epsilon_j)_{j=1}^n\) be families in \(P_1\) satisfying \(\sum_{i=1}^m e_ie_i^\ast = \sum_{j=1}^n \epsilon_j^\ast \epsilon_j = 1\), and define \(e_I \coloneqq e_{i_1} \dotsc e_{i_k}\) for all \(k \in \bN\) and \(I = (i_1,\dotsc,i_k) \in \Set{1,\dotsc,m}^k\) and  \(\epsilon_J \coloneqq \epsilon_{j_1}\dotsc\epsilon_{j_k}\) for all \(k \in \bN\) and \(J = (j_1,\dotsc,j_k) \subset \Set{1,\dotsc,n}^k\).
	Then, for each \(k \in \bN\), it follows that \((e_I^\ast)_{I \in \Set{1,\dotsc,m}^k}\) is a cobasis for \(\cL(P)(-k)\), that \((\conj{\epsilon_J^\ast})_{J \in \Set{1,\dotsc,n}^k}\) is a cobasis for \(\conj{\cL(P)(-k)}\), that \((\epsilon_J)_{J \in \Set{1,\dotsc,n}^k}\) is a cobasis for \(\cL(P)(k)\), and that \((\conj{e_I})_{I \in \Set{1,\dotsc,m}^k}\) is a cobasis for \(\conj{\cL(P)(k)}\).
	From here, monoidality of \(\cL(P)\) follows from elementary algebraic properties of \(P\): coherence with respect to unitors follows from multiplicativity of the isometric \(\ast\)-isomorphism \(\iota_P : B \to P^{\Unit(1)}\), while coherence with respect to associators follows from associativity of multiplication in \(P\).
	Similarly, if \(f : P \to Q\) is an arrow in \(\Pic(B)\), then \(\cL(f)\) intertwines \(\cL(P)^{(0)}\) and \(\cL(Q)^{(0)}\)  since \(f\) intertwines the given isometric \(\ast\)-isomorphisms \(\iota_P : B \to P^{\Unit(1)}\) and \(\iota_Q : B \to Q^{\Unit(1)}\), while coherence of \(\cL(f)\) with respect to \(\cL(P)^{(2)}\) and \(\cL(Q)^{(2)}\) follows from multiplicativity of \(f\).
\end{proof}

For example, in Example \ref{ex:classicaltotal1}, \(C^\infty_{\mathrm{alg}}(X)\) defines an object of \(\Circ(C^\infty(Y))\) with respect to \(\pi^\ast : C^\infty(Y) \to C_{\mathrm{alg}}^\infty(X)^{\Unit(1)}\).
	By Serre--Swan duality, for each \(k \in \bZ\), the Hermitian line \(B\)-bimodule \(\cL_k(C^\infty_{\mathrm{alg}}(X))\) recovers the associated Hermitian line bundle of winding number \(-k\).

Likewise, in Example \ref{ex:hopf1}, the homomorphism \(\cE : \bZ \to \Pic(\cO_q(\mathbf{CP}^1))\) given by
\(
	\cE \coloneqq \cL(\cO_q(\SU(2)))
\)
recovers (up to a sign convention) the canonical line bundles on \(\cO_q(\mathbf{CP}^1)\) as studied by Landi--Reina--Zampini \cite{LRZ}.
In fact, by a result of Carotenuto--\'{O} Buachalla \cite[Prop.\ 4.4]{COB}, the homomorphism \(\cE\) exhausts the left \(\cO_q(\SU(2))\)-covariant Hermitian line \(\cO_q(\CP^1)\)-bimodules up to isomorphism.

We now recover the known result that the functor \(\cL\) is an equivalence of categories.
As a preliminary, recall that a \emph{conditional expectation} of a unital pre-\Cstar-algebra \(A_2\) onto a unital pre-\Cstar-algebra \(A_1\) with respect to an isometric \(\ast\)-homomorphism \(\iota : A_1 \to A_2\) is a contractive unit-preserving and \(\ast\)-preserving \(A_1\)-bimodule map \(\bE : A_2 \to A_1\) satisfying \(\bE((A_2)_+) \subseteq (A_1)_+\) and \(\bE \circ \iota = \id_{A_1}\).
In this case, we say that \(\bE\) is \emph{faithful} whenever it satisfies
\(
	\Set{a \in (A_2)_+ \given \bE(a) = 0} = \Set{0}.
\)

\begin{proposition}\label{prop:conditional}
	Let \(P\) be a topological quantum principal \(\U\)-bundle over \(B\).
	Define a complex-linear map \(\bE_P : P \to B\) by setting \(\rest{\bE_P}{P_j} \coloneqq \left(p \mapsto\iota_P^{-1}\mleft(\delta^{j,0}p\mright)\right)\) for all \(j \in \bZ\).
	Then \(\bE\) is a \(\U\)-invariant faithful conditional expectation of \(P\) onto \(B\) with respect to \(\iota_P\).
\end{proposition}

\begin{proof}
	Let \(\sigma\) denote the \(\U\)-action on \(P\), and let \(m\) denote the normalised Haar measure on \(\U\).
	Note that \(\bE_P\) is manifestly \(\U\)-invariant, unit-preserving, \(\ast\)-preserving, and \(B\)-bilinear and satisfies \(\bE_P \circ \iota_P = \id_B\).
	Since \(\sigma\) is of finite type, we may use Bochner integration on \(\U\) to write
	\(\bE_P = \left(p \mapsto \iota_P\mleft(\int_{\U} \sigma_z(p) \, \du{m}(z)\mright)\right)\).
	Since \(\sigma\) acts isometrically on \(P\), it follows that \(\bE\) is contractive; since \(\sigma\) acts by unital \(\ast\)-automorphisms, it follows that the \(\bE_P\) maps \(P_+\) to \(B_+\).
	
	Let us now show that \(\bE\) is faithful.\footnote{This elementary argument, which is surely folkloric, was found in an anonymous answer to a MathOverflow question (\url{https://mathoverflow.net/q/72624}).}
	Let \(p \in P_+ \setminus \Set{0}\), so that there exists a bounded state \(\phi : P \to \bC\), such that \(\phi(p) > 0\).
	Since \(\left(z \mapsto \phi(\sigma_z(p))\right) : \U \to [0,\infty)\) is continuous, there exists an open neighbourhood \(I\) of \(1\), such that \(\phi(\sigma_z(p)) > \tfrac{1}{2}\phi(p)\) for all \(z \in I\).
	Hence, by norm-continuity of \(\bE_P\), it follows that
	\[
		(\phi \circ \iota_P)\mleft(\bE_P(p)\mright) = \int_{\U} (\phi \circ \sigma_z)(p) \, \du{m}(z) \geq \frac{1}{2}\phi(p)m(I) > 0. \qedhere
	\]
\end{proof}

\begin{theorem}[{Buss--Meyer--Zhu~\cite[Thm 3.3]{BMS}, Schwieger--Wagner~\cite[Thmm 4.21 \& 5.2]{SW1}}]\label{thm:fell}
	The following defines a a weak inverse \(\Sigma : \grp{Hom}(\bZ,\Pic(B)) \to \Circ(B)\) of the functor \(\cL\).
	\begin{enumerate}[leftmargin=*]
		\item Given a homomorphism \(F : \bZ \to \Pic(B)\), construct a topological quantum principal \(\U\)-bundle \(\Sigma(F)\) over \(B\) as follows:
		\begin{enumerate}[leftmargin=*]
			\item define the unital \(\ast\)-algebra \(\Sigma(F)\) by equipping the complex vector space \(\bigoplus_{k \in \bZ} F(k)\) with the multiplication and \(\ast\)-operation defined, respectively, by
			\begin{align}
				\forall m,n \in \bZ, \, \forall p \in F(m), \, \forall q \in F(n), && p  q &\coloneqq F^{(2)}_{m,n}(p \otimes q),\\
				\forall m \in \bZ, \, \forall p \in F(m), && p^\ast &\coloneqq F^{(-1)}_m(\conj{p});
			\end{align}
			\item equip \(\Sigma(F)\) with the unique \Cstar-norm \(\norm{}_{\Sigma(F)}\), such that
			\begin{equation}\label{eq:cstarnorm}
				\forall k \in \bZ, \, \forall p \in F(k), \quad \norm{p}^2_{\Sigma(F)} = \norm{\hp{p}{p}};
			\end{equation}
			\item define a \(\Unit(1)\)-action of finite type \(\alpha\) on \(\Sigma(F)\) by
			\begin{equation}\label{eq:circleaction}
				\forall z \in \Unit(1), \, \forall m \in \bZ, \, \forall p \in F(m), \quad \alpha_z(p) \coloneqq z^m p;
			\end{equation}
			\item set \(\iota_{\Sigma(F)} \coloneqq (F^{(0)})^{-1}\).
		\end{enumerate}
		\item Given a \(2\)-isomorphism \(\eta : R \to S\), construct \(\Sigma(\eta) : \Sigma(R) \to \Sigma(S)\) by
		\begin{equation}
			\forall k \in \bZ, \, \forall p \in R(k), \quad \Sigma(\eta)(p) \coloneqq \eta_k(p).
		\end{equation}
	\end{enumerate}
	Hence, in particular, the category \(\Circ(B)\) is essentially small.
\end{theorem}

\begin{proof}
	We supply a proof that we can (and shall) adapt to other contexts.
	We first show that \(\Sigma\) is well-defined on objects.
	Let \(F : \bZ \to \Pic(B)\) be a given homomorphism, which is a bar functor by Theorem \ref{thm:barfunctor}.
	This now implies that \(\Sigma(F)\) is a unital \(\ast\)-algebra and that \(\iota_{\Sigma(F)}\) is a \(\ast\)-isomorphism.
	Indeed, coherence of \(F\) with respect to associators implies associativity of \(\Sigma(F)\), while coherence of \(F\) with respect to unitors implies that \(\Sigma(F)\) is unital and that \(\iota_{\Sigma(F)}\) is a unital homomorphism.
	Hence, commutativity of \eqref{eq:bar3}, \eqref{eq:bar2}, and \eqref{eq:bar1} implies that the \(\ast\)-operation is antimultiplicative, involutive, and unital, respectively, while commutativity of \eqref{eq:bar1} also implies that \(\iota_{\Sigma(F)}\) is a \(\ast\)-homomorphism.
	
	Now, recall from Example~\ref{ex:fell} that \(F\) canonically defines a pre-Fell bundle \(\mathcal{F}\) over \(\bZ\) in the sense of Exel~\cite[Def.\ 24.2]{exelbook}; it follows that \(\Sigma(F)\) is precisely the \(\ast\)-algebra of compactly supported cross-sections of \(\mathcal{F}\).
	Thus, by~\cite[Propp.\ 17.9.(iv) \& 19.8]{exelbook}, the \Cstar-norm on the reduced cross-sectional \Cstar-algebra~\cite[Def.\ 17.6]{exelbook} of the Fell bundle completion~\cite[Def.\ 24.7]{exelbook} of \(\mathcal{F}\) yields the unique \Cstar-norm on \(\Sigma(F)\) satisfying \eqref{eq:cstarnorm}; since \(F^{(0)}\) satisfies \eqref{eq:unitary}, this implies that \(\iota_{\Sigma(F)}\) is isometric.
	Finally, by \cite[Thm 3]{Raeburn}, it follows \eqref{eq:circleaction} defines a \(\Unit(1)\)-action of finite type on the unital pre-\Cstar-algebra \(\Sigma(F)\); that \((\Sigma(F);\alpha)\) defines a topological quantum principal \(\U\)-bundle over \(B\) now follows from the existence of cobases for \(F(1)\) and \(\conj{F(1)}\).
	
	Next, we show that \(\Sigma\) is well-defined on arrows.
	Let \(\eta : R \Rightarrow S\) be a \(2\)-isomorphism in \(\grp{Hom}(\bZ,\Pic(B))\), so that \(\eta\) is a bar natural transformation by Theorem~\ref{thm:barfunctor}.
	This now implies that the \(\Unit(1)\)-equivariant vector space isomorphism \(\Sigma(\eta) : \Sigma(R) \to \Sigma(S)\) is a unital \(\ast\)-isomorphism intertwining \(\iota_{\Sigma(R)}\) and \(\iota_{\Sigma(S)}\).
	Indeed, coherence of \(\eta\) with respect to \(R^{(2)}\) and \(S^{(2)}\) implies that \(\Sigma(\eta)\) is multiplicative, that \(\eta_1\) intertwines \(R^{(0)}\) and \(S^{(0)}\) implies that \(\Sigma(\eta)\) is unital and intertwines \(\iota_{\Sigma(R)}\) and \(\iota_{\Sigma(S)}\), and the fact that \(\eta\) is a bar natural transformation implies that \(\Sigma(\eta)\) is \(\ast\)-preserving.
	Since \(\eta_k\) and \(\eta_k^{-1}\) both satisfy \eqref{eq:unitary} for each \(k \in \bZ\), the bar natural transformation \(\eta\) induces a isomorphism of the Fell bundle completions of the pre-Fell bundles induced by \(R\) and \(S\) respectively, so that \(\Sigma(\eta)\) is isometric by \cite[Prop.\ 21.3]{exelbook}.
	
	Now, functoriality of \(\Sigma\) is easily checked, so it remains to construct natural isomorphisms \(\mu : \id_{\Circ(B)} \Rightarrow \Sigma \circ \cL\) and \(\nu : \id_{\grp{Hom}(\bZ,Pic(B))} \Rightarrow \cL \circ \Sigma\).
	On the one hand, let \(P\) be a topological quantum principal \(\U\)-bundle over \(B\).
	Since the \(\bZ\)-grading \(P = \bigoplus_{k\in\bZ}P_k\) is strong, the spectral subspaces of \(P\) define a pre-Fell bundle over \(\bZ\) fibrewise-isometrically isomorphic (\emph{mutatis mutandis}) over \(\iota_P^{-1}\) to the pre-Fell bundle over \(\bZ\) induced by \(\cL(P)\); note that this \(\bZ\)-grading is topological in the sense of Exel \cite[Def.\ 19.2]{Exel} by Proposition \ref{prop:conditional} and that averaging over the \(\U\)-action yields a faithful conditional expectation of the \Cstar-completion of \(P\) onto the \Cstar-completion of \(B\), cf.\ \cite[\S 4]{ADL}.
	Hence, by \cite[Prop.\ 21.3]{exelbook}, there exists unique \(\Unit(1)\)-equivariant isometric \(\ast\)-isomorphism \(\mu_P : P \to \Sigma \circ \cL(P)\) that satisfies \(\mu_P \circ \iota_P = \iota_{\Sigma \circ \cL(P)}\), namely, set \(\rest{\mu_P}{P_k} \coloneqq \id\) for \(k \in \bZ \setminus \Set{0}\) and \(\rest{\mu_P}{P_0} \coloneqq \iota_P^{-1}\).
	Naturality of \(\mu \coloneqq (\mu_P : P \to \Sigma\circ \cL(P))_{P \in \Obj(\Circ(B))}\) now follows by uniqueness.
	On the other hand, given monoidal \(F : \bZ \to \Pic(B)\), define \(\nu_F : F \Rightarrow \cL \circ \Sigma(F)\) as follows: for each \(k \in \bZ\), let \((\nu_F)_k : F(k) \to \left(\cL \circ \Sigma(F)\right)\!(k)\) be the inclusion of \(F(k)\) in \(\Sigma(F)\) as a direct sum of \(B\)-bimodules.
	Naturality of \(\nu \coloneqq (\nu_F : F \Rightarrow \cL \circ \Sigma(F))_{F \in \Obj(\grp{Hom}(\bZ,\Pic(B)))}\) follows from the fact that direct sums in \(\grp{Bimod}(B)\) are coproducts. 
\end{proof}

\begin{remark}
	Let \(T\) be a compact Abelian group with Pontrjagin dual \(\hat{T}\); suppose that \(B\) admits polar decompositions.
	The results above generalise to yield an equivalence of categories between \(\grp{Hom}(\hat{T},\Pic(B))\) and an analogous category of topological quantum principal \(T\)-bundles over \(B\), thereby recovering the relevant classification results of Schwieger--Wagner~\cite{SW1} in a manner that is adaptable to NC differential geometry.
\end{remark}

The construction of the natural isomorphism \(\mu : \id_{\Circ(B)} \Rightarrow \Sigma \circ \cL\) in the proof of Theorem \ref{thm:fell} implies the following useful characterisation of relevant \Cstar-norms.

\begin{corollary}[{Arici--Kaad--Landi~\cite[Thm.\ 3.10]{AKL}}]\label{cor:fell}
	Let \(P\) be a topological quantum principal \(\U\)-bundle on \(B\); let \(\norm{}\) denote its \Cstar-norm.
	Let \(\norm{}^\prime\) be a \(\U\)-invariant \Cstar-norm on \(P\).
	Then \(\norm{}^\prime = \norm{}\) if and only if \(\rest{\norm{}^\prime}{P^{\U}} = \rest{\norm{}}{P^{\U}}\).
\end{corollary}

Combining Theorem \ref{thm:fell} with Corollary \ref{cor:abstractpimsner} recovers Arici--Kaad--Landi's characterisation of topological quantum principal \(\U\)-bundles.

\begin{corollary}[{Arici--Kaad--Landi~\cite[\S 3]{AKL}; cf.\ Abadie--Eilers--Exel~\cite[Thm 3.1]{ABE}, Beggs--Brzezi\'{n}ski~\cite[Thm 7.3]{BeBrz14}}]
	The functor \(\epsilon_1 \circ \cL : \Circ(B) \to \Pic(B)\) is an equivalence of categories.
\end{corollary}

\begin{definition}[Abadie--Eilers--Exel~\cite{ABE}]
	Let \(E\) be a Hermitian line \(B\)-bimodule.
	The \emph{crossed product} of \(B\) by \(E\) is the essentially unique topological quantum principal \(\Unit(1)\)-bundle \(B \rtimes_E \bZ\) over \(B\), such that
	\(
		\cL(B \rtimes_E \bZ)(1) \cong E
	\).
\end{definition}

One may justify this terminology as follows.
Let \(\phi \in \Aut(B)\), so that its \emph{algebraic crossed product} \(B \rtimes_\phi^{\mathrm{alg}} \bZ\) is the unital \(\ast\)-algebra obtained from \(B\) by adjoining a unitary \(U\) satisfying \(UbU^\ast = \phi(b)\) for all \(b \in B\).
Then \(B \rtimes_\phi^{\mathrm{alg}} \bZ\) defines topological quantum \(\U\)-bundle over \(B\) when equipped with the reduced crossed product \Cstar-norm and the unique \(\U\)-action of finite type \(\alpha\), such that \(\rest{\alpha_z}{B} = \id\) and \(\alpha_z(U) = zU\) for all \(z \in \U\).
Since \(\left(b_\phi \mapsto U \phi^{-1}(b)\right) : B_\phi \to \cL(B \rtimes_\phi \bZ)(1)\) is an isomorphism of Hermitian line \(B\)-bimodules, we may take
\(
	B \rtimes_{\tau(\phi)} \bZ \coloneqq B \rtimes_\phi^{\mathrm{alg}}\bZ
\).

\subsection{Horizontal calculi as generalised crossed products}

As promised, we now adapt the considerations of the last subsection to the setting of NC differential geometry by replacing the Picard \(2\)-group with the differential Picard \(2\)-group.
However, in the absence of additional constraints, we can only reconstruct the \emph{horizontal calculus} of a quantum principal \(\U\)-bundle.
In what follows, let \(B\) be a given unital pre-\Cstar-algebra with \(\ast\)-exterior algebra \((\Omega_B,\du_B)\).

Let \(P\) be a \(\U\)-pre-\Cstar-algebra of finite type with \(\U\)-action \(\alpha\).
We define a \emph{\(\U\)-\(\ast\)-quasi-\textsc{dga} of finite type} over \(P\) to be a \(\ast\)-quasi-\textsc{dga} \((\Omega,\du)\) over \(P\) together with a pointwise extension of \(\alpha\) to a group homomorphism \(\hat{\alpha} : \U \to \Aut(\Omega,\du)\), such that, for each \(k \in \bN_0\), the restriction of \(\hat{\alpha}\) to a \(U\)-action on the complex vector space \(\Omega^k\) is of finite type.
In this case, we call \((\Omega,\du)\) a \emph{\(\U\)-\(\ast\)-exterior algebra of finite type} over \(P\) whenever the underlying \(\ast\)-quasi-\textsc{dga} is a \(\ast\)-exterior algebra.
At last, we denote by \(\grp{QDGA}^{\U}\) the concrete category whose objects \((P;\Omega,\du)\) consist of a \(\U\)-pre-\Cstar-algebra of finite type \(P\) together with a \(\U\)-\(\ast\)-quasi-\textsc{dga} of finite type \((\Omega,\du)\) over \(P\) and whose arrows \(f : (P_1;\Omega_1,\du_1) \to (P_2;\Omega_2,\du_2)\) are \(\U\)-equivariant morphisms of \(\ast\)-quasi-\textsc{dga}.

The following definition characterises the differentiable structure that a Hermitian line \(B\)-bimodule with connection can generally induce on the corresponding topological quantum principial \(\U\)-bundle over \(B\).

\begin{definition}[{\DJ{}ur\dj{}evi\'{c}~\cite[\S 2]{Dj98}, cf.\ \'{C}a\'{c}i\'{c}~\cite{Cacic}}]
	Let \(P\) be a topological quantum principal \(\U\)-bundle over \(B\).
	A \emph{horizontal calculus} for \(P\) is a \(\Unit(1)\)-\(\ast\)-quasi-\textsc{dga} \((\Omega_{P,\hor},\du_{P,\hor})\) of finite type over \(P\) together with an isomorphism of quasi-\(\ast\)-\textsc{dga}
	\(
		\hat{\iota}_P : (B;\Omega_B,\du) \to (P^{\U},(\Omega_{P,\hor})^{\Unit(1)},\rest{\du_{P,\hor}}{(\Omega_{P,\hor})^{\Unit(1)}})
	\)
	extending the isometric \(\ast\)-isomorphism \(\iota_P : B \to P^{\U}\), such that
	\(
		\Omega_{P,\hor} = P \cdot (\Omega_{P,\hor})^{\Unit(1)} \cdot P
	\).
\end{definition}

\begin{example}[{Majid~\cite[\S 3]{Majid05}}]\label{ex:hopf2}
	We continue from Example \ref{ex:hopf1}.
	Let \(\Omega_{q,\hor}(\SU(2))\) be the graded \(\ast\)-algebra over \(\cO_q(\SU(2))\) generated by \(e^+ \in \Omega^1_{q,\hor}(\SU(2))\) and \(e^- \coloneqq -(e^+)^\ast\) subject to the relations
	\begin{gather*}
		e^\pm  a = q^{-1} a  e^\pm, \quad e^\pm  a^\ast = q a^\ast  e^\pm, \quad e^\pm  c = q^{-1} c  e^\pm, \quad e^\pm  c^\ast = q c^\ast  e^\pm,\\
		(e^\pm)^2 = 0, \quad e^-  e^+ + q^{-2} e^+  e^- = 0.
	\end{gather*}
	Define complex-linear maps \(\partial_\pm : \cO_q(\SU(2)) \to \cO_q(\SU(2)\) by
	\begin{align*}
		\partial_+(a) &\coloneqq -qc^\ast, && \partial_+(a^\ast) \coloneqq 0, && \partial_+(c) \coloneqq a^\ast, && \partial_+(c^\ast) \coloneqq 0,\\
		\partial_-(a) &\coloneqq 0, && \partial_-(a^\ast) \coloneqq c, && \partial_-(c) \coloneqq 0, && \partial_-(c^\ast) \coloneqq -q^{-1} a,
	\end{align*}
	together with the twisted Leibniz rule
	\[
		\forall x \in \cO_q(\SU(2)), \, \forall j \in \bZ, \, \forall y \in \cO_q(\SU(2))_j, \quad \partial_\pm(xy) = \partial_{\pm}xyq^{-j} + x\partial_\pm(y);
	\]
	hence, define \(\du_{q,\hor} : \cO_q(\SU(2)) \to \Omega^1_{q,\hor}(\SU(2))\) by setting
	\[
		\forall p \in \cO_q(\SU(2)), \quad \du_{q,\hor}(p) \coloneqq \partial_+(p)  e^+ + \partial_-(p)  e^-,
	\]
	and extend \(\du_{q,\hor}\) to \(\Omega_{q,\hor}(\SU(2))\) by setting \(\du_{q,\hor}(e^\pm) \coloneqq 0\).
	Finally, extend the \(\Unit(1)\)-action from \(\cO_q(\SU(2))\) to \(\Omega_{q,\hor}(\SU(2))\) by setting \(\alpha_z(e^{\pm}) = z^{\pm 2} e^{\pm}\) for all \(z \in \U\).
	Then \((\Omega_{q,\hor}(\SU(2)),\du_{P,\hor})\) defines a horizontal calculus for the topological quantum principal \(\Unit(1)\)-bundle \(\cO_q(\SU(2))\) over \(\cO_q(\CP^1)\) with respect to
	\(
		(\Omega_q(\CP^1),\du) \coloneqq \left(\Omega_{q,\hor}(\SU(2))^{\Unit(1)}, \rest{\du_{q,\hor}}{\Omega_{q,\hor}(\SU(2))^{\Unit(1)}}\right)
	\),
	which, by Majid's result, recovers the \(2\)-dimensional calculus on \(\cO_q(\CP^1)\) first constructed by Podle\'{s}~\cite{Podles92}. 
\end{example}

We now define the concrete category \(\dCirc_\hor(B)\) of \emph{horizontally differentiable quantum principal \(\U\)-bundles over \(B\)} and their isomorphisms as follows:
\begin{enumerate}[leftmargin=*]
	\item an object \((P;\Omega_{P,\hor},\du_{P,\hor})\) consists of a topological quantum principal \(\U\)-bundle \(P\) over \(B\) together with a horizontal calculus \((\Omega_{P,\hor},\du_{P,\hor})\) on \(P\);
	\item an arrow \(f : (P;\Omega_{P,\hor},\du_{P,\hor}) \to (Q;\Omega_{Q,\hor},\du_{Q,\hor})\) is an isomorphism of \(\U\)-\(\ast\)-quasi-\textsc{dga}, such that \(\hat{\iota}_Q \circ f = f \circ \hat{\iota}_P\).
\end{enumerate}
It is useful to observe that the forgetful functor \(\dCirc_\hor(B) \to \Circ(B)\) is faithful: an arrow \(f : (P;\Omega_{P,\hor},\du_{P,\hor}) \to (Q;\Omega_{Q,\hor},\du_{Q,\hor})\) in \(\dCirc(B)\) is uniquely determined by the corresponding arrow \(\rest{f}{P} : P \to Q\) in \(\Circ(B)\) precisely because \(\Omega_{P,\hor}\) is generated as an algebra by \(P\) and \(\hat{\iota}_P(\du(B)) \subset \du_{P,\hor}(P)\).
We can now make precise sense of associated line bundles with connection in the NC setting.

\begin{proposition}[{cf.\ \'{C}a\'{c}i\'{c}--Mesland~\cite[Appx B]{CaMe}}]\label{prop:assoclineconn}
	Let \((P,\Omega_{P,\hor},\du_{P,\hor})\) be a horizontally differentiable quantum principal \(\U\)-bundle over \(B\).
	\begin{enumerate}[leftmargin=*]
		\item Observe that \(\Omega_{P,\hor}\) defines a \(B\)-bimodule with respect to \(\iota_P : B \to P^{\Unit(1)}\).
		There exists a unique \(\Unit(1)\)-equivariant isomorphism \(\hat{\ell}_P : \Omega_{P,\hor} \to P \otimes_B \Omega_B\) of \(B\)-bimodules, such that
			\begin{equation}\label{eq:stronguniversal}
				\forall p \in P, \, \forall \beta \in \Omega_B, \quad \hat{\ell}_P^{-1}(p \otimes \beta) = p  \hat{\iota}_P(\beta).
			\end{equation}
		\item Let \(k \in \bZ\) be given.
			Define functions \(\sigma_{P;k} : \Omega_B \otimes_B \cL(P)(k) \to \cL(P)(k) \otimes_B \Omega_B\) and \(\nabla_{P;k} : \cL(P)(k) \to \cL(P) \otimes_B \Omega^1_B\) by
			\begin{align}
				\forall \beta \in \Omega_B, \, \forall p \in P_k, && \sigma_{P;k}(\beta \otimes p) \coloneqq \hat{\ell}_P\mleft(\hat{\iota}_P(\beta)  p\mright),\\
				\forall p \in P_k, && \nabla_{P;k}(p) \coloneqq \hat{\ell}_P\mleft(\du_{P,\hor}(p)\mright),
			\end{align}
			respectively.
			Then \((\sigma_{P;k},\nabla_{P;k})\) defines a Hermitian bimodule connection on the Hermitian line \(B\)-bimodule \(\cL(P)(k)\).
	\end{enumerate}
\end{proposition}

\begin{proof}
	We first show that \(\hat{\ell}_P\) is well-defined; uniqueness and \(\Unit(1)\)-equivariance will then follow by construction.
	Given \(k \in \bZ\), define \(\hat{\ell}_{P;k} : (\Omega_{P,\hor})_k \to \cL(P)(k) \otimes_B \Omega_B\) by \(\hat{\ell}_{P;k} \coloneqq \left(\omega \mapsto \sum_{i=1}^m e_i \otimes \hat{\iota}_P^{-1}(e_i^\ast  \omega)\right)\), where \((e_i)_{i=1}^m\) be a basis for \(\cL(P)(k)\); that \(\hat{\ell}_{P;k}\) is an isomorphism of \(B\)-bimodules with inverse given by \eqref{eq:stronguniversal} now follows from observing that \((e_i)_{i=1}^m\) satisfies
	\(
		1 = \iota_P\mleft(\sum_{i=1}^m \hp{\conj{e_i}}{\conj{e_i}}\mright) = \sum_{i=1}^m e_i  e_i^\ast
	\).
	We may now set \(\hat{\ell}_P \coloneqq \bigoplus_{k \in \bZ} \hat{\ell}_{P;k}\).
	
	We now fix \(k \in \bZ\) and show that \((\sigma_{P;k},\nabla_{P;k})\) defines a Hermitian bimodule connection on the Hermitian line \(B\)-bimodule \(\cL(P)(k)\).
	Let \((e_i)_{i=1}^m\) be a basis and let \((\epsilon_j)_{j=1}^n\) be a strict cobasis for \(\cL(P)(k)\).
	Recall that \(\sum_{i=1}^m e_i  e_i^\ast = 1\) and observe that \(\sum_{j=1}^n \epsilon_j^\ast \epsilon_j = \iota_P\mleft(\sum_{j=1}^n \hp{\epsilon_j}{\epsilon_j}\mright) = 1\).
	On the one hand, the fact that \(\sum_{j=1}^n \epsilon_j^\ast  \epsilon_j = 1\) implies that \(\sigma_{P;k}\) is indeed an isomorphism of graded \(B\)-bimodules with inverse \(\sigma_{P;k}^{-1} = \left(p \otimes \beta \mapsto \sum_{j=1}^n \hat{\iota}_P^{-1}\mleft(p  \beta  \epsilon_j^\ast\mright) \otimes \epsilon_j\right)\).
	On the other hand, the fact that \(\sum_{i=1} e_i e_i^\ast = 1\) implies, that for all \(\alpha,\beta \in \Omega_B\) and \(p \in P_k\),
	\[
		\hat{\ell}_P^{-1} \circ \sigma_{P;k}(\alpha  \beta \otimes p)
		= \hat{\iota}_P(\alpha  \beta)  p
		= \hat{\ell}_P^{-1}\mleft(\sigma_{P;k}(\alpha \otimes \leg{\sigma_{P;k}(\beta \otimes p)}{0})  \leg{\sigma_{P;k}(\beta \otimes p)}{1}\mright),
	\]
	which yields \eqref{eq:braidassoc}.
	Thus, \(\sigma_{P;k}\) defines a Hermitian generalised braiding; it remains to show that \(\nabla_{P;k}\) is a right Hermitian connection satisfying \eqref{eq:leftleibniz} with respect to \(\sigma_{P;k}\).
	However, we may again use the maps \(\hat{\ell}_P\) and \(\hat{\iota}_P\) together with the equality \(\sum_{i=1}^m e_i e_i^\ast = 1\) to derive \eqref{eq:rightleibniz}, \eqref{eq:hermconn}, and \eqref{eq:leftleibniz} from the Leibniz rule for \(\du_{P,\hor}\).
\end{proof}

\begin{proposition}[{cf.\ Beggs--Majid~\cite[Prop.\ 5.56]{BeMa18}, Salda\~{n}a~\cite[\S 3]{Saldana20}}]\label{prop:assoclineconnfunct}
	The functor \(\cL\) of Proposition \ref{prop:assocline} lifts to the functor \(\hat{\cL} : \grp{DCirc}_{\hor}(B) \to \grp{Hom}(\bZ,\dPic(B))\) defined as follows. 
	\begin{enumerate}[leftmargin=*]
		\item Let \((P,\Omega_{P,\hor},\du_{P,\hor})\) be a horizontally differentiable quantum principal \(\U\)-bundle over \(B\).
		Define \(\hat{\cL}(P,\Omega_{P,\hor},\du_{P,\hor}) : \bZ \to \dPic(B)\) as follows:
		\begin{enumerate}[leftmargin=*]
			\item given \(k \in \bZ\), let \(\hat{\cL}(P,\Omega_{P,\hor},\du_{P,\hor})(k) \coloneqq (\cL(P)(k),\sigma_{P,k},\nabla_{P,k})\), where \((\sigma_{P;k},\nabla_{P;k})\) is the Hermitian bimodule connection of Proposition \ref{prop:assoclineconn};
			\item let \(\hat{\cL}(P,\Omega_{P,\hor},\du_{P,\hor})^{(0)}\) be the unique lift of \(\id_{P_0} \eqqcolon \cL(P)^{(0)}\);
			\item given \(m,n \in \bZ\), let \(\hat{\cL}(P,\Omega_{P,\hor},\du_{P,\hor})^{(2)}_{m,n}\) be the unique lift of \(\cL(P)^{(2)}_{m,n}\).
		\end{enumerate}
		\item Given an isomorphism \(f : (P,\Omega_{P,\hor},\du_{P,\hor}) \to (Q,\Omega_{Q,\hor},\du_{Q,\hor})\) of horizontally differentiable quantum principal \(\U\)-bundles over \(B\), let
		\[\hat{\cL}(f) : \hat{\cL}(P,\Omega_{P,\hor},\du_{P,\hor}) \Rightarrow \hat{\cL}(Q,\Omega_{Q,\hor},\du_{Q,\hor})\] be the unique lift of the \(2\)-isomorphism \(\cL(f) : \cL(P) \Rightarrow \cL(Q)\).
	\end{enumerate}
\end{proposition}

\begin{proof}
	First, let \((P,\Omega_{P,\hor},\du_{P,\hor})\) be a horizontally differentiable quantum principal \(\U\)-bundle over \(B\).
	For notational simplicity, set \(F \coloneqq \cL(P)\) and denote our would-be homomorphism \(\hat{\cL}(P,\Omega_{P,\hor},\du_{P,\hor})\) by \(\hat{F}\).
	The functor \(\hat{F} : \bZ \to \dPic(B)\) is well defined by Proposition \ref{prop:assoclineconn}; that the arrow \(F^{(0)} : F(0) \to B\) satisfies \eqref{eq:intertwineconn} follows from the fact that \(\hat{\iota}_P \circ \du = \du_{P,\hor}\).
	Given \(m,n \in \bZ\), the arrow \(F^{(2)}_{m,n} : F(m) \otimes_B F(n) \to F(m+n)\) satisfies \eqref{eq:intertwineconn} by applying the isomorphism \(\hat{\ell}_P^{-1}\) of Proposition \ref{prop:assoclineconn} to both sides of the desired equality and then applying the Leibniz rule for \(\du_{P,\hor}\) in \(\Omega_{P,\hor}\); thus, the natural isomorphism \(\hat{F}^{(2)}\) is well defined.
	Commutativity of the relevant commutative diagrams now follows from observing that the forgetful functor \(\dPic(B) \to \Pic(B)\) is faithful.
	
	Now, let \(f : (P,\Omega_{P,\hor},\du_{P,\hor}) \to (Q,\Omega_{Q,\hor},\du_{Q,\hor})\) be an isomorphism of horizontally differentiable quantum principal \(\U\)-bundle over \(B\).
	Again, for notational simplicity, set \(R \coloneqq \cL(P)\), \(\hat{R} \coloneqq \hat{\cL}(P,\Omega_{P,\hor},\du_{P,\hor})\), \(S \coloneqq \cL(Q)\), and \(\hat{S} \coloneqq \hat{\cL}(Q,\Omega_{Q,\hor},\du_{Q,\hor})\).
	Observe that \(f \otimes \id_{\Omega_B}\) necessarily intertwines the isomorphisms \(\hat{\ell}_P\) and \(\hat{\ell}_Q\) of Proposition \ref{prop:assoclineconn}, so that for each \(k \in \bZ\), the arrow \(\cL(f)_k : R(k) \to S(k)\) in \(\Pic(B)\) satisfies \eqref{eq:intertwineconn} precisely since \(\du_{Q,\hor} \circ f = f \circ \du_{P,\hor}\); it follows that \(\hat{\cL}(f) : \hat{R} \to \hat{S}\) is well defined as a natural transformation.
	Once more, commutativity of the relevant commutative diagrams now follows from observing that the forgetful functor \(\dPic(B) \to \Pic(B)\) is faithful.
\end{proof}

\begin{example}[{Landi--Reina--Zampini~\cite{LRZ}, Khalkhali--Landi--Van Suijlekom~\cite{KLVS}}]\label{ex:hopf3}
	We continue from Example \ref{ex:hopf2}; in particular, we now equip \(\cO_q(\CP^1)\) with Podle\'{s}'s \(2\)-dimensional calculus \((\Omega_q(\CP^1),\du)\).
	The homomorphism \(\cE \coloneqq \cL(\cO_q(\SU(2)))\) lifts to \(\hat{\cE} : \bZ \to \dPic(\cO_q(\CP^2))\) by setting
	\(
		\hat{\cE} \coloneqq \hat{\cL}(\cO_q(\SU(2)),\Omega_{q,\hor}(\SU(2)), \du_{q,\hor})
	\).
	In fact, given \(k \in \bZ\), it follows that \(\hat{\cE}(k) = (\cE(k),\sigma_{k},\nabla_{k})\), where \(\nabla_k\) and \(\sigma_k\) respectively recover the canonical connection \cite[\S 4.1]{LRZ} and `twisted flip' \cite[\S\S 3.5-6]{KLVS} on \(\cE(k)\).
\end{example}

At last, we show that the functor \(\hat{\cL}\) is, indeed, an equivalence of categories.

\begin{proposition}\label{prop:horizontalconstruction}
	Let \(\hat{F} : \bZ \to \dPic(B)\) be a homomorphism, let \(F : \bZ \to \Pic(B)\) be its image under the forgetful functor \(\grp{Hom}(\bZ,\dPic(B)) \to \grp{Hom}(\bZ,\Pic(B))\), and let \(P \coloneqq \Sigma(F)\).
	The following defines a horizontal calculus \((\Omega_{P,\hor},\du_{P,\hor})\) on \(P\):
	\begin{enumerate}[leftmargin=*]
		\item define the graded \(\ast\)-algebra \(\Omega_{P,\hor}\) by equipping the complex vector space \(P \otimes_B \Omega_B\) with the multiplication and \(\ast\)-operation defined, respectively, by
		\begin{align}
			\forall \alpha,\beta \in \Omega_B, \, \forall p \in \bZ, \,\forall k \in \bZ, \, \forall q \in F(n), && (p \otimes \alpha)  (q \otimes \beta) &\coloneqq p  \sigma_{F(k)}(\alpha \otimes q)  \beta,\\
			\forall \alpha \in \Omega_B, \, \forall k \in \bZ, \, \forall p \in F(k), && (p \otimes \alpha)^\ast &\coloneqq \sigma_{F(-k)}(\alpha^\ast \otimes p^\ast),
		\end{align}
		and with the grading induced by the grading on \(\Omega_B\);
		\item define \(\du_{P,\hor} : \Omega_{P,\hor} \to \Omega_{P,\hor}\) by
		\begin{equation}
			\forall k \in \bZ, \, \forall p \in F(k), \, \forall \beta \in \Omega_B, \quad \du_{P,\hor}(p \otimes \beta) \coloneqq \nabla_{F(k)}(p) \otimes \beta + p \otimes \du \beta;
		\end{equation}
		\item extend the \(\Unit(1)\)-action \(\alpha\) on \(P\) pointwise to \(\hat{\alpha} : \Unit(1) \to \Aut(\Omega_{P,\hor})\) by
		\begin{equation}
			\forall z \in \Unit(1), \, \forall p \in P, \, \forall \beta \in \Omega_B, \quad \hat{\alpha}_z(p \otimes \beta) \coloneqq \alpha_z(p) \otimes \beta;
		\end{equation}
		\item let \(\hat{\iota}_P : (\Omega_B,\du) \to (\Omega_{P,\hor}^{\Unit(1)},\rest{\du_{P,\hor}}{\Omega_{P,\hor}^{\Unit(1)}})\) be induced by \((F^{(0)} \otimes \id) \circ \lambda_{\Omega_B}\).
	\end{enumerate}
\end{proposition}

\begin{proof}
	The construction of \(\Omega_{P,\hor}\), \(\hat{\alpha}\), and \(\hat{\iota}_P\) from \(\hat{F}\) follows, \emph{mutatis mutandis}, from the construction of \(P \coloneqq \Sigma(F)\), \(\alpha\), and \(\iota_P\) from \(F\) in the proof of Theorem \ref{thm:fell}.
	Indeed, recall that \(\hat{F}\) canonically defines a bar functor by Theorem \ref{thm:barfunctor}.
	Hence, each definitional commutative diagram satisfied by the bar functor \(\hat{F}\) yields a corresponding commutative diagram satisfied by the family of Hermitian generalised braidings \((\sigma_{F(k)})_{k \in \bZ}\), which, in turn, yields the corresponding properties of \(\Omega_{P,\hor}\).
	
	We now turn to \(\du_{P,\hor}\), which, by construction, is \(\Unit(1)\)-equivariant, is \(\bC\)-linear, and satisfies \(\du_{P,\hor} \circ \hat{\iota}_P = \hat{\iota}_P \circ \du\).
	Given \(m,n \in \bZ\), the fact that \(F^{(2)}_{m,n}\) satisfies \eqref{eq:intertwineconn} implies that \(\du_{P,\hor}\) satisfies the Leibniz rule on \((\Omega_{P,\hor})_m \cdot (\Omega_{P,\hor})_n = (\Omega_{P,\hor})_{m+n}\).
	Finally, given \(k \in \bZ\), the fact that \(F^{(-1)}_k\) satisfies \eqref{eq:intertwineconn} implies that \(\du_{P,\hor}\) is \(\ast\)-preserving on \(\ast\mleft((\Omega_{P,\hor})_k\mright) = (\Omega_{P,\hor})_{-k}\).
\end{proof}

\begin{theorem}\label{thm:horizontal}
	The functor \(\Sigma : \grp{Hom}(\bZ,\Pic(B)) \to \Circ(B)\) of Theorem \ref{thm:fell} lifts to the weak inverse \(\hat{\Sigma} :  \grp{Hom}(\bZ,\dPic(B)) \to \grp{DCirc}_{\hor}(B)\) of \(\hat{\cL}\) defined as follows.
	\begin{enumerate}[leftmargin=*]
		\item Given a homomorphism \(\hat{F}: \bZ \to \dPic(B)\) descending to \(F : \bZ \to \Pic(B)\), let \[\hat{\Sigma}(\hat{F}) \coloneqq (\Sigma(F),\Omega_{\Sigma(F),\hor},\du_{\Sigma(F),\hor}),\] where \((\Omega_{\Sigma(F),\hor},\du_{\Sigma(F),\hor})\) is the horizontal calculus of Proposition \ref{prop:horizontalconstruction}.
		\item Given a \(2\)-isomorphism \(\hat{\eta} : \hat{R} \Rightarrow \hat{S}\) descending to a \(2\)-isomorphism \(\eta : R \Rightarrow S\), let
		\(
			\hat{\Sigma}(\hat{\eta}) : \hat{\Sigma}(\hat{R}) \to \hat{\Sigma}(\hat{S})
		\)
		be the unique lift of \(\Sigma(\eta) : \Sigma(R) \to \Sigma(S)\).
	\end{enumerate}
	Hence, in particular, the category \(\grp{DCirc}_{\hor}(B)\) is essentially small.
\end{theorem}

\begin{proof}
	We have seen that \(\Sigma\) is well defined on objects, so let us check that it is well defined on arrows.
	Let \(\hat{\eta} : \hat{R} \Rightarrow \hat{S}\) be an arrow in \(\grp{Hom}(\bZ,\dPic(B))\) descending to \(\eta : R \Rightarrow S\) in \(\grp{Hom}(\bZ,\Pic(B))\), so that \(\hat{\eta}\) and \(\eta\) define bar natural transformations by Theorem \ref{thm:barfunctor}.
	We can extend \(\Sigma(\eta) : \Sigma(R) \to \Sigma(S)\) to a \(\Unit(1)\)-equivariant isomorphism of graded \(B\)-bimodules \(\hat{\Sigma}(\hat{\eta}) : \Omega_{\Sigma(R),\hor} \to \Omega_{\Sigma(S),\hor}\) by setting \(\hat{\Sigma} \coloneqq \Sigma(\eta) \otimes \id_{\Omega_B}\).
	Coherence of \(\eta\) with respect to \(R^{(2)}\) and \(S^{(2)}\) implies that \(\hat{\Sigma}(\hat{\eta})\) is multiplicative, that \(\eta_1\) intertwines \(R^{(0)}\) and \(S^{(0)}\) implies that \(\hat{\Sigma}(\hat{\eta})\) is unital, and the fact that \(\hat{\eta}\) is a bar functor implies that \(\hat{\Sigma}(\eta)\) is \(\ast\)-preserving.
	Finally, given \(k \in \bZ\), the fact that \(\eta(k)\) satisfies \eqref{eq:intertwineconn} implies that \(\hat{\Sigma}(\eta)\) satisfies \(\hat{\Sigma}(\eta) \circ \du_{\Sigma(R),\hor} = \du_{\Sigma(S),\hor} \circ \hat{\Sigma}(\eta)\) on \((\Omega_{\Sigma(R),\hor})_k\).
	The rest now follows from Theorem \ref{thm:fell}, \emph{mutatis mutandis}.
\end{proof}

\begin{remark}
	Building on a proposal of~\DJ{}ur\dj{}evi\'{c} \cite[\S 4.4]{Dj96}, Salda\~{n}a proves analogues of Proposition \ref{prop:horizontalconstruction}~\cite[Thm 3.11]{Saldana20} and Theorem \ref{thm:horizontal}~\cite[Thm 3.12]{Saldana20} for quantum principal bundles with structure quantum group given by a Hopf \(\ast\)-algebra in terms of certain heavily structured functors.
	By contrast, in the special case of quantum principal \(\U\)-bundles, Theorem \ref{thm:barfunctor} allows us to use monoidal functors \emph{simpliciter}.
	Indeed, after suitable generalisation, the same will still be true in the more general case where the structure quantum group is a group ring.
\end{remark}

By combining Theorem \ref{thm:horizontal} with Corollary \ref{cor:abstractpimsner}, we obtain the differentiable analogue of Arici--Kaad--Landi's characterisation of topological quantum principal \(\U\)-bundles---and hence a differentiable analogue of Pimsner's construction---in the absence of any further constraints.

\begin{corollary}\label{cor:cphor}
	The functor \(\epsilon_1 \circ \hat{\cL} : \dCirc_\hor(B) \to \dPic(B)\) is an equivalence.
\end{corollary}

\begin{definition}
	The \emph{horizontal crossed product} of \((B;\Omega_B,\du)\) by a Hermitian line \(B\)-bimodule with connection \((E,\sigma_E,\nabla_E)\) is the essentially unique horizontally differentiable quantum principal \(\U\)-bundle
	\(
		(B; \Omega_B,\du) \rtimes^{\hor}_{(E,\sigma_E,\nabla_E)} \bZ
	\)
	over \((B;\Omega_B,\du)\), such that
	\(
		\hat{\cL}\mleft((B; \Omega_B,\du) \rtimes_{(E,\sigma_E,\nabla_E)}^{\hor} \bZ\mright)(1) \cong (E,\sigma_E,\nabla_E)
	\).
\end{definition}

One may justify this terminology as follows.
Let \((\omega,\phi) \in \tDiff(B)\), so that \(B \rtimes_\phi^{\mathrm{alg}} \bZ\) admits the horizontal calculus \((\Omega_B \rtimes_\phi^{\mathrm{alg}} \bZ,\du_{(\omega,\phi)})\), where the graded \(\ast\)-algebra \(\Omega_B \rtimes_\phi \bZ\) is obtained from \(\Omega_B\) by adjoining a unitary \(U \in (\Omega_B \rtimes_\phi \bZ)^0\) that satisfies \(U_\phi \beta U_\phi^{-1} = \phi(\beta)\) for all \(\beta \in \Omega_B\), the \(\ast\)-derivation \(\du_{(\omega,\phi)}\) is determined by requiring \(\rest{\du_{(\omega,\phi)}}{\Omega_B} \coloneqq \du_B\) and \(\du_{(\omega,\phi)}(U_\phi) \coloneqq \iu{}\omega  U_\phi\), and the \(\Unit(1)\)-action \(\hat{\alpha}\) on \(\Omega_B \rtimes_\phi^{\mathrm{alg}} \bZ\) is determined by \(\rest{\hat{\alpha}_z}{\Omega_B} = \id_{\Omega_B}\) and \(\alpha_z(U_\phi) \coloneqq zU_\phi\) for all \(z \in \U\).
Since \((b_\phi \mapsto U\phi^{-1}(b)) : \hat{\tau}(\omega,\phi) \to \hat{\cL}(B \rtimes^{\mathrm{alg}}_\phi \bZ;\Omega_B \rtimes_\phi^{\mathrm{alg}} \bZ,\du_{(\omega,\phi)})(1)\) is an isomorphism in \(\dPic(B)\), we may therefore take
\(
	(B; \Omega_B,\du) \rtimes^\hor_{\hat{\tau}(\omega,\phi)} \bZ \coloneqq (B \rtimes^{\mathrm{alg}}_\phi \bZ;\Omega_B \rtimes_\phi^{\mathrm{alg}} \bZ,\du_{(\omega,\phi)})
\).

We conclude this subsection by discussing curvature.
In general, the \emph{curvature} of a \(\ast\)-quasi-\textsc{dga} \((\Omega,\du)\) is the map \(\du^2\), which vanishes for a \(\ast\)-exterior algebra.
Thus, the curvature (in this sense) of a horizontally differentiable quantum principal \(\U\)-bundle \((P,\Omega_{P,\hor},\du_{P,\hor})\) over \(B\) is the map \(\du_{P,\hor}^2\), which is a \(\U\)-equivariant \(\ast\)-derivation that vanishes on \(\Omega_B\) and hence, in particular, is left and right \(\Omega_B\)-linear.
Passing this notion of curvature through the lens of Proposition \ref{prop:assoclineconnfunct} and Theorem \ref{thm:horizontal} yields the following more refined definition.

\begin{propositiondefinition}[{cf.\ \DJ{}ur\dj{}evi\'{c}~\cite[Lemma 2.2]{Dj98}}]\label{propdef:vertical}
	Let \((P,\Omega_{P,\hor},\du_{P,\hor})\) be a horizontally differentiable quantum principal \(\U\)-bundle over \(B\).
	\begin{enumerate}[leftmargin=*]
		\item Its \emph{Fr\"{o}hlich automorphism} is the unique \(\U\)-equivariant automorphism \(\hat{\Phi}_P\) of the \(\U\)-\(\ast\)-quasi-\textsc{dga} of finite type \((\Zent(\Omega_B),\rest{\du}{\Zent(\Omega_B)})\), such that
			\begin{equation}
				\forall k \in \bZ, \, \forall p \in P_k, \, \forall \beta \in \Zent(\Omega_B), \quad \hat{\iota}_P\mleft(\hat{\Phi}_P^k(\beta)\mright)  p = p  \hat{\iota}_P(\beta).
			\end{equation}
		\item Its \emph{curvature \(1\)-cocycle} is the unique group \(1\)-cocycle \(\bF_P : \bZ \to \cS(B)\) for the right \(\bZ\)-action generated by \(\hat{\Phi}_P^{-1}\), such that
			\begin{equation}
				\forall k \in \bZ, \, \forall p \in P_k, \quad \du_{P,\hor}^2(p) = p \cdot \hat{\iota}_P\mleft(\iu{}\bF_P(k)\mright).
			\end{equation}
	\end{enumerate}
	Hence, its \emph{curvature data} is the pair \((\Phi_P,\bF_P)\).
\end{propositiondefinition}

\begin{proof}
	By Proposition \ref{prop:assoclineconnfunct} together with Proposition-Definition \ref{propdef:curve}, we can and must take
	\(
		\hat{\Phi}_P \coloneqq \Phi[\hat{\cL}(P,\Omega_{P,\hor},\mathrm{d}_{P,\hor})](1)\) and \(\bF_P \coloneqq \bF[\hat{\cL}(P,\Omega_{P,\hor},\du_{P,\hor})]
	\).
\end{proof}

Suppose that \((P,\Omega_{P,\hor},\du_{P,\hor})\) is a horizontally differentiable quantum principal \(\U\)-bundle over \(B\) with curvature data \((\Phi_P,\bF_P)\).
On the one hand, by Theorem \ref{thm:horizontal}, every homomorphism \(\hat{F} : \bZ \to \dPic(B)\) that is \(2\)-isomorphic to \(\hat{\cL}(P,\Omega_{P,\hor},\du_{P,\hor})\) satisfies
\(
	\hat{\Phi} \circ \pi_0(\hat{F})(1) = \hat{\Phi}_P\) and \(\bF \circ \pi_0(\hat{F}(1)) = \bF_P
\).
On the other hand, by Corollary \ref{cor:cphor}, every Hermitian line \(B\)-bimodule \((E,\sigma_E,\nabla_E)\) that is isomorphic to \(\hat{\cL}(P,\Omega_{P,\hor},\du_{P,\hor})(1)\) satisfies 
\(
	\hat{\Phi}_{[E,\nabla_E]} = \hat{\Phi}_P\) and \(\bF_{[E,\nabla_E]} = \bF_P(1)
\);
in other words, for every Hermitian line \(B\)-bimodule \((E,\sigma_E,\nabla_E)\), the resulting horizontal crossed product \((B,\Omega_B,\du) \rtimes_{(E,\sigma_E,\nabla_E)} \bZ\) has curvature data \((\Phi_{[E,\nabla_E]},\bF_{[E,\nabla_E]})\).

\begin{example}[{Landi--Reina--Zampini~\cite[Prop.\ 4.2]{LRZ}}]\label{ex:hopf4}
	Continuing from Example \ref{ex:hopf3}, let us determine the curvature data \((\Phi_{\cO_q(\SU(2))},\bF_{\cO_q(\SU(2))})\) of the horizontally differentiable quantum principal \(\U\)-bundle \((\cO_q(\SU(2)),\Omega_{q,\hor}(\SU(2)),\du_{q,\hor})\).
	Using the \textsc{pbw} basis for \(\cO_q(\SU(2))\), one may show that \(\Zent(\Omega_q(\CP^1)) = \bC[\iu{}e^+  e^-]\).
	Since the generators \(a,c \in \cO_q(\SU(2))_1\) satisfy \(aa^\ast + (qc)(qc)^\ast = a^\ast a + c^\ast c = 1\), one may therefore compute
	\begin{equation}\label{eq:hopfcurv}
		\hat{\Phi}_{\cO_q(\SU(2))}(\iu{}e^+  e^-) = q^2 \iu{} e^+  e^-, \quad \bF_{\cO_q(\SU(2))}(1) = q^{-2} \iu{} e^+ e^-.
	\end{equation}
\end{example}

\subsection{Reconstruction of total calculi}\label{sec:3.4}

At last, we leverage structural results of \DJ{}ur\dj{}evi\'{c}~\cite{DJ97} and Beggs--Majid~\cite{BeMa} to obtain the promised NC generalisation of the classical correspondence between Hermitian line bundles with unitary connection and principal \(\U\)-bundles with principal connection.
Once more, let \(B\) be a unital pre-\Cstar-algebra with \(\ast\)-exterior algebra \((\Omega_B,\du_B)\), which we view as a fixed NC base manifold.
In what follows, given \(q \in (0,\infty)\), we define the corresponding \emph{\(q\)-integers} by setting \([k]_q \coloneqq \tfrac{1-q^k}{1-q}\) for \(k \in \bZ\) when \(q \neq 1\) and \([k]_q \coloneqq k\) for \(k \in \bZ\) when \(q = 1\).

We begin by noting that \(\U\) will not always appear with its usual smooth structure as a Lie group.
Instead, we must allow for all possible \(1\)-dimensional bi-invariant \(\ast\)-exterior algebras on the unital pre-\Cstar-algebra \(\cO(\U)\) of trigonometric polynomials---the following conveniently generalises their construction.

\begin{definition}
	Let \(\kappa \in (0,\infty)\).
	We define \emph{\(\kappa\)-deformed Chevalley--Eilenberg extension} to be the faithful functor \(\CE_\kappa : \grp{QDGA}^{\Unit(1)} \to \grp{QDGA}^{\Unit(1)}\) constructed as follows.
	\begin{enumerate}[leftmargin=*]
		\item Given an object \((P;\Omega,\du)\), let
		\(
			\CE_\kappa(P;\Omega,\du) \coloneqq \left(P;\CE_\kappa(\Omega),\CE_\kappa(\du)\right)
		\),
		where \(\CE_\kappa(\Omega)\) is the graded \(\ast\)-algebra obtained from \(\Omega\) by adjoining a self-adjoint element \(e_\kappa\) of degree \(1\) satisfying the relations \(e_\kappa^2 = 0\) and
			\begin{equation}\label{eq:ce}
				\forall (n,k) \in \bN_0 \times \bZ, \, \forall \omega \in \Omega^n_k, \quad e_\kappa  \omega = (-1)^n \kappa^{-k} \omega  e_\kappa,
			\end{equation}
		where \(\CE_\kappa(\du)\) is defined by setting \(\CE_\kappa(\du)(e_\kappa) \coloneqq 0\) and 
		\begin{equation}
			\forall (n,k) \in \bN_0 \times \bZ, \, \forall \omega \in \Omega^n_k, \quad \CE_\kappa(\du)(\omega) \coloneqq (-1)^n\kappa^{-k}2\pi\iu{}[k]_{\kappa}\omega  e_\kappa + \du\omega. \label{eq:dce}
		\end{equation}
		and where the \(\Unit(1)\)-action on \(\CE_\kappa(\Omega)\) is the unique extension of the \(\Unit(1)\)-action on \(\Omega\) leaving \(e_\kappa\) invariant.
		\item Given an arrow \(f : (P,\Omega_P,\du_P) \to (Q,\Omega_Q,\du_Q)\), let
		\[
			\CE_\kappa(f) : \CE_\kappa(P,\Omega_P,\du_P) \to \CE_\kappa(Q,\Omega_Q,\du_Q)
		\]
		be the unique extension of \(f : \Omega_P \to \Omega_Q\) satisfying \(\CE_\kappa(f)(e_\kappa) = e_\kappa\).
	\end{enumerate}
\end{definition}

Given \(\kappa > 0\), the \(\ast\)-exterior algebra \((\Omega_\kappa(\U),\du_\kappa) \coloneqq (\CE_\kappa(\cO(\U)),\CE_\kappa(0))\) on \(\cO(\U)\) is the essentially unique \(\ast\)-exterior algebra on \(\cO(\U)\) of dimension \(1\) that satisfies the relation
\(
	\du_\kappa(z) \cdot z = \kappa z \cdot \du_\kappa(z),
\)
where \(\du_\kappa(z) = 2\pi\iu{}e_\kappa \cdot z\).
Note that \(\kappa = 1\) recovers the usual de Rham calculus on \(\U\) as a Lie group.
In general, differentiability of a \(\U\)-action with respect to the \(\ast\)-exterior algebra \((\Omega_\kappa(\U),\du_\kappa)\) may now be characterised as follows.

\begin{definition}[{cf.\ \DJ{}ur\dj{}evi\'{c}~\cite[\S 3]{DJ97}, Beggs--Brzezi\'{n}ski~\cite[\S 7]{BB}}]
	Let \(P\) be a \(\Unit(1)\)-pre-\Cstar-algebra of finite type and let \((\Omega,\du)\) be a \(\Unit(1)\)-\(\ast\)-exterior algebra over \(P\).
	We say that \((\Omega,\du)\) is \emph{\(\kappa\)-vertical} whenever there exists a (necessarily unique) lift of \(\id_P\) to a morphism of \(\U\)-\(\ast\)-quasi-\textsc{dga} \(\ver : (P,\Omega,\du) \to \CE_\kappa(P,\Omega,\du)\), the \emph{vertical coevaluation} on \((\Omega,\du)\).
	In this case, we define \emph{horizontal form} in \(\Omega\) to be an element of the \(\Unit(1)\)-invariant graded \(\ast\)-subalgebra
	\(
		\Omega_{\hor} \coloneqq \Set{\omega \in \Omega \given \ver(\omega) = \omega}
	\) of \(\Omega\),
	and a \emph{basic form} to be an element of the \(\Unit(1)\)-invariant and \(\du\)-invariant graded \(\ast\)-subalgebra
	\(
		\Omega_{\bas} \coloneqq (\Omega_{\hor})^{\Unit(1)}
	\) of \(\Omega\).
\end{definition}

At last, given \(\kappa > 0\), we can make precise sense of \textsc{nc} differentiable principal \(\U\)-bundles, where \(\U\) carries the bi-invariant \(\ast\)-exterior algebra \((\Omega_\kappa(\U),\du_\kappa))\).

\begin{definition}[{Brzezi\'{n}ski--Majid~\cite[\S 4]{BrM}, Hajac~\cite{Hajac}, \DJ{}ur\dj{}evi\'{c}~\cite[\S 3]{DJ97}, Beggs--Brzezi\'{n}ski~\cite[\S 7]{BB}, Beggs--Majid~\cite[\S 5.5]{BeMa}; cf.\ \'{C}a\'{c}i\'{c}~\cite{Cacic}}]
	Let \(\kappa \in (0,\infty)\).
	A \emph{\(\kappa\)-differentiable quantum principal \(\Unit(1)\)-bundle} over \(B\) is a triple \((P,\Omega_P,\du_P)\), where \(P\) is a topological quantum principal \(\U\)-bundle over \(B\) and \((\Omega_P,\du_P)\) is a \(\kappa\)-vertical \(\Unit(1)\)-\(\ast\)-exterior algebra over \(P\) together with an isomorphism of \(\ast\)-quasi-\textsc{dga} \(\hat{\iota}_P : (\Omega_B,\du_B) \to (\Omega_{P,\bas},\rest{\du_P}{\Omega_{P,\bas}})\) extending \(\iota_P\), such that
	\(
		\Omega_{P,\hor} = P \cdot \Omega_{P,\bas} \cdot P
	\).
\end{definition}

\begin{example}\label{ex:classicaltotal2} 
	Continuing from Example \ref{ex:classicaltotal1}, let
	\[
		\Omega_{\mathrm{alg}}(X) \coloneqq \bigoplus\nolimits_{k \in \bZ}^{\mathrm{alg}} \Set{\omega \in \Omega(X) \given \forall z \in \Unit(1), \, (\sigma_z)^\ast \omega = z^{-k} \omega},
	\]
	which we equip with the \(\Unit(1)\)-action \(z \mapsto (\sigma_{z^{-1}})^\ast\) and the usual exterior derivative.
	Then \((C^\infty_{\mathrm{alg}}(X),\Omega_{\mathrm{alg}}(X),\du)\) defines a \(1\)-differentiable quantum principal \(\Unit(1)\)-bundle over \((C^\infty(Y),\Omega(Y),\du)\) with respect to \(\pi^\ast : \Omega(Y) \to \Omega_{\mathrm{alg}}(X)^{\Unit(1)}\).
	Note that the vertical coevaluation reduces to the map
	\(
		\Omega_{\mathrm{alg}}(X) \to \Omega(\U)^{\U} \hotimes_\bC \Omega_{\mathrm{alg}}(X)
	\)
	that dualises contraction with the fundamental vector field \(\tfrac{\partial}{\partial t}\) of the \(\U\)-action on \(X\).
\end{example}

The following necessary and sufficient conditions are of both theoretical and practical importance.
Note that they involve the \emph{strong connection condition} first identified by Hajac \cite{Hajac}.

\begin{proposition}[{Beggs--Majid~\cite[Cor.\ 5.53 \& Lemma 5.60]{BeMa}}]
	Let \(\kappa \in (0,\infty)\), let \(P\) be a topological quantum principal \(\U\)-bundle over \(B\), let \((\Omega_P,\du_P)\) be a \(\kappa\)-vertical \(\Unit(1)\)-\(\ast\)-exterior algebra over \(P\), and let \(\hat{\iota}_P : (\Omega_B,\du_B) \to (\Omega_{P,\bas},\rest{\du_P}{\Omega_{P,\bas}})\) be an injective morphism of \(\ast\)-quasi-\textsc{dga} extending \(\iota_P\).
	Then \((P,\Omega_P,\du_P)\) defines a \(\kappa\)-differentiable quantum principal \(\Unit(1)\)-bundle over \(B\) with respect to \(\hat{\iota}_P\) if and only if
	\begin{equation}\label{eq:strong}
		\Omega_{P,\hor} = P \cdot \hat{\iota}_P(\Omega_B).
	\end{equation}
	Moreover, if \(\Omega^n_B\) is flat as a left \(B\)-module for all \(n \in \bN_0\), then \((P,\Omega_P,\du_P)\) defines a \(\kappa\)-differentiable quantum principal \(\Unit(1)\)-bundle over \(B\) with respect to \(\hat{\iota}_P\) if and only if
	\(
		\Omega^1_{P,\hor} = P \cdot \hat{\iota}_P(\Omega^1_B)
	\).
\end{proposition}

We now recall the notions of principal Ehresmann connection and connection \(1\)-form appropriate to our NC setting; as we shall see, the familiar bijection between principal connections and connection \(1\)-forms persists.

\begin{definition}[{Brzezi\'{n}ski--Majid~\cite[\S 4.2 \& Appx.\ \textsc{a}]{BrM}, Hajac~\cite[\S 4]{Hajac}, \DJ{}ur\dj{}evi\'{c}~\cite[\S 4]{DJ97}, Beggs--Majid~\cite[\S 5.5]{BeMa}}]
	Let \(\kappa \in (0,\infty)\), and let \((P,\Omega_P,\du_P)\) be a \(\kappa\)-differentiable quantum principal \(\Unit(1)\)-bundle over \(B\) with respect to \((\Omega_B,\du)\).
	\begin{enumerate}[leftmargin=*]
		\item A \emph{connection} on \((P,\Omega_P,\du_P)\) is a surjective \(\Unit(1)\)-equivariant grading- and \(\ast\)-preserving algebra homomorphism \(\Pi : \Omega_P \to \Omega_{P,\hor}\), such that \(\Pi^2 = \Pi\) and
		\begin{equation}
			\forall \omega \in \Omega^1_P, \quad (\id-\Pi)(\omega)^2 = 0. \label{eq:connnil}
		\end{equation}
		\item A \emph{connection \(1\)-form} on \((P,\Omega_P,\du_P)\) is self-adjoint \(\vartheta \in (\Omega^1_P)^{\U}\) satisfying
		\begin{gather}
			\forall (n,k) \in \bN_0 \times \bZ, \, \forall \omega \in (\Omega^n_P)_k, \quad \vartheta  \omega = (-1)^n \kappa^{-k} \omega  \vartheta, \label{eq:conncent}\\
			\ver(\vartheta) = e_\kappa + \vartheta. \label{eq:connver}
		\end{gather}
	\end{enumerate}
\end{definition}

\begin{remark}
	In the terminology of Brzezi\'{n}ski--Majid~\cite[\S 4.2 \& Appx.\ \textsc{a}]{BrM}, Hajac~\cite[\S 4]{Hajac}, and Beggs--Majid~\cite[\S 5.5]{BeMa}, the restriction of a connection \(\Pi\) to \(1\)-forms is a \emph{\(\ast\)-preserving strong bimodule connection}.
	In the terminology of \DJ{}ur\dj{}evi\'{c}~\cite[\S 4]{DJ97}, the datum of a connection \(1\)-form is equivalent to the datum of a \emph{multiplicative regular connection}.
\end{remark}

\begin{proposition}[{cf.\ Brzezi\'{n}ski--Majid~\cite[Propp.\ 4.4 \& 5.10]{BrM}, \DJ{}ur\dj{}evi\'{c}~\cite[Proof of Thm 4.12]{DJ97}}]\label{prop:connection}
	Let \(\kappa \in (0,\infty)\); let \((P,\Omega_P,\du_P)\) be a \(\kappa\)-differentiable quantum principal \(\Unit(1)\)-bundle over \(B\).
	For every connection \(\Pi\) on \((P,\Omega_P,\du_P)\), there exists a unique connection \(1\)-form \(\vartheta\), such that
	\begin{equation}\label{eq:connection}
		\forall k \in \bZ, \, \forall p \in P_k, \quad (\id - \Pi) \circ \du_P (p) = 2\pi\iu{}[k]_{\kappa} \kappa^{-k} p  \vartheta. 
	\end{equation}
	Conversely, for every connection \(1\)-form \(\vartheta\) on \((P,\Omega_P,\du_P)\), there exists a unique connection \(\Pi\) that satisfies \eqref{eq:connection}.
\end{proposition}

\begin{proof}[Proof of Prop.\ \ref{prop:connection}]
	We begin with preliminary observations.
	By a lemma of Beggs--Majid~\cite[Lemma 5.59]{BeMa}, the vertical coevaluation of \((P,\Omega_P,\du_P)\) satisfies
	\begin{equation}\label{eq:verrange}
		\forall n \in \bN, \quad (\ver - \id)(\Omega^n_P) \subseteq \Omega^{n-1}_{P,\hor} \cdot e_\kappa.
	\end{equation}
	Together with \eqref{eq:strong}, this yields a short exact sequence
	\begin{equation}\label{eq:ses}
		0 \to \Omega_{P,\hor} \to \Omega_P \xrightarrow{\ver - \id} \Omega_{P,\hor} \cdot e_\kappa \to 0
	\end{equation}
	of \(\ast\)-closed \(\Unit(1)\)-invariant \(\Omega_{P,\hor}\)-sub-bimodules of \(\Omega_P\) and \(\Unit(1)\)-equivariant left and right \(\Omega_{P,\hor}\)-linear maps preserving both the ambient \(\ast\)-operation and \(\bN_0\)-grading.
	
	First, suppose that \(\Pi\) is a connection on \(P,\Omega_P,\du_P)\).
	Then \(\Pi\) is a left splitting of \eqref{eq:ses}, so that
	\(
		\rest{(\ver-\id)}{\ran(\id-\Pi)} : \ran(\id-\Pi) \to \Omega_{P,\hor} \cdot e_\kappa
	\)
	is a \(\Unit(1)\)-equivariant isomorphism of \(\Omega_{P,\hor}\)-bimodules preserving both the ambient \(\ast\)-operation and the ambient \(\bN_0\)-grading.
	Hence, let
	\(
		\vartheta \coloneqq \left(\rest{(\ver - \id)}{\ran(\id-\Pi)}\right)^{-1}(e_\kappa)
	\),
	which is thus a \(\Unit(1)\)-invariant self-adjoint element of \(\Omega^1_P\) satisfying \eqref{eq:connver} by construction and \eqref{eq:connection} by \eqref{eq:dce} applied to \(\du_P(P)\).
	It remains to show that \(\vartheta\) satisfies \eqref{eq:conncent}.
	Since \(\vartheta^2 = 0\) by \eqref{eq:connnil}, it suffices to show that \eqref{eq:conncent} holds for horizontal \(\omega\), but this now follows from the fact that \(\rest{(\ver-\id)}{\ran(\id-\Pi)}\) is an isomorphism of \(\Omega_{P,\hor}\)-bimodules.
	Finally, let us show that \(\vartheta\) is uniquely determined by \(\Pi\).
	Let \((\epsilon_j)_{j=1}^n\) be a finite family in \(P_1\) satisfying \(\sum_{j=1}^n \epsilon_j^\ast \epsilon_j = 1\).
	Then
	\[
		\vartheta = \sum\nolimits_{j=1}^n \epsilon_j^\ast \epsilon_j  \vartheta = \frac{\kappa}{2\pi\iu{}}\sum\nolimits_{j=1}^n \epsilon_j^\ast (2\pi\iu{}[1]_\kappa \kappa^{-1} \epsilon_j  \vartheta) = (\id-\Pi)\mleft(\kappa\sum\nolimits_{j=1}^n \epsilon_j^\ast \du_P(\epsilon_j)\mright).
	\]
	
	Now, suppose that \(\vartheta\) is a connection \(1\)-form on \((P,\Omega_P,\du_P)\).
	On the one hand, by construction of \(\CE_\kappa(\Omega_P)\), the element \(e_\kappa\) freely generates the left \(\Omega_{P,\hor}\)-submodule \(\Omega_P \cdot e_\kappa \subseteq \CE_\kappa(\Omega_P)\).
	On the other hand, by \eqref{eq:conncent} and \eqref{eq:connver}, the element \(\vartheta\) satisfies the same relations in \(\Omega_P\) that \(e_\kappa\) satisfies in \(\CE_\kappa(\Omega_P)\).
	Hence, \(\id_{\Omega_P}\) extends to a surjective \(\Unit(1)\)-equivariant algebra homomorphism \(\psi_\vartheta : \CE_\kappa(\Omega_P) \to \Omega_P\) intertwining \(\ast\)-operations and \(\bN_0\)-gradings by setting \(\psi_\vartheta(e_\kappa) \coloneqq \vartheta\).
	We show that \(\Pi \coloneqq \id_{\Omega_P} - \psi_\vartheta \circ (\ver-\id_{\Omega_P})\) is a connection satisfying \eqref{eq:connection} with respect to \(\vartheta\).
	
	First, by construction, the map \(\Pi\) is \(\Unit(1)\)-equivariant and unital, is left and right \(\Omega_{P,\hor}\)-linear, and is \(\ast\)- and grading-preserving; moreover, \(\rest{\Pi}{\Omega_{P,\hor}} = \id_{\Omega_{P,\hor}}\) by definition of \(\Omega_{P,\hor}\).
	Next,
	\(
		(\ver-\id) \circ \Pi = (\ver - \id) - (\ver - \id) \circ \psi_\vartheta \circ (\ver - \id) = 0
	\) by \eqref{eq:verrange} together with \eqref{eq:connver},
	so that \(\ran\Pi \subset \Omega_{P,\hor}\); from this, it follows that \(\Pi^2 = \Pi\), and hence, in particular, that \(\ran(\id-\Pi) = \Omega_{P,hor} \cdot \vartheta\), so that \eqref{eq:connnil} follows since \(\vartheta^2=0\).
	Multiplicativity now follows from left and right \(\Omega_{P,\hor}\)-linearity of \(\Pi\) together with the decomposition \(\Omega_P = \Omega_{P,\hor} \oplus \Omega_{P,\hor} \cdot \vartheta\) of \(\Omega_{P,\hor}\)-bimodules.
	Finally, that \(\Pi\) is uniquely determined by \(\vartheta\) follows from multiplicativity of \(\Pi\) and the fact that \(P\) and \(\du_P(P)\) generate \(\Omega_P\).
\end{proof}

Hence, just as in the classical case, one may now use the connection \(1\)-form to define the curvature \(2\)-form of a principal connection.

\begin{definition}
	Let \((P,\Omega_P,\du_P)\) be a \(\kappa\)-differentiable quantum principal \(\Unit(1)\)-bundle over \(B\).
	Let \(\Pi\) be a connection on \((P,\Omega_P,\iota_P)\) with connection \(1\)-form \(\vartheta\).
	The \emph{curvature} of \(\Pi\) is the closed self-adjoint \(2\)-form \(\cF_\Pi \coloneqq -\hat{\iota}^{-1}_P(\du_P(\vartheta)) \in \Zent(\Omega_B)^2\).
\end{definition}

\begin{example}\label{ex:classicaltotal3}
	We continue from Example \ref{ex:classicaltotal2}.
	Let \(H^\ast X \to X\) be the horizontal cotangent bundle of \(X\), whose fibre at \(x \in X\) is the annihilator of \(\tfrac{\partial}{\partial t}\) at \(x\), so that 
	\[
		\Omega_{\mathrm{alg}}(X)_{\hor} = \bigoplus_{k \in \bZ} \Set*{\omega \in \Gamma\left(\bigwedge H^\ast X \otimes \bC\right) \given \forall z \in \Unit(1), \, (\sigma_z)^\ast \omega = z^{-k}\omega}.
	\]
	Hence, let \(\Pi\) be a principal connection on \(\pi : X \to Y\), which we view as a \(\Unit(1)\)-equivariant real vector bundle endomorphism \(\Pi : T^\ast X \to T^\ast X\) satisfying \(\Pi^2 = \Pi\) and \(\ran \Pi = H^\ast X\).
	Then \(\Pi\) induces a connection on \((C^\infty_{\mathrm{alg}}(X),\Omega_{\mathrm{alg}}(X),\du)\), whose connection \(1\)-form and curvature \(2\)-form respectively recover the usual connection \(1\)-form and curvature \(2\)-form of \(\Pi\).
\end{example}

We now leverage structural results of \DJ{}ur\dj{e}vi\'{c}~\cite{DJ97} to obtain the promised correspondence between NC Hermitian line bundles with connection and NC principal \(\U\)-bundles with principal connection.

Let \(\kappa \in (0,\infty)\).
Define the concrete category \(\grp{Gauge}_\kappa(B)\) of \emph{\(\kappa\)-differentiable quantum principal \(\U\)-bundle with connection over \(B\)} as follows:
\begin{enumerate}[leftmargin=*]
	\item an object is a triple \((P,\Omega_P,\du_P;\Pi)\) consisting of a \(\kappa\)-differentiable quantum principal \(\Unit(1)\)-bundle \((P,\Omega_P,\du_P)\) over \(B\) and a connection \(\Pi_P\) on \((P,\Omega_P,\du_P)\);
	\item an arrow \(f : (P,\Omega_P,\du_P;\Pi_P) \to (Q,\Omega_Q,\du_Q;\Pi_Q)\) is an isomorphism of \(\U\)-\(\ast\)-quasi-\textsc{dga} \(f: (P,\Omega_P,\du_P) \to (Q,\Omega_Q,\du_Q)\) that satisfies both \(f \circ \hat{\iota}_P = \hat{\iota}_Q\) and \(f \circ \Pi_P = \Pi_Q \circ f\).
\end{enumerate}
Hence, we may define a functor \(\Hor_\kappa : \grp{Gauge}_\kappa(B) \to \dCirc_\hor(B)\) as follows:
\begin{enumerate}[leftmargin=*]
	\item given an object \((P,\Omega,\du;\Pi)\), let
	\(
		\Hor_\kappa(P,\Omega,\du;\Pi) \coloneqq \left(P,\Omega_{\hor},\rest{\Pi \circ \du}{\Omega_{\hor}}\right)
	\);
	\item given an arrow \(f : (P,\Omega_P,\du_P;\Pi_P) \to (Q,\Omega_Q,\du_Q;\Pi_Q)\), let
	\[
		\Hor_\kappa(f) : \Hor_\kappa(P,\Omega_P,\du_P,\Pi_P) \to \Hor_\kappa(Q,\Omega_Q,\du_Q,\Pi_Q)
	\]
	be given by the map \(\rest{f}{\Omega_{P,\hor}} : \Omega_{P,\hor} \to \Omega_{Q,\hor}\).
\end{enumerate}
Thus, the functor \(\Hor_\kappa\) takes a \(\kappa\)-differentiable quantum principal \(\U\)-bundle with connection and extracts the horizontal calculus induced by the choice of connection.
A straightforward calculation shows that its essential range satisfies a simple algebraic constraint.

\begin{proposition}[{cf.\ \DJ{}ur\dj{}evi\'{c}~\cite[\S 6.6]{DJ97}}]
	Let \(\kappa \in (0,\infty)\), and let \((P,\Omega_P,\du_P;\Pi)\) be a \(\kappa\)-differentiable quantum principal \(\U\)-bundle with connection over \(B\).
	Let \(\cF_\Pi\) be the curvature \(2\)-form of \(\Pi\), and let \((\hat{\Phi}_{P,\Pi},\bF_{P,\Pi})\) be the curvature data of \(\Hor_\kappa(P,\Omega_P,\du_P;\Pi)\), so that
	\[
		\forall (n,k) \in \bN_0 \times \bZ, \, \forall \beta \in (\Omega^n_{P,\hor})_k, \quad (\Pi \circ \du_P)^2(\beta) = \beta  \cdot  \hat{\iota}_P\mleft(\iu{}\,\bF_{P,\Pi}(k)\mright).
	\]
	Then \(\bF_{P,\Pi} : \bZ \to \cS(B)\) is given by \(\bF_{P,\Pi} = \left(k \mapsto 2\pi[k]_\kappa \kappa^{-k} \cF_\Pi\right)\), so that
	\begin{equation}\label{eq:eigenvalue}
		\hat{\Phi}_{P,\Pi}\mleft(\bF_{P,\Pi}(1)\mright) = \kappa \bF_{P,\Pi}(1).
	\end{equation}
\end{proposition}

\begin{definition}
	Let \((P,\Omega_{P,\hor},\du_{P,\hor})\) be a horizontally differentiable quantum principal \(\U\)-bundle over \(B\) with curvature data \((\hat{\Phi}_P,\bF_P)\).
	We say that \((P,\Omega_{P,\hor},\du_{P,\hor})\) is \emph{flat} whenever \(\bF_P = 0\).
	When \(\bF_P(1)\) is an eigenvector of \(\hat{\Phi}_P\), the \emph{vertical deformation parameter} \(\kappa_P \in \bR^\times\) of \((P,\Omega_{P,\hor},\du_{P,\hor})\) is defined to be the corresponding eigenvalue of \(\hat{\Phi}_P\).
\end{definition}

Remarkably, the algebraic constraint of \eqref{eq:eigenvalue} suffices to characterize the essential range of the functor \(\Hor_\kappa\), which therefore yields an equivalence of categories.

\begin{theorem}[{\DJ{}ur\dj{}evi\'{c}~\cite[Thm 4.12 \& \S 6.5]{DJ97}}]\label{thm:dj}
	Let \(\kappa \in (0,\infty)\), and let \(\dCirc_{\hor,\kappa}(B)\) denote the strictly full subcategory of \(\dCirc_\hor(B)\) whose objects are flat or have vertical deformation parameter \(\kappa\).
	Then the functor \(\Hor_\kappa\) restricts to an equivalence of categories \(\grp{Gauge}_\kappa(B) \to \dCirc_{\hor,\kappa}(B)\) with weak inverse \(\Tot_\kappa : \dCirc_{\hor,\kappa}(B) \to \grp{Gauge}_\kappa(B)\) defined as follows.
	\begin{enumerate}[leftmargin=*]
		\item Given an object \((P,\Omega_{P,\hor},\du_{P,\hor})\) with curvature \(1\)-cocycle \(\bF_\Pi\), let
		\[
			\Tot_\kappa(P,\Omega_{P,\hor},\du_{P,\hor}) \coloneqq (P,\CE_\kappa(\Omega_{P,\hor}),\CE_\kappa(\du_{P,\hor}) + \sci_\Pi,\Pi_\kappa),
		\]
		where \(\sci_\Pi : \CE_\kappa(\Omega_{P,\hor}) \to \CE_\kappa(\Omega_{P,\hor})\) is the complex-linear map defined by 
		\[
			\forall \omega_1,\omega_2 \in \Omega_{P,\hor}, \quad \sci_\Pi(\omega_1+\omega_2 e_\kappa) \coloneqq -\frac{\kappa}{2\pi}\omega_2  \bF_\Pi(1),
		\]
		and where \(\Pi_\kappa : \CE_\kappa(\Omega_{P,\hor}) \to \CE_\kappa(\Omega_{P,\hor})\) is the unique algebra homomorphism satisfying \(\rest{\Pi_\kappa}{\Omega_{P,\hor}} = \id_{\Omega_{P,\hor}}\) and \(\Pi_\kappa(e_\kappa) = 0\).
		\item Given an isomorphism \(f : (P,\Omega_{P,\hor},\du_{P,\hor}) \to (Q,\Omega_{Q,\hor},\du_{Q,\hor})\) of horizontally differentiable quantum principal \(\U\)-bundles over \(B\), let
		\[
			\Tot_\kappa(f) : \Tot_\kappa(P,\Omega_{P,\hor},\du_{P,\hor}) \to \Tot_\kappa(Q,\Omega_{Q,\hor},\du_{Q,\hor})
		\]
		be given by the map \(\CE_\kappa(f) : \CE_\kappa(\Omega_{P,\hor}) \to \CE_\kappa(\Omega_{Q,\hor})\).
	\end{enumerate}
	In particular, a canonical natural isomorphism \(\text{\textup{\DJ}} : \id_{\grp{Gauge}_\kappa(B)} \Rightarrow \Tot_\kappa \circ \Hor_\kappa\) is defined as follows: given an object \((P;\Omega_P,\du_P,\Pi)\) of \(\grp{Gauge}_\kappa(B)\), define the isomorphism \(\text{\textup{\DJ}}_{(P;\Omega_P,\du_P;\Pi)} : (P;\Omega_P,\du_P;\Pi) \to \Tot_\kappa \circ \Hor_\kappa(P;\Omega_P,\du_P;\Pi)\) by
	\begin{equation}\label{eq:djnat}
		\forall \omega \in \Omega_P, \quad \text{\textup{\DJ}}_{(P;\Omega_P,\du_P;\Pi)}(\omega) \coloneqq (\ver-\id) \circ (\id-\Pi)(\omega) + \Pi(\omega).
	\end{equation}
\end{theorem}

By combining this theorem with Theorem \ref{thm:horizontal}, Proposition-Definition \ref{propdef:vertical}, and Corollary \ref{cor:abstractpimsner}, we obtain a precise NC generalisation of the classical correspondence between Hermitian line bundles with unitary connection and principal \(\U\)-bundles with principal connection.

\begin{definition} 
	Let \((E,\sigma_E,\nabla_E)\) be a Hermitian line \(B\)-bimodule with connection.
	We say that \((E,\sigma_E,\nabla_E)\) is \emph{flat} whenever \(\bF_{[E,\nabla_E]} = 0\).
	When \(\bF_{[E,\nabla_E]}\) is an eigenvector of the automorphism \(\hat{\Phi}_{[E,\nabla_E]}\), the \emph{vertical deformation parameter} \(\kappa_{[E,\nabla_E]} \in \bR^\times\) of \((E,\sigma_E,\nabla_E)\) is defined to be the corresponding eigenvalue of \(\hat{\Phi}_{[E,\nabla_E]}\).
\end{definition}

\begin{corollary}\label{cor:diffpimsner}
	Let \(\kappa \in (0,\infty)\), and let \(\dPic_\kappa(B)\) be the strictly full subcategory of \(\dPic(B)\) whose objects are full or have vertical deformation parameter \(\kappa\).	
	Then \(\dPic_\kappa(B)\) is the essential image of \(\dCirc_{\hor,\kappa}(B)\) under \(\epsilon_1 \circ \hat{\cL}\), so that the functor \(\epsilon_1 \circ \hat{\cL} \circ \Hor_\kappa : \dCirc_{\kappa,\tot}(B) \to \dPic_\kappa(B)\) is an equivalence of categories.
\end{corollary}

\begin{definition}
	Let \(\kappa \in (0,\infty)\), and let \((E,\sigma,\nabla)\) be a Hermitian line \(B\)-bimodule with connection that is flat or has vertical deformation parameter \(\kappa\).
	The \emph{\(\kappa\)-total crossed product} of \((B;\Omega_B,\du)\) by \((E,\sigma_E,\nabla_E)\) is the essentially unique \(\kappa\)-differentiable quantum principal \(\U\)-bundle with connection \((B;\Omega_B,\du) \rtimes_{(E,\sigma_E,\nabla_E)}^{\kappa,\tot}\bZ\) on \(B\), such that
	\(
		(\hat{\cL} \circ \Hor_\kappa)\mleft((B;\Omega_B,\du) \rtimes_{(E,\sigma_E,\nabla_E)}^{\kappa,\tot}\bZ\mright) \cong (E,\sigma_E,\nabla_E)
	\);
	in this case, we define a \(\ast\)-exterior algebra \((\Omega_B,\du_B) \rtimes_{(E,\sigma_E)}^{\kappa,\tot} \bZ\) and connection \(\Pi_{(E,\sigma_E,\nabla_E)}\) by
	\[
		\left(B \rtimes_E \bZ; (\Omega_B,\du_B) \rtimes_{(E,\sigma_E)}^{\kappa,\tot} \bZ;\Pi_{(E,\sigma_E,\nabla_E)}\right) \coloneqq (B;\Omega_B,\du) \rtimes_{(E,\sigma_E,\nabla_E)}^{\kappa,\tot}\bZ.
	\]
\end{definition}

Thus, given \(\kappa \in (0,\infty)\) and \((E,\sigma_E,\nabla_E)\) that is flat or has vertical deformation parameter \(\kappa\), we may take
\(
	(B;\Omega_B,\du) \rtimes_{(E,\sigma_E,\nabla_E)}^{\kappa,\tot}\bZ \coloneqq \Tot_\kappa((B;\Omega_B,\du) \rtimes_{(E,\sigma_E,\nabla_E)}^{\hor}\bZ),
\)
where \((B;\Omega_B,\du) \rtimes_{(E,\sigma_E,\nabla_E)}^{\hor}\bZ\) is any horizontal crossed product of \((B;\Omega_B,\du)\) by \((E,\sigma_E,\nabla_E)\).
Note that \((E,\sigma_E,\nabla_E)\) is flat if and only if \((B;\Omega_B,\du) \rtimes_{(E,\sigma_E,\nabla_E)}^{\hor}\bZ\) is flat and that \((E,\sigma_E,\nabla_E)\) has vertical deformation parameter \(\kappa\) if and only if \((B;\Omega_B,\du) \rtimes_{(E,\sigma_E,\nabla_E)}^{\hor}\bZ\) has vertical deformation parameter \(\kappa\).

There are certain examples that are naturally described in terms of homomorphisms \(\bZ \to \dPic(B)\) or that yield homomorphisms of particular interest.
In such cases, it convenient to have a straightforward algebraic characterization of the essential range of the composite functor \(\hat{\cL} \circ \Hor_\kappa\).

\begin{corollary}
	Let \(\kappa \in (0,\infty)\), and let \(\grp{Hom}_\kappa(\bZ,\dPic(B))\) be the essential image of the subcategory \(\dCirc_{\hor,\kappa}(B)\) under the equivalence \(\hat{\cL} : \dCirc_\hor(B) \to \grp{Hom}(\bZ,\dPic(B))\).
	Then a homomorphism \(\hat{F} : \bZ \to \dPic(B)\) defines an object of \(\grp{Hor}_\kappa(\bZ,\dPic(B))\) if and only if \(\hat{F}(1)\) is flat or has vertical deformation parameter \(\kappa\), so that, in the latter case,
	\begin{equation}\label{eq:kappahom}
		\forall m \in \bZ, \quad \bF \circ \pi_0(\hat{F})(m) = \kappa^{-m+1}[m]_\kappa \bF_{[\hat{F}(1)]}.
	\end{equation}	
\end{corollary}

\begin{proof}
	Relative to the discussion after Proposition-Definition \ref{propdef:vertical}, it remains to check \eqref{eq:kappahom}.
	Suppose that \(\hat{F} : \bZ \to \dPic(B)\) is a homomorphism, such that \(\hat{F}(1)\) has vertical deformation parameter \(\kappa\).
	The right \(1\)-cocycle identity for curvature \(1\)-cocycle of \(B\) specialises to
	\[
		\forall m,n \in \bZ, \quad \bF \circ \pi_0(\hat{F})(m+n) = \hat{\Phi}_{[\hat{F}(1)]}^{-n}\mleft(\bF \circ \pi_0(\hat{F})(m)\mright) + \bF \circ \pi_0(\hat{F})(n).
	\]
	By induction together with the equation \(\hat{\Phi}_{[\hat{F}(1)]}(\bF_{[\hat{F}(1)]}) = \kappa\bF_{[\hat{F}(1)]}\), it follows that \(\bF \circ \pi_0(\hat{F})\) satisfies \(\bF \circ \pi_0(\hat{F}) = \left(m \mapsto [m]_{\kappa^{-1}} \bF_{[\hat{F}(1)]}\right)\), which yields \eqref{eq:kappahom}.
\end{proof}

\begin{example}[{\DJ{}ur\dj{}evi\'{c}~\cite[\S 4]{DJ97}}]\label{ex:hopf5}
	We continue from Example \ref{ex:hopf4}. 
	By \eqref{eq:hopfcurv}, it follows that \((\cO_q(\SU(2)),\Omega_{q,\hor}(\SU(2)),\du_{q,\hor})\) has deformation parameter \(q^2\); hence, by \eqref{eq:hopfcurv} and \eqref{eq:kappahom}, the homomorphism \(\hat{\cE} : \bZ \to \dPic(\cO_q(\CP^1))\) of Example \ref{ex:hopf3} satisfies \(\bF \circ \pi_0(\hat{\cE}) = \left(m \mapsto [m]_{q^2} q^{-2m} \iu{}e^+ e^-\right)\).
	At last, by results of \DJ{}ur\dj{}evi\'{c},
	\[
		(\cO_q(\SU(2)),\Omega_q(\SU(2)),\du_q;\Pi_q) \coloneqq \Tot_{q^2}(\cO_q(\SU(2)),\Omega_{q,\hor}(\SU(2)),\du_{q,\hor}))
	\]
	recovers the \(3\)-dimensional calculus \((\Omega_q(\SU(2)),\du_q)\) on \(\cO_q(\SU(2))\) of Woronowicz~\cite{Woronowicz87} and the non-universal \(q\)-monopole connection \(\Pi_q\) of Brzezi\'{n}ski--Majid~\cite{BrM}.
	In other words, we may obtain \(\Omega_q(\SU(2))\) from \(\Omega_{q,\hor}(\SU(2))\) by adjoining the skew-adjoint \(\U\)-invariant \(1\)-form
	\(
		e^0 = 2\pi\iu{}q^{-2}e_{q^2}
	\)
	subject to the relations \((e^0)^2 = 0\) and
	\[
		\forall (n,k) \in \bN_0 \times \bZ, \, \forall \omega \in \Omega_{q,\hor}^n(\SU(2))_k, \quad e^0  \omega = (-1)^n q^{-2k} \omega e^0,
	\]
	and we may obtain \(\du_q\) from \(\du_{q,\hor}\) by setting \(\du_q(e^0) \coloneqq q^{-2} e^+  e^-\) and 
	\begin{equation*}
		\forall (n,k) \in \bN_0 \times \bZ, \, \forall \omega \in \Omega_{q,\hor}^n(\SU(2))_k, \quad \du_{q}(\omega) \coloneqq (-1)^n [k]_{q^{-2}} \omega e^0 + \du_{q,\hor}(\omega).
	\end{equation*}
	From now on, we shall refer to \((\cO_q(\SU(2)),\Omega_q(\SU(2)),\du_q;\Pi_q)\) as the \emph{\(q\)-monopole}.
\end{example}

\begin{example}\label{ex:heis6}
	We continue from Example \ref{ex:heis5}.
	Since \((g \mapsto g_{21}\theta + g_{22}) : \Gamma_\theta \to \bR^\times\) is an injective homomorphism~\cite[Thm 5.2.10]{HK}, there exists a unique generator \(\gamma\) of the infinite cyclic group \(\Gamma_\theta\) satisfying \(\gamma_{21}\theta+g_{22}>1\); hence, let \(\epsilon_\theta \coloneqq \gamma_{21}\theta+\gamma_{22}\), which recovers the norm-positive fundamental unit of the real quadratic field generated by \(\theta\).
	Next, since 
	\(
		\hat{\Phi}_{[\hat{E}(\gamma)]}(e^1 e^2) = (\gamma_{21}\theta+\gamma_{22})^2e^1 e^2 = \epsilon_\theta^2 \, e^1 e^2
	\),
	it follows that \(\hat{E}(\gamma)\) has vertical deformation parameter \(\epsilon_\theta^2\).
	Thus, the composite homomorphism  \(\hat{E} \circ (k \mapsto \gamma^k)\) is an object of \(\grp{Hom}_{\epsilon_\theta^2}(\bZ,\dPic(C^\infty_\theta(\bT^2)))\), so that \(\hat{\Sigma}\mleft(\hat{E} \circ (k \mapsto \gamma^k)\mright)\) defines an object of \(\dCirc_{\hor,\epsilon_\theta^2}(C^\infty_\theta(\bT^2))\).
	Hence, at last, we define the \emph{real multiplication instanton} to be the \(\epsilon^2_\theta\)-differentiable quantum principal \(\U\)-bundle over \(C^\infty_\theta(\bT^2)\) given by
	\(
		(P_\theta; \Omega_{P_\theta}, \du_{P_\theta}; \Pi_{P_\theta}) \coloneqq \Tot_{\epsilon^2_\theta} \circ \hat{\Sigma}\mleft(\hat{E} \circ (k \mapsto \gamma^k)\mright)
	\),
	which recovers a construction of \'{C}a\'{c}i\'{c} \cite{Cacic}.
	Note that \Cstar-algebraic completion of \(P_\theta\) is part of a family of Cuntz--Pimsner algebras first considered by Nawata \cite{Nawata}.
\end{example}

\section{Lifting problems for NC Riemannian structures}\label{sec:4}

In the commutative case, a Riemannian metric on the base space of a principal \(\U\)-bundle with principal connection lifts canonically to the total space.
Here, we study the analogous lifting problems for two closely interrelated notions of Riemannian structure on NC manifolds, which are based, respectively, on generalised Hodge star operators and formal spectral triples.
In particular, we show that these lifted Riemannian structures inexorably involve modular phenomena in both vertical and horizontal directions that are generally non-trivial and distinct.
Along the way, we  show that quantum \(\SU(2)\) \emph{qua} total space of the \(q\)-monopole does not admit a non-pathological \(\U\)-equivariant twisted spectral triple and obtain a geometric formal derivation of Kaad--Kyed's compact quantum metric space \cite{KK} on quantum \(\SU(2)\) for a canonical choice of parameters.

For the entirety of this section, let \(B\) be a unital separable pre-\Cstar-algebra with \(\ast\)-differential calculus \((\Omega_B,\du_B)\), which we assume has dimension \(N \in \bN\); let \(\gamma_B : \Omega_B \to \Omega_B\) denote the \(\bZ/2\bZ\)-grading on \(\Omega_B\) by parity of degree.
Moreover, given a horizontal quantum principal \(\U\)-bundle \((P;\Omega_{P,\hor};\du_{P,\hor})\) over \(B\), we suppress the isomorphism \(\hat{\iota}_P : \Omega_B \to \Omega_{P,\hor}^{\U}\).

\subsection{Hodge operators and conformality}\label{sec:4.1}

We begin by introducing the minimum of Riemannian structure required for classical \(\U\)-gauge theory on \textsc{nc} manifolds: the Hodge star operator and integration against the Riemannian volume form.
Such an approach was first proposed by Kustermans--Murphy--Tuset for quantum groups \cite{KMThodge}, it has since attained its fullest expression in the setting of NC K\"{a}hler geometry in the sense of \'{O} Buachalla \cite{OB}.
We show that it permits robust generalisation of the notion of conformal orientation-preserving diffeomorphism to the entire differential Picard group.

We begin with a straightforward generalisation of the Hodge star operator.

\begin{definition}[Kustermans--Murphy--Tuset \cite{KMThodge}, Majid \cite{Majid05}, Zampini \cite{Zampini}, \'{O} Buachalla \cite{OB}]
	A \emph{Hodge operator} on \((\Omega_B,\du_B)\) is a \(\ast\)-preserving \(B\)-bimodule morphism \(\star : \Omega_B \to \Omega_B\), such that, for every \(k \in \Set{0,\dotsc,N}\),
	\begin{gather}
		\star{}(\Omega^k_B) \subseteq \Omega^{N-k}_B, \quad \rest{\star^2}{\Omega^k_B} = (-1)^{k(N-k)}\id_{\Omega^k_B}, \label{eq:hodge1}\\
		\forall \omega,\eta \in \Omega^k_B, \quad \omega \cdot {\star}(\eta) = {\star}^{-1}(\omega) \cdot \eta . \label{eq:hodge2}
	\end{gather}
	Hence, the \emph{inverse metric} induced by \(\star\) is the right \(B\)-valued inner product \(g\) on \(\Omega_B\) defined by 
	\begin{equation}\label{eq:hodge3}
		\forall \omega, \eta \in \Omega_B, \quad g(\omega,\eta) \coloneqq \star{}\mleft(\omega^\ast \cdot {\star}(\eta)\mright).
	\end{equation}	
\end{definition}

By combining a generalised Hodge star operator with a suitable generalisation of integration against the corresponding Riemannian volume form, we obtain our first notion of NC Riemannian structure; following Connes \cite{Connes85} and Kustermans--Murphy--Tuset \cite{KMT}, we impose the divergence theorem as a requirement.

\begin{definition}[cf.\ Kustermans--Murphy--Tuset~\cite{KMThodge}, \'{O} Buachalla~\cite{OB}, Salda\~{n}a~\cite{Saldana21Yang}]
	A \emph{Riemannian geometry} on \((B;\Omega_B,\du_B)\) is a pair \((\star,\tau)\), where \(\star\) is a Hodge operator on \((\Omega_B,\du_B)\) whose inverse metric \(g\) admits a basis as a right \(B\)-valued inner product on \(\Omega_B\) and satisfies
	\begin{equation}\label{eq:metricbdd}
		\forall b \in B, \, \forall \omega \in \Omega_B, \quad g(b  \omega,b  \omega) \leq \norm{b}^2 g(\omega,\omega),
	\end{equation}
	and where \(\tau\) is a bounded state on \(B\) that satisfies
	\begin{align}
		\forall \omega \in \Omega^{N-1}_B, && (\tau \circ \star \circ \du_B)(\omega) &= 0, \label{eq:closed}\\
		\forall b \in B, && \sup\Set{\tau(a^\ast b^\ast b a) \given a \in A, \, \tau(a^\ast a) \leq 1} &= \norm{b}^2. \label{eq:faithfulstate} 
	\end{align}
\end{definition}

\begin{example}\label{ex:classical4} 
	We continue from Example \ref{ex:classical3}.
	Suppose that \(X\) is orientable.
	Equip \(X\) with an orientation and a Riemannian metric \(g\); let \(\star_g\) and \(\vol_g\) respectively denote the resulting Hodge star operator and Riemannian volume form.
	Then \((\star_g,\int_X (\cdot) \vol_g)\) defines a Riemannian geometry on \((C^\infty(X);\Omega(X,\bC),\du)\), whose inverse metric is the usual inverse Riemannian metric.
	Note that a basis for \(\Omega(X,\bC)\) with respect to the inverse metric can be constructed from local orthonormal frames using a smooth partition of unity.
\end{example}

\begin{example}\label{ex:hopf6}
	We continue from Example \ref{ex:hopf5}.
	Let \(h_q\) denote Woronowicz's Haar state on \(\cO_q(\SU(2))\), which is faithful on \(C_q(\SU(2))\) by a result of Nagy~\cite{Nagy}.
	Since \((\cO_q(\CP^1),\Omega_q(\CP^1),\du_q)\) is an NC K\"{a}hler manifold \`{a} la \'{O} Buachalla \cite[\S\S 4.4, 5.4]{OB}, it admits a canonical Riemannian geometry \((\star_q,\rest{h_q}{\cO_q(\CP^1)})\), where \(\star_q(1) \coloneqq \iu{}e^+e^-\) and \(\star_q\) restricts to \(\pm\iu{}\id\) on \(\cO_q(\SU(2))_{\mp 2}\cdot e^{\pm}\).
	Note that \(\star_q\) recovers Zampini's modification~\cite[Eq.\ 5.14]{Zampini} of Majid’s Hodge operator~\cite[\S 4]{Majid05} for \(\alpha^{\prime\prime} = -q^2\).
\end{example}

\begin{example}\label{ex:heis7}
	We continue from Example \ref{ex:heis6}.
	The canonical Riemannian geometry \((\star,\tau)\) on \((C^\infty_\theta(\bT^2);\Omega_\theta(\bT^2),\du)\) is given by
	\[
		\star(1) \coloneqq e^1  e^2, \,\,\, \star(e^1) \coloneqq e^2, \,\,\, \star(e^2) \coloneqq -e^1; \quad
		\forall (m,n) \in \bZ^2, \,\,\, \tau(U^m V^n)  \coloneqq \delta^{m,0} \delta^{n,0};
	\]
	so that \(\tau\) recovers the canonical \(\U\)-invariant faithful trace on \(C_\theta(\bT^2)\).
\end{example}

Just as in the classical case, we may now equip \(\Omega_B\) with an \(L^2\)-inner product and compute the (formal) adjoint of the exterior derivative \(\du_B\) in terms of the Hodge star operator.

\begin{proposition}[{\'{O} Buachalla \cite[\S\S 5.2--3]{OB}}]\label{prop:basicseparable}
	Let \((\star,\tau)\) be a Riemannian geometry on \((B;\Omega_B,\du_B)\); let \(g\) be the resulting inverse metric.
	Then \(\Omega_B\) defines a \(B\)-self-correspondence of finite type with respect to \(g\) that decomposes as an orthogonal direct sum \(\Omega_B = \bigoplus_{k=0}^N \Omega^k_B\) of sub-\(B\)-bimodules.
	Hence, the \(\bC\)-vector space \(\Omega_B\) defines a separable pre-Hilbert space with respect to the inner product \(\ip{}{}_\tau\) defined by
	\begin{equation}\label{eq:l2ip}
		\forall \omega, \eta \in \Omega_B, \quad \ip{\omega}{\eta}_\tau \coloneqq \tau\mleft(g(\omega,\eta)\mright), 
	\end{equation} 
	with respect to which the left \(B\)-module structure on \(\Omega_B\) defines an isometric \(\ast\)-representation of \(B\), the direct sum decomposition \(\Omega_B = \bigoplus_{k=0}^N \Omega^k_B\) is orthogonal, the Hodge operator \(\star\) is unitary, and \(\du_B^\ast = {\star}^{-1} \circ {\du_B} \circ {\star} \circ \gamma_B\).
\end{proposition}

\begin{proof}
	Relative to the references (cf.\ the proof of Proposition \ref{prop:totalprehilbert} below), it remains to show that \(\Omega_B\) is separable as a pre-Hilbert space and that the left \(B\)-module structure is isometric as a \(\ast\)-homomorphism.
	Let \(\mathfrak{B}\) be the \Cstar-algebraic completion of \(B\), so that \(\tau\) extends to a state on \(\mathfrak{B}\).
	Let \(m \in \Set{0,\dotsc,N}\), and let \((e_i)_{i=1}^n\) be a basis for \(\Omega^m_B\) with respect to \(g\), so that \(X \coloneqq \mleft(g(e_i,e_j)\mright)_{i,j=1}^n \in M_n(B)\) is positive with unique positive square root \(\sqrt{X} \in M_n(\mathfrak{B})\).
	Let \(a \coloneqq (a_1,\dotsc,a_n) \in B^n \subset \mathfrak{B}^n\) and set \(\omega \coloneqq \sum_{i=1}^n e_i  a_i\).
	Then
	\begin{multline*}
		\ip{\omega}{\omega}_\tau = \tau\mleft(\sum\nolimits_{i,j=1}^n a_i^\ast  g(e_i,e_j)  a_j \mright) = \tau\mleft(\hp{a}{X a}_{B^n}\mright) \leq \tau\mleft(\norm*{\sqrt{X}}^2 \hp{a}{a}_{\mathfrak{B}^n}\mright)\\ \leq \norm{X}\sum\nolimits_{i=1}^n \norm{a_i}^2.
	\end{multline*}
	Since \(B\) is separable as a normed vector space and since \((e_i)_{i=1}^m\) generates \(\Omega^m_B\) as a right \(B\)-module, it follows that \(\Omega^m_B\) is separable as a pre-Hilbert space.
	Hence, \(\Omega_B = \bigoplus_{m=0}^N \Omega^m_B\) is also separable as a pre-Hilbert space.
	
	We now show that the left \(B\)-module structure \(\pi : B \to \bL(\Omega_B)\) is isometric.
	Let \(b \in B\) be given.
	Since \(\pi\) is bounded, it is contractive, so that, by \eqref{eq:faithfulstate}, 
	\[
		\norm{b}^2 \geq \norm{\pi(b)^2} \geq \sup\Set{\ip{\pi(b)a}{\pi(b)a}_\tau \given a \in B,\, \ip{a}{a}_\tau \leq 1} = \norm{b}^2. \qedhere
	\]
\end{proof}

We now generalise the notion of conformal orientation-preserving diffeomorphism to our NC setting.
For convenience, let \(\mathcal{Z}_{>0}(B)\) denote the multiplicative group of all positive invertible elements of \(\Zent(\Omega_B)^0\), so that \(\mathcal{Z}_{>0}(B)\) admits a canonical right action of the differential Picard group \(\dpic(B)\) defined by
\begin{equation}\label{eq:confact}
	\forall \mu \in \mathcal{Z}_{>0}(B), \, \forall [E,\nabla_E] \in \dpic(B), \quad \mu \ract [E,\nabla_E] \coloneqq \hat{\Phi}^{-1}_{[E,\nabla_E]}(\mu).
\end{equation}

Note from Examples \ref{ex:twist2} and \ref{ex:classical3} and the proof of Theorem \ref{thm:stabilizer} that the dynamical content of a Hermitian line \(B\)-bimodule with connection is encoded by its generalised braiding.
Hence, we promote the behaviour of the usual Hodge star operator under orientation-preserving conformal diffeomorphisms~\cite[Thm.\ 1.159.h]{Besse} to the following definition.

\begin{definition}
	Let \(\star_B\) be a Hodge operator on \((\Omega_B,\du_B)\).
	A Hermitian line \(B\)-bimodule with connection \((E,\sigma_E,\nabla_E)\) is \emph{\(\star_B\)-conformal} when there exists (necessarily unique) \(\mu \in \mathcal{Z}_{>0}(B)\), the \emph{conformal factor} of \((E,\sigma_E,\nabla_E)\), such that
	\begin{multline}\label{eq:conformal}
		\forall x \in E, \, \forall k \in \Set{0,\dotsc,N}, \, \forall \alpha \in \Omega^k_B, \\ \sigma_E\mleft(\star_B(\alpha) \otimes x\mright) = \leg{\sigma_E\mleft(\alpha \otimes x\mright)}{0} \otimes \star_B{}\mleft(\leg{\sigma_E\mleft(\alpha \otimes x\mright)}{1}\mright)  \mu^{N-2k}.
	\end{multline}
	We denote by \(\dPic(B;\star_B)\) the strictly full subcategory of \(\dPic(B)\) whose objects are \(\star_B\)-conformal, we denote by \(\dpic(B;\star_B)\) the corresponding subset of \(\dpic(B)\), and we define
	\(
		\mu : \dpic(B;\star_B) \to \mathcal{Z}_{>0}(B)
	\)
	by mapping \([E,\nabla_E] \in \dpic(B;\star_B)\) to the conformal factor \(\mu_{[E,\nabla_E]}\) of any (and hence every) representative.
\end{definition}

In the classical case, orientation-preserving conformal diffeomorphisms form a group and their conformal factors define a multiplicative \(1\)-cocycle on this group.
The same is true in the NC setting.

\begin{proposition}
	Suppose that \(\star_B\) is a Hodge operator on \((\Omega_B,\du_B)\).
	Then \(\dPic(B;\star_B)\) defines a sub-\(2\)-group of \(\dPic(B)\), the subset \(\dpic(B;\star_B)\) defines a subgroup of \(\dpic(B)\), and the function \(\mu : \dpic(B;\star_B) \to \mathcal{Z}_{>0}(B)\) defines a group \(1\)-cocycle with respect to the restriction to \(\dpic(B;\star_B)\) of the right \(\dpic(B)\)-action on \(\mathcal{Z}_{>0}(B)\) defined by \eqref{eq:confact}.
\end{proposition}

\begin{proof}
	First, note that the monoidal unit \((B,\sigma_B,\nabla_B)\) is \(\star_B\)-conformal with conformal factor \(\mu_{[B,\nabla_B]} = 1\).
	On the one hand, suppose that \((E,\sigma_E,\nabla_E)\) and \((F,\sigma_F,\nabla_F)\) are \(\star_B\)-conformal. Then, given \(k \in \Set{0,\dotsc,N}\), \(\alpha \in \Omega^k_B\), \(x \in E\), and \(y \in F\),
	\begin{align*}
		&\sigma_{E \otimes_B F}\mleft(\star_B(\alpha) \otimes (x \otimes y)\mright)\\
		&= \left(\leg{\sigma_E(\alpha \otimes x)}{0} \otimes \leg{\sigma_F\mleft(\star_B(\leg{\sigma_E(\alpha \otimes x)}{1})  \mu_{[E,\nabla_E]}^{N-2k} \otimes y\mright)}{0}\right)\\ &\quad\quad \otimes \leg{\sigma_F\mleft(\star_B(\leg{\sigma_E(\alpha \otimes x)}{1})  \mu_{[E,\nabla_E]}^{N-2k} \otimes y\mright)}{1}\\
		&= \left(\leg{\sigma_E(\alpha \otimes x)}{0} \otimes \leg{\sigma_F\mleft(\leg{\sigma_E(\alpha \otimes x)}{1} \otimes y\mright)}{0}\right) \\ &\quad\quad \otimes \star_B\mleft(\leg{\sigma_F(\leg{\sigma_E(\alpha \otimes x)}{1}\otimes y)}{1}\mright)  \hat{\Phi}_{[F,\nabla_F]}^{-1}(\mu_{[E,\nabla_E]}^{N-2k})  \mu_{[F,\nabla_F]}^{N-2k}\\
		& = \leg{\sigma_{E \otimes_B F}(\alpha \otimes (x\otimes y))}{0}\\ &\quad\quad \otimes \star_B\mleft(\leg{\sigma_{E \otimes_B F}(\alpha \otimes (x\otimes y))}{1}\mright)  \left(\hat{\Phi}_{[F,\nabla_F]}^{-1}(\mu_{[E,\nabla_E]})  \mu_{[F,\nabla_F]}\right)^{N-2k}.
	\end{align*}
	Hence, the subcategory \(\dPic(B;\star_B)\) is closed under the monoidal product and the map \(\mu\) satisfies the required \(1\)-cocycle identity.
	On the other hand, suppose that \((E,\sigma_E,\nabla_E)\) is \(\star_B\)-conformal.
	Then, given \(k \in \Set{0,\dotsc,N}\), \(\alpha \in \Omega^k_B\), and \(x \in E\),
	\begin{align*}
		\sigma_{\conj{E}}\mleft(\star_B(\alpha) \otimes \conj{x}\mright)
		&= \conj{\leg{\sigma^{-1}_E(x \otimes \star_B(\alpha)^\ast)}{0}} \otimes \leg{\sigma^{-1}_E(x \otimes \star_B(\alpha)^\ast)}{-1}^\ast\\
		&= \conj{\leg{\sigma^{-1}_E(x \otimes \alpha^\ast)}{0}  \mu_{[E,\nabla_E]}^{-N+2k}} \otimes \star_B\mleft(\leg{\sigma^{-1}_E(x \otimes \alpha^\ast)}{-1}^\ast\mright)\\
		&= \mu_{[E,\nabla_E]}^{-N+2k}  \leg{\sigma_{\conj{E}}(\alpha \otimes \conj{x})}{0} \otimes \star_B\mleft(\leg{\sigma_{\conj{E}}(\alpha \otimes \conj{x})}{1}\mright)\\
		&= \leg{\sigma_{\conj{E}}(\alpha \otimes \conj{x})}{0} \otimes \star_B\mleft(\leg{\sigma_{\conj{E}}(\alpha \otimes \conj{x})}{1}\mright)  \hat{\Phi}_{[E,\nabla_E]}(\mu_{[E,\nabla_E]}^{-1})^{N-2k}.
	\end{align*}
	Hence, the the subcategory \(\dPic(B;\star_B)\) is also closed under monoidal inversion.
\end{proof}


\begin{example}\label{ex:classical6}
	Continuing from Example \ref{ex:classical4}, let \(\operatorname{Conf}_+(X,g)\) denote the group of conformal orientation-preserving diffeomorphisms of \((X,g)\).
	On the one hand, for every Hermitian line bundle with unitary connection \((\cE,\nabla_{\cE})\) on \(X\), \((\Gamma(\cE),\operatorname{flip},\nabla_{\cE})\) is \(\star_g\)-conformal with conformal factor \(1\).
	On the other, for every \(\phi \in \Diff(X)\), \(\hat{\tau}(0,(\phi^{-1})^\ast)\) is \(\star_g\)-conformal if and only if \(\phi\) is conformal and orientation-preserving, in which case
	\(
		\mu \circ \pi_0(\hat{\tau})(0,(\phi^{-1})^\ast) = \sqrt{\tfrac{\phi^\ast g}{g}}
	\).
	Hence, the isomorphism of Example \ref{ex:classical3} restricts to an isomorphism
	\(
		\operatorname{Conf}_+(X,g) \ltimes \check{H}^2(X) \to \dPic(C^\infty(X);\star_g),
	\) 
	with respect to which \(\mu : \dPic(C^\infty(X);\star_g) \to C^\infty(X,(0,\infty))\) reduces to the map
	\[
		\mleft((\phi,[\cE,\nabla_{\cE}]) \mapsto \sqrt{\tfrac{\phi^\ast g}{g}}\mright) : \operatorname{Conf}_+(X,g) \ltimes \check{H}^2(X) \to C^\infty(X,(0,\infty)).
	\]
\end{example}

In light of Theorem \ref{thm:horizontal} and Proposition-Definition \ref{propdef:vertical}, we may also consider conformality of horizontally differentiable quantum principal \(\U\)-bundles over \(B\).

\begin{definition}
	Let \(\star_B\) be a Hodge operator on \((\Omega_B,\du_B)\).
	Let \((P;\Omega_{P,\hor},\du_{P,\hor})\) be a horizontally differentiable quantum principal \(\U\)-bundle over \(B\) with Fr\"{o}hlich automorphism \(\hat{\Phi}_P\); define a right \(\bZ\)-action on \(\mathcal{Z}_{>0}(B)\) by
	\begin{equation}\label{eq:horizontalact}
		\forall \mu \in \mathcal{Z}_{>0}(B), \, \forall k \in \bZ, \quad \mu \ract k \coloneqq \hat{\Phi}_P^{-k}(\mu).
	\end{equation}
	Then \((P;\Omega_{P,\hor},\du_{P,\hor})\) is \emph{\(\star_B\)-conformal} if there exists a (necessarily unique) group \(1\)-cocycle \(\mu_P : \bZ \to \mathcal{Z}_{>0}(B)\), the \emph{conformal factor} of \((P;\Omega_{P,\hor},\du_{P,\hor})\), such that
	\begin{multline}\label{eq:conformalhor}
		\forall (m,j) \in \bN_0 \times \bZ, \, \forall \alpha \in \Omega^m_B, \, \forall p \in P_j,\\
		{\star_B}(\alpha) p = \leg{\hat{\ell}_P\mleft(\alpha p\mright)}{0} {\star_B} \mleft(\leg{\hat{\ell}_P\mleft(\alpha p\mright)}{1}\mright)  \mu_P(j)^{N-2k},
	\end{multline}
	where \(\hat{\ell}_P : \Omega_{P,\hor} \to P \otimes_B \Omega_B\) is the \(B\)-bimodule isomorphism of Proposition \ref{prop:assoclineconn}.
	We denote by \(\dCirc_{\hor}(B;\star_B)\) the strictly full subcategory of \(\dCirc_\hor(B)\) with \(\star_B\)-conformal objects.
\end{definition}

\begin{proposition}\label{prop:conformal}
	Let \(\star_B\) be a Hodge operator on \((\Omega_B,\du_B)\).
	The essential image of \(\dCirc_{\hor}(B;\star_B)\) under the functor \(\hat{\cL}\) is \(\grp{Hom}(\bZ,\dPic(B;\star_B))\), so that the functor \(\epsilon_1 \circ \hat{\cL} : \dCirc_{\hor}(B;\star_B) \to \dPic(B;\star_B)\) is an equivalence of categories.
	In particular, if \((P;\Omega_{P,\hor},\du_{P,\hor})\) is a \(\star_B\)-conformal horizontally differentiable quantum principal \(\U\)-bundle over \(B\), then its conformal factor \(\mu_P\) satisfies
	\begin{equation}
		\mu_P = \mu \circ \pi_0\mleft(\hat{\cL}(P;\Omega_{P,\hor},\du_{P,\hor})\mright).
	\end{equation}
\end{proposition}

Thus, a Hermitian line \(B\)-bimodule with connection \((E,\sigma_E,\nabla_E)\) is \(\star_B\)-conformal if and only if \((P;\Omega_{P,\hor},\du_{P,\hor}) \coloneqq (B;\Omega_B,\du_B) \rtimes_{(E,\sigma_E,\nabla_E)} \bZ\) is \(\star_B\)-conformal, in which case, the conformal factor \(\mu_P\) of \((P;\Omega_{P,\hor},\du_{P,\hor})\) is determined by \(\mu_P(1) = \mu_{[E,\nabla_E]}\).

\subsection{The lifting problem for Riemannian structures \texorpdfstring{\emph{via}}{via} Hodge operators}\label{sec:4.2}

We now attack the problem of lifting Riemannian geometries in terms of Hodge operators to the total spaces of NC principal \(\U\)-bundles with connection.
The existence of such lifts will be governed by conformality in our NC sense, and the resulting lifted Riemannian geometries will necessarily involve modular phenomena in both vertical and horizontal directions that are generally non-trivial and distinct.

In what follows, let \(\kappa > 0\), let \((P;\Omega_P,\du_P;\Pi)\) be a \(\kappa\)-differentiable quantum principal \(\U\)-bundle with connection over \(B\), let \(\vartheta\) be the connection \(1\)-form of \(\Pi\), and let \(\hat{\Phi}_P\) be the Fr\"{o}hlich automorphism of \(\Hor_\kappa(P;\Omega_P,\du_P;\Pi) = (P,\Omega_{P,\hor},\du_{P,\hor})\).

We begin with a general definition of \(\U\)-equivariant Hodge operator on a NC total space that draws on standard requirements from the classical case: that the canonical surjection onto the base be a Riemannian submersion, that the principal Ehresmann connection be fibrewise orthogonal, that the fibres all have unit length, and that the total space have the 'fibre-first' orientation.
However, we carefully control possible failure of the Hodge operator to be right linear and \(\ast\)-preserving in terms of (possibly distinct) modular automorphisms in the vertical and horizontal directions.

On the one hand, we define a \emph{modular automorphism} of \(\Omega_P\) is a \(\U\)-equivariant automorphism \(\Delta\) of \(\Omega_P\) as a unital graded \(\bC\)-algebra satisfying \(\rest{\Delta}{\Omega_P^{\U}} = \id\) and 
\begin{equation}
	\forall j \in \bZ, \, \forall p \in P_j, \quad p^\ast \Delta(p) \geq 0;
\end{equation}
for example, given \(t \in (0,\infty)\), we may define a modular automorphism \(\Lambda_t\) of \(\Omega_P\) by 
\begin{equation}\label{eq:modularscale}
	\forall (m,j) \in \bN_0 \times \bZ, \, \forall \omega \in \Omega^m_j, \quad \Lambda_t(\omega) \coloneqq t^{-j}\omega.
\end{equation}
On the other hand, we use the connection \(\Pi\) to define a convenient bigrading \((\Omega^{j,k}_P)_{(j,k) \in \bN_0^2}\) of \(\Omega_P\) as follows.
For each \(k \in \Set{0,\dotsc,N}\), let
\begin{equation}\label{eq:bigrading}
	\Omega^{0,k}_P \coloneqq \Pi(\Omega^k_P) = \Omega^k_{P,\hor}, \quad \Omega^{1,k}_P \coloneqq (\id-\Pi)(\Omega^{k+1}_P) = \vartheta \cdot \Omega^k_{P,\hor},
\end{equation}
and for \((j,k) \notin \Set{0,1} \times \Set{0,\dotsc,N}\), set \(\Omega^{j,k}_P \coloneqq 0\).
Then \((\Omega^{j,k}_P)_{(j,k) \in \bN_0^2}\) satisfies:
\begin{align*}
	\forall m \in \Set{0,\dotsc,N+1}, && \bigoplus\nolimits_{j=0}^1 \bigoplus\nolimits_{k=0}^{m-1} \Omega^{j,k}_P &= \Omega^m_P,\\
	\forall (j_1,k_1),(j_2,k_2) \in \bN_0^2, && \Omega^{j_1,k_1}_P \cdot \Omega^{j_2,k_2}_P &= \Omega^{j_1+j_2,k_1+k_2}_P,\\
	\forall (j,k) \in \bN_0^2, && \left(\Omega^{j,k}_P\right)^\ast &= \Omega^{j,k}_P.
\end{align*}

\begin{definition}
	Let \((\Delta_\ver,\Delta_\hor)\) be a commuting pair of modular automorphisms of \(\Omega_P\) that commute with \(\Pi\).
	A \emph{\((\Delta_\ver,\Delta_\hor)\)-modular Hodge operator} on \((\Omega_P,\du_P)\) with respect to \(\Pi\) is a \(\U\)-equivariant left \(P\)-linear map that commutes with both \(\Delta_\ver\) and \(\Delta_\hor\), satisfies
	\begin{align}
		\forall (j,k) \in \Set{0,1} \times \Set{0,\dotsc,N}, && {\star}\mleft(\Omega^{j,k}_P\mright) &\subseteq \Omega^{1-j,N-k}_P, \label{eq:totalhodgeswap}\\
		\forall m \in \Set{0,\dotsc,N+1}, && \rest{\star^2}{\Omega^m_P} &= (-1)^{m(N+1-m)}\id_{\Omega^m_P}, \label{eq:totalhodgeinv}
	\end{align}
	and satisfies, for every \((j,k) \in \Set{0,1} \times\Set{0,\dotsc,N}\),
	\begin{align}
		\forall p \in P, \, \forall \omega \in \Omega^{j,k}_P, && {\star}(\omega p) &= {\star}(\omega) \cdot (\Delta_\hor^{2k-N} \circ \Delta_\ver^{2j-1})(p) ,\label{eq:totalhodgeright}\\
		\forall \omega \in \Omega^{j,k}_P, && {\star}(\omega)^\ast &= \star\mleft((\Delta_\hor^{2k-N} \circ \Delta_\ver^{2j-1})(\omega)^\ast\mright),\label{eq:totalhodgestar}\\
		\forall \omega \in \Omega^{j,k}_P, && {\star}\mleft(\delta^{j,0}\omega\mright) &= (-1)^{N-k}  {\star}(\omega\vartheta)\vartheta, \label{eq:totalhodgelift}\\ 
		\forall \omega,\eta \in \Omega^{j,k}_P, && \omega \cdot {\star}(\eta) &= {\star}^{-1}(\omega) \cdot (\Delta^{2k-N}_\hor \circ \Delta^{2j-1}_\ver)(\eta) \label{eq:totalhodgesym}.
	\end{align}
	Hence, in this case, the \emph{inverse metric} induced by the \((\Delta_\ver,\Delta_\hor)\)-modular Hodge operator \(\star\) is the \(\bR\)-bilinear map \(g : \Omega_P \times \Omega_P \to P\) defined by
	\begin{equation}
		\forall \omega,\eta \in \Omega_P, \quad g(\omega,\eta) \coloneqq {\star}\mleft(\omega^\ast \cdot {\star}(\eta)\mright).
	\end{equation}
\end{definition}

Notwithstanding the appearance of modular automorphisms, the following properties of inverse metrics will suffice for our purposes.

\begin{proposition}\label{prop:totalinversemetric}
	Let \(\Delta_\ver\) and \(\Delta_\hor\) be a commuting pair of modular automorphisms of \(\Omega_P\) that commute with \(\Pi\), and let \(\star\) be a \((\Delta_\ver,\Delta_\hor)\)-modular Hodge operator on \((\Omega_P,\du_P)\) with respect to \(\Pi\).
	Then the inverse metric \(g : \Omega_P \times \Omega_P \to P\) is \(\U\)-equivariant in the sense that
	\begin{equation}
		\forall (m,j), (n,k) \in \bN_0 \times \bZ, \, \forall \omega \in (\Omega^m_P)_j, \, \forall \eta \in (\Omega^n_P)_k, \quad g(\omega,\eta) \in P_{-j+k},
	\end{equation}
	makes \(\Pi\) into an orthogonal projection in the sense that
	\begin{equation}\label{eq:totalorthogonality}
		\forall \omega,\eta \in \Omega_P, \quad g(\omega,\eta) = g(\Pi(\omega),\Pi(\eta)) + g((\id-\Pi)(\omega),(\id-\Pi)(\eta)),
	\end{equation}
	and satisfies, for each \((j,k) \in \Set{0,1} \times \Set{0,\dotsc,N}\),
	\begin{align}
		\forall \omega,\eta \in \Omega^{j,k}_P, \, \forall p \in P, && g(\omega,\eta \cdot p) &= g(\omega,\eta) \cdot (\Delta_\ver^{2j} \circ \Delta_\hor^{2k})(p),\\
		\forall \omega,\eta \in \Omega^{j,k}_P, && g(\omega,\eta)^\ast &= (\Delta_\ver^{2j} \circ \Delta_\hor^{2k})\mleft(g(\eta,\omega)\mright).\label{eq:totalmodularsymmetry}
	\end{align}	
\end{proposition}

\begin{proof}
	The non-trivial claims are equations \ref{eq:totalorthogonality} and \ref{eq:totalmodularsymmetry}.
	On the one hand, let \(\omega,\eta \in \Omega_P\) be given; since \((\id-\Pi)(\Omega_P)^2 = 0\), it now follows by \eqref{eq:totalhodgeswap} that
	\begin{align*}
		\star{}\mleft(g(\omega,\eta)\mright)
		&= \left(\Pi(\omega^\ast) + (\id-\Pi)(\omega^\ast)\right) \star{}\mleft(\Pi(\eta) + (\id-\Pi)(\eta)\mright)\\
		&= \left(\Pi(\omega^\ast)\right) \star{}\mleft(\Pi(\eta)\mright) + \left(\id-\Pi)(\omega^\ast)\right) \star{}\mleft((\id-\Pi)(\eta)\mright)\\
		&= \star{}\mleft(g(\Pi(\omega),\Pi(\eta)) + g((\id-\Pi)(\omega),(\id-\Pi)(\eta))\mright).
	\end{align*}
	On the other hand, let \((j,k) \in \Set{0,1} \times \Set{0,\dotsc,N}\) and \(\omega,\beta \in \Omega^{j,k}_P\); by \eqref{eq:totalhodgesym},
	\begin{align*}
		{\star}\mleft(g(\omega,\eta)^\ast\mright)
		&= \Delta_\ver \circ \Delta_\hor^N\mleft({\star}(g(\omega,\eta))^\ast\mright)\\
		&= \Delta_\ver \circ \Delta_\hor^N\mleft((-1)^{(j+k)(N+1-j-k)}{\star}(\eta)^\ast \omega\mright)\\
		&= \Delta_\ver \circ \Delta_\hor^N\mleft((-1)^{(j+k)(N+1-j-k)}({\star} \circ \Delta^{2j-1}_\ver \circ \Delta^{2k-N}_\hor)(\eta^\ast) \cdot \omega\mright)\\
		&= \Delta^{2j}_\ver \circ \Delta^{2k}_\hor\mleft(\eta^\ast \cdot {\star}(\omega)\mright)\\
		&= {\star}\mleft(\Delta^{2j}_\ver \circ \Delta^{2k}_\hor\mleft(g(\eta,\omega)\mright)\mright). \qedhere
	\end{align*}
\end{proof}

At last, we give our proposed notion of lifted Riemannian geometry.

\begin{definition}
	A \emph{total Riemannian geometry} on \((P;\Omega_P,\du_P;\Pi)\) is a quadruple \((\Delta_\ver,\Delta_\hor,\star,\tau)\), where \((\Delta_\ver,\Delta_\hor)\) is a commuting pair of modular automorphisms of \(\Omega_P\) that commute with \(\Pi\), where \(\star\) is a \((\Delta_\ver,\Delta_\hor)\)-modular Hodge operator on \((\Omega_P,\du_P)\) with respect to \(\Pi\) whose inverse metric restricts, for each \((m,j) \in \bN_0 \times \bZ\), to a \(B\)-valued inner product on \((\Omega^m_P)_j\) admits a basis and satisfies
	\begin{equation}\label{eq:totalmetricbdd}
		\forall b \in B, \, \forall \omega \in (\Omega^m_P)_j, \quad g(b  \omega, b \omega) \leq \norm{b}^2 g(\omega,\omega),
	\end{equation}
	and where \(\tau\) is a \(\U\)-equivariant bounded state on \(P\) that satisfies
	\begin{align}
		\forall \omega \in \Omega^N_{P}, && (\tau \circ \star \circ \du)(\omega) &= 0; \label{eq:totalclosed}\\
		\forall p \in P, && \sup\Set{\tau(q^\ast p^\ast p q) \given q \in P, \, \tau(q^\ast q) \leq 1} &= \norm{p}^2. \label{eq:totalfaithful}
	\end{align}
\end{definition}

\begin{definition}
	Let \((\star_B,\tau_B)\) be a Riemannian geometry on \(B\) with respect to \((\Omega_B,\du_B)\), let \((\Delta_\ver,\Delta_\hor,\star,\tau)\) be a total Riemannian geometry on \((P;\Omega_P,\du_P;\Pi)\), and suppose that
	\begin{align}
		\forall \beta \in \Omega_B, && {\star}(\vartheta\beta) &= {\star_B}(\beta), \label{eq:hodgerest}\\
		\forall b \in B, && \tau(b) &= \tau_B(b). \label{eq:staterest}
	\end{align} 
	We call \((\star_B,\tau_B)\) a \emph{restriction} of \((\Delta_\ver,\Delta_\hor,\star,\tau)\) to \((B;\Omega_B,\du_B)\) and we call \((\Delta_\ver,\Delta_\hor,\star,\tau)\) a \emph{lift} of \((\star_B,\tau_B)\) to \((P;\Omega_P,\du_P;\Pi)\).
\end{definition}

Our definitions are justified by the following existence and uniqueness theorem, which characterizes existence of lifts in terms of conformality and demonstrates the inexorability of non-trivial modular phenomena outside of an narrow r\'{e}gime.

\begin{theorem}\label{thm:totalgeometry}
	Let \(\Lambda_\kappa\) denote the modular automorphism of \(\Omega_P\) defined by \eqref{eq:modularscale}.
	\begin{enumerate}[leftmargin=*]
		\item Suppose that \((\Delta_\ver,\Delta_\hor,\star,\tau)\) is a total Riemannian geometry on \((P;\Omega_P,\du_P;\Pi)\).
			There exists a unique restriction \((\star_B,\tau_B)\) of \((\Delta_\ver,\Delta_\hor,\star,\tau)\) to \((B;\Omega_B,\du_B)\).
			Moreover, it follows that \(\Delta_\ver = \Lambda_\kappa\) and that \((P;\Omega_{P,\hor},\du_{P,\hor})\) is \(\star_B\)-conformal with conformal factor \(\mu_P\) satisfying
			\begin{equation}\label{eq:liftedmod}
				\forall (m,j) \in \bN_0 \times \bZ, \, \forall \omega \in (\Omega^m_P)_j, \quad \Delta_\hor(\omega) = \omega  \mu_P(j)
			\end{equation}
		\item Let \((\star_B,\tau_B)\) be a Riemannian geometry on \((B;\Omega_B,\du_B)\), and suppose that \((P;\Omega_{P,\hor},\du_{P,\hor})\) is \(\star_B\)-conformal with conformal factor \(\mu_P\).
			Hence, define a modular automorphism \(\Delta_\hor\) of \(\Omega_P\) by \eqref{eq:liftedmod}.
			There exists a unique \((\Lambda_\kappa,\Delta_\hor)\)-modular Hodge operator \(\star\) on \((\Omega_P,\du_P)\) with respect to \(\Pi\) and faithful \(\U\)-equivariant bounded state \(\tau\) on \(P\) making \((\Lambda_\kappa,\Delta_\hor,\star,\tau)\) into a lift of \((\star_B,\tau_B)\) to \((P;\Omega_P,\du_P;\Pi)\), namely
		\begin{align}
			\forall p \in P, \, \forall k \in \Set{0,\dotsc,N}, \, \forall \beta \in \Omega^k_B, && {\star_P}\mleft(p\beta\mright) &\coloneqq (-1)^k p \vartheta {\star_B}(\beta), \label{eq:totalstara}\\
			\forall p \in P, \, \forall k \in \Set{0,\dotsc,N}, \, \forall \beta \in \Omega^k_B, && {\star_P}\mleft(p \vartheta\beta\mright) &\coloneqq p {\star_B}(\beta), \label{eq:totalstarb}\\
			\forall j \in \bZ, \, \forall p \in P_j, && \tau_P(p) &\coloneqq \tau_B\mleft(\delta^{j,0}p\mright).\label{eq:totalstate}
		\end{align}
	\end{enumerate}
\end{theorem}

\begin{lemma}\label{lem:modular}
	For every modular automorphism \(\Delta\) of \(\Omega_{P}\), there exists a unique group \(1\)-cocycle \(\mu : \bZ \to \mathcal{Z}_{>0}(B)\) 
	for the right \(\bZ\)-action defined by \eqref{eq:horizontalact}, such that
	\begin{equation}\label{eq:modular}
				\forall (m,j) \in \bN_0 \times \bZ, \, \forall \omega \in (\Omega^m_P)_j, \quad \Delta(\omega) = \omega  \mu(j).
	\end{equation}
	Conversely, for every group \(1\)-cocycle \(\mu : \bZ \to \mathcal{Z}_{>0}(B)\), Equation \ref{eq:modular} defines a modular automorphism \(\Delta\) of \(\Omega_{P}\).
\end{lemma}

\begin{proof}
	Let \(\Delta\) be a modular automorphism.
	For each \(j \in \bZ\), the map \(\Delta\) restricts to a \(B\)-bimodule morphism \(\cL(P)(j) \to \cL(P)(j)\), so that, by Proposition \ref{prop:imprimitivity}, there exists unique \(\mu(j) \in \Zent(B)\) that satisfies \eqref{eq:modular} for \(m=0\); given any cobasis \((e_i)_{i=1}^N\) for \(\cL(P)_j\), it follows that \(0 \leq \sum_{i=1}^N e_i^\ast  \Delta(e_i) = \sum_{i=1}^N e_i^\ast  e_i  \mu(j) = \mu(j)\) and
	\(
		\mu(j)  \alpha = \sum_{i=1}^N e_i^\ast \Delta(e_i)  \alpha = \sum_{i=1}^N e_i^\ast \Delta\mleft(e_i \alpha\mright) = \alpha  \mu_P(j)
	\)
	for all \(\alpha \in \Omega_B\),
	so that \(\mu_{P}(j) \in \Zent(\Omega_B)^0\).
	Given \(j,k \in \bZ\), for all \(x \in P_j\), and \(y \in P_k\), we find that
	\(
		xy  \mu(j+k) = \Delta(xy) = \Delta(x)\Delta(y) = x  \mu_P(j)\cdot y \cdot \mu_P(k) = xy  \Phi_P^{-k}(\mu(j))\mu(k)
	\),
	so that \(\mu_P(j+k) = \hat{\Phi}_P^{-k}(\mu_P(j))\mu_P(k)\) by the equality \(P_{j+k} = P_j \cdot P_k\) together with uniqueness of \(\mu_P(j+k)\).
	Since \(\Delta\) acts as the identity on \(P_0\), it follows that \(\mu(0) = 1\); hence, for each \(j \in \bZ\), it follows that \(\mu(j) \in \mathcal{Z}_{>0}(B)\)  with inverse \(\Phi_P^{-j}(\mu(-j))^{-1}\).
	Thus, we obtain a unique group \(1\)-cocycle \(\mu: \bZ \to \mathcal{Z}(B)^\times_{\geq 0}\) satisfying \eqref{eq:modular} for \(m = 0\).
	Finally, since \(\Omega_{P} = P \cdot \Omega_B \oplus P \cdot \vartheta \cdot \Omega_B\) and since \(\Delta\) acts as the identity on \(\Omega_B\) and \(\vartheta\), it follows that \(\mu: \bZ \to \mathcal{Z}_{>0}(B)\) is the unique group \(1\)-cocycle satisfying \eqref{eq:modular} in general.
	Reversing this argument almost suffices to show that a group \(1\)-cocycle \(\mu: \bZ \to \mathcal{Z}_{>0}(B)\) defines a modular automorphism \(\Delta\) by \eqref{eq:modular}; all that is left is that
	\(
		p^\ast \Delta(p) = p^\ast  p \cdot \iota_P\mleft(\mu_P(j)\mright) = p^\ast \Phi^{j}_P(\mu_P(j))  p \geq 0
	\)
	for all \(j \in \bZ\) and \(p \in P_j\).
\end{proof}

\begin{proof}[Proof of Theorem \ref{thm:totalgeometry}]
	First, suppose that \((\Delta_\ver,\Delta_\hor,\star,\tau)\) is a total Riemannian geometry on \((P;\Omega_P,\du_P;\Pi)\).
	We begin by showing that \(\Delta_\ver = \Lambda_\kappa\).
	By Lemma \ref{lem:modular}, let \(\mu : \bZ \to \mathcal{Z}_{>0}(B)\) be the unique group \(1\)-cocycle satisfying \eqref{eq:modular} with respect to \(\Delta = \Delta_\ver\).
	Then \(\Delta_\ver^2 = \Lambda_\kappa^2\) since, for every \(p \in P\),
	\begin{multline*}
		{\star}(p) = (-1)^N {\star}(p\vartheta)\vartheta
		= (-1)^N {\star}(\vartheta) \cdot (\Delta^{-N}_\hor \circ \Delta_\ver \circ \Lambda_\kappa^{-1})(p) \cdot \vartheta\\
		= {\star}(1) \cdot (\Delta^{-N}_\hor \circ \Delta_\ver \circ \Lambda^{-2}_\kappa)(p)
		= {\star}\mleft(\Delta^2_\ver \circ \Lambda_\kappa^{-2}(p)\mright).
	\end{multline*}
	Let \(j \in \bZ\) and let \((e_i)_{i=1}^N\) be a cobasis for \(\cL(P)(j)\).
	Then \(\mu(j) = \kappa^{-j}\) since \(\mu(j) \geq 0\) and
	\(
		\kappa^{-2j} = \sum_{i=1}^N e_i^\ast  \Lambda_\kappa^2(e_i) = \sum_{i=1}^N e_i^\ast \Delta_\ver^2(e_i) = \mu(j)^2
	\).
	
	Next, we show that there is a unique Hodge operator \(\star_B\) on \(B\) with respect to \((\Omega_B,\du_B)\) satisfying \eqref{eq:hodgerest}.
	By \eqref{eq:totalhodgelift} and \eqref{eq:totalhodgeinv}, there exists a unique \(\bC\)-linear map \(\star_B : \Omega_{B} \to \Omega_B\) satisfying \eqref{eq:hodgerest}, which is given by \({\star_B}(\beta) \coloneqq {\star}(\vartheta\beta)\) for all \(k \in \Set{0,\dotsc,N}\) and \(\beta \in \Omega^k_B\).
	The map \(\star_B\) is left \(B\)-linear by construction and \(\U\)-equivariant by \(\U\)-equivariance of \(\star\) and \(\U\)-invariance of \(\vartheta\).
	Moreover, since both \(\Lambda_\kappa\) and \(\Delta_\hor\) act as the identity on \(\Omega_B\) and on \(\vartheta\) and since \(\vartheta\) supercommutes with \(\Omega_B\), it follows that \(\star_B\) is right \(B\)-linear by \eqref{eq:totalhodgeright}, is \(\ast\)-preserving by \eqref{eq:totalhodgestar}, satisfies \eqref{eq:hodge1} by \eqref{eq:totalhodgeswap} and \eqref{eq:totalhodgeinv}, and satisfies \eqref{eq:hodge2} by \eqref{eq:totalhodgesym}.
	
	Next, we show that the pair \((\star_B,\rest{\tau}{B})\) defines a Riemannian geometry on \(B\) with respect to \((\Omega_B,\du_B)\).
	On the one hand, since both \(\Lambda_\kappa\) and \(\Delta_\hor\) act trivially on \(\Omega_B\), Proposition \ref{prop:totalinversemetric} together with the \(j=0\) case of \eqref{eq:totalmetricbdd} shows that the inverse metric induced by \(\star_B\) satisfies \eqref{eq:metricbdd}; indeed, for each \(m \in \Set{0,\dotsc,M}\), one can obtain a basis for \(\Omega^m_B\) from a basis for \((\Omega_P^m)^{\U}\) by applying \(\Pi\) retaining any non-zero vectors.
	On the other hand,
	\(
		\du_P(\vartheta  \beta) = \du_P(\vartheta)  \beta - \vartheta  \du_B(\beta) = -\cF_\Pi  \beta - \vartheta  \du_B(\beta) = -\vartheta  \du_B(\beta)
	\)
	for every \(\beta \in \Omega^{N-1}_B\), so that
	\(
		\tau \circ {\star_B} \circ \du_B(\beta) = \tau\mleft(-{\star}\mleft(\vartheta  \du_B(\beta)\mright)\mright) = \tau \circ {\star} \circ \du(\vartheta  \beta) = 0
	\).
	
	Finally, by Lemma \ref{lem:modular}, let \(\mu_P : \bZ \to \mathcal{Z}_{>0}(B)\) be the unique group \(1\)-cocycle satisfying \eqref{eq:liftedmod}.
	Given \eqref{eq:liftedmod} and \(\Delta_\hor(\vartheta) = \vartheta\), that \((P;\Omega_{P,\hor},\du_{P,\hor})\) is \(\star_B\)-conformal with conformal factor \(\mu_P\) follows from \eqref{eq:totalhodgeright}.
	Uniqueness of \(\tau_B\) is trivial.
	
	Now, let \((\star_B,\tau_B)\) be a Riemannian geometry on \((B;\Omega_B,\du_B)\), and suppose that \((P;\Omega_{P,\hor},\du_{P,\hor})\) is \(\star_B\)-conformal with conformal factor \(\star_B\).
	By Lemma \ref{eq:modular}, the modular automorphism \(\Delta_\hor\) of \(\Omega_P\) constructed from \(\mu_P\) by \eqref{eq:liftedmod} is well-defined.
	
	We first show that there is a unique \((\Delta_\ver,\Delta_\hor)\)-modular Hodge operator on \((\Omega_P,\du_P)\) with respect to \(\Pi\) satisfying \eqref{eq:hodgerest}.
	First, by Proposition \ref{prop:assoclineconn}, \eqref{eq:totalstara} and \eqref{eq:totalstarb} define the unique \(\U\)-equivariant left \(P\)-linear map \(\star : \Omega_P \to \Omega_P\) satisfying \eqref{eq:hodgerest} and \eqref{eq:totalhodgelift}.
	Next, the map \(\star\) satisfies \eqref{eq:totalhodgeswap} by construction, satisfies \eqref{eq:totalhodgeright} by \eqref{eq:liftedmod} and \eqref{eq:conformalhor}, satisfies \eqref{eq:totalhodgeinv} by \eqref{eq:hodge1} and left \(P\)-linearity, and satisfies \eqref{eq:totalhodgestar} by the fact that \(\star_B\) is \(\ast\)-preserving together with left \(P\)-linearity of \(\star\) and \eqref{eq:totalhodgeright}.
	Finally, the map \(\star\) satisfies \eqref{eq:totalhodgesym} by a case-by-case application of \eqref{eq:hodge2} together with \eqref{eq:totalhodgeright} and left \(P\)-linearity of \(\star\).
	
	Next, we turn to the inverse metric \(g\) induced by \(\star\).
	Let \(g_B\) be the inverse metric induced by \(\star_B\), let \((m,j) \in \Set{0,\dotsc,N} \times \bZ\), and let \(\hp{}{}_j\) be the \(B\)-valued inner product on \(\cL(P)(j)\).
	Let \(p_1,p_2 \in P_j\) and \(\alpha_1,\alpha_2 \in \Omega^m_B\).
	On the one hand,
	\[
		\star\mleft(g(p_1\alpha_1,p_2\alpha_2)\mright) = \alpha_1^\ast p_1^\ast {\star}(p_2\alpha_2) = \alpha_1^\ast p_1^\ast p_2 (-1)^N  {\star}_B(\alpha_2)\vartheta = \star\mleft(g_B(\alpha_1,\hp{p_1}{p_2}_j\alpha_2)\mright),
	\]
	while on the other, \(g(p_1\alpha_1\vartheta,p_2\alpha_2\vartheta) = g_B(\alpha_1,\hp{p_1}{p_2}_j\alpha_2)\) by a similar calculation.
	Thus, the \(B\)-bimodule isomorphism \(\hat{\ell}_P\) of Proposition \ref{prop:assoclineconn} yields a \(B\)-bimodule isomorphism
	\(
		\left(\omega \mapsto \hat{\ell}_P(\omega)\right) : (\Omega^m_{P,\hor})_j \to \cL(P)(j) \otimes_B \Omega^m_B\) and a \(B\)-bimodule isomorphism \(\left(\omega \cdot \vartheta \mapsto \hat{\ell}_P(\omega)\right) : \vartheta \cdot (\Omega^m_{P,\hor})_j \to \cL(P)(j) \otimes_B \Omega^m_B
	\)
	that respectively realise the restrictions of \(g\) to \((\Omega^m_{P,\hor})_j\) and \(\vartheta \cdot (\Omega^m_{P,\hor})_j\) as pullbacks of the \(B\)-valued inner product on the tensor product \(\cL(P)(j) \otimes_B \Omega^m_B\) of \(B\)-self-correspondences of finite type.
	Hence, both \((\Omega^m_{P,\hor})_j = \Pi((\Omega^m_P)_j)\) and \(\vartheta (\Omega^m_{P,\hor})_j = (\id-\Pi)((\Omega^{m+1}_P)_j)\) define \(B\)-self-correspondences of finite type with respect to \(g\), which suffices for us.

	Finally, we show that \eqref{eq:totalstate} defines the unique \(\U\)-equivariant bounded state \(\tau\) on \(P\) satisfying \(\rest{\tau}{B} = \tau_B\), \eqref{eq:totalclosed} and \eqref{eq:totalfaithful}.
	Recall the bounded faithful conditional expectation \(\bE_P : P \to B\) of Proposition \ref{prop:conditional}.
	First, the map \(\tau : P \to \bC\) defined by \eqref{eq:totalstate} can now be rewritten as \(\tau = \tau_B \circ \bE_P\), which therefore defines a faithful bounded \(\U\)-equivariant state on \(P\) restricting to \(\tau_B\) on \(B\).
	Next, by continuity and \(\U\)-invariance, any faithful bounded \(\U\)-equivariant state \(\tau^\prime\) on \(P\) satisfying \(\rest{\tau^\prime}{B} = \tau_B\) must satisfy \(\tau^\prime = \tau^\prime \circ \bE_P = \tau_B \circ \bE_P = \tau\).
	Finally, we show that \(\tau\) satisfies \eqref{eq:totalclosed} with respect to \(\star\).
	On the one hand, let \(j \in \bZ\), \(p \in P_j\), and \(\beta \in \Omega^N_B\).
	Since
	\(
		\delta^{j,0}\du_P(p\beta) = \delta^{j,0}\left(2\pi\iu{}[j]_\kappa \vartheta  p\beta + \du_{P,\hor}(p)\beta + p\du_B\beta\right) = 0 
	\),
	it follows by \eqref{eq:totalstara} that
	\(
		\tau \circ \star \circ \du(p\beta) = \tau_B \circ \star\mleft(\delta^{j,0}\du(p\beta)\mright) = 0
	\).
	On the other hand, let \(j \in \bZ\), \(p \in P_j\), and \(\alpha \in \Omega^{N-1}_B\) be given.
	Since
	\[
		\delta^{j,0}\du_P(p \alpha \vartheta) = \du_B\mleft((\delta^{j,0}p)\alpha\mright) \cdot \vartheta  + (-1)^N (\delta^{j,0}p)\alpha \cF_\Pi  = \du_B\mleft((\delta^{j,0}p)\alpha\mright)\vartheta,
	\]
	it follows by \eqref{eq:totalstarb} and \eqref{eq:totalstate} that
	\[	\tau \circ \star \circ \du(p\alpha\vartheta) = \tau_B \circ \star\mleft(\du_P(\delta^{j,0}p\alpha\vartheta)\mright) = (-1)^N\tau_B \circ \star_B \circ \du_B\mleft((\delta^{j,0}p)\alpha\mright) = 0.
	\]
	Thus, either way, the composition \(\tau \circ {\star} \circ \rest{\du}{\Omega^N_P}\) vanishes.
	
	Let us now show that \(\tau\) satisfies \eqref{eq:totalfaithful}.
	Define \(\norm{}^\prime : P \to [0,\infty)\) by
	\[
		\forall p \in P, \quad (\norm{p}^\prime)^2 \coloneqq \sup\Set{\tau(q^\ast p^\ast p q) \given q \in P, \, \tau(q^\ast q) \leq 1}.
	\]
	Since \(\norm{}^\prime\) is the operator norm on \(P\) with respect to the \textsc{gns} representation of \(P\) induced by the faithful bounded state \(\tau\), it follows that \(\norm{}^\prime\) is a \Cstar-norm bounded from above by \(\norm{}\); since \(\tau\) is \(\U\)-invariant, it follows that \(\norm{}^\prime\) is a \(\U\)-invariant \Cstar-norm on \(P\).
	Hence, by Corollary \ref{cor:fell}, it suffices to show that \(\tau\) satisfies \eqref{eq:totalfaithful} on \(P^{\U} = B\).
	But now, given \(b \in B\), it follows from \eqref{eq:faithfulstate} applied to \(\tau_B\) that
	\begin{align*}
		(\norm{b}^\prime)^2 &= \sup\Set{\tau(q^\ast p^\ast p q) \given q \in P, \, \tau(q^\ast q) \leq 1}\\ &\geq \sup\Set{\tau(c^\ast p^\ast p c) \given q \in B, \, \tau(c^\ast c) \leq 1}\\ &= \norm{b}^2. \qedhere
	\end{align*}
\end{proof}

The construction of Lemma \ref{lem:modular} will be used frequently enough to warrant the following definition.

\begin{definition}
	Equip \(\mathcal{Z}_{>0}(B)\) with the right \(\bZ\)-action constructed from \(\hat{\Phi}_P\) by \eqref{eq:horizontalact}.
	The \emph{symbol} of a modular automorphism \(\Delta\) is the unique group \(1\)-cocycle \(\mu : \bZ \to \mathcal{Z}_{>0}(B)\) that satisfies \eqref{eq:modular} with respect to \(\Delta\).
\end{definition}

In particular, we may use Proposition \ref{prop:conformal} to rewrite Theorem \ref{thm:totalgeometry} as follows.

\begin{corollary}
	Let \((\star_B,\tau_B)\) be a Riemannian geometry on \((B;\Omega_B,\du_B)\).
	Let \((E,\sigma_E,\nabla_E)\) be a Hermitian line \(B\)-bimodule with connection that is flat or has vertical deformation parameter \(\kappa\).
	Then \((\star_B,\tau_B)\) admits a lift \((\Delta_\ver,\Delta_\hor,\star,\tau)\) to \((B;\Omega_B,\Sigma_B) \rtimes_{(E,\sigma_E,\nabla_E)}^{\kappa,\tot} \bZ\) if and only if \((E,\sigma_E,\nabla_E)\) is \(\star_B\)-conformal, in which case the lift is unique, \(\Delta_\ver = \Lambda_\kappa\), and \(\Delta_\hor\) has symbol \(\mu \circ \left(m \mapsto [E,\nabla_E]^m\right)\).
\end{corollary}

\begin{example}\label{ex:hopf7}
	Continuing from Examples \ref{ex:hopf5} and \ref{ex:hopf6}, observe that
	\[
		(\cO_q(\SU(2)),\Omega_{q,\hor}(\SU(2)),\du_{q,\hor}) = \Hor_\kappa(\cO_q(\SU(2)),\Omega_q(\SU(2)),\du_q,\Pi_q)
	\]
	is \(\star_q\)-conformal with conformal factor \(k \mapsto q^{-k}\); compare \cite[Lemma 5.6]{Zampini}.
	Moreover, recall that the usual basis for the free left \(\cO_q(\SU(2))\)-module \(\Omega_q(\SU(2))\) is \(\Set{e^0,e^+,e^-}\), where \(e^0 \coloneqq 2\pi\iu{}q^{-2}e_{q^2}\).
	Hence, by Theorem \ref{thm:totalgeometry}, the unique lift of \((\star_q,\rest{h_q}{\cO_q(\CP^1)})\) to \((\cO_q(\SU(2));\Omega_q(\SU(2)),\du_q;\Pi_q)\) is \((\Lambda_{q^2},\Lambda_q,\widetilde{\star}_q,h_q)\), where \(\widetilde{\star}_q\) is uniquely determined by \({\widetilde{\star}_q}(1) \coloneqq \tfrac{q^2}{2\pi} e^0 e^+  e^-\) and
	\[
		{\widetilde{\star}_q}(e^0) \coloneqq -\tfrac{2\pi}{q^2} e^+  e^-, \quad {\widetilde{\star}_q}(e^+) \coloneqq -\tfrac{q^6}{2\pi}e^0e^+,  \quad {\widetilde{\star}_q}(e^-) \coloneqq \tfrac{1}{2\pi q^2}e^0 e^-,
	\]
	and \(h_q\) is Woronowicz's Haar state on \(\cO_q(\SU(2))\).
	In fact, \({\widetilde{\star}_q}\) recovers the Hodge operator of Zampini \cite[Eq.\ 4.20]{Zampini} for \((\alpha^\prime,\beta,\nu,\gamma) = (\tfrac{-2\pi}{q^{4}},1,q^{-2},\tfrac{4\pi^2}{q^{4}})\), which satisfies his canonical constraints \cite[Remark 5.7]{Zampini} with respect to the choice of parameter \(\alpha^{\prime\prime} = -q^2\) from Example \ref{ex:hopf7}; it also necessarily recovers the Hodge operator of Kustermans--Murphy--Tuset \cite[Thm.\ 8.1 \emph{et seq.}]{KMThodge} up to suitable rescaling in each respective degree.
	Note that \(\Lambda_{q^2} \neq \Lambda_q\) since \(q \neq 1\).
\end{example}

\begin{example}\label{ex:heis9}
	We continue from Examples \ref{ex:heis6} and \ref{ex:heis7}.
	By Example \ref{ex:heis5}, the homomorphism \(\hat{E}\) of Example \ref{ex:heis4} is a homomorphism \(\hat{E} : \Gamma_\theta \to \dPic(C^\infty_\theta(\bT^2);\star)\) satisfying
	\(
		\mu \circ \pi_0(\hat{E}) = \left(g \mapsto (g_{21}\theta+g_{22})^{-1}\right)
	\).
	Hence, by Proposition \ref{prop:conformal} and Theorem \ref{thm:totalgeometry}, the unique lift of \((\star,\tau)\) from Example \ref{ex:heis7} to the real multiplication instanton \((P_\theta,\Omega_{P_\theta},\du_{P_\theta},\Pi_{P_\theta})\) is \((\Lambda_{\epsilon_\theta^2},\Lambda_{\epsilon_\theta},\widetilde{\star},\widetilde{\tau})\), where \(\widetilde{\star}\) is determined by
	\[
		{\widetilde{\star}}(1) \coloneqq e^0e^1e^2, \quad {\widetilde{\star}}(e^0) \coloneqq e^1e^2, \quad {\widetilde{\star}}(e^1) \coloneqq -e^0e^2, \quad {\widetilde{\star}}(e^2) \coloneqq e^0e^1,
	\]
	and \(\tilde{\tau}\) is determined by \(\rest{\tilde{\tau}}{C^\infty_\theta(\bT^2)} = \tau\).
	Note that \(\Lambda_{\epsilon_\theta^2} \neq \Lambda_{\epsilon_\theta}\) since \(\epsilon_\theta > 1\).
\end{example}

We conclude this section by showing that modular phenomena are no obstacle to equipping \(\Omega_P\) with an \(L^2\)-inner product and computing the formal adjoint of the exterior derivative \(\du_P\) in terms of the Hodge star operator.

\begin{proposition}\label{prop:totalprehilbert}
	Suppose that \((\Delta_\ver,\Delta_\hor,\star,\tau)\) be a total Riemannian geometry on \((P;\Omega_P,\du_P;\Pi)\) with inverse metric \(g\).
	Then \(\Omega_P\) defines a pre-Hilbert space with respect to the inner product \(\ip{}{}_\tau\) defined by
	\begin{equation}
		\forall \omega,\eta \in \Omega_P, \quad \ip{\omega}{\eta} \coloneqq \tau\mleft(g(\omega,\eta)\mright).
	\end{equation}
	Moreover, with respect to this pre-Hilbert space structure, the \(\U\)-action on \(\Omega_P\) defines a unitary representation of finite type, \(\Omega_P = \bigoplus_{j=0}^1 \bigoplus_{k=0}^N \Omega^{j,k}_P\) is an orthogonal direct sum, the left \(P\)-module structure on \(\Omega_P\) defines a \(\U\)-equivariant isometric \(\ast\)-representation on \(P\), the connection \(\Pi\) defines an orthogonal projection, the operator \(\star_P\) is unitary, and
	\(
		\du_P^\ast = {\star}^{-1} \circ {\du_P} \circ {\star} \circ \chi_P
	\).
\end{proposition}

\begin{proof}
	Recall the faithful conditional expectation \(\bE_P : P \to B\) of Proposition \ref{prop:conditional}, so that the state \(\tau\) satisfies \(\tau = \tau \circ \bE_P = \tau_B \circ \bE_P\).
	
	First, let \((m,j) \in \Set{0,\dotsc,N} \times \bZ\) be given, so that \(g\) makes \((\Omega^m_P)_j\) and hence its direct summands \((\Omega^{0,m}_P)_j = \Pi((\Omega^m_P)_j)\) and \((\Omega^{1,m-1}_P)_j = (\id-\Pi)((\Omega^m_P)_j)\) into \(B\)-self-correspondences of finite type.
	Since the state \(\tau\) is faithful and positive, \(\ip{}{}_\tau\) restricts to \(\U\)-invariant positive-definite inner products on both \((\Omega^{0,m}_P)_j\) and \((\Omega^{1,m-1}_P)_j\); moreover, the proof of Proposition \ref{prop:basicseparable} shows that both \((\Omega^{0,m}_P)_j\) and \((\Omega^{1,m-1}_P)_j\) are separable as pre-Hilbert spaces by separability of \(B\) as a pre-\Cstar-algebra.
	But now, by \eqref{eq:totalstate}, for every \((m,j),(n,k) \in \Set{0,\dotsc,N} \times \bZ\), \(\omega \in (\Omega^{0,m}_{P,\hor})_j\), and \(\eta \in (\Omega^{0,n}_{P,\hor})_k\), we find that
	\(
		\ip{\omega}{\eta}_\tau = \tau\mleft(g(\omega,\eta)\mright) = \tau\mleft(\delta^{m,n}g(\omega,\eta)\mright) = \tau\mleft(\delta^{m,n}\delta^{j,k}g(\omega,\eta)\mright)
	\)
	and similarly that \(\ip{\omega\vartheta}{\eta\vartheta}_\tau = \tau\mleft(\delta^{m,n} g(\omega\vartheta,\eta\vartheta)\mright)\), while \(\ip{\omega}{\eta\vartheta}_\tau = \ip{\omega\vartheta}{\eta}_\tau = 0\) by \eqref{eq:totalorthogonality}.
	This shows that \(\Omega_{P} = \bigoplus_{j=0}^1 \bigoplus_{k=0}^N \bigoplus_{\ell=-\infty}^\infty (\Omega^{j,k}_P)_\ell\) is orthogonal with respect to \(\ip{}{}_\tau\), so that \(\ip{}{}_\tau\) defines a \(\U\)-invariant positive-definite inner product on \(\Omega_{P}\), with respect to which \(\Omega_{P}\) is separable and the \(\U\)-action is unitary and of finite type.	

	Next, we show that the left \(P\)-module structure on \(\Omega_{P}\) yields a \(\U\)-equivariant isometric \(\ast\)-representation of \(P\); note that \(\U\)-equivariance is automatic.
	First, we show that each \(p \in P\) acts as an adjointable operator on \(\Omega_{P,\hor}\) with adjoint given by \(p^\ast\).
	Indeed, let \(p \in P\).
	Then, for all \(\omega,\eta \in \Omega_P\),
	\[
		{\star}\mleft(g(p\omega,\eta)\mright) = (p \omega)^\ast \cdot {\star}(\eta) = \omega^\ast  p^\ast {\star}(\eta) = \omega^\ast {\star}(p^\ast  \eta) = {\star}\mleft(g(\omega,p^\ast  \eta)\mright)
	\]
	so that \(\ip{p\omega}{\eta}_\tau = \ip{\omega}{p^\ast  \eta}\).
	Now, let us show that each \(p \in P\) acts as a bounded operator on \(\Omega_{P}\).
	Indeed, let \(p \in P\) be given, and write \(p = \sum_{k\in\bZ} \hat{p}(k)\), where \(\hat{p}(k) \in P_k\) for each \(k \in \bZ\), so that
	\(
		E(p^\ast p) = E\mleft(\sum_{k,\ell \in \bZ} \hat{p}(k)^\ast \hat{p}(\ell)\mright) = \sum_{k \in \bZ} \hat{p}(k)^\ast \hat{p}(k).
	\)
	Let \((m,j) \in \Set{0,\dotsc,N} \times \bZ\) and let \(\omega \in (\Omega^m_{P})_j\).
	Then, by adjointability of \(p\), Equation \ref{eq:totalmetricbdd}, the proof of Proposition \ref{prop:basicseparable}, and contractivity of \(E\),
	\[
		(\bE_P \circ g)(p \omega,p \omega)) =  (\bE_P \circ g)\mleft(\omega,\sum_{k,\ell \in \bZ}\hat{p}(k)^\ast \hat{p}(\ell) \omega\mright) = g\mleft(\omega,\bE_P(p^\ast p) \omega\mright)  \leq \norm{p}^2 g(\omega,\omega).
	\]
	Thus, the left \(P\)-module structure defines a bounded \(\ast\)-representation of \(P\), which is isometric, \emph{mutatis mutandis}, by the proof of Proposition \ref{prop:basicseparable}.
	
	Next, we show that \(\du_{P}\) is adjointable with adjoint \({\star}^{-1} \circ {\du_P} \circ {\star} \circ \chi_P\).
	Let \(m \in \Set{0,\dotsc,N+1}\), let \(\omega \in \Omega^{m-1}_P\), and let \(\eta \in \Omega^m_P\).
	Then, since \(\tau_P \circ {\star} \circ \du_P\mleft(\omega^\ast \eta\mright) = 0\) by \eqref{eq:totalclosed}, it follows that
	\begin{align*}
		\du_P(\omega)^\ast  {\star}(\eta)
		&= \du_P\mleft(\omega^\ast \eta\mright) + (-1)^m	 \omega^\ast (\du_P \circ {\star})(\eta)\\
		&= \du_P\mleft(\omega^\ast \eta\mright) + \omega^\ast  {\star}\mleft(({\star}^{-1} \circ \du_P \circ {\star} \circ \gamma_P)(\eta)\mright).
	\end{align*}
	
	Finally, we show that \(\star_P\) is unitary. 
	Let \((j,k) \in \Set{0,1} \times \Set{0,\dotsc,N}\); let \(\omega,\eta \in \Omega^{j,k}_P\).
	Then \(\ip{{\star}(\omega)}{{\star}(\eta)}_\tau = \ip{\omega}{\eta}_\tau\) since
	\begin{align*}
		{\star}(\omega)^\ast \cdot {\star}\mleft({\star}(\eta)\mright)
		&= {\star}\mleft((\Delta^{1-2j}_\ver \circ \Delta^{N-2k}_\hor)(\omega^\ast)\mright) \cdot (-1)^{m(N+1-m)}\eta\\
		&= (\Delta^{1-2j}_\ver \circ \Delta^{N-2k}_\hor)\mleft({\star}^{-1}(\omega^\ast) \cdot (\Delta^{2k-N}_\hor \circ \Delta^{2j-1}_\ver)(\eta)\mright)\\
		&= (\Delta^{1-2j}_\ver \circ \Delta^{N-2k}_\hor)\mleft(\omega^\ast \cdot {\star}(\eta)\mright). \qedhere
	\end{align*}
\end{proof}

\subsection{Unbounded lifts of commutator representations}\label{sec:4.3}

We now consider the analogous lifting problem for Connes's NC Riemannian geometry in terms of \emph{spectral triples}~\cite{Connes95}.
Here, analogues of Dirac-type operators simultaneously encode differential calculus (to first order), index theory, Riemannian geometry, and metric geometry.
Following Schm\"{u}dgen \cite{Schmuedgen}, we restrict our attention to the first aspect and consider \emph{commutator representations} of \(\ast\)-exterior algebras through degree \(1\).
However, the resulting lifted commutator representations will once more invovle modular phenomena in the form of unboundedness of represented \(1\)-forms.

Just as before, let \(\kappa > 0\), let \((P;\Omega_P,\du_P;\Pi)\) be a \(\kappa\)-differentiable quantum principal \(\U\)-bundle with connection over \(B\), let \(\vartheta\) be the connection \(1\)-form of \(\Pi\), and let \(\hat{\Phi}_P\) be the Fr\"{o}hlich automorphism of \(\Hor_\kappa(P;\Omega_P,\du_P;\Pi) = (P,\Omega_{P,\hor},\du_{P,\hor})\).
Moreover, given a pre-Hilbert space \(H\), let \(\bL(H)\) denote the unital pre-\Cstar-algebra of bounded adjointable operators on \(H\), which is \(\bZ/2\bZ\)-graded as a \(\ast\)-algebra whenever the \(H\) is as a pre-Hilbert space.

We begin with a simplified version of the notion of spectral triple, which we shall apply to the NC base manifold \((\Omega_B,\du_B)\).
In short, it generalises Clifford actions of \(1\)-forms in terms of bounded commutators with a Dirac-type operator~\cite{CPR}.

\begin{definition}[{Baaj--Julg~\cite{BJ}, Connes~\cite{Connes95}, Schm\"{u}dgen~\cite{Schmuedgen}}]
	Let \(H\) be a separable \(\bZ/2\bZ\)-graded pre-Hilbert space equipped with a bounded \(\ast\)-homomorphism \(\pi : B \to \bL(H)\) and an odd symmetric \(\bC\)-linear map \(D : H \to H\), so that \(\bL(H)\) defines a \(B\)-bimodule with respect to \(\pi\).
	We call \((H,\pi,D)\) a \emph{bounded commutator representation} of \((B;\Omega_B,\du_B)\) whenever there exists a (necessarily unique) \(B\)-bimodule homomorphism \(\pi_D : \Omega^1_B \to \bL(H)\), such that
	\begin{equation}
		\forall b \in B, \quad \pi_D \circ \du_B(b) = \iu{}[D,\pi(b)];
	\end{equation}
	hence, we call \((H,\pi,D)\) \emph{faithful} whenever \(\pi\) is isometric and \(\pi_D\) is injective.
\end{definition}

\begin{remark}\label{rem:kk1}
	Let \(\mathfrak{B}\) denote the \Cstar-algebra completion of \(B\).
	A bounded commutator representation \((H,\pi,D)\) of \((B;\Omega_B,\du_B)\) defines an even spectral triple for \(\mathfrak{B}\) if and only if \(D\) is essentially self-adjoint and has compact resolvent.
\end{remark}

\begin{example}[D\k{a}browski--Sitarz \cite{DS01}, Majid \cite{Majid05}]\label{ex:hopf8}
	We continue from Example \ref{ex:hopf6}.
	Let \(\slashed{S}_{q,\pm}(\CP^1) \coloneqq \cO_q(\SU(2))_{\mp 1}\) with inner product \(\ip{}{}\) given by
	\(
		\ip{s_1}{s_2} \coloneqq h_q(s_1^\ast s_2)
	\)
	for all \(s_1,s_2 \in \slashed{S}_{q,\pm}(\CP^1)\); hence, by the proof of Proposition \ref{prop:basicseparable} and faithfulness of \(h_q\) on \(C_q(\SU(2))\) \cite{Nagy}, each of \(\slashed{S}_{q,\pm}(\CP^1)\) defines a separable pre-Hilbert space admitting isometric \(\pi_{\pm} : \cO_q(\CP^1) \to \bL(\slashed{S}_{q,\pm}(\CP^1))\), respectively, given by left multiplication in \(\cO_q(\SU(2))\).
	Hence, let \(\slashed{S}_q(\CP^1) \coloneqq \slashed{S}_{q,+}(\CP^1) \oplus \slashed{S}_{q,-}(\CP^1)\) as an orthogonal direct sum of pre-Hilbert spaces with \(\bZ/2\bZ\)-grading \(\id \oplus (-\id)\) and define \(\pi : \cO_q(\CP^1) \to \bL(\slashed{S}_q(\CP^1))\) by setting \(\pi \coloneqq (b \mapsto \pi_{+}(b) \oplus \pi_{-}(b))\).
	Finally, let \(\slashed{D}_q : \slashed{S}_q(\CP^1) \to \slashed{S}_q(\CP^1)\) be Majid's spin Dirac operator \cite[Prop.\ 5.5]{Majid05}, which is constructed from the maps \(\partial_+\) and \(\partial_-\) of Example \ref{ex:hopf3} by \(\slashed{D}_q \coloneqq \left(\begin{smallmatrix}0& q^{-1}\partial_+\\q\partial_- & 0\end{smallmatrix}\right)\).
	Then \((\slashed{S}_q(\CP^1),\pi,\slashed{D}_q)\) is a faithful bounded commutator representation of \((\Omega_q(\CP^1),\du_q)\) that recovers Majid's \(q\)-deformed Clifford action \cite[\S 5]{Majid05} as the induced map \(\pi_{\slashed{D}_q}\).
	Moreover, a straightforward calculation now shows that \((\slashed{S}_q(\CP^1),\pi,\slashed{D}_q)\) recovers the \emph{spin Dirac} spectral triple on \(\cO_q(\CP^1)\) of D\k{a}browski--Sitarz \cite{DS01} as reformulated by Neshveyev--Tuset \cite[\S 3]{NT}.
\end{example}

The following proposition shows that NC Riemannian geometry in terms of spectral triples generalises NC Riemannian geometry in terms of abstract Hodge star operators.

\begin{proposition}[{cf.\ Das--\'{O} Buachalla--Somberg~\cite[\S 3.2]{DOBS}}]\label{prop:hodgederham}
	Let \((\star,\tau)\) be a Riemannian geometry on \((B;\Omega_B,\du_B)\), so that \(\Omega_B\) defines a \(\bZ/2\bZ\)-graded separable pre-Hilbert space with respect to the inner product \(\ip{}{}_\tau\) induced by \((\star,\tau)\) and the \(\bZ/2\bZ\)-grading \(\gamma_B\).
	Let \(\pi : B \to \bL(\Omega_B)\) denote the bounded \(\ast\)-representation of \(B\) on \(\Omega_B\) defined by left multiplication.
	The triple \((\Omega_B,\pi,\du_B+\du_B^\ast)\) defines a faithful bounded commutator representation of \((\Omega_B,\du_B)\) that satisfies
	\begin{equation}
		\forall \alpha \in \Omega^1_B, \, \forall \beta \in \Omega_B \quad \pi_{\du_B+\du_B^\ast}(\alpha)\beta = \iu{}\alpha \cdot \beta + \iu{}{\star}^{-1}\mleft(\alpha \cdot ({\star} \circ \gamma_B)(\beta)\mright).
	\end{equation}
\end{proposition}

\begin{lemma}\label{lem:basicadjoint}
	Under the hypotheses of Proposition \ref{prop:hodgederham}, given \(k \in \Set{0,\dotsc,N}\) and \(\omega \in \Omega^k_B\), define \(\sce(\omega) : \Omega_B \to \Omega_B\) to be left multiplication by \(\omega\) in \(\Omega_B\).
	Then, for all \(k \in \Set{0,\dotsc,N}\) and \(\omega \in \Omega^k_B\), the map \(\sce(\omega)\) defines a bounded adjointable operator on the pre-Hilbert space \(\Omega_B\) that satisfies
	\(
		\sce(\omega)^\ast = (-1)^k {\star}^{-1} \circ \sce(\omega^\ast) \circ {\star} \circ \gamma_B^k.
	\)
\end{lemma}

\begin{proof}
	Let \(k \in \Set{0,\dotsc,N}\) and \(\omega \in \Omega^k_B\) be given.
	Let \(g\) be the inverse metric induced by \((\star,\tau)\), so that \(\Omega_B\) defines a \(B\)-self-correspondence of finite type with respect to \(g\).
	Hence, the right \(B\)-linear map \(\sce(\omega)\) is adjointable and bounded as an operator on \((\Omega_B,g)\) with operator norm \(\norm{\sce(\omega)} < +\infty\).
	Moreover, given \(j \in \Set{0,\dotsc,N}\), \(\alpha \in \Omega^j_B\), and \(\beta \in \Omega^{k+j}_B\), we see that
	\[
		(\sce(\omega)\alpha)^\ast {\star}(\beta) = (-1)^{jk} \alpha^\ast  \omega^\ast  {\star}(\beta) = \alpha^\ast \cdot ((-1)^k {\star}^{-1} \circ \sce(\omega^\ast) \circ {\star} \circ \gamma_B^k)(\beta),
	\]
	so that \(\sce(\omega)^\ast = (-1)^k {\star}^{-1} \circ \sce(\omega^\ast) \circ {\star} \circ \gamma_B\) for \(\sce(\omega)\) as operators on the \(B\)-self-correspondence of finite type \(\Omega_B\).
	But now, recall that \(\ip{}{}_\tau = \tau \circ g\), which immediately implies that \(\sce(\omega)^\ast\) remains the adjoint of \(\sce(\omega)\) as an operator on the pre-Hilbert space \(\Omega_B\).
	Then \(\sce(\omega)\) is also bounded with operator norm bounded by \(\norm{\sce(\omega)}\), since, for all \(\alpha \in \Omega_B\),
	\[
		\ip{\sce(\omega)\alpha}{\sce(\omega)\alpha}_\tau = \tau\mleft(g(\sce(\omega)\alpha,\sce(\omega)\alpha)\mright) \leq \tau\mleft(\norm{\sce(\omega)}^2 g(\alpha,\alpha)\mright) = \norm{\sce(\omega)}^2 \ip{\alpha}{\alpha}_\tau. \qedhere
	\]
\end{proof}

\begin{proof}[Proof of Proposition \ref{prop:hodgederham}]
In light of Proposition \ref{prop:basicseparable} and Lemma \ref{lem:basicadjoint}, it suffices to show that
\(
	[\du_B+\du_B^\ast,\pi(b)]\beta = \du_B(b) \cdot \beta + {\star}^{-1}\mleft(\du_B(b) \cdot ({\star} \circ \gamma_B)(\beta)\mright).
\)
for all \(b \in B\) and \(\beta \in \Omega_B\).
Hence, let \(b \in B\), \(k \in \Set{0,\dotsc,N}\), and \(\beta \in \Omega^k_B\) be given.
On the one hand, the Leibniz rule for \(\du_B\) immediately implies that \([\du_B,\pi(b)]\beta = \du_B(b) \cdot \beta\).
On the other hand, together with left \(B\)-linearity of \(\gamma_B\) and \(\star\), it also implies that \([\du_B^\ast,\pi(b)]\beta = {\star}^{-1}\mleft(\du_B(b) \cdot ({\star} \circ \gamma_B)(\beta)\mright)\), since
\begin{multline*}
	\du_B^\ast\mleft(b \cdot \beta\mright) 
	= {\star}^{-1} \circ {\du_B} \circ {\star} \circ \gamma_B(b \cdot \beta)
	= {\star}^{-1} \circ {\du_B}\mleft(b \cdot ({\star} \circ \gamma_B)(\beta)\mright)\\
	= {\star}^{-1}\mleft(\du_B(b) \cdot ({\star}\circ \gamma_B)(\beta)\mright) + b \cdot \du_B^\ast\beta. 
\end{multline*}
Hence, in the notation of Lemma \ref{lem:basicadjoint}, we find that
\[
	\iu{}[\du_B+\du_B^\ast,\pi(b)] = \iu{}\sce(\du_B b) + \iu{}({\star}^{-1} \circ \sce(\du_B b) \circ {\star} \circ \gamma_B) = \iu{}\sce(\du_B b) + (\iu{}\sce(\du_B b^\ast))^\ast. \qedhere
\]
\end{proof}

\begin{definition}
	Let \((\star,\tau)\) be a Riemannian geometry on \((B;\Omega_B,\du_B)\).
	The \emph{Hodge--de Rham commutator representation} induced by \((\star,\tau)\) is the faithful bounded commutator representation \((\Omega_B,\pi_B,\du_B+\du_B^\ast)\) of \((B;\Omega_B,\du_B)\) constructed from \((\star,\tau)\) by Proposition \ref{prop:hodgederham}.
\end{definition}

We now turn to the construction of commutator representations for \((P;\Omega_P,\du_P;\Pi)\).
The following proposition shows that faithful bounded commutator representations of \((P;\Omega_P,\du_P)\) do not exist when \(\kappa \neq 1\).
This forces us to consider commutator representations where \(1\)-forms may be represented by unbounded operators.

\begin{proposition}[{cf.\ Schm\"{u}dgen \cite[Lemma 6]{Schmuedgen}}]\label{prop:nogobdd}
	Let \((H,\pi,D)\) be a bounded commutator representation of \((P,\Omega_P,\du_P)\). 
	If  \(\kappa \neq 1\), then \((\id-\Pi)(\Omega^1_P) \subseteq \ker \pi_D\).
\end{proposition}

\begin{proof}
	Let \((e_i)_{i=1}^m\) and \((\epsilon_j)_{j=1}^n\) be finite families in \(P_1\) satisfying \(\sum_{i=1}^m e_ie_i^\ast = 1\) and \(\sum_{j=1}^n \epsilon_j^\ast \epsilon_j = 1\); define bounded completely positive \(\phi_{\pm} : \bL^{\U}(H) \to \bL^{\U}(H)\) by
	\[
		\forall T \in \bL^{\U}(H), \quad \phi_+(T) \coloneqq \sum\nolimits_{i=1}^m \pi(e_i) T \pi(e_i^\ast), \quad \phi_{-}(T) \coloneqq \sum\nolimits_{j=1}^n \pi(\epsilon_j^\ast) T \pi(\epsilon_j),
	\]
	which are unit preserving and hence contractive.
	Since \(\kappa^{-1}\sum_{i=1}^m e_i \vartheta e_i^\ast = \vartheta\) and \(\kappa \sum_{j=1}^m \epsilon_j^\ast \vartheta \epsilon_j = \theta\), it follows that
	\(
		\norm{\pi_D(\vartheta)} = \kappa^{\mp 1} \norm{\phi_{\pm} \circ \pi_D(\vartheta)} \leq \kappa^{\mp 1}\norm{\pi_D(\vartheta)}
	\).
	Thus, if \(\kappa \neq 1\), then \(\pi_D(\vartheta) = 0\), so that \(\pi_D\) vanishes on \((\id-\Pi)(\Omega^1_P) = P \cdot \vartheta\).
\end{proof}

\begin{example}[{Schm\"{u}dgen \cite[Thm.\ 3]{Schmuedgen}}]\label{ex:hopf9}
	Continuing from Example \ref{ex:hopf5}, let \((H,\pi,D)\) be a bounded commutator representation of \((\cO_q(\SU(2));\Omega_q(\SU(2)),\du_q)\).
	On the one hand, since \(q^2 \neq 1\), Proposition \ref{prop:nogobdd} shows that \(\pi_D(e^0) = 0\).
	On the other hand, since \(q \neq 1\), the proof of Proposition \ref{prop:nogobdd}, \emph{mutatis mutandis}, shows that \(\pi_D(e^\pm) = 0\).
	Hence, it follows that \(\pi_D = 0\).
\end{example}

\begin{example}\label{ex:heis11}
	Continuing from Example \ref{ex:heis6}, let us suppose that \((H,\pi,D)\) is a bounded commutator representation of \((P_\theta,\Omega_{P_\theta},\du_{P_\theta})\).
	On the one hand, since \(\epsilon_\theta^2 \neq 1\), Proposition \ref{prop:nogobdd} shows that \(\pi_D(e^0) = 0\).
	On the other hand, since \(\epsilon_\theta \neq 1\), the proof of Proposition \ref{prop:nogobdd}, \emph{mutatis mutandis}, shows that \(\pi_D(e^1) = 0\) and \(\pi_D(e^2) = 0\).
	Hence, it follows that \(\pi_D = 0\).
\end{example}

This catastrophic failure of bounded commutator representations to accommodate important examples of NC differentiable principal \(\U\)-bundles forces us to consider a more general notion of commutator representation where elements of \(\Omega^1_P\) may be represented by unbounded operators of the following kind.

\begin{definition}
	Let \(H\) be a separable \(\bZ/2\bZ\)-graded pre-Hilbert space equipped with a unitary representation \(V : \U \to \Unit(H)_{\even}\) of finite type.
	We say that an operator \(T : H \to H\) is \emph{locally bounded} if it satisfies both of the following conditions:
	\begin{enumerate}
		\item for all \(j,k \in \bZ\), the map \(\rest{P_j T P_k}{H_k} : H_k \to H_j\) is bounded and adjointable;
		\item the set \(\Set{c \in \bZ \given \exists k \in \bZ, \, P_{k+c}TP_k \neq 0}\) is finite.
	\end{enumerate}
	Hence, \(\bL_{\loc}^{\U}(H)\) is the \(\bZ/2\bZ\)-graded unital \(\ast\)-algebra of locally bounded operators on \(H\), where the \(\ast\)-operation is given by taking operator adjoints and the \(\bZ/2\bZ\)-grading is induced by the \(\bZ/2\bZ\)-grading on \(H\).
	At last, set \(\bL^{\U}(H) \coloneqq \bL(H) \cap \bL_{\loc}^{\U}(H)\).
\end{definition}

\begin{example}
	Let \(H\) be a separable \(\bZ/2\bZ\)-graded pre-Hilbert space equipped with a unitary representation \(V : \U \to \U(\bL(H))_{\even}\) of finite type.
	Then each \((T_j)_{j \in \bZ} \in \prod_{j \in \bZ}\bL(H_j)\) yields \(\U\)-equivariant \(\bigoplus_{j\in\bZ}T_j \in \bL^{\U}_{\loc}(H)^{\U}\).
	In particular, given \(\kappa > 0\), we define even \(\U\)-equivariant \(\Lambda_\kappa\), \(\partial_\kappa \in \bL_{\loc}^{\U}(H)\) by
	\begin{equation}
		\Lambda_\kappa \coloneqq \bigoplus_{j \in \bZ} \kappa^{-j} \id_{H_j}, \quad \partial_\kappa \coloneqq \bigoplus_{j \in \bZ} 2\pi\iu{}[j]_\kappa\id_{H_j},
	\end{equation}	
	so that \(\Lambda_\kappa\) is formally self-adjoint while \(\partial_\kappa\) is formally skew-adjoint.
\end{example}

We now weaken the definition of bounded commutator representation accordingly.

\begin{definition}
	Let \(A\) be a \(\U\)-pre-\Cstar-algebra of finite type, and let \((\Omega,\du)\) be a \(\U\)-\(\ast\)-quasi-\textsc{dga} of finite type over \(A\).
	Let \(H\) be a separable \(\bZ/2\bZ\)-graded pre-Hilbert space equipped with a unitary representation \(V : \U \to \Unit(H)_\even\) of finite type, a \(\U\)-equivariant bounded \(\ast\)-automorphism \(\pi : A \to \bL^{\U}(H)_{\even}\), and a \(\U\)-invariant odd formally self-adjoint \(\bC\)-linear map \(D : H \to H\), so that \(\bL^{\U}_{\loc}(H)\) defines a \(A\)-bimodule with respect to \(\pi\).
	We call \((H,\pi,D)\) a \emph{locally bounded commutator representation} of \((A;\Omega,\du)\) whenever there exists a (necessarily unique) \(A\)-bimodule homomorphism \(\pi_D : \Omega^1 \to \bL^{\U}_{\loc}(H)\), such that
	\begin{equation}
		\forall a \in A, \quad \pi_D \circ \du(a) = \iu{}[D,\pi(a)].
	\end{equation}
	Hence, we call \((H,\pi,D)\) \emph{faithful} whenever \(\pi\) is isometric and \(\pi_D\) is injective.
\end{definition}

At last, we propose a refined notion of locally bounded commutator representation for \(\kappa\)-differentiable quantum principal \(\U\)-bundles with connection over \(B\).
When \(\kappa = 1\), it reduces to a multigraded variation on a D\k{a}browski--Sitarz's definition of principal \(\U\)-spectral triples~\cite{DS} in the spirit of \'{C}a\'{c}i\'{c}--Mesland~\cite{CaMe}.

\begin{definition}
	A \emph{projectable commutator representation} of \((P;\Omega_P,\du_P;\Pi)\) is a quadruple of the form \((H,\pi,D,\Gamma)\), where:
	\begin{enumerate}[leftmargin=*]
		\item \((H,\pi,D)\) is a locally bounded commutator representation of \((P,\Omega_P,\du_P)\), such that
		\(
			\left(p \otimes \xi \mapsto \pi(p)\xi\right) : P \otimes_B H^{\U} \to H
		\)
		is bijective and \(\pi_D(\vartheta)^2 = \Lambda_\kappa^2\);
		\item \(\Gamma \in \bL^{\U}(H)\) is an even \(\U\)-invariant self-adjoint unitary commuting with \(\ran\pi\) and  anticommuting with \(\pi_D(\vartheta)\), such that the \emph{horizontal Dirac operator}
		\begin{equation}
			D_{\hor} \coloneqq \tfrac{1}{2}(D+\Gamma D \Gamma)
		\end{equation}
		supercommutes with \(\pi_D(\vartheta)\) and the \emph{remainder}
		\begin{equation}
			Z \coloneqq \tfrac{1}{2}(D-\Gamma D \Gamma) + \iu{}\pi_D(\vartheta)\partial_\kappa
		\end{equation}
		is bounded and supercommutes with \(\ran\pi\).
	\end{enumerate}
	Hence, we call \((H,\pi,D,\Gamma)\) \emph{faithful} whenever \((H,\pi,D)\) is faithful and the maps
	\[
		\left(b \mapsto \rest{\pi(b)}{H^{\U}}\right) : B \to \bL(H^{\U}), \quad \left(\beta \mapsto \rest{\pi_D(\beta)}{H^{\U}}\right) : \Omega^1_B \to \bL(H^{\U})
	\]
	are isometric and injective, respectively.
\end{definition}

\begin{remark}\label{rem:kk2}
	Let \(\mathfrak{P}\) be the \(C^\ast\)-algebra completion of \(P\).
	A projectable commutator representation \((H,\pi,D,\Gamma)\) of \(P\) can be viewed as defining a formal \(\U\)-equivariant unbounded \(\mathrm{KK}_1\)-cycle \((P,H,D)\) for \((\mathfrak{P},\bC)\), where the \(\U\)-invariant odd self-adjoint unitary \(-\iu{}\Gamma\pi_D(\vartheta)\Lambda_\kappa^{-1}\) generates the \(1\)-multigrading.
	If \(\kappa = 1\), the horizontal Dirac operator \(D_\hor\) has bounded commutators with \(\pi(B)\), and the operator \(D\) is essentially self-adjoint with compact resolvent, then \((P,H,D)\) defines a genuine \(\U\)-equivariant odd spectral triple for \(\mathfrak{P}\).
	Otherwise, the formal unbounded \(\mathrm{KK}_1\)-cycle \((P,H,D)\) generally lies outside the current scope of unbounded \(\mathrm{KK}\)-theory.\end{remark}

The following shows that a total Riemannian geometry on \((P,\Omega_P,\du_P;\Pi)\) induces a canonical projectable commutator representation just as a Riemannian geometry on \((B;\Omega_B,\du_B)\) induces a bounded commutator representation.

\begin{proposition}\label{prop:totalderham}
	Suppose that \((\Delta_\ver,\Delta_\hor,\star,\tau)\) be a total Riemannian geometry on  \((P,\Omega_P,\du_P,\Pi)\).
	Hence, view \(\Omega_P\) as a \(\bZ/2\bZ\)-graded separable pre-Hilbert space with respect to the inner product \(\ip{}{}_\tau\) induced by \((\Delta_\ver,\Delta_\hor,\star,\tau)\) and the \(\bZ/2\bZ\)-grading \(\gamma_P\), so that the \(\U\)-action \(\hat{\sigma}\) on \(\Omega_P\) defines a unitary \(\U\)-representation of finite type by even operators.
	Let \(\pi : P \to \bL(\Omega_P)\) denote the isometric \(\ast\)-representation of \(P\) on \(\Omega_P\) defined by left multiplication.
	Then \((\Omega_P,\pi,\du_P+\du_P^\ast,2\Pi-\id)\) defines a faithful projectable commutator representation of \((P,\Omega_P,\du_P,\Pi)\) that satisfies
	\begin{equation}\label{eq:totalderhamcommutator}
		\forall \omega \in \Omega^1_P, \, \forall \eta \in \Omega_P, \quad \pi_{\du_P+\du_P^\ast}(\omega)\eta = \iu{}\omega \cdot \eta + {\star^{-1}}\mleft(\iu{}\omega \cdot ({\star}\circ\gamma_P)(\eta)\mright).
	\end{equation}
	Moreover, the remainder \(Z\) of \((\Omega_P,\pi,\du_P+\du_P^\ast,2\Pi-\id)\) is given by
	\begin{equation}
		\forall p_1,p_2 \in P, \, \forall \alpha_1,\alpha_2 \in \Omega_B, \quad Z(p_1\alpha_1 + p_2\vartheta\alpha_2) = -p_2 \cF_\Pi \alpha_2 - p_1 \vartheta {\star_B^{-1}}\mleft(\cF_\Pi {\star_B}(\alpha_1)\mright),
	\end{equation}
	where \(\cF_\Pi\) is the curvature \(2\)-form of the connection \(\Pi\) and where \((\star_B,\tau_B)\) is the restriction of \((\Delta_\ver,\Delta_\hor,\star,\tau)\) to \((B;\Omega_B,\du_B)\).
\end{proposition}

\begin{lemma}\label{lem:totaladjoint}
	Under the hypotheses of Proposition \ref{prop:totalderham}, let \(m \in \Set{0,\dotsc,N+1}\) and \(\omega \in \Omega^m_P\) be given, and let \(\sce(\omega) : \Omega_P \to \Omega_P\) be left multiplication by \(\omega\) in \(\Omega_P\). Then \(\sce(\omega) \in \bL_{\loc}^{\U}(H)\) and
	\(
		\sce(\omega)^\ast = (-1)^m {\star^{-1}} \circ \sce(\omega^\ast) \circ {\star} \circ \gamma_P^m.
	\)
\end{lemma}

\begin{proof}
	The proof of Lemma \ref{lem:basicadjoint} applies almost verbatim; all that remains is to show that \(\sce(\omega) \in \bL_{\loc}^{\U}(H)\).
	
	First, suppose that \(m=1\) and \(\omega = \vartheta\).
	Let \((m,j) \in \Set{0,\dotsc,N+1} \times \bZ\) and let \(\eta \in (\Omega^m_P)_j\), so that \(\eta = \eta_1 + \vartheta\eta_2\) for \(\eta_1 \in (\Omega^m_{P,\hor})_j\) and \(\eta_2 \in (\Omega^{m-1}_{P,\hor})_j\).
	Then \(\sce(\omega)(\eta) = \vartheta\eta_1 \in (\Omega^{m+1}_P)_j\), so that
	\(
		\ip{\sce(\omega)(\eta)}{\sce(\omega)(\eta)}_\tau = \ip{\vartheta\eta_1}{\vartheta\eta_1}_\tau = \kappa^{-2j}\ip{\eta_1}{\eta_1}_\tau \leq \kappa^{-2j}\ip{\eta}{\eta}_\tau.
	\)
	by the proof of Proposition \ref{prop:totalprehilbert}.
	Thus, given \(j,k \in \bZ\), we see that \(\bE_j \sce(\omega) \bE_k \neq 0\) only if \(j = k\), in which case \(\norm{\bE_j \sce(\omega) \bE_j} \leq \kappa^{-j}\).
	
	Next, suppose that \(\omega \in \Omega^m_B\).
	Let \((r,s,j) \in \Set{0,1} \times \Set{0,\dotsc,N} \times \bZ\) be given, and recall from Proposition \ref{prop:totalprehilbert} that both \((\Omega^{r,s}_P)_{j}\) and \((\Omega^{r,s+m}_P)_{j}\) are \(B\)-self-correspondences of finite type with respect to the inverse metric \(g\) induced by \(\star\).
	Let \(S^{r,s}_j\) denote the restriction of \(\sce(\omega)\) to \((\Omega^{r,s}_P)_{j}\), whose range is therefore contained in \((\Omega^{r,s+m}_P)_{j}\).
	Since \(S^{r,s}_j : (\Omega^{r,s}_P)_j \to (\Omega^{r,s+m}_P)_{j}\) is right \(B\)-linear, it is bounded as a map of right pre-Hilbert \(B\)-modules, and hence, since \(\ip{}{}_\tau = \tau \circ g\), as a map of pre-Hilbert spaces.
	Thus, given \(j,k \in \bZ\), it follows that \(\bE_j \sce(\omega) \bE_k \neq 0\) only if \(j = k\), in which case \(\norm{\bE_j \sce(\omega) \bE_j} \leq \sup\Set{\norm{S^{r,s}_j} \given (r,s) \in \Set{0,1} \times \Set{0,\dotsc,N}}\).
	
	Let us finally consider the general case.
	Without loss of generality, there exist \(p_1,p_2 \in P\), \(\alpha_1 \in \Omega^m_B\), and \(\alpha_2 \in \Omega^{m-1}_B\), such that \(\omega = p_1\alpha_1+p_2\vartheta\alpha_2\).
	Then \(\sce(\omega) = \pi(p_1)\sce(\alpha_1) + \pi(p_2)\sce(\vartheta)\sce(\alpha_2) \in \bL_{\loc}^{\U}(\Omega_P)\) since \(\pi(p_1),\pi(p_2) \in \bL^{\U}(\Omega_P)\) by Proposition \ref{prop:totalprehilbert}.
\end{proof}

\begin{proof}[Proof of Prop.\ \ref{prop:totalderham}]
	Let \(\sce : \Omega_P \to \bL_{\loc}^{\U}(\Omega_P)\) be the \(\U\)-equivariant \(\bC\)-linear map defined by Lemma \ref{lem:totaladjoint} and linearity, let \(\sci : \Omega_P \to \bL_{\loc}^{\U}(\Omega_P)\) be the \(\U\)-equivariant \(\bC\)-linear map defined by \(\sci(\omega) \coloneqq (-1)^m {\star^{-1}} \circ \sce(\omega) \circ {\star} \circ \gamma_P^m = \sce(\omega^\ast)^\ast\) for \(m \in \Set{0,\dotsc,N+1}\) and \(\omega \in \Omega^1_P\), let \(c \coloneqq \iu{}(\sce - \sci)\), and set \(D \coloneqq \du_P+\du_P^\ast\),
	By analogy, define \(\sce_B,\sci_B,c_B : \Omega_B \to \bL(\Omega_B)\) and set \(D_B \coloneqq \du_B+\du_B^\ast\), so that \(\pi_{D_B} = \rest{c_B}{\Omega^1_B}\).
	Finally, let \(\vartheta\) denote the connection \(1\)-form of \(\Pi\), let \(\nabla \coloneqq \hat{\ell}_P \circ \Pi \circ \rest{\du_P}{P}\), where  \(\hat{\ell}_P : \Omega_{P,\hor} \to P \otimes_B \Omega_B\) is the \(\U\)-equivariant isomorphism of \(B\)-bimodules of Proposition \ref{prop:assoclineconn}, let \(D_\ver \coloneqq -\iu{}\pi_D(\vartheta)\partial_\kappa\), and let \(\Gamma \coloneqq 2\Pi-\id\).

	First, after substituting Proposition \ref{prop:totalprehilbert} for Proposition \ref{prop:basicseparable} and Lemma \ref{lem:totaladjoint} for Lemma \ref{lem:basicadjoint}, the proof of Proposition \ref{prop:hodgederham} shows that \((\Omega_P,\pi,D)\) defines a faithful locally bounded commutator representation satisfying \eqref{eq:totalderhamcommutator}.
	Moreover, Proposition \ref{prop:assoclineconn} combined with Theorem \ref{thm:dj} yields bijectivity of the multiplication map \((p \otimes \omega \mapsto p\omega) : P \otimes_B \Omega_P^{\U} \to \Omega_P\).
	
	Let us now consider \(\pi_D(\vartheta)\), the would-be horizontal Dirac operator \(D_\hor\), and the would-be remainder \(Z\).
	Note that \(\sce(\vartheta)\) maps \(\Omega_{P,\hor}\) to \(\vartheta \cdot \Omega_{P,\hor}\) and vanishes on \(\vartheta \cdot \Omega_{P,\hor} = \Omega_{P,\hor}^\perp\), so that its adjoint \(\sci(\vartheta)\) maps \(\vartheta \cdot \Omega_{P,\hor}\) to \(\Omega_{P,\hor}\) and vanishes on \(\Omega_{P,\hor}\); since \(\Gamma\) acts as \(\id\) on \(\Omega_{P,\hor}\) and as \(-\id\) on \(\vartheta \cdot \Omega_{P,\hor}\), this suffices to show that \(\Gamma\) anticommutes with \(\pi_D(\vartheta)\).
	Now, let \(p \in P\), \(m \in \Set{0,\dotsc,N}\), and \(\alpha \in \Omega^m_B\) be given.
	On the one hand, we find that \(\sce(\vartheta)(p\alpha) = \Lambda_\kappa(p)\vartheta\alpha\) and
	\begin{align*}
		D(p \alpha) &= \du_P(p\alpha) + (-1)^m {\star^{-1}} \circ \du_P\mleft(p \vartheta \star_B(\alpha)\mright) \\
		&= (\Lambda_\kappa \circ \partial_\kappa)(p)\vartheta\alpha + \leg{\nabla(p)}{0} \leg{\nabla(p)}{1} \alpha + p \du_B(\alpha)\\ &\quad+ {\star^{-1}}\mleft(-\leg{\nabla(p)}{0}\vartheta\leg{\nabla(p)}{1}\star_B(\alpha) - p \cF_\Pi {\star_B}(\alpha) + p\vartheta(\du_B\circ {\star_B})(\alpha)\mright)\\
		&= (\Lambda_\kappa \circ \partial_\kappa)(p)\vartheta\alpha + \leg{\nabla(p)}{0} \sce_B(\leg{\nabla(p)}{1})(\alpha) + p \du_B(\alpha)\\ &\quad- \leg{\nabla(p)}{0} \sci_B(\leg{\nabla(p)}{1})(\alpha) - p \vartheta\sci_B(\cF_\Pi)(\alpha) + p \du_B^\ast(\alpha)\\
		&= \left(- \iu{}\leg{\nabla(p)}{0}c_B(\leg{\nabla(p)}{1})(\alpha) + p D_B(\alpha)\right) + \left((\Lambda_\kappa \circ \partial_\kappa)(p)\vartheta\alpha  - p \vartheta\sci_B(\cF_\Pi)(\alpha)\right),
	\end{align*}
	so that
	\(
		D_\hor(p\alpha) = - \iu{}\leg{\nabla(p)}{0}c_B(\leg{\nabla(p)}{1})(\alpha) + p D_B(\alpha)
	\), 
	and hence
	\[
		Z(p\alpha) = (\Lambda_\kappa \circ \partial_\kappa)(p)\vartheta\alpha  - p \vartheta\sci_B(\cF_\Pi)(\alpha) - \iu{}c(\vartheta)\partial_\kappa(p\alpha) = -p\vartheta\sci_B(\cF_\Pi)(\alpha).
	\]
	On the other hand,
	\begin{align*}
		\sci(\vartheta)(p\vartheta\alpha) &= (-1)^m {\star^{-1}}\mleft(\vartheta \cdot p {\star_B}(\alpha)\mright)
		= (-1)^m {\star^{-1}}\mleft(\Lambda_\kappa(p) \vartheta {\star_B}(\alpha)\mright)
		= \Lambda_\kappa(p) \alpha,\\
		D(p\vartheta\alpha)
		&= \du_P(p\vartheta\alpha) + (-1)^{m+1}{\star^{-1}} \circ \du_P\mleft(p{\star_B}(\alpha)\mright)\\
		&= \leg{\nabla(p)}{0}\!\! \leg{\nabla(p)}{1}\vartheta\alpha - p \cF_\Pi \alpha -p\vartheta\du_B(\alpha) + (-1)^{m+1}{\star^{-1}}\!(\Lambda_\kappa \circ \partial_\kappa)(p){\star_B}(\alpha)\\&\quad\quad + \leg{\nabla(p)}{0} \leg{\nabla(p)}{1}{\star_B}(\alpha) + p(\du_B \circ {\star_B})(\alpha))\\
		&= \leg{\nabla(p)}{0} \leg{\nabla(p)}{1}\vartheta\alpha - p \cF_\Pi \alpha -p\vartheta\du_B(\alpha)\\ &\quad-(\Lambda_\kappa \circ \partial_\kappa)(p)\alpha + \leg{\nabla(p)}{0}  \vartheta  \sci_B(\leg{\nabla(p)}{1})(\alpha) - p\vartheta\du_B^\ast(\alpha)\\
		&= \left(-(\Lambda_\kappa \circ \partial_\kappa)(p)\alpha-p \cF_\Pi \alpha\right) + \left(\iu{}\leg{\nabla(p)}{0} \vartheta c_B(\leg{\nabla(p)}{1})(\alpha) - pD_B(\alpha)\right),
	\end{align*}
	so that \(D_\hor(p\vartheta\alpha) = \iu{}\leg{\nabla(p)}{0} \vartheta c_B(\leg{\nabla(p)}{1})(\alpha) - pD_B(\alpha)\), and hence
	\[
		Z(p\vartheta\alpha) = (\Lambda_\kappa \circ \partial_\kappa)(p)\alpha-p \cF_\Pi \alpha - \iu{}c(\vartheta)\partial_\kappa(p\vartheta\alpha) = -p \sce_B(\cF_\Pi)(\alpha).
	\]
	Thus, for all \(p_1,p_2 \in P\) and \(\alpha_1,\alpha_2 \in \Omega_B\),
	\begin{align*}
		\pi_D(\vartheta)(p_1\alpha_1+p_2\vartheta\alpha_2) &= \Lambda_\kappa(p_2)\alpha_2 + \Lambda_\kappa(p_1)\vartheta\alpha_1,\\
		D_\hor(p_1\alpha_1+p_2\vartheta\alpha_2) &= - \iu{}\leg{\nabla(p_1)}{0}c_B(\leg{\nabla(p_1)}{1})(\alpha_1) + p_1 D_B(\alpha_1)\\ &\quad + \iu{}\leg{\nabla(p_2)}{0} \vartheta c_B(\leg{\nabla(p_2)}{1})(\alpha_2) - p_2 \vartheta D_B(\alpha_2)\\
		Z(p_1\alpha_1+p_2\vartheta\alpha_2) &= -p_2 \sce_B(\cF_\Pi)(\alpha_2) - p_1\vartheta\sci_B(\cF_\Pi)(\alpha_1).
	\end{align*}
	These expressions for \(\pi_D(\vartheta)\), \(D_\hor\), and \(Z\) now make clear that \(\pi_D(\vartheta)^2 = \Lambda_\kappa^2\), that \(D_\hor\) supercommutes with \(\pi_D(\vartheta)\), and that \(Z\) supercommutes with \(\ran\pi\).
	
	Finally, we show that \(Z\) is bounded.
	Recall the faithful conditional expectation \(\bE_P : P \to B\) of Proposition \ref{prop:conditional}, so that \(P \otimes \bC^2\) defines a countably generated right pre-Hilbert \(B\)-module with respect to the \(B\)-valued inner product \(\hp{}{}\) given by
	\[
		\forall p_1,p_2 \in P, \, \forall v_1,v_2 \in \bC^2, \quad \hp{p_1 \otimes v_1}{p_2 \otimes v_2} \coloneqq \ip{v_1}{v_2} \bE_P(p_1^\ast p_2).
	\]
	Thus, \((P \otimes \bC^2) \otimes_B \Omega_B\) defines a pre-Hilbert space with respect to the inner product defined, \emph{mutatis mutandis}, by \eqref{eq:tensorip}.
	Moreover, by Proposition \ref{prop:assoclineconn}, Theorem \ref{thm:dj}, and the proof of Proposition \ref{prop:totalprehilbert}, define unitary \(M : (P \otimes \bC^2) \otimes_B \Omega_B \to \Omega_P\) by setting
	\(
		M \coloneqq \left(\left(\begin{smallmatrix}p_1\\p_2\end{smallmatrix}\right) \otimes \alpha \mapsto \coloneqq (p_1+p_2\vartheta)\alpha\right)
	\).
	Since the left \(B\)-linear maps \(\sce(\cF_\Pi)\) and \(\sci(\cF_\Pi)\) are both bounded as operators on the pre-Hilbert space \(\Omega_B\), standard Hilbert \(C^\ast\)-module lore implies that
	\(
		Z = -M\left(\left(\begin{smallmatrix}0&1\\0&0\end{smallmatrix}\right) \otimes \sce(\cF_\Pi) + \left(\begin{smallmatrix}0&0\\1&0\end{smallmatrix}\right) \otimes\sci(\cF_\Pi)\right)M^\ast
	\)
	is bounded and symmetric as an operator on the pre-Hilbert space \(\Omega_P\).
	
	We conclude by showing that the maps
	\[
		\left(b \mapsto \rest{\pi(p)}{\Omega_P^{\U}}\right) : B \to \bL(\Omega_P^{\U}), \quad \left(\beta \mapsto \rest{\pi_{\du_P+\du_P^\ast}(\beta)}{\Omega_P^{\U}}\right) : \Omega^1_B \to \bL(\Omega_P^{\U})
	\]
	are isometric and injective, respectively.
	First, let \(b \in B\).
	On the one hand, \(\rest{\pi(b)}{\Omega_P^{\U}}\) block-diagonal with respect to orthogonal decomposition \(\Omega_P^{\U} = \Omega_B \oplus \vartheta \Omega_B\), where \(\rest{\pi(b)}{\Omega_B}\) is left multiplication by \(b\) on \(\Omega_B\).
	On the other hand, by Proposition \ref{prop:hodgederham}, left multiplication of \(B\) on \(\Omega_B\) defines an isometric \(\ast\)-homorphism \(B \to \bL(\Omega_B)\).
	Hence, it follows that
	\(
		\norm{b} = \norm{\rest{\pi(p)}{\Omega_B}} \leq \norm{\rest{\pi(b)}{\Omega_P^{\U}}} \leq \norm{\pi(b)} \leq \norm{b}
	\).
	Now, let \(\beta \in \Omega^1_B\) be given.
	On the one hand, both \(\rest{\sce(\beta)}{\Omega_P^{\U}}\) and \(\rest{\sce(\beta^\ast)}{\Omega_P^{\U}}\) are both block-diagonal with respect to the orthogonal decomposition \(\Omega_P^{\U} = \Omega_B \oplus \vartheta \Omega_B\), where \(\rest{\sce(\beta)}{\Omega_B} = \sce_B(\beta)\) and \(\rest{\sce(\beta^\ast)}{\Omega_B} = \sce_B(\beta^\ast)\), so that \(\rest{c(\beta)}{\Omega_P^{\U}}\) is similarly block-diagonal with \(\rest{c(\beta)}{\Omega_B} = c_B(\beta)\).
	On the other hand, by Proposition \ref{prop:hodgederham}, the map \(c_B : \Omega^1_B \to \bL(\Omega_B)\) is injective.
	Hence, it follows that \(\rest{c(\beta)}{\Omega_P^{\U}} = 0\) only if \(\beta = 0\).
\end{proof}

\begin{definition}
	Suppose that \((\Delta_\ver,\Delta_\hor,\star,\tau)\) is a total Riemannian geometry on \((P;\Omega_P,\du_P;\Pi)\). 
	We define the \emph{total Hodge--de Rham commutator representation} induced by \((\Delta_\ver,\Delta_\hor,\star,\tau)\) to be the faithful projectable commutator representation \((\Omega_P,\pi_P,\du_P+\du_P^\ast,2\Pi-\id)\) of \((P;\Omega_P,\du_P;\Pi)\) constructed from \((\Delta_\ver,\Delta_\hor,\star,\tau)\) by Proposition \ref{prop:totalderham}.
\end{definition}

We now justify our terminology by showing that a faithful projectable commutator representation of \((P;\Omega_P,\du_P;\Pi)\) canonically projects to a faithful bounded commutator representation of \((B;\Omega_B,\du_B)\).
This will make precise the notion of lifting a faithful bounded commutator representation of \((B;\Omega_B,\du_B)\) to \((P;\Omega_P,\du_P;\Pi)\).

On the one hand, define the concrete category \(\grp{BCRep}(B)\) of faithful bounded commutator representations of \((B;\Omega_B,\du_B)\) and their isomorphisms as follows:
\begin{enumerate}[leftmargin=*]
	\item an object is a faithful bounded commutator represention \((H,\pi,D)\) of \((\Omega_B,\du_B)\);
	\item an arrow \(U : (H_1,\pi_1,D_1) \to (H_2,\pi_2,D_2)\) is a unitary \(U : H_1 \to H_2\) that satisfies
	\(
		U \pi_1(\cdot) U^\ast = \pi_2\) and \(U D_1 U^\ast = D_2
	\).
\end{enumerate}
On the other hand, given \(\kappa > 0\) and \((P;\Omega_P,\du_P;\Pi)\) a \(\kappa\)-differentiable quantum principal \(\U\)-bundle with connection over \(B\), define the concrete category \(\grp{PCRep}(P;\Pi)\) of faithful projectable commutator representations of \((P;\Omega_P,\du_P;\Pi)\) and their isomorphisms as follows:
\begin{enumerate}[leftmargin=*]
	\item an object of \(\grp{PCRep}(P;\Pi)\) is a faithful projectable commutator representation \((H,\pi,D,\Gamma)\) of \((P;\Omega_P,\du_P;\Pi)\);
	\item an arrow \((U,Z) : (H_1,\pi_1,D_1,\Gamma_1) \to (H_2,\pi_2,D_2,\Gamma_2)\) of \(\grp{PCRep}(P;\Pi)\) consists of an even \(\U\)-equivariant unitary \(U : H_1 \to H_2\) and odd \(\U\)-invariant symmetric \(Z \in \bL^{\U}(H_1)\) supercommuting with \(\ran \pi\) and \(\Gamma\), such that
	\[
		U \pi_1(\cdot) U^\ast = \pi_1, \quad U (D_1-Z) U^\ast = D_2, \quad U \Gamma_1 U^\ast = \Gamma_2;
	\]
	\item given objects \((H_1,\pi_1,D_1,\Gamma_1)\), \((H_2,\pi_2,D_2,\Gamma_2)\), \((H_3,\pi_3,D_3,\Gamma_3x)\), and arrows
	\begin{align*}
		(U_1,Z_1) &: (H_1,\pi_1,D_1,\Gamma_1) \to (H_2,\pi_2,D_2,\Gamma_2),\\
		(U_2,Z_2) &: (H_2,\pi_2,D_2,\Gamma_2) \to (H_3,\pi_3,D_3,\Gamma_3),
	\end{align*}
 	the composition \((U_2,Z_2) \circ (U_1,Z_1) : (H_1,\pi_1,D_1,\Gamma_1) \to (H_3,\pi_3,D_3,\Gamma_3)\) is given by
	\(
		(U_2,Z_2) \circ (U_1,Z_1) \coloneqq \left(U_2U_1,U_1^\ast Z_2 U_1+Z_1\right)
	\);
	\item the identity arrow of an object \((H,\pi,D,\Gamma)\) is given by \((\id,0)\).
\end{enumerate}
Note that an arrow \((U,Z) : (H_1,\pi_1,D_1,\Gamma_1) \to (H_2,\pi_2,D_2,\Gamma_2)\) in \(\grp{PCRep}(P;\Pi)\) encodes \(\U\)-equivariant unitary equivalence of \((H_1,\pi_1,D_1,\Gamma_1)\) and \((H_2,\pi_2,D_2,\Gamma_2)\) after perturbation by the \emph{relative remainder} \(Z\).

\begin{proposition}\label{prop:restfunct}
	The following gives a functor \(\iota_P^\ast : \grp{PCRep}(P;\Pi) \to \grp{BCRep}(B)\).
	\begin{enumerate}[leftmargin=*]
		\item Given an object \((H,\pi,D,\Gamma)\), let
		\(
			\iota_P^\ast (H,\pi,D,\Gamma) \coloneqq \left(P H^{\U},P\pi(\cdot)P,PD_{\hor}P\right)
		\),
		where \(P \coloneqq \rest{\tfrac{1}{2}(\id+\Gamma)}{H^{\U}}\) and \(D_\hor\) is the horizontal Dirac operator of the faithful projectable commutator representation \((H,\pi,D,\Gamma)\).
		\item Given an arrow \(U : (H_1,\pi_1,D_1,\Gamma_1) \to (H_2,\pi_2,D_2,\Gamma_2)\), let \(\iota_P^\ast U\) be given by \(P_2 U P_1\), where \(P_1 \coloneqq \rest{\tfrac{1}{2}(\id+\Gamma_1)}{H_1^{\U}}\) and \(P_2 \coloneqq \rest{\tfrac{1}{2}(\id+\Gamma_2)}{H_2^{\U}}\).
	\end{enumerate}
\end{proposition}

\begin{proof}
	This is a routine verification except for one subtlety.
	Let \((H,\pi,D,\Gamma)\) be a faithful projectable commutator representation of \((P;\Omega_P,\du_P;\Pi)\).
	It remains to show that the bounded commutator representation \((H_B,\pi_B,D_B) \coloneqq \iota_P^\ast(H,\pi,D,\Gamma)\) of \((B;\Omega_B,\du_B)\) is faithful.
	Observe that \(H^{\U}\) admits the orthogonal decomposition \(H^{\U} = H_B \oplus \Gamma H_B\), where \(\Gamma\) restricts to a unitary \(V : H_B \to \Gamma H_B\).
	Hence, it follows that \(\rest{\pi(b)}{H^{\U}} = \pi_B(b) \oplus (V \pi_B(b) V^\ast)\) for all \(b \in B\), so that \(\pi_B\) is isometric since \(\left(b \mapsto \rest{\pi(b)}{H^{\U}}\right) : B \to \bL(H^{\U})\) is isometric.
	A qualitatively identical argument shows that \(\pi_D\) is injective.
\end{proof}

\begin{definition}
	Let \((H,\pi,D)\) be a faithful bounded commutator representation of \((B;\Omega_B,\du_B)\), and let \((\tilde{H},\tilde{\pi},\tilde{D},\tilde{\Gamma})\) be a faithful projectable commutator representation of \((P;\Omega_P,\du_P;\Pi)\).
	We say that \((\tilde{H},\tilde{\pi},\tilde{D},\tilde{\Gamma})\) is a \emph{lift} of \((B;\Omega_B,\du_B)\) to \((P;\Omega_P,\du_P;\Pi)\) whenever \(\iota_P^\ast (\tilde{H},\tilde{\pi},\tilde{D},\tilde{\Gamma})\) and \((H,\pi,D)\) are isomorphic in \(\grp{BCRep}(B)\).
\end{definition}

\begin{example}
	Suppose that \((\Delta_\ver,\Delta_\hor,\star,\tau)\) is a total Riemannian geometry on \((P,\Omega_P,\du_P,\Pi)\), and let \((\star_B,\tau_B)\) be its restriction to a Riemannian geometry on \((B;\Omega_B,\du_B)\).
	Then the total Hodge--de Rham commutator representation \((\Omega_P,\pi_P,\du_P+\du_P^\ast,2\Pi-\id)\) induced by \((\Delta_\ver,\Delta_\hor,\star,\tau)\)  is a lift of the Hodge--de Rham commutator representation \((\Omega_B,\pi_B,\du_B+\du_B^\ast)\) induced by \((\star_B,\tau_B)\).
	Indeed, the inclusion map \(\hat{\iota}_P : \Omega_B \xrightarrow{\sim} \Omega_{P,\hor}^{\U} = \Pi(\Omega_P^{\U})\) defines an isomorphism 
	\[
		\hat{\iota}_P : (\Omega_B,\pi_B,\du_B+\du_B^\ast) \to \iota_P^\ast(\Omega_P,\pi_P,\du_P+\du_P^\ast,2\Pi-\id).
	\]
\end{example}

At last, we show that every faithful bounded commutator representation of \((B;\Omega_B,\du_B)\) has an essentially unique lift to \((P;\Omega_P,\du_P;\Pi)\), namely, up to \(\U\)-equivariant unitary equivalence after perturbation by a relative remainder.
Note that we cannot use Schwieger--Wagner's lifting construction~\cite{SW2}, even after generalisation to \(\kappa \neq 1\), since it requires unnatural choices of representation-theoretic data that need not even yield locally bounded commutator representations of \((P;\Omega_P,\du_P)\).

We first show that lifts always exist.
When \(\kappa = 1\), the right \(B\)-module \(\Omega^1_B\) is free with basis consisting of self-adjoint elements of \(\Zent(\Omega_B)^1\), the Fr\"{o}hlich automorphism \(\hat{\Phi}_P\) is the identity map, and the faithful bounded commutator representation of \((B;\Omega_B,\du_B)\) takes a certain restrictive form, our construction recovers a lifting construction for spectral triples first proposed by Gabriel--Grensing \cite{GG}.

In what follows, recall the self-adjoint Pauli matrices
\[
	\sigma^1 \coloneqq \begin{pmatrix}0&1\\1&0\end{pmatrix}, \quad \sigma^2 \coloneqq \begin{pmatrix}0&-\iu{}\\ \iu{}&0\end{pmatrix}, \quad \sigma^3 \coloneqq \begin{pmatrix} 1 & 0 \\ 0 & -1\end{pmatrix} = -\iu{}\sigma^1\sigma^2.
\]

\begin{proposition}\label{prop:liftcommrepconstruct}
	Let \((H,\pi,D)\) be a faithful bounded commutator representation of \((B;\Omega_B,\du_B)\).
	Define a map \(\nabla : P \to P \otimes_B \Omega_B\) by \(\nabla \coloneqq \hat{\ell}_P \circ \Pi \circ \rest{\du_P}{P}\),
	where \(\hat{\ell}_P : \Omega_{P,\hor} \to P \otimes_B \Omega_B\) is the \(B\)-bimodule isomorphism of Proposition \ref{prop:assoclineconn}, and let \(\bE_P : P \to B\) be the faithful conditional expectation of Proposition \ref{prop:conditional}.
	Equip the left \(P\)-module \(P \otimes \bC^2\) with the right \(B\)-module structure
	\[
		\forall p \in P, \, \forall x \in \bC^2, \, \forall b \in B, \quad (p \otimes x) \cdot b \coloneqq pb \otimes x,
	\]
	equip \(H\) with the left \(B\)-module structure defined by \(\pi\), and equip \((P \otimes \bC^2) \otimes_B H\) with the inner product \(\ip{}{}\) defined by
	\begin{multline*}
		\forall p_1,p_2 \in P, \, \forall x_1,x_2 \in \bC^2, \, \forall h_1,h_2 \in H, \\ \ip{p_1 \otimes x_1 \otimes h_1}{p_2 \otimes x_2 \otimes h_2} \coloneqq \ip{x_1}{x_2}\ip{h_1}{\pi(\bE_P(p_1^\ast p_2))h_2},
	\end{multline*}
	the \(\bZ/2\bZ\)-grading \(\id \otimes \sigma^3 \otimes \chi_H\) and the linear \(\U\)-representation induced by the \(\U\)-action on \(P\).
	Finally, define an operator \((\id \otimes \sigma^3) \otimes_\nabla D\) on \((P \otimes \bC^2) \otimes_B H\) by
	\begin{multline*}
		\forall p \in P, \, \forall x \in \bC^2, \, \forall h \in H, \\ (\id \otimes \sigma^3) \otimes_\nabla D\mleft(p \otimes x \otimes h\mright) \coloneqq -\iu{}\leg{\nabla(p)}{0} \otimes \sigma^3 x \otimes \pi_D\mleft(\leg{\nabla(p)}{1}\mright)h + p \otimes \sigma^3 x \otimes D h.
	\end{multline*}
	Then
	\[
		\left((P \otimes \bC^2) \otimes_B H, \id \otimes \id \otimes \pi(\cdot),\iu{}(\Lambda_\kappa \circ \partial_\kappa) \otimes \sigma^2 \otimes \id + (\id \otimes \sigma^3) \otimes_\nabla D,\id \otimes \sigma^3 \otimes \id\right)
	\]
	is a lift of \((H,\pi,D)\) to \((P;\Omega_P,\du_P;\Pi)\) with horizontal Dirac operator \((\id \otimes \sigma^3) \otimes_\nabla D\) and remainder \(0\).
\end{proposition}

\begin{lemma}\label{lem:commreplift}
	Let \((H,\pi,D)\) be a bounded commutator representation of \((B;\Omega_B,\du_B)\), and let \((E,\sigma,\nabla)\) be a Hermitian line \(B\)-bimodule with connection.
	Equip \(E \otimes_B H\) with the positive definite inner product defined, \emph{mutatis mutandis}, by \eqref{eq:tensorip}.
	\begin{enumerate}[leftmargin=*]
		\item For every \(x \in E\), we obtain contractive \(\phi_E[x] : E \otimes_B H \to H\) by setting
		\[
			\forall y \in E, \, \forall h \in H, \quad \phi_E[x](y \otimes h) = \pi(\hp{x}{y}_E)h
		\]
		Hence, in particular, the pre-Hilbert space \(E \otimes_B H\) is separable.
		\item We obtain formally self-adjoint \(\id \otimes_\nabla D : E \otimes_B H \to E \otimes_B H\) by setting
		\[
			\forall y \in E, \, \forall h \in H, \quad (\id \otimes_\nabla D)(y \otimes h) = -\iu{}\leg{\nabla(y)}{0} \otimes \pi_D(\leg{\nabla(y)}{1})h + y \otimes Dh.
		\]
		\item For every \(\alpha \in \Omega^1_B\), we obtain bounded \(\rho_E[\alpha] : E \otimes_B H \to E \otimes_B H\) by setting
		\[
			\forall y \in E, \, \forall h \in H, \quad \rho_\alpha(y \otimes h) = \leg{\sigma(\alpha \otimes y)}{0} \otimes \pi(\leg{\sigma(\alpha \otimes y)}{1}) h.
		\]
	\end{enumerate}
\end{lemma}

\begin{proof}
	Before continuing, let us fix a basis \((e_i)_{i=1}^N\) for \(E\).
	Recall, moreover, that by the proof of Proposition \ref{prop:basicseparable}, \emph{mutatis mutandis}, every positive \(X \in M_n(B)\) satisfies
	\begin{equation}\label{eq:matrix}
		\forall h = (h_i)_{i=1}^N \in H^N, \quad \ip{h}{\pi_n(X)h} \leq \norm{X} \sum\nolimits_{i=1}^N \norm{h_i}^2,
	\end{equation}
	where \(\pi_n : M_n(B) \to \bL(H^n)\) is the bounded \(\ast\)-homomorphism canonically induced by \(\pi : B \to \bL(H)\).
	Note that this applies, in particular, to \(X \coloneqq (\hp{e_i}{e_j})_{i,j=1}^N\).
	
	First, let \(x \in E\) be given.
	Define \(\psi_E[x] : H \to E \otimes_B H\) by \(\psi_E[x] \coloneqq (h \mapsto x \otimes h)\).
	A standard calculation show that \(\phi_E[x] = \psi_E[x]^\ast\) and that  \(\psi_E[x]\) is bounded with operator norm \(\norm{\psi_E[x]} = \norm{\pi(\ip{x}{x})}^{1/2} \leq 1\), so that \(\phi_E[x]\) is contractive.
	Since \((e_i)_{i=1}^N\) is a basis for \(E\), it now follows that
	\begin{equation}\label{eq:vectorframe}
		\forall \xi \in E \otimes_B H, \quad \xi = \sum\nolimits_{i=1}^N e_i \otimes \phi_{E}[e_i]\xi.
	\end{equation}
	
	Next, let \(V\) be a countable dense subset of \(H\); we claim that \(E \otimes_B H\) admits the dense subset \(\Set{\sum_{i=1}^N e_i \otimes v_i \given v_1,\dotsc,v_N \in V}\).
	Let \(\xi \in E \otimes_B H\) and \(\epsilon > 0\) be given.
	Let \(X \coloneqq (\hp{e_i}{e_j})_{i,j=1}^N\), and choose \(v_1,\dotsc,v_N \in V\), such that \(\norm{\phi_E[e_i]\xi - v_i}^2 < \tfrac{\epsilon^2}{CN+1}\).
	Then, by \eqref{eq:vectorframe} and \eqref{eq:matrix},
	\begin{multline*}
		\norm*{\xi-\sum\nolimits_{i=1}^N e_i \otimes v_i}^2 = \sum\nolimits_{i,j=1}^N \ip{\phi_{e_i}\xi - v_i}{\pi(\hp{e_i}{e_j})(\phi_E[e_j]\xi-v_j}\\ \leq \norm{X} \sum\nolimits_{i=1}^N \norm{\phi_{E}[e_i]\xi-v_i}^2 < \epsilon^2.
	\end{multline*}
	
	Next, that \(\id \otimes_\nabla D\) is well-defined and formally self-adjoint is well-known in the literature on unbounded \(\mathrm{KK}\)-theory---see, e.g., \cite[Lemma 2.28]{BMS}.
	
	Finally, let \(\alpha \in \Omega^1_B\) be given.
	Then right \(B\)-linearity of the generalised braiding \(\sigma\) guarantees that \(\rho_E[\alpha]\) is a well-defined map.
	Hence, by \eqref{eq:vectorframe}, for every \(\xi \in E \otimes_B H\), 
	\begin{align*}
		\norm{\rho_E[\alpha]\xi} &\leq \sum\nolimits_{i,j=1}^N \norm*{e_i \otimes \pi_D\mleft(\hp{e_i}{\sigma(\alpha \otimes e_j)}\mright)\phi_{E}(e_j)\xi}\\ &\leq \left(\sum\nolimits_{i,j=1}^N \norm{e_i} \cdot \norm{\pi_D\mleft(\hp{e_i}{\sigma(\alpha \otimes e_j)}\mright)}\right) \norm{\xi}. \qedhere
	\end{align*}
\end{proof}

\begin{proof}[Proof of Proposition \ref{prop:liftcommrepconstruct}]
	For notational convenience, we conflate the isotypical subspace \(P_j\) with the Hermitian line \(B\)-bimodule \(P_j\) for each \(j \in \bZ\).
	Moreover, we shall also use the notation of Lemma \ref{lem:commreplift} and its proof. 
	Let us first show that 
	\[
		(\tilde{H},\tilde{\pi},\tilde{D}) \coloneqq \left((P \otimes \bC^2) \otimes_B H, \id \otimes \id \otimes \pi(\cdot),\iu{}(\Lambda_\kappa \circ \partial_\kappa) \otimes \sigma^2 \otimes \id + (\id \otimes \sigma^3) \otimes_\nabla D\right)
	\]
	defines a faithful locally bounded commutator representation of \((P;\Omega_P,\du_P)\).
	Let \(\mathfrak{B}\) be the \(C^\ast\)-algebraic completion of \(B\), and let \(\tau : \tilde{H} \to \bC^2 \otimes (P \otimes_B H)\) be the canonical unitary defined by \(\tau \coloneqq \left(p \otimes x \otimes h \mapsto x \otimes p \otimes h\right)\).
	
	First, we show that \(\tilde{H}\) is separable.
	By Lemma \ref{lem:commreplift}, the pre-Hilbert space \(P \otimes_B H = \bigoplus_{j \in \bZ} P_j \otimes_B H\) is separable, so that \(\tilde{H} \cong (P \otimes_B H)^2\) is also separable.
	Hence, \(\chi_{\tilde{H}} \coloneqq {\id} \otimes \sigma^3 \otimes \chi_H\) defines a \(\bZ_2\)-grading on \(\tilde{H}\) and \(\sigma_\cdot \otimes \id \otimes \id\) defines a unitary \(\U\)-representation of finite type on \(\tilde{H}\) with \(\tilde{H}_j = (P_j \otimes \bC^2) \otimes_B H \cong (P \otimes_B H)^2\) for \(j \in \bZ\).
	
	Next, we show that \(\tilde{\pi}\) is well-defined.
	It suffices to show that the left \(P\)-module structure on \(P \otimes_B H\) defines a bounded \(\ast\)-homomorphism \(\lambda : P \to \bL(P \otimes_B H)\), since this will imply boundedness of \(\tilde{\pi} = \tau^\ast \circ (\id \otimes \lambda(\cdot)) \circ \tau\); the other properties of \(\tilde{\pi}\) will follow by routine checks.
	In turn, the only non-trivial points are that \(\lambda\) is well-defined and bounded as a map of Banach spaces.
	Let \(p \in P\) be given, so that there exists \(N \in \bN\), such that \(p \in \bigoplus_{j=-N}^N P_j\); hence, we may uniquely write \(p = \sum_{j=-N}^N \hat{p}(j)\), where \(\hat{p}(j) \in P_j\) for each \(j \in \Set{-N,\dotsc,N}\), so that \(\bE_P(p^\ast p) = \sum_{j=-N}^N \hat{p}(j)^\ast\hat{p}(j)\).
	Let \(k \in \bN\) and \(\xi \in P_j \otimes_B H\) be given. Let \((e_i)_{i=1}^M\) be a basis for \(P_j\).
	Then 
	\begin{align*}
		\norm{\lambda(p)\xi}^2 &= \sum\nolimits_{m,n=1}^M \ip{\phi_{P_j}[e_m]\xi}{\pi(\bE_P(e^\ast_m p^\ast p e_n))\phi_{P_j}[e_n]\xi}\\ &= \sum\nolimits_{m,n=1}^M \ip{\phi_{P_j}[e_m]\xi}{\pi(e^\ast_m \bE_P(p^\ast p) e_n)}.
	\end{align*}
	Now, let \(b \coloneqq \sqrt{\bE_P(p^\ast p)} \in \mathfrak{B}\).
	After passing to Hilbert \(C^\ast\)-module and Hilbert space completions, we may apply \cite[p.\ 42]{Lance} to conclude that
	\begin{align*}
		\sum_{m,n=1}^M \ip{\phi_{P_j}[e_m]\xi}{\pi(e^\ast_m \bE_P(p^\ast p) e_n)} 
		&= \sum_{m,n=1}^N \ip{\phi_{P_j}[e_m]\xi}{\pi(\hp{b e_m}{b e_n})\phi_{P_j}[e_n]\xi}\\
		&\leq \norm{b}^2 \sum_{m,n=1}^N \ip{\phi_{P_j}[e_m]\xi}{\pi(\hp{e_m}{e_n}_j)\phi_{P_j}[e_n]\xi}\\
		&\leq \norm{p}^2 \norm{\xi}^2.
	\end{align*}
	
	Next, we show that \(\tilde{D}\) is \(\U\)-invariant, odd, and symmetric.
	Define \(S\) and \(T\) satisfying \(\tilde{D} = S + T\) by \(S \coloneqq \iu{}(\Lambda_\kappa \circ \partial_\kappa) \otimes \sigma^2 \otimes \id\) and \(T \coloneqq (\id \otimes \sigma^3) \otimes_\nabla D\), respectively.
	On the one hand, the block-diagonal operator \(S = \bigoplus_{j\in\bZ}(-2\pi[j]_\kappa\kappa^{-j}\id) \otimes \sigma^2 \otimes \id\) is \(\U\)-invariant, odd, and formally self-adjoint by construction.
	On the other hand,the operator \(T\) is likewise \(\U\)-invariant and odd by construction; by \(\U\)-invariance, it follows that \(T = \bigoplus_{j \in \bZ} \rest{T}{\tilde{H}_j}\), where
	\(
		\rest{T}{\tilde{H}_j} = \tau^\ast \circ \left(\sigma^3 \otimes (\id \otimes_{\nabla_{P,j}} D)\right) \circ \rest{\tau}{\tilde{H}_j}
	\)
	is symmetric for each \(j \in \bZ\) by Lemma \ref{lem:commreplift}.
	
	Next, we show that \(\tilde{\pi}_{\tilde{D}} : \Omega^1_P \to \bL_{\loc}^{\U}(\tilde{H})\) is well-defined.
	On the one hand, recall that  \((\id-\Pi)(\Omega^1_P)\) is freely generated as a left \(P\)-module by the connection \(1\)-form \(\vartheta\) of \(\Pi\); hence, we may define a map \(\tilde{\pi}_\ver : (\id-\Pi)(\Omega^1_P) \to \bL_{\loc}^{\U}(\tilde{H})\) by
	\[
		\forall p \in P, \quad \tilde{\pi}_\ver(p \vartheta) \coloneqq \tilde{\pi}(p) \cdot (\Lambda_\kappa \otimes \sigma^2 \otimes \id).
	\]
	On the other hand, since \((p \otimes \alpha \mapsto p\alpha) : P \otimes_B \Omega^1_B \to \Pi(\Omega^1_P)\) is a \(B\)-bimodule isomorphism, we may use Lemma \ref{lem:commreplift} to define \(\tilde{\pi}_{\hor} : \Pi(\Omega^1_P) \to \bL_{\loc}^{\U}(\tilde{H})\) by
	\[
		\forall p \in P, \, \forall \alpha \in \Omega^1_B, \, \forall j \in \bZ, \, \forall  \quad \rest{\tilde{\pi}_\hor(p\alpha)}{\tilde{H}_j} \coloneqq \tilde{\pi}(p) \circ \tau^\ast \circ (\sigma^3 \otimes \rho_{P_j}[\alpha]) \circ \rest{\tau}{\tilde{H}_j}.
	\]
	Since \(\tilde{D} = S+T\), it now suffices to show that
	\[
		\forall p \in P, \quad \iu{}[S,\tilde{\pi}(p)] = \tilde{\pi}_\ver \circ (\id-\Pi) \circ \du_P(p), \quad \iu{}[T,\tilde{\pi}(p)] = \tilde{\pi}_\hor \circ \Pi \circ \du_P(p).
	\]
	
	Finally, let \(j,k \in \bZ\), \(p \in P_j\), \(q \in P_k\), \(x \in \bC^2\), and \(h \in H\).
	On the one hand,
	\begin{align*}
		[S,\tilde{\pi}(p)](q \otimes x \otimes h) 
		&= -2\pi [j+k]_\kappa \kappa^{-j-k} pq \otimes \sigma^2 x \otimes h + 2\pi[k]_\kappa \kappa^{-k}pq \otimes \sigma^2 x\otimes h\\
		&= 2\pi[j]_\kappa\kappa^{-j}pq \otimes \sigma^2 x \otimes h\\
		& -\iu{}\tilde{\pi}_\ver(2\pi\iu{}[j]_\kappa\kappa^{-j}p\vartheta)(q \otimes x \otimes h)\\
		&= -\iu{}(\tilde{\pi}_\ver \circ (\id-\Pi) \circ \du_P)(p)(q \otimes x \otimes h).
	\end{align*}
	On the other hand, since \(\Pi \circ \du_P\) is a derivation, it follows that
	\begin{align*}
		\nabla_{P;j+k}(pq) &= \leg{\nabla_{P;j}(p)}{0} \leg{\leg{\sigma_{P;k}(\nabla_{P;j}(p))}{1}\otimes q)}{0} \otimes \leg{\leg{\sigma_{P;k}(\nabla_{P;j}(p)}{1}\otimes q)}{1}\\ &\quad\quad\quad\quad + p \leg{\nabla_{P;k}(q)}{0} \otimes \leg{\nabla_{P;k}(q)}{1},
	\end{align*}
	so that \(\iu{}[T,\tilde{\pi}(p)](q \otimes x \otimes h) = (\tilde{\pi}_\hor \circ \Pi \circ \du_P)(p)(q \otimes x \otimes h)\) since
	\begin{align*}
		&\iu{}T(pq \otimes x \otimes h)\\
		&\quad\quad=	\leg{\nabla_{P;j}(p)}{0} \leg{\leg{\sigma_{P;k}(\nabla_{P;j}(p)}{1}\otimes q)}{0} \otimes \sigma^3 x \otimes \pi(\leg{\leg{\sigma_{P;k}(\nabla_{P;j}(p)}{1}\otimes q)}{1})h\\
		&\quad\quad\quad + p \leg{\nabla_{P;k}(q)}{0} \otimes \sigma^3 x \otimes\pi(\leg{\nabla_{P;k}(q)}{1})h + pq \otimes \sigma^3 x \otimes Dh\\
		&\quad\quad= (\tilde{\pi}_\hor \circ \Pi \circ \du_P)(p)(q \otimes x \otimes h) + \iu{}\tilde{\pi}(p)T(q \otimes x \otimes h).
	\end{align*}
	
	Finally, we show that \((\tilde{H},\tilde{\pi},\tilde{D})\) is faithful.
	We first show that \(\tilde{\pi}\) is isometric.
	Since \(\tilde{\pi}\) is bounded, faithful, and \(\U\)-equivariant, it suffices by Corollary \ref{cor:fell} to show that \(\rest{\tilde{\pi}}{B}\) is isometric.
	Indeed, let \(b \in B\).
	Since \(\tau^\ast(B \otimes_B H) \cong H\) is an orthogonal direct summand of \(\tilde{H}\) and \(\pi\) is isometric, 
	\(
		\norm{b} \geq \norm{\tilde{\pi}(b)} \geq \norm{\rest{\tau \tilde{\pi}(b) \tau^\ast}{B \otimes_B H}} = \norm{\pi(b)} = \norm{b}
	\).
	Now, let us show that \(\tilde{\pi}_{\tilde{D}}\) is injective; to do so, it suffices to show that \(\tilde{\pi}_\ver\) and \(\tilde{\pi}_\hor\) are both injective.
	On the one hand, \(\tilde{\pi}_\ver\) is injective since \(\tilde{\pi}\) is injective and \(\tilde{\pi}_\ver(\vartheta)\) is invertible.
	On the other, to show that \(\tilde{\pi}_\hor\) is injective, it suffices to show injectivity of \(f : P \otimes_B \Omega^1_B \to \End_{\bC}(H,P \otimes_B H)\) defined by \(f(p \otimes \beta)h \coloneqq \pi(p)\pi_D(\beta)h\) for \(p \in P\), \(\beta \in \Omega^1_B\), and \(h \in H\).
	Indeed, fix \(j \in \bZ\), and note that \(\rest{f_j}{P_j \otimes_B \Omega^1_B} = r_j \circ s_j\), where \(s_j : P_j \otimes_B \Omega^1_B \to P_j \otimes_B \bL(H)\) and \(r_j : P_j \otimes_B \bL(H) \to \End_{\bC}(H,P_j \otimes_B H)\) are given by \(s_j \coloneqq \id \otimes \pi_D\) and \(r_j \coloneqq (p \otimes S \mapsto \psi_{P_j}[p]S)\), respectively.
	Then \(s_j\) is injective by flatness of the projective right \(B\)-module \(P_j\), while \(r_j\) is injective by existence of the left inverse \(T \mapsto \sum_{i=1}^N e_i \otimes \phi_{P_j}[e_i]T\), where \((e_i)_{i=1}^N\) is any basis for the Hermitian line \(B\)-bimodule \(P_j\).
	Hence, the map \(\rest{f_j}{P_j \otimes_B \Omega^1_B} : P_j \otimes_B \Omega^1_B \to \End_{\bC}(H,P_j \otimes_B H)\) is also injective.
	
	Now, let \(\tilde{\Gamma} \coloneqq \id \otimes \sigma^3 \otimes \id\), which is an even \(\U\)-invariant self-adjoint unitary commuting with \(\ran\tilde{\pi}\).
	Let us check that \((\tilde{H},\tilde{\pi},\tilde{D},\tilde{\Gamma})\) defines a lift of \((H,\pi,D)\).
	
	First, let \(M : P \otimes_B \tilde{H}^{\U} \to \tilde{H}\) be given by \(M \coloneqq (p \otimes \xi \mapsto \tilde{\pi}\xi)\).
	Define a left \(B\)-linear unitary \(\Phi : \bC^2 \otimes H \to \tilde{H}^{\U}\) by \(\Phi \coloneqq (x \otimes h \mapsto 1 \otimes x \otimes h)\), and observe that
	\[
		\forall p \in P, \, \forall x \in \bC^2, \, \forall h \in H, \quad \tau \circ M \circ (\id \otimes \Phi)(p \otimes x \otimes h) = x \otimes p \otimes h,
	\]
	so that \(\tau \circ M \circ (\id \otimes \Phi) : P \otimes_B (\bC^2 \otimes H) \to \bC^2 \otimes (P \otimes_B H)\) is manifestly bijective, which implies that \(M\) is bijective as well.
	Next, note that \(\tilde{pi}_{\tilde{D}}(\vartheta)^2 = \Lambda_\kappa^2\) since \(\tilde{\pi}_{\tilde{D}}(\vartheta) = \tilde{\pi}_\ver(\vartheta) = \Lambda_\kappa \otimes \sigma^2 \otimes \id\), which also shows that \(\tilde{\Gamma}\) anticommutes with \(\tilde{\Gamma}\).
	Next, observe that \(\tilde{\Gamma}\) anticommutes with \(S\) and commutes with \(T\), so that \(\tilde{D}_\hor = T\) and \(Z = S + \iu{}\tilde{\pi}_{\tilde{D}}(\vartheta)\partial_\kappa = 0\).	
	Next, since \(\tau^\ast (B \otimes H) \cong H\) is an orthogonal direct summand of \(H^{\U}\), the proof that \(\tilde{\pi}\) is isometric also shows that the map \((b \mapsto \rest{\tilde{\pi}(b)}{\tilde{H}^{\U}}) : B \to \bL(\tilde{H}^{\U})\) is isometric.
	Likewise, the proof that \(\tilde{\pi}_\hor\) is injective, specialised to \(j=0\), shows that \((\beta \mapsto \rest{\tilde{\pi}_{\tilde{D}}(\beta)}{\tilde{H}^{\U}}) : \Omega^1_B \to \bL(\tilde{H}^{\U})\) is injective.
	Finally, we may construct an arrow \(V : (H,\pi,D) \to \iota_P^\ast(\tilde{H},\tilde{\pi},\tilde{D},\tilde{\Gamma})\) by setting \(V \coloneqq \left(h \mapsto 1 \otimes \left(\begin{smallmatrix}1\\0\end{smallmatrix}\right) \otimes h\right)\).
\end{proof}

Having proved existence of lifts, we now show that they are indeed unique up to \(\U\)-equivariant unitary equivalence after perturbation by a relative remainder.

\begin{theorem}\label{thm:liftequivalence}
	The functor \(\iota_P^\ast\) of Proposition \ref{prop:restfunct} is an equivalence of categories with weak inverse \((\iota_P)_! : \grp{BCRep}(B) \to \grp{PCRep}(P;\Pi)\) defined as follows.
	\begin{enumerate}[leftmargin=*]
		\item Given an object \((H,\pi,D)\), let \((\iota_P)_!(H,\pi,D)\) be the projectable commutator representation of \((P;\Omega_P,\du_P;\Pi)\) constructed from \((H,\pi,D)\) by Proposition \ref{prop:liftcommrepconstruct}.
		\item Given an arrow \(U : (H_1,\pi_1,D_1) \to (H_2,\pi_2,D_2)\), let \((\iota_P)_!(U) \coloneqq ({\id} \otimes {\id} \otimes U,0)\).
	\end{enumerate}
	Thus, in particular, every bounded commutator representation of \((B;\Omega_B,\du_B)\) has an essentially unique lift to \((P;\Omega_P,\du_P;\Pi)\).
\end{theorem}

\begin{proof}
	It remains to construct natural isomorphisms \(U : \id_{\grp{PCRep}(P;\Pi)} \Rightarrow (\iota_P)_! \circ \iota_P^\ast\) and \(V : \id_{\grp{BCRep}(B)} \Rightarrow \iota_P^\ast \circ (\iota_P)_!\).
	
	First, let \((H,\pi,D,\Gamma)\) be an object of \(\grp{PCRep}(P;\Pi)\); let \(D_\hor\) be its horizontal Dirac operator and \(Z\) its remainder, and let \((H_B,\pi_B,D_B) \coloneqq \iota_P^\ast(H,\pi,D,\Gamma)\).
	Define an even \(\U\)-equivariant unitary \(\Upsilon : (P \otimes \bC^2) \otimes_B H_B \to H\) by
	\begin{multline*}
		\forall p \in P, \, \forall \begin{pmatrix}v_1\\v_2\end{pmatrix} \in \bC^2, \, \forall h \in H_B, \\ \Upsilon\mleft(p \otimes \begin{pmatrix}v_1\\v_2\end{pmatrix} \otimes h\mright) \coloneqq \pi(p)\left(v_1\id - \iu{}v_2\Gamma\pi_D(\vartheta)\Lambda_\kappa^{-1}\right)h.
	\end{multline*}
	A straightforward if tedious calculation generalising the proof of Proposition \ref{prop:totalderham} now shows, in the notation of Proposition \ref{prop:liftcommrepconstruct}, that
	\begin{gather*}
		\Upsilon^\ast (-\iu{}\pi_D(\vartheta)\partial_\kappa) \Upsilon = \iu{}(\Lambda_\kappa \circ \partial_\kappa) \otimes \sigma^2 \otimes \id, \quad \Upsilon^\ast D_\hor \Upsilon = ({\id} \otimes \sigma^3) \otimes_\nabla D_B, \\ \Upsilon^\ast \Gamma \Upsilon = {\id} \otimes \sigma^3 \otimes \id,
	\end{gather*}
	so that we may take \(U_{(H,\pi,D,\Gamma)} \coloneqq (\Upsilon^\ast,Z)\).
	
	Now, let \((H,\pi,D)\) be an object of \(\grp{BCRep}(B)\).
	Then, as in the proof of Proposition \ref{prop:liftcommrepconstruct}, we may define \(V_{(H,\pi,D)} : (H,\pi,D) \to \iota_P^\ast \circ (\iota_P)_!(H,\pi,D)\) by setting \(V_{(H,\pi,D)} \coloneqq \left(h \mapsto 1 \otimes \left(\begin{smallmatrix}1\\0\end{smallmatrix}\right) \otimes h\right)\).
\end{proof}

Note that Corollary \ref{cor:diffpimsner} and Theorem \ref{thm:liftequivalence} combine to yield a formalisation of the constructions of Bellissard--Marcolli--Reihani \cite{BMR} and Gabriel--Grensing \cite{GG} for (generalised) crossed product spectral triples.
Indeed, let \((H,\pi,D)\) be a faithful bounded commutator representation of \((B;\Omega_B,\du_B)\).
Let \((E,\sigma_E,\nabla_E)\) be a Hermitian line \(B\)-bimodule with connection, such that \(\epsilon_1 \circ \hat{\cL} \circ \Hor_\kappa(P;\Omega_P,\du_P;\Pi) \cong (E,\sigma_E,\nabla_E)\).
Then the \emph{\(\kappa\)-total crossed product} of \((H,\pi,D)\) by \((E,\sigma_E,\nabla_E)\) is the canonical lift
\(
	(H,\pi,D) \rtimes_{(E,\sigma_E,\bZ)}^{\kappa,\tot} \bZ \coloneqq (\iota_P)_!(H,\pi,D)
\)
of \((H,\pi,D)\) to \((P;\Omega_P,\du_P)\).

\begin{remark}\label{rem:kk3}
	We continue from Remarks \ref{rem:kk1} and \ref{rem:kk2}.
	The right pre-Hilbert \(B\)-module \(P \otimes \bC^2 \cong \bigoplus_{j \in \bZ} \cL(P)(j)^2\) yields a formal \(\U\)-equivariant unbounded \(KK_1\)-cycle \((P,P \otimes \bC^2,\iu{}(\Lambda_\kappa \circ \partial_\kappa) \otimes \sigma^2)\) for \((\mathfrak{P},\mathfrak{B})\), where \({\id} \otimes \sigma^1\) generates the \(1\)-multigrading.
	This defines a genuine \(\U\)-equivariant unbounded \(KK_1\)-cycle for \((\mathfrak{P},\mathfrak{B})\) if and only if \(\kappa=1\), in which case, it recovers a well-known construction of Carey--Neshveyev--Nest--Rennie \cite[Cor.\ 2.10]{CNNR} up to \(1\)-multigrading; in all cases, its formal bounded transform recovers, up to \(1\)-multigrading, the canonical representative of the extension class \([\partial] \in KK_1(\mathfrak{P},\mathfrak{B})\) of \(\mathfrak{P}\) as a Pimsner algebra \cite[\S 2.2]{AKL}.
	Moreover, we may now reinterpret Theorem \ref{thm:liftequivalence} in terms of formal unbounded Kasparov products \cite{Mesland,KL}:
	\begin{enumerate}[leftmargin=*]
		\item given an object \((H,\pi,D)\) of \(\grp{BCRep}(B)\), we may write
		\[
			(P,\tilde{H},\tilde{D}) \coloneqq (P, P \otimes \bC^2, \iu{}(\Lambda_\kappa \circ \partial_\kappa) \otimes \sigma^2; \nabla) \otimes_B (B,H,D);
		\]
		\item given an object \((H,\pi,D,\Gamma)\) of \(\grp{PCRep}(P;\Pi)\), the isomorphism \(U_{(H,\pi,D,\Gamma)}\) of the proof of Theorem \ref{thm:liftequivalence} yields
		\[
			(P,H,D-Z) \cong (P, P \otimes \bC^2, \iu{}(\Lambda_\kappa \circ \partial_\kappa) \otimes \sigma^2; \nabla) \otimes_B (B,H_B,D_B),
		\]
		where \(Z\) is the remainder of \((H,\pi,D,\Gamma)\) and \((H_B,\pi_B,D_B) \coloneqq \iota_P^\ast(H,\pi,D,\Gamma)\).
	\end{enumerate}
	In both cases, \(\nabla\) is the represented connection on \((P,P \otimes \bC^2,\iu{}(\Lambda_\kappa \circ \partial_\kappa) \otimes \sigma^2)\) constructed from \(\Pi\) in Proposition \ref{prop:liftcommrepconstruct}.
	In the second case, if \((P,H,D)\) defines a genuine \(\U\)-equivariant unbounded \(KK_1\)-cycle for \((\mathfrak{P},\bC)\), then this formal unbounded Kasparov product defines a genuine unbounded Kasparov product \cite[Thm.\ 2.44]{CaMe}.
	Otherwise, the \(KK\)-theoretic significance of Theorem \ref{thm:liftequivalence} is an open question.
\end{remark}

\begin{corollary}
	Suppose that \((\star_B,\tau_B)\) is a Riemannian geometry on \((B;\Omega_B,\du_B)\) that lifts to a total Riemannian geometry \((\Delta_\ver,\Delta_\hor,\star,\tau)\) on \((P,\Omega_P,\du_P,\Pi)\).
	Then the total Hodge--de Rham commutator representation \((P;\pi_P,\du_P+\du_P^\ast;2\Pi-\id)\) induced by \((\Delta_\ver,\Delta_\hor,\star,\tau)\) is the essentially unique lift of the Hodge--de Rham commutator representation \((B;\pi_B,\du_B+\du_B^\ast)\) induced by \((\star_B,\tau_B)\) to \((P;\Omega_P,\du_P;\Pi)\).
\end{corollary}

\begin{example}\label{ex:hopf10}
	Continuing from Examples \ref{ex:hopf5} and \ref{ex:hopf8}, we use Proposition \ref{prop:liftcommrepconstruct} to construct
	\(
		\iota^\ast_{\cO_q(\SU(2))}(\slashed{S}_q(\CP^1),\pi,\slashed{D}_1).
	\)
	First, let \(\slashed{S}_q(\SU(2)) \coloneqq \bC^2 \otimes \bC^2 \otimes \cO_q(\SU(2))\) with the inner product \(\ip{}{}\) given by
	\begin{multline*}
		\forall x_1,x_2,y_1,y_2 \in \bC^2, \, \forall p_1,p_2 \in \cO_q(\SU(2)), \\ \ip{x_1\otimes y_1 \otimes p_1}{x_2 \otimes y_2 \otimes p_2} \coloneqq \ip{x_1}{x_2}\ip{y_1}{y_2}h_q(p_1^\ast p_2),
	\end{multline*}
	the \(\bZ_2\)-grading \(\sigma^3 \otimes \sigma^3 \otimes \id\), and the unitary \(\U\)-representation of finite type \(\tilde{U} \coloneqq \left(z \mapsto {\id} \otimes \left(\begin{smallmatrix}z & 0 \\ 0 & z^{-1}\end{smallmatrix}\right) \otimes \alpha_z\right)\); hence, define \(\Lambda_{q^2}\) and \(\partial_{q^2}\) on \(\slashed{S}_q(\SU(2))\) in terms of \(\tilde{U}\).
	Next, let \(\tilde{\pi} : \cO_q(\SU(2)) \to \bL^{\U}(\slashed{S}_q(\SU(2)))_{\even}\) be induced by multiplication from the left in \(\cO_q(\SU(2))\).
	At last, let \(\tsD_q \coloneqq \tsD_{q,\ver} + \tsD_{q,\hor}\), 
	\begin{align*}
		\tsD_{q,\ver} &\coloneqq \iu{}(\sigma^2 \otimes \id \otimes \id) \circ \Lambda_{q^2} \circ \partial_{q^2}, \\
		\tsD_{q,\hor} &\coloneqq \sigma^3 \otimes \left(\tfrac{1}{2}(\sigma^1-\iu{}\sigma^2) \otimes q\partial_{-} + \tfrac{1}{2}(\sigma^1+\iu{}\sigma^2) \otimes q\partial_{+}\right),
	\end{align*}
	and let
	\(
		\Gamma_q \coloneqq \sigma^3 \otimes {\id} \otimes \id
	\).
	Since \((p \otimes x \mapsto p \cdot x) : \cO_q(\SU(2)) \otimes_{\cO_q(\CP^1)} \slashed{S}_{q,\pm}(\CP^1)\) are left \(\cO_q(\SU(2))\)-module isomorphisms by Proposition \ref{prop:assocline}, we may construct an even \(\U\)-equivariant unitary \(\Phi : (\cO_q(\SU(2)) \otimes \bC^2) \otimes_{\cO_q(\CP^1)} \slashed{S}_q(\CP^1) \to \slashed{S}_q(\SU(2))\) by
	\begin{multline*}
		\forall p \in \cO_q(\SU(2)), \, \forall x \in \bC^2, \, \forall \left(\begin{smallmatrix}s_+\\s_-\end{smallmatrix}\right) \in \slashed{S}_q(\CP^1), \\ \Phi\mleft(p \otimes x \otimes \left(\begin{smallmatrix}s_+\\s_-\end{smallmatrix}\right)\mright) \coloneqq x \otimes \left(\left(\begin{smallmatrix}1\\0\end{smallmatrix}\right)\otimes p \cdot s_+ + \left(\begin{smallmatrix}0\\1\end{smallmatrix}\right) \otimes p \cdot s_-\right),
	\end{multline*}
	which yields the desired isomorphism of projective commutator representations
	\[
		(\Phi,0) : \iota_{\cO_q(\SU(2))}(\slashed{S}_q(\CP^1),\pi,\slashed{D}_q) \to \left(\slashed{S}_q(\SU(2)),\tilde{\pi},\tsD_q,\Gamma_q\right).
	\]
	Note that \(\left(\slashed{S}_q(\SU(2)),\tilde{\pi},\tsD_q,\Gamma_q\right)\) is faithful since \((\slashed{S}_q(\CP^1),\pi,\slashed{D}_1)\) is.
\end{example}

\subsection{Twisted boundedness of lifted commutator representations}\label{sec:4.4}

We have solved the lifting problem for faithful bounded commutator representations, but at a cost: faithful projectable commutator representations typically involve unbounded represented \(1\)-forms.
We now control this unboundedness in the spirit of Connes--Moscovici's \emph{twisted spectral triples} \cite{CM} by permitting distinct vertical and horizontal twists.
One upshot is that quantum \(\SU(2)\) \emph{qua} total space of the \(q\)-monopole does not admit a non-pathological \(\U\)-equivariant twisted spectral triple. 
The other is that Kaad--Kyed's compact quantum metric space \cite{KK} on quantum \(\SU(2)\) for a canonical choice of parameters can be geometrically derived, up to equivalence of Lipschitz seminorms, from the spin Dirac spectral triple on quantum \(\CP^1\) using the \(q\)-monopole.

Once more, let \(\kappa > 0\), let \((P,\Omega_P,\du_P,\Pi)\) be a \(\kappa\)-differentiable quantum principal \(\U\)-bundle over \(B\), let \(\vartheta\) be the connection \(1\)-form of \(\Pi\), let \(\hat{\Phi}_P\) be the Fr\"{o}hlich automorphism of \(\Hor_\kappa(P;\Omega_P,\du_P;\Pi) = (P,\Omega_{P,\hor},\du_{P,\hor})\), and let \(\Phi_P\) be the Fr\"{o}hlich automorphism of the Hermitian line \(B\)-bimodule \(\cL(P)(1)\), so that \(\hat{\Phi}_P\) and \(\Phi_P\) agree on \(\Zent(\Omega_B)^0\).
Hence, recall that \(\hat{\Phi}_P\) induces the right \(\bZ\)-action on \(\mathcal{Z}_{>0}(B) \coloneqq (\Zent(\Omega^B)^0)^\times_{+}\) defined by \eqref{eq:horizontalact}, which therefore extends, \emph{mutatis mutandis}, to a right \(\bZ\)-action on \(\Zent(B)^{\times}_{+}\) in terms of \(\Phi_P\).

We begin with the analogue for locally bounded commutator representations of modular automorphisms.

\begin{definition}
	Suppose that \((H,\pi,D)\) is a locally bounded commutator representation of \((P;\Omega_P,\du_P)\).
	A \emph{modular symmetry} of \((H,\pi,D)\) is an even positive \(\U\)-invariant invertible operator \(N \in \bL_{\loc}^{\U}(H)\) the restricts to the identity on \(H^{\U}\), commutes with \(\pi(B)\), and satisfies \(N \ran(\pi) N^{-1} = \ran(\pi)\).
\end{definition}

\begin{remark}\label{rem:twist}
Let \(\mathfrak{P}\) denote the \(C^\ast\)-completion for \(P\).
Suppose that \((H,\pi,D)\) is a locally bounded commutator representation of \((P;\Omega_P,\du_P)\), that \(\nu\) is a modular automorphism of \(\Omega_P\), and that \(N\) is a modular symmetry of \((H,\pi,D)\) that satisfies \(N^{-1} \pi(\cdot) N = \pi \circ \nu\).
Hence, let \(D^N \coloneqq NDN\).
Since, for all \(p \in P\),
\[
	N[D,\pi(p)]N = D^N \pi\mleft(\nu(p)\mright) - \pi\mleft(\nu^{-1}(p)\mright) D^N = D^N \pi\mleft(\nu(p)\mright) - \pi\mleft(\nu^{-2}\mleft(\nu(p)\mright)\mright) D^N,
\]
it follows that \((P,H,D^N)\) defined a \(\U\)-equivariant even \(\nu^{-2}\)-twisted spectral triple for \(\mathfrak{P}\) only if \(N \ran(\pi_D) N \subseteq \bL^{\U}(H)\).
\end{remark}

In light of the above remark, the following theorem will exclude the existence of non-pathological \(\U\)-equivariant twisted spectral triples that faithfully represent the total spaces of the \(q\)-monopole or the real multiplication instanton of Example \ref{ex:heis6}, respectively.
In particular, it will imply that faithful projectable commutator representations of these two examples cannot naturally be accommodated by the theory of twisted spectral triples.

\begin{theorem}\label{thm:nogo}
	Suppose that \(\Zent(B) = \bC\).
	Let \((H,\pi,D)\) be a locally bounded commutator representation of \((P;\Omega_P,\du_P)\), such that \(\pi\) is injective, \(\pi(P) \cdot H^{\U}\) is dense in \(H\), and there exists a modular symmetry \(N\) of \((H,\pi,D)\), such that \(N  \pi_D(\Omega^1_P)  N \subseteq \bL^{\U}(H)\).
	Suppose that  \(\eta \in \Omega^1_{P,\hor} \setminus \Set{0}\) and \(t \in (0,\infty) \setminus \Set{\kappa}\) satisfy
	\begin{equation}\label{eq:tcentral}
		\forall p \in P, \quad \eta \cdot p = \Lambda_t(p) \cdot \eta.
	\end{equation}
	Then \((\id-\Pi)(\Omega^1_P) \subseteq \ker \pi_D\) or \(\pi_D(\eta) = 0\).
\end{theorem}

\begin{lemma}\label{lem:modularsym}
	Let \((H,\pi,D)\) be a \(\U\)-equivariant commutator representation of \((P;\Omega_P,\du_P)\).
	If \(\pi\) is injective, \(\left(b \mapsto \rest{\pi(b)}{H^{\U}}\right) : \Zent(B) \to \bL(H^{\U})\) is isometric, and \(\pi(P) \cdot H^{\U}\) is dense in \(H\), then for every modular symmetry \(N\) of \((H,\pi,D)\), there exists a unique right \(1\)-cocycle \(\mu : \bZ \to \Zent(B)^\times_+\), such that 
	\begin{equation}\label{eq:modularsymm}
		N = \bigoplus_{j \in \bZ} \rest{\pi(\mu(-j))}{H_j}.
	\end{equation}
	Conversely, for every \(1\)-cocycle \(\mu : \bZ \to \Zent(B)^\times_+\), Equation \ref{eq:modularsymm} defines a modular symmetry \(N\) of \((H,\pi,D)\).
\end{lemma}

\begin{proof}
	We prove the non-trivial direction.
	Suppose that \(\pi\) is injective, the map \(b \mapsto \rest{\pi(b)}{H^{\U}}\) restricts to an isometry on \(\Zent(B)\), and \(\pi(P) \cdot H^{\U}\) is dense in \(H\).
	Let \(N\) be a modular symmetry of \((H,\pi,D)\).
	Since \(\pi\) is injective, there exists a unique \(\U\)-equivariant algebra automorphism \(\Delta\) of \(P\), such that \(\pi \circ \Delta = N^{-1} \pi(\cdot) N\); in particular, \(\rest{\Delta}{B} = \id_B\) since \(N\) commutes with \(\pi(B)\).
	Hence, by Lemma \ref{lem:modular}, \emph{mutatis mutandis}, there exists a unique \(1\)-cocycle \(\mu : \bZ \to \Zent(B)^\times\), such that
	\(\Delta(p) = p \cdot \mu(j)\) for all \(j \in \bZ\) and \(p \in P_j\).
	Hence, for all \(j \in \bZ\), \(p \in P_j\) and \(h \in H^{\U}\),
	\(
		N \pi(p) h = N \pi(p) N^{-1} h = \pi(\Delta^{-1}(p)) h = \pi(p \mu(j)^{-1}) h = \pi(\mu(-j))\pi(p)h
	\)
	since \(\hat{\Phi}_P^j(\mu(j)^{-1}) = \mu(-j)\),
	so that \((N,\mu)\) satisfies \eqref{eq:modularsymm} since \(\pi(P) \cdot H^{\U}\) is dense in \(H\).
	Finally, let \((\epsilon_i)_{i=1}^N\) be a finite family in \(P_1\) satisfying \(\sum_{i=1}^N \epsilon_i^\ast \epsilon_i = 1\).
	Then
	\[
		0 \leq \sum\nolimits_{i=1}^N \pi(\epsilon_i)^\ast N \pi(\epsilon_i) = \sum\nolimits_{i=1}^N \pi(\epsilon_i)^\ast \pi(\epsilon_i \mu(1)^{-1}) N = \pi(\mu(1)^{-1}) N,
	\]
	so that \(\rest{\pi(\mu(1)}{H^{\U})} \geq 0\).
	Hence, since \(\left(b \mapsto \rest{\pi(b)}{H^{\U}}\right) : \Zent(B) \to \bL(H^{\U})\) is isometric, it follows that \(\mu(1) \geq 0\), so that \(\mu\) takes its values in \(\Zent(B)^{\times}_{\geq 0}\).
\end{proof}

\begin{lemma}\label{lem:tcentral}
	Suppose that \(\Zent(B) = \bC\).
	Let \(\eta \in \Omega^1_{P,\hor} \setminus \Set{0}\) and \(t \in (0,\infty) \setminus \Set{\kappa}\), and suppose that \(\eta\) and \(t\) satisfy \eqref{eq:tcentral}.
	Let \((H,\pi,D)\) be a locally bounded commutator representation of \((P;\Omega_P,\du_P)\), such that \(\pi\) is injective and \(\pi(P) \cdot H^{\U}\) is dense in \(H\).
	For every modular symmetry \(N\) of \((H,\pi,D)\), the operator \(N \pi_D(\omega) N\) is bounded only if \(N = \Lambda_{t^{-1/2}}\) or \(\pi_D(\omega) = 0\).
\end{lemma}

\begin{proof}
	Let \(N\) be a modular symmetry of \((H,\pi,D)\); suppose that \(N \pi_D(\omega) N\) is bounded.
	Since \(\Zent(B) = \bC\), it follows from Lemma \ref{lem:modular} that there exists unique \(s \in (0,\infty)\), such that \(N = \Lambda_s\).
	Now, let \((e_i)_{i=1}^m\) and \((\epsilon_j)_{j=1}^n\) be finite families in \(P_1\), such that \(\sum_{i=1}^m e_ie_i^\ast = 1\) and \(\sum_{j=1}^n \epsilon_j^\ast \epsilon_j = 1\); hence, let \(\phi_{\pm} : \bL^{\U}(H) \to \bL^{\U}(H)\) be the unit-preserving contractions from the proof of Proposition \ref{prop:nogobdd} induced by \((e_i)_{i=1}^m\) and \((\epsilon_j)_{j=1}^n\), respectively.
	Then
	\begin{align*}
		\phi_+\mleft(N \pi_D(\eta) N\mright) &= \sum\nolimits_{i=1}^m \pi(e_i) \Lambda_s \pi_D(\eta) \Lambda_s \pi(e_i^\ast)\\ &= \pi\mleft(\sum\nolimits_{i=1}^m e_i \cdot (\Lambda_s \circ \Lambda_t \circ \Lambda_s)(e_i^\ast)\mright) N \pi_D(\eta) N\\ &= s^2 t N \pi_D(\eta) N,
	\end{align*}
	while a similar calculation shows that \(\phi_-\mleft(N \pi_D(\eta) N \mright) = (s^2 t)^{-1} N \pi_D(\eta) N\).
	Hence,
	\[
		\norm{N \pi_D(\eta) N} = (s^2 t)^{\mp 1} \norm{\phi_{\pm}\mleft(N \pi_D(\eta) N\mright)} \leq (s^2 t)^{\mp 1} \norm{N \pi_D(\eta) N},
	\]
	so that \(N \pi_D(\eta) N = 0\) or \(s^2 t = 1\).
\end{proof}

\begin{proof}[Proof of Theorem \ref{thm:nogo}]
	Let \(N\) be a modular symmetry of \((H,\pi,D)\) that satisfies \(N \cdot \pi_D(\Omega^1_P) \cdot N \subseteq \bL^{\U}(H)\).
	Suppose that \(\pi_D(\eta) \neq 0\).
	By Lemma \ref{lem:tcentral} applied to \(\eta\), it follows that \(N = \Lambda_{t^{-1/2}} \neq \Lambda_{\kappa^{-1/2}}\); hence, \(\pi_D(\vartheta) = 0\) by Lemma \ref{lem:tcentral} applied to \(\vartheta\), so that \(\pi_D\) annihilates \((\id-\Pi)(\Omega^1_P) = P \cdot \vartheta\).
\end{proof}

\begin{example}\label{ex:hopf11}
	Continuing from Example \ref{ex:hopf9}, let \((H,\pi,D)\) be a locally bounded commutator representation of \((\cO_q(\SU(2)),\Omega_q(\SU(2)),\du_q)\), such that \(\pi\) is injective and \(\pi(\cO_q(\SU(2))) \cdot H^{\U}\) is dense in \(H\).
	If there exists a modular symmetry \(N\) of \((H,\pi,D)\) satisfying \(N \ran(\pi_D) N \subseteq \bL^{\U}(H)\), then \((\id-\Pi_q)(\Omega^1_P) \subseteq \ker \pi_D\) or \(\Omega^1_{q,\hor}(\SU(2)) \subseteq \ker \pi_D\).
	Suppose that \(N\) is such a modular symmetry and \((\id-\Pi_q)(\Omega^1_q(\SU(2))) \setminus \ker \pi_D \neq \emptyset\).
	Note that \((\eta,t) \coloneqq (e^\pm,q)\) satisfy \eqref{eq:tcentral}, where \(q \neq q^2\), so that \(\pi_D(e^\pm) = 0\) by Theorem \ref{thm:nogo}.
	Since \(\Set{e^+,e^-}\) generates \(\Omega^1_{q,\hor}(\SU(2))\) as a left \(\cO_q(\SU(2))\)-module, it follows that \(\Omega^1_{q,\hor}(\SU(2)) \subseteq \ker \pi_D\).
\end{example}

\begin{example}\label{ex:heis12}
	Continuing from Example \ref{ex:heis11}, let \((H,\pi,D)\) be a locally bounded commutator representation of \((P_\theta,\Omega_{P_\theta},\du_{P_\theta})\), such that the map \(\pi\) is injective and \(\pi(P_\theta) \cdot H^{\U}\) is dense in \(H\),
	If there exists a modular symmetry \(N\) of \((H,\pi,D)\), such that \(N \ran(\pi_D) N \subseteq \bL^{\U}(H)\), then \((\id-\Pi_{P_\theta})(\Omega^1_{P_\theta}) \subseteq \ker \pi_D\) or \(\Omega^1_{P_\theta,\hor} \subseteq \ker \pi_D\).
	Indeed, note that the left \(P\)-module \(\Omega^1_{P,\hor}\) is freely generated by \(e^1,e^2 \in \Zent(\Omega_\theta(\bT^2))^1\), where \eqref{eq:tcentral} is satisfied by \((\eta,t) = (e^i,\epsilon_\theta)\) for \(i=1,2\) by Example \ref{ex:heis5}.
	Since \(\epsilon_\theta \neq \epsilon_\theta^2\), we may apply Theorem \ref{thm:nogo} exactly as in Example \ref{ex:hopf10}.
\end{example}

Since a single modular symmetry cannot generally be used to control the unboundedness of represented \(1\)-forms, we must to allow for distinct modular symmetries in the vertical and horizontal directions.
Recall that \((\id-\Pi)(\Omega^1_P) = P \cdot \vartheta\) and \(\Pi(\Omega^1_P) = P \cdot \Omega^1_B\).

\begin{definition}
	Suppose that \((H,\pi,D,\Gamma)\) is a projectable commutator representation of \((P;\Omega_P,\du_P;\Pi)\).	
	\begin{enumerate}[leftmargin=*]
		\item A \emph{vertical twist} for \((H,\pi,D,\Gamma)\) is a pair \((N_\ver,\nu_\ver)\), where \(N_\ver\) is a modular symmetry of \((H,\pi,D)\) commuting with both \(\Gamma\) and \(\pi_D(\vartheta)\) and \(\nu_\ver\) is a modular automorphism of \(\Omega_P\), such that \[N_\ver^{-1}\pi(\cdot) N_\ver = \pi \circ \rest{\nu_\ver}{P}, \quad N_\ver \pi_D(\vartheta) N_\ver \in \bL^{\U}(H).\]
		\item A \emph{horizontal twist} for \((H,\pi,D,\Gamma)\) is a pair \((N_\hor,\nu_\hor)\), where \(N_\hor\) is a modular symmetry of \((H,\pi,D)\) commuting with both \(\Gamma\) and \(\pi_D(\vartheta)\) and \(\nu_\hor\) is a modular automorphism of \(\Omega_P\), such that \[N_\hor^{-1}\pi(\cdot)N_\hor = \pi \circ \rest{\nu_\hor}{P},\quad N_\hor \pi_D\mleft(\Omega^1_B\mright) N_\hor \subseteq \bL^{\U}(H).\]
	\end{enumerate}
\end{definition}

Thus, if \((H,\pi,D,\Gamma)\) is a projectable commutator representation of \((P;\Omega_P,\du_P;\Pi)\), then any vertical twist \((N_\ver,\nu_\ver)\) satisfies \(N_\ver \pi_D\mleft((\id-\Pi)(\Omega^1_P)\mright) N_\ver \subseteq \bL^{\U}(H)\), and any horizontal twist \((N_\hor,\nu_\hor)\) satisfies \(N_\hor \pi_D\mleft(\Pi(\Omega^1_P)\mright) N_\hor \subseteq \bL^{\U}(H)\).

We now study the existence of vertical and horizontal twists for faithful projectable commutator representations.
Lemmata \ref{lem:modular} and \ref{lem:modularsym} justify the following definition.

\begin{definition}
	Suppose that \((H,\pi,D,\Gamma)\) is faithful projectable commutator representation of \((P;\Omega_P,\du_P;\Pi)\).
	We define a \emph{modular pair} for \((H,\pi,D,\Gamma)\) to be a pair \((N,\nu)\), where \(N\) is a modular symmetry of \((H,\pi,D)\) and \(\nu\) is a modular automorphism of \(\Omega_P\) satisfying the equation \(N^{-1}  \pi(\cdot)  N = \pi \circ \rest{\nu}{P}\).
	In this case, the \emph{symbol} of \((N,\nu)\) is the unique right \(1\)-cocycle \(\lambda : \bZ \to \mathcal{Z}_{>0}(B)\), such that 
	\[
		\forall j \in \bZ, \quad \rest{N}{H_j} = \rest{\pi\mleft(\lambda(-j)\mright)}{H_j}, \quad \rest{\nu}{(\Omega_P)_j} = \left(\omega \mapsto \omega \lambda(j)\right).
	\]
\end{definition}

We first show that there is a canonical choice of vertical twist.

\begin{proposition}\label{prop:verticaltwist}
	Let \((H,\pi,D,\Gamma)\) be a faithful projectable commutator representation of \((P;\Omega_P,\du_P;\Pi)\).
	Then \((\Lambda_{\kappa^{-1/2}},\Lambda_{\kappa^{1/2}})\) defines a vertical twist of \((H,\pi,D,\Gamma)\), which is unique whenever \(\Zent(B) = \bC\).
\end{proposition}

\begin{proof}
	Let \((N,\nu)\) be a modular pair for \((H,\pi,D,\Gamma)\) with symbol \(\lambda\).
	Since \(\pi_D(\theta)\) satisfies \(\pi_D(\theta)^2 = \Lambda_\kappa^4\), it follows that \((N \pi_D(\theta) N)^2 = \Lambda_\kappa^2 N^4\), so that \((N,\nu)\) is a vertical twist for \((H,\pi,D,\Gamma)\) if and only if
	\(
		\sup_{j \in \bZ}\kappa^{-j/2}\norm{\rest{\pi(\lambda(-j))}{H_j}} < +\infty
	\),
	which is certainly satisfied by \(\lambda \coloneqq (j \mapsto \kappa^{-j/2})\).
	Moreover, if \(\Zent(B) = \bC\), then \(\lambda = (j \mapsto t^j)\) for unique real \(t > 0\), so that \((N,\nu)\) is a vertical twist for \((H,\pi,D,\Gamma)\) if and only if \(\sup_{j\in\bZ} (\kappa^{1/2}t)^{-j} < +\infty\), if and only if \(t = \kappa^{-1/2}\).
\end{proof}

To characterize existence of horizontal twists, we shall need the following broad generalisation of a definition from the literature on spectral triples for crossed products due to Bellissard--Marcolli--Reihani \cite{BMR}.

\begin{definition}\label{def:equicontinuity}
	Suppose that \((H,\pi,D)\) is a faithful bounded commutator representation of \((B;\Omega_B,\du_B)\).
	Let \(\Gamma\) be a group, and let \(\hat{F} : \Gamma \to \dPic(B)\) be a homomorphism, so that the right \(\dpic(B)\)-action on \(\mathcal{Z}_{>0}(B)\) of \eqref{eq:confact} pulls back via \(\hat{\Phi} \circ \pi_0(\hat{F})\) to a right \(\Gamma\)-action.
	For each \(\gamma \in \bZ\), let \((F(\gamma),\sigma_\gamma,\nabla_\gamma) \coloneqq \hat{F}(\gamma)\), and equip \(F(\gamma) \otimes_B H\) with the inner product defined by \eqref{eq:tensorip}; hence, for each \(\beta \in \Omega^1_B\), define \(\rho_\gamma[\beta] : F(\gamma) \otimes_B H  \to F(\gamma) \otimes_B H\) by
	\begin{equation}\label{eq:equicontinuity}
		\forall x \in F(\gamma), \, \forall h \in H, \quad \rho_\gamma[\beta](x \otimes h) \coloneqq \leg{\sigma_\gamma(\beta \otimes x)}{0} \otimes \pi_D\mleft(\leg{\sigma_\gamma(\beta \otimes x)}{1}\mright)h,
	\end{equation}
	and let \(\norm{\rho_\gamma[\beta]}\) denote the resulting operator norm of \(\rho_\gamma[\beta]\), which we set to equal \(+\infty\) whenever \(\rho_\gamma[\beta]\) is not bounded.
	Given a \(1\)-cocycle \(\lambda : \Gamma \to \mathcal{Z}_{>0}(B)\), we say that \(\hat{F}\) is \emph{\(\lambda\)-metrically equicontinuous} with respect to \((H,\pi,D)\) whenever
	\begin{equation}
		\forall \beta \in \Omega^1_B, \quad \sup_{\gamma \in \Gamma} \norm*{\rho_\gamma\!\left[\lambda(\gamma^{-1})^2\beta\right]} < +\infty.
	\end{equation}
\end{definition}

\begin{example}\label{ex:heis10}
	Recall the homomorphism \(\hat{E} : \Gamma_\theta \to \dPic(C^\infty_\theta(\bT^2))\) of Example \ref{ex:heis4}; define \(\lambda : \Gamma_\theta \to \bR_{>0}\) by \(\lambda \coloneqq \left(g \mapsto (g_{21}\theta+g_{22})^{-1/2}\right)\).
	Then \(\hat{E}\) is \(\lambda\)-metrically equicontinuous with respect to every faithful bounded commutator representation of \((C^\infty_\theta(\bT^2),\Omega_\theta(\bT^2),\du)\).
	Indeed, let \((H,\pi,D)\) be such a bounded commutator representation.
	Recall that the left \(C^\infty_\theta(\bT^2)\)-module \(\Omega^1_\theta(\bT^2)\) is generated by \(\Set{e^1,e^2} \subset \Zent(\Omega_\theta(\bT^2))^1\).
	Given \(i=1,2\) and \(g \in \Gamma_\theta\), it follows by construction of \(\hat{E}\) that \(\rho_g[e^i] = \tfrac{1}{g_{21}\theta+g_{22}} \id \otimes \pi_D(e^i)\), so that
	\[
		\norm{\rho_g[\lambda(g^{-1})^{2}e^i]} = \norm{\id \otimes \pi_D(e^i)} \leq \norm{\id}\norm{\pi_D(e^i)} \leq \norm{\pi_D(e^i)}.
	\]
\end{example}

\begin{example}\label{ex:hopf12a}
	Recall from Example \ref{ex:hopf3} the homomorphisms of coherent \(2\)-groups \(\cE : \bZ \to \Pic(\cO_q(\CP^1))\) and \(\hat{\cE} : \bZ \to \dPic(\cO_q(\CP^1))\), and define the homomorphism \(\lambda  : \bZ \to \bR_{>0}\) by \(\lambda \coloneqq (k \mapsto q^{-j})\).
	Then \(\hat{\cE}\) is \(\lambda\)-metrically equicontinuous with respect to every faithful bounded commutator representation of \((\cO_q(\CP^1);\Omega_q(\CP^1),\du_q)\).
	Indeed, let \((H,\pi,D)\) be such a bounded commutator representation.
	Recall that \(\Omega_q(\CP^1) = \cE_{-2} \cdot e^+ \oplus \cE_2 \cdot e^-\).
	Choose a cobasis \((\eta^\mp_i)_{i=1}^{N_\mp}\) for \(\cE_{\mp 2}\), and define \(\tau_\pm : H \to \cE_{\pm 2} \otimes_{\cO_q(\CP^1)} H\) by \(\tau_\pm(h) \coloneqq \sum_{i=1}^{N_\mp} (\eta^\mp_i)^\ast \otimes \pi_D(\eta^\mp_i e^\pm)h\) for \(h \in H\);
	note that \(\tau_\pm\) is bounded and left \(B\)-linear since, for all \(h,k \in H\) and \(x \in \cE_{\pm 2}\),
	\[
		\ip{x \otimes k}{\tau_\pm(h)} = \ip*{k}{\pi_D\mleft(x^\ast \left(\sum\nolimits_{i=1}^{N_\mp} (\eta^\mp_i)^\ast \eta^\mp_i\right) e^\pm\mright) h} = \ip{k}{\pi_D(x^\ast e^\pm) h}.
	\]
	For \(i,j \in \bZ\), define \(V_{i,j} : \cE_i \otimes_{\cO_q(\CP^1)} (\cE_j \otimes_{\cO_q(\CP^1)} H) \to (\cE_i \otimes_{\cO_q(\CP^1)} \cE_j) \otimes_{\cO_q(\CP^1)} H\) by \(V_{i,j} \coloneqq( x \otimes (y \otimes h) \mapsto (x \otimes y) \otimes h)\) and, for each \(p \in \cE_i\), the bounded adjointable map \(\pi_{i,j}(p) : \cE_j \otimes_{\cO_q(\CP^1)} H \to \cE_{i+j} \otimes_{\cO_q(\CP^1)} H\) by \(\pi_{i,j}(p) \coloneqq (y \otimes h \mapsto p \cdot y \otimes h)\), which satisfies \(\norm{\lambda_{i,j}(p)} \leq \norm{p}\).
	At last, let \(j \in \bZ\) and \(p \in \cE_{\mp 2}\) be given.
	Since each \(x \in \cE_j\) satisfies \(p e^\pm x = q^{-j} \sum_{i=1}^{N_\mp} px(\eta^\mp)_i^\ast \eta^\mp_i e^\pm\), it follows that
	\[
		\rho_j[\lambda(-j)^2 p e^\pm] = \pi_{\mp 2,j\pm 2}(p) \circ (\cE^{(2)}_{j,\pm 2} \otimes \id) \circ V_{j,\pm 2} \circ (\id{} \otimes \tau_\pm),
	\]
	and hence
	\[
		\norm{\rho_j[\lambda(-j)^2 p e^\pm]} \leq \norm{\pi_{\mp 2,j\pm 2}(p)} \norm{\cE^{(2)}_{j,\pm 2} \otimes \id} \norm{V_{j,\pm 2}} \norm{\id{} \otimes \tau_\pm} \leq \norm{p} \norm{\tau_\pm}.
	\]
\end{example}

In light of Corollary \ref{cor:abstractpimsner}, one may ask whether our generalised notion of metric equicontinuity makes sense at the level of Hermitian line \(B\)-bimodules with connection.
The following proposition answer this question in the affirmative.

\begin{proposition}
	Suppose that \((H,\pi,D)\) is a faithful bounded commutator representation of \((B;\Omega_B,\du_B)\).	
	Let \(\Gamma\) be a group, let \(\hat{F}_1,\hat{F}_2 : \Gamma \to \dPic(B)\) be homomorphisms, and suppose that \(\hat{F}_1 \cong \hat{F}_2\) in \(\grp{Hom}(\Gamma,\dPic(B))\).
	Hence, let \(\lambda : \Gamma \to \mathcal{Z}_{>0}(B)\) be a right \(1\)-cocycle for the pullback of the \(\dpic(B)\)-action of \eqref{eq:confact} by \(\pi_0(\hat{F}_1) = \pi_0(\hat{F}_2)\).
	Then \(\hat{F}_1\) is \(\lambda\)-metrically equicontinuous if and only if \(\hat{F}_2\) is.
\end{proposition}

\begin{proof}
	Choose a natural isomorphism \(\eta : \hat{F}_1 \rightarrow \hat{F}_2\).
	Let \(\gamma \in \Gamma\).
	For \(i=1,2\), let \((F_i(\gamma),\sigma_{i;\gamma},\nabla_{i;\gamma}) \coloneqq \hat{F}_i(\gamma))\), and define \(\rho_{i;\gamma}[\beta] : F_i(\gamma) \otimes_B H \to F_i(\gamma) \otimes_B H\) for each \(\beta \in \Omega^1_B\) by \eqref{eq:equicontinuity}.
	Since \(\eta_\gamma : (F_1(\gamma),\sigma_{1;\gamma},\nabla_{1;\gamma}) \to (F_2(\gamma),\sigma_{2;\gamma},\nabla_{2;\gamma})\) is an isomorphism is \(\dPic(B)\), the map \(\eta_\gamma \otimes \id : F_1(\gamma) \otimes_B H \to F_2(\gamma) \otimes_B H\) is a unitary that satisfies \((\eta_\gamma \otimes \id) \circ \rho_{1;\gamma}[\beta] = \rho_{2;\gamma}[\beta] \circ (\eta_\gamma \otimes \id)\) for all \(\beta \in \Omega^1_B\).
\end{proof}

Hence, given a Hermitian line \(B\)-bimodule with connection \((E,\sigma_E,\nabla_E)\) and a group \(1\)-cocycle \(\lambda : \bZ \to \mathcal{Z}_{>0}(B)\) for the right \(\bZ\)-action generated by \(\hat{\Phi}_E^{-1}\), we define \((E,\sigma_E,\nabla_E)\) to be \emph{\(\lambda\)-metrically equicontinuous} whenever some (and hence every) homomorphism \(\hat{F} : \bZ \to \dPic(B)\) satisfying \(\hat{F}(1) \cong (E,\sigma_E,\nabla_E)\) is \(\lambda\)-metrically equicontinuous.
The following characterisation of metric equicontinuity in our general sense for crossed products by extended diffeomorphisms now shows that metric equicontinuity (in our sense) with respect to a trivial \(1\)-cocycle corresponds to the existing definition in the literature on crossed product spectral triples.

\begin{proposition}[cf.\ Bellissard--Marcolli--Reihani \cite{BMR}]\label{prop:equicont}
	Let \((H,\pi,D)\) be a bounded commutator representation of \((B;\Omega_B,\du_B)\).
	Let \((\omega,\phi) \in \tDiff(B)\), and suppose that \(\lambda : \bZ \to \mathcal{Z}_{>0}(B)\) is a right \(1\)-cocycle for the right \(\bZ\)-action generated by \(\phi^{-1}\).
	Then \(\hat{\tau}(\omega,\phi)\) is \(\lambda\)-metrically equicontinuous with respect to \((H,\pi,D)\) if and only if
	\begin{equation}
		\forall b \in B, \quad \sup_{k \in \bZ} \norm*{\pi(\lambda(k)^{-1}) \cdot [D,\pi(\phi^{-k}(b))] \cdot \pi(\lambda(k)^{-1})} < +\infty.
	\end{equation}
\end{proposition}

\begin{proof}
	By Proposition \ref{prop:equicont}, it suffices to check that \(\hat{\tau} \circ (k \mapsto (\omega,\phi)^k)\) is \(\lambda\)-metrically equicontinuous.
	Let \(k \in \bZ\) be given.
	Define a unitary \(V_k : B_{\phi}^k \otimes_B H \to H\) by \(V \coloneqq (b_{\phi}^k \otimes h \mapsto \pi(\phi^{-k}(b))h)\).
	By construction of \(\hat{\tau}((\omega,\phi)^k) \coloneqq (B_\phi,\sigma_{\phi^k},\nabla_{(\omega,\phi)^k})\), it follows that \(V_k \rho_k[\beta] V_k^\ast = \pi_D\mleft(\phi^{-k}(\beta)\mright)\) for all \(\beta \in \Omega^1_B\), so that
	\begin{align*}
		V_k \rho_k\mleft[\hat{\Phi}_{[\hat{\tau}(\omega,\phi)]}(\lambda(-k)^2) \cdot \du_B(b)\mright] V_k^\ast 
		&= \pi_D\mleft(\lambda(k)^{-2} \du_B \phi^{-k}(b)\mright)\\
		&= \iu{} \pi(\lambda(k)^{-1}) [D,\pi(\phi^{-k}(b))] \pi(\lambda(k)^{-1}).
	\end{align*}
	for every \(b \in B\).
	Since \(\du_B(B)\) generates \(\Omega^1_B\) as a left \(B\)-module, this suffices.
\end{proof}

At last, we characterise horizontal twists among all modular pairs.

\begin{proposition}\label{prop:equicontinuous}
	Let \((H,\pi,D)\) be a faithful bounded commutator representation of \((B;\Omega_B,\du_B)\), and let \((\tilde{H},\tilde{\pi},\tilde{D},\tilde{\Gamma})\) be a lift of \((H,\pi,D)\) to \((P;\Omega_P,\du_P,\Pi)\).
	Let \((N,\nu)\) be a modular pair for \((\tilde{H},\tilde{\pi},\tilde{D},\tilde{\Gamma})\) with symbol \(\lambda\).
	Then \((N,\nu)\) defines a horizontal twist for \((\tilde{H},\tilde{\pi},\tilde{D},\tilde{\Gamma})\) if and only if \(\hat{\cL} \circ \Hor_\kappa(P;\Omega_P,\du_P;\Pi)\) is \(\lambda\)-metrically equicontinuous with respect to \((H,\pi,D)\).
\end{proposition}

\begin{proof}
	By Theorem \ref{thm:liftequivalence}, we may assume that \((\tilde{H},\tilde{\pi},\tilde{D},\tilde{\Gamma}) = (\iota_P)_!(H,\pi,D)\) without any loss of generality.
	Hence, we reprise the notation of the proof of Proposition \ref{prop:liftcommrepconstruct}.
	In particular, \(N \tilde{\pi}(\cdot) N^{-1} = \tilde{\pi} \circ \nu\) since
	\(
		\tau \circ N \circ \tau^\ast = \id \otimes \nu \otimes \id
	\).

	Let \(\beta \in \Omega^1_B\) be given.
	Fix \(j \in \bZ\).
	Then 
	\(
		\rest{\tilde{\pi}_{\tilde{D}}(\beta)}{\tilde{H}_j} = \tau^\ast \circ (\sigma^3 \otimes \rho_j[\beta]) \circ \rest{\tau}{\tilde{H}_j}
	\)
	by the proof of Proposition \ref{prop:liftcommrepconstruct}.
	Hence, for every \(x \in \bC^2\), \(p \in P_j\), and \(h \in H\), 
	\begin{align*}
		&\tau N \tilde{\pi}_{\tilde{D}}(\beta) N \tau^\ast(x \otimes p \otimes h) \\
		&\quad\quad= (\id \otimes \nu \otimes \id)\mleft(\sigma^3 x \otimes \leg{\sigma_j(\beta \otimes p)}{0} \otimes \pi_D\mleft(\leg{\sigma_j(\beta \otimes p)}{1} \lambda(j)\mright)h\mright)\\
		&\quad\quad= \sigma^3 \otimes \leg{\sigma_j(\beta \otimes p)}{0} \lambda(j) \otimes \pi_D\mleft(\leg{\sigma_j(\beta \otimes p)}{1} \lambda(j)\mright)h\\
		&\quad\quad= \sigma^3 \otimes \rho_j\mleft[\lambda(-j)^2\beta\mright](x \otimes p \otimes h),
	\end{align*}
	which implies that
	\(
		\norm{\rest{N \tilde{\pi}_{\tilde{D}}(\beta) N}{\tilde{H}_j}} = \norm*{\sigma^3 \otimes \rho_j\mleft[\lambda(-j)^2\beta\mright]} = \norm*{\rho_j\mleft[\lambda(-j)^2\beta\mright]}
	\)
	by unitarity of \(\sigma^3 \in M_2(\bC)\).
	Thus, \(N \tilde{\pi}_{\tilde{D}}(\beta) N \in \bL_{\loc}^{\U}(H)\) is bounded if and only if \(\Set*{\norm*{\rho_j\mleft[\lambda(-j)^2\beta\mright]} \given j \in \bZ}\) is bounded from above.
\end{proof}

\begin{example}\label{ex:heis13}
	Continuing from Examples \ref{ex:heis6} and \ref{ex:heis10}, let \((H,\pi,D)\) be a faithful bounded commutator representation of \((C^\infty_\theta(\bT^2),\Omega_\theta(\bT^2),\du)\), and let \((\tilde{H},\tilde{\pi},\tilde{D},\tilde{\Gamma})\) be a lift of \((H,\pi,D)\) to the real multiplication instanton \((P_\theta,\Omega_{P_\theta},\du_{P_\theta},\Pi_{P_\theta})\).
	On the one hand, by Proposition \ref{prop:verticaltwist}, \((\Lambda_{\epsilon_\theta^{-1}},\Lambda_{\epsilon_\theta})\) is the unique vertical twist of \((\tilde{H},\tilde{\pi},\tilde{D},\tilde{\Gamma})\).
	On the other hand, by Example \ref{ex:heis10}, the homomorphism \(\hat{\cL} \circ \Hor_{\epsilon_\theta^2}(P_\theta,\Omega_{P_\theta},\du_{P_\theta},\Pi_{P_\theta})\) is \((m \mapsto \epsilon_\theta^{-m/2})\)-equicontinuous with respect to \((H,\pi,D)\), so that \((\Lambda_{\epsilon_\theta^{-1/2}},\Lambda_{\epsilon_\theta^{1/2}})\) is the unique horizontal twist of \((\tilde{H},\tilde{\pi},\tilde{D},\tilde{\Gamma})\) by Proposition \ref{prop:equicontinuous} together with Lemma \ref{lem:tcentral} applied to \((\eta,t) = (e^1,\epsilon_\theta)\).
	Note that \((\Lambda_{\epsilon_\theta^{-1}},\Lambda_{\epsilon_\theta})\) and \((\Lambda_{\epsilon_\theta^{-1/2}},\Lambda_{\epsilon_\theta^{1/2}})\) are non-trivial and distinct since \(\epsilon_\theta \neq 1\).
\end{example}

\begin{example}\label{ex:hopf12}
	Continuing from Example \ref{ex:hopf10}, let \((H,\pi,D)\) be a faithful bounded commutator representation of \((\cO_q(\CP^1),\Omega_q(\CP^1),\du_q)\), and let \((\tilde{H},\tilde{\pi},\tilde{D},\tilde{\Gamma})\) be a lift of \((H,\pi,D)\) to the \(q\)-monopole \((\cO_q(\SU(2));\Omega_q(\SU(2)),\du_q;\Pi_q)\).
	On the one hand, by \ref{prop:verticaltwist}, the modular pair \((\Lambda_{q^{-1}},\Lambda_q)\) is the unique vertical twist of \((\tilde{H},\tilde{\pi},\tilde{D},\tilde{\Gamma})\).
	On the other hand, by Example \ref{ex:hopf12a}, the homomorphism \(\hat{\cE} : \bZ \to \dPic(\cO_q(\CP^1))\) of Example \ref{ex:hopf3} is \((m \mapsto q^{-m/2})\)-equicontinuous with respect to \((H,\pi,D)\), so that \((\Lambda_{q^{-1/2}},\Lambda_{q^{1/2}})\) is the unique horizontal twist of \((\tilde{H},\tilde{\pi},\tilde{D},\tilde{\Gamma})\) by Proposition \ref{prop:equicontinuous} together with Lemma \ref{lem:tcentral} applied to \((\eta,t) = (e^\pm,q)\).	
	Note that \((\Lambda_{q^{-1}},\Lambda_q)\) and \((\Lambda_{q^{-1/2}},\Lambda_{q^{1/2}})\) are non-trivial and distinct since \(q \neq 1\).
\end{example}

We now show that total Hodge--de Rham commutator representations admit canonical horizontal twists under a mild hypothesis on \(B\). 

\begin{theorem}\label{thm:derhamtwist}
	Suppose that \((\Delta_\ver,\Delta_\hor,\star,\tau)\) is a total Riemannian geometry on \((P;\Omega_P,\du_P;\Pi)\).
	Let \(\mu_P\) be the symbol of \(\Delta_\hor\), and suppose that \(\mu_P(1)\) has a square root in \(\mathcal{Z}_{>0}(B)\); hence, let \(\mu_P^{1/2} : \bZ \to \mathcal{Z}_{>0}(B)\) be the unique right \(1\)-cocycle for the right \(\bZ\)-action of \eqref{eq:horizontalact} that satisfies \(\mu_P^{1/2}(\cdot)^2 = \mu_P\), and let \((N_\hor,\nu_\hor)\) be the modular pair with symbol \(\mu_P^{1/2}\) for total Hodge--de Rham commutator representation \((\Omega_P,\pi_P,\du_P+\du_P^\ast,2\Pi-\id)\) induced by \((\Delta_\ver,\Delta_\hor,\star,\tau)\).
	Then \((N_\hor,\nu_\hor)\) defines a horizontal twist for \((\Omega_P,\pi_P,\du_P+\du_P^\ast,2\Pi-\id)\).
\end{theorem}

\begin{proof}
	We use the notation of the proof of Proposition \ref{prop:totalderham}.
	Hence, it suffices to show that \(N_\hor\) satisfies \(N_\hor \cdot \sce(\Omega^1_B) \cdot N_\hor \subseteq \bL^{\U}(\Omega_P)\).
	
	Let \(\beta \in \Omega^1_B\).
	Note that \(\Omega_P = \bigoplus_{j=-\infty}^\infty \bigoplus_{r=0}^1 \bigoplus_{s=0}^N (\Omega^{r,s}_P)_j\) is an orthogonal direct sum of pre-Hilbert spaces.
	Fix \((r,s,j) \in \Set{0,1} \times \Set{0,\dotsc,N} \times \bZ\), so that \(\sce(\beta)\) maps \((\Omega^{r,s}_P)_j\) to \((\Omega^{r,s+1}_P)_j\), and let
	\(
		T^{r,s}_j \coloneqq \rest{N_\hor \sce(\beta) N_\hor}{(\Omega^{r,s}_P)_j} = \rest{\sce(\hat{\Phi}_P^j(\mu_P(j)^{-1})\beta)}{(\Omega^{r,s}_P)_j}
	\).
	It therefore suffices to bound the operator norm of \(T^{r,s}_j\) uniformly in \(j \in \bZ\).
	
	Let \(E \coloneqq (\Omega^{r,s}_P)_j\), \(F \coloneqq (\Omega^{r,s+1}_P)_j\), \(V \coloneqq (\Omega^{r,s}_P)^{\U}\), and \(W \coloneqq (\Omega^{r,s+1}_P)^{\U}\), which we view as orthogonal direct summands of the pre-Hilbert space \(\Omega_P\); note that each of these pre-Hilbert spaces also defines a \(B\)-self-correspondence of finite type by the proof of Proposition \ref{prop:totalprehilbert}, where each pre-Hilbert space norm is bounded from above by the corresponding norm as a \(B\)-self-correspondence of finite type.
	Write \(\hat{L} \circ \Hor_\kappa(P;\Omega_P,\du_P;\Pi) \eqqcolon (P_j,\sigma_{P;j},\nabla_{P;j})\), where we conflate the Hermitian line \(B\)-bimodule \(\cL(P)(j)\) with \(P_j\); recall that \(\Omega^1_B\) defines a \(B\)-self-correspondence of finite type by Proposition \ref{prop:basicseparable}, so that \(\Omega^1_B \otimes_B P_j\) and \(P_j \otimes_B \Omega^1_B\) both define \(B\)-self-correspondences of finite type.
	Hence, we can view each of \(\Omega^1_B \otimes_B E\), \((\Omega^1_B \otimes_B P_j) \otimes_B V\), \((P_j \otimes_B \Omega^1_B) \otimes_B V\), \(P_j \otimes_B (\Omega^1_B \otimes V)\) and \(P_j \otimes_B W\)
	as pre-Hilbert spaces with respect to the inner product defined, \emph{mutatis mutandis}, by \eqref{eq:tensorip}.
	Finally, define \(\tau_j : \Omega^1_B \otimes_B P_j \to P_j \otimes_B \Omega^1_B\) by \(\tau_j(\eta) \mapsto \sigma_{P;j}\mleft(\mu_P(-j))\eta\mright) = \sigma_{P;j}(\eta)\mu_P(j)^{-1}\).
	It now follows that \(T^{r,s}_j : E \to F\) factorizes as the composition
	\begin{multline*}
		E \xrightarrow{\beta \otimes {-}} \Omega^1_B \otimes_B E \xrightarrow{\cong} (\Omega^1_B \otimes_B P_j) \otimes_B V \xrightarrow{\tau_j \otimes \id} (P_j \otimes_B \Omega^1_B) \otimes_B V\\ \xrightarrow{\cong} P_j \otimes_B (\Omega^1_B \otimes V) \xrightarrow{\id \otimes m_{r,s}} P_j \otimes_B W \xrightarrow{\cong} F,
	\end{multline*}
	where the first two arrows denoted by \(\cong\) are the usual (inverse) associators, which are unitary \cite[\S 8.2.12]{BL}, where \(P_j \otimes_B W \xrightarrow{\cong} F\) is given by multiplication in \(\Omega_P\) and hence is unitary by the proof of Proposition \ref{prop:totalprehilbert}, and where \(m_{r,s} : \Omega^1_B \otimes V \to W\) is given by multiplication in \(\Omega_P\).
	
	Let us now look at the non-trivial arrows in this composition.
	First, an explicit calculation shows that \(\beta \otimes {-} \coloneqq (\xi \mapsto \beta \otimes \xi)\) is bounded with operator norm \(\norm{\beta \otimes {-}} \leq \norm{\beta}\), where \(\norm{\beta} = \norm{g_B(\beta,\beta)}^{1/2}\) is the norm of \(\beta\) as an element of the \(B\)-self-correspondence of finite type \(\Omega_B\).
	Next, since \(\tau_j\) is right \(B\)-linear map between pre-Hilbert \(B\)-modules of finite type, it is necessarily bounded and adjointable, so that \(\tau_j \otimes \id\) is bounded as a map between pre-Hilbert spaces with operator norm \(\norm{\tau_j \otimes \id} \leq \norm{\tau_j} \norm{\id} = \norm{\tau_j}\) by standard results \cite[\S 8.2.12]{BL}.
	Finally, since \(m_{r,s} : \Omega^1_B \otimes_B V \to W\) is a right \(B\)-linear map of pre-Hilbert \(B\)-modules of finite type, it is bounded and adjointable, and hence bounded as a map of pre-Hilbert spaces with operator norm \(\norm{m_{r,s}}\), so that \(\id \otimes m_{r,s}\) is also bounded as a map of pre-Hilbert spaces with operator norm \(\norm{\id \otimes m_{r,s}} \leq \norm{\id}\norm{m_{r,s}} = \norm{m_{r,s}}\).
	Thus, the operator norm of \(T^{r,s}_j\) is bounded from above by \(\norm{\beta}\norm{\tau_j}\norm{m_{r,s}}\), so that, at last, it suffices to show that \(\norm{\tau_j} = 1\).
	
	Finally, let \(\eta \in \Omega^1_B \otimes_B P_j\) be given, so that \(\eta = \sum_{i=1}^n \alpha_i \otimes p_i\) for \(\alpha_1,\dotsc,\alpha_n \in \Omega^1_B\) and \(p_1,\dotsc,p_n \in P_j\); hence, let \(\tilde{\eta} \coloneqq \sum_{i=1}^n \alpha_i \cdot p_i \in (\Omega^{0,1}_P)_j\), so that
	\begin{multline*}
		{\star}(\tilde{\eta}) = \sum\nolimits_i \star(\alpha_i p_i) = -\sum\nolimits_i \vartheta\star_B(\alpha_i) p_i \mu_P(j)^{2-N} \kappa^j\\ = (-1)^N \sum\nolimits_i \star_B(\alpha_i) p_i \mu_P(j)^{-N} \vartheta \mu_P(j)^2
	\end{multline*}
	by \eqref{eq:totalhodgeright}. 
	On the one hand,
	\begin{multline*}
		{\star_P}\mleft(\hp{\eta}{\eta}\mright) = \sum\nolimits_{i,j} p_i^\ast g_B(\alpha_i,\alpha_j) p_j \vartheta\star_B(1) = (-1)^N\sum\nolimits_{i,j} p_i^\ast \alpha_i \star_B(1)p_j\mu_P(j)^{-N}\vartheta\\ = \tilde{\eta}^\ast \star_P(\tilde{\eta}) \mu_P(j)^{-2},
	\end{multline*}
	while on the other,
	\begin{multline*}
		{\star_P}\mleft(\hp{\sigma_{P;j}(\eta)}{\sigma_{P;j}(\eta)}\mright) = \sum\nolimits_{i,j} g_B(\alpha_i,p_i^\ast p_j \alpha_j) \vartheta \star_B(1)\\ = (-1)^N \sum\nolimits_{i,j} \alpha_i^\ast p_i^\ast p_j \star_B(\alpha_j) \vartheta = \tilde{\eta}^\ast \star_P(\tilde{\eta}),
	\end{multline*}
	so that 
	\begin{multline*}
		\hp{\tau_j(\eta)}{\tau_j(\eta)} = (\mu_P(j)^{-1})^\ast \hp{\sigma_{P;j}(\eta)}{\sigma_{P;j}(\eta)} \mu_P(j)^{-1}\\ = \hp{\sigma_{P;j}(\eta)}{\sigma_{P;j}(\eta)} \mu_P(j)^{-2} = \hp{\eta}{\eta}. \qedhere
	\end{multline*}
\end{proof}

We conclude with a first step towards relating our constructions to Rieffel's compact quantum metric spaces~\cite{Rieffelmemoir}.
We show that a faithful projectable commutator representation of \((P;\Omega_P,\du_P;\Pi)\) equipped with vertical and horizontal twists yields a Lipschitz seminorm~\cite[Def.\ 2.1]{KK} on the \(C^\ast\)-algebra completion of \(P\) that satisfies a twisted Leibniz inequality~\cite[Lemma 4.8]{KK}.
This, in turn, will recover, up to equivalence of seminorm, Kaad--Kyed's compact quantum metric space on quantum \(\SU(2)\) for a canonical choice of parameters~\cite[\S 4]{KK}.

\begin{proposition}[{cf.\ Kaad--Kyed~\cite[Lemma 48]{KK}}]\label{prop:lipschitz}
	Let \((H,\pi,D,\Gamma)\) be a faithful projectable commutator representation of \((P;\Omega_P,\du_P;\Pi)\) with vertical twist \((N_\ver,\nu_\ver)\) and horizontal twist \((N_\hor,\nu_\hor)\).
	Define a \(\U\)-invariant norms \(\norm{}_{\tau}\) and \(\norm{}_{\tau,\tot}\) on \(P\) and \(\Omega^1_P\), respectively, by
	\begin{align*}
		\forall p \in P, && \norm{p}_{\tau} &\coloneqq \max\Set*{\norm{\nu_\ver(p)} + \norm{\nu_\hor(p)},\norm{\nu_\ver^{-1}(p)} + \norm{\nu_\hor^{-1}(p)}},\\
		\forall \omega \in \Omega^1_P, && \norm{\omega}_{\tau;\tot} &\coloneqq \norm*{N_\ver (\pi_D \circ (\id-\Pi))(\omega)  N_\ver + N_\hor  (\pi_D \circ \Pi)(\omega)  N_\hor}.
	\end{align*}
	Then \(\norm{}_{\tau}\) makes \(P\) into a normed \(\ast\)-algebra, while \(\norm{}_{\tau;\tot}\) is invariant under the \(\ast\)-operation and satisfies
	\begin{equation}\label{eq:twistedmodule}
		\forall p_1,p_2 \in P, \, \forall \omega \in \Omega^1_P, \quad \norm{p_1  \omega  p_2}_{\tau;\tot} \leq \norm{p_1}_{\tau} \norm{\omega}_{\tau;\tot} \norm{p_2}_{\tau}.
	\end{equation}
	Hence, the \(\U\)-invariant seminorm \(L_\tau \coloneqq \norm{\du_P(\cdot)}_{\tau;\tot}\) on \(P\) annihilates \(\bC \subseteq P\), is invariant under the \(\ast\)-operation, and satisfies
	\begin{equation}\label{eq:twistedleibniz}
		\forall p_1,p_2 \in P, \quad L_\tau(p_1p_2) \leq L_\tau(p_1) \norm{p_2}_{\tau} + \norm{p_1}_{\tau} L_\tau(p_2).
	\end{equation}
\end{proposition}

\begin{lemma}[{cf.\ Kaad--Kyed \cite[Rem.\ 4.6]{KK}}]\label{lem:lipschitz}
	Under the hypotheses of Proposition \ref{prop:lipschitz}, the \(\U\)-invariant seminorms \(\norm{}_{\tau,\ver}\) and \(\norm{}_{\tau,\hor}\) on \(\Omega^1_P\) defined by
	\[
		\norm{}_{\tau,\ver} \coloneqq \norm*{N_\ver (\pi_D \circ (\id-\Pi))(\cdot)  N_\ver}, \quad \norm{}_{\tau,\hor} \coloneqq \norm*{N_\hor  (\pi_D \circ \Pi)(\cdot)  N_\hor}
	\]
	satisfy the inequality \( \max\Set{\norm{\omega}_{\tau,\ver},\norm{\omega}_{\tau,\hor}} \leq \norm{\omega}_{\tau,\tot}\) for all \(\omega \in \Omega^1_P\).
\end{lemma}

\begin{proof}
	Let \(\omega \in \Omega^1_P\) be given; let \(\omega_\hor \coloneqq \Pi(\omega)\), and write \((\id-\Pi)(\omega) = p\vartheta\) for unique \(p \in P\).
	Let \(c \coloneqq \pi_D(\vartheta)\Lambda_\kappa^{-1}\), which is a \(\U\)-invariant self-adjoint unitary by definition of a projectable commutator representation.
	On the one hand, \(c\) manifestly commutes with the operator \(\pi_D(p\vartheta) = \pi(p)c\Lambda_\kappa\).
	On the other hand, since \(\Omega^1_{P,\hor} = P \cdot \du_B(B)\) and \([D,\pi(b)] = [\pi_D(\vartheta)\partial_\kappa+D_\hor,\pi(b)] = [D_\hor,\pi(b)]\) for all \(b \in B\), it follows that \(c\) anticommutes with \(\pi_D(\omega_\hor)\) as well.
	Setting \(E_\pm \coloneqq \tfrac{1}{2}(\id \pm c)\), we may now decompose \(N_\ver \pi_D(p\vartheta) N_\ver + N_\hor \pi_D(\omega_\hor) N_\hor\) with respect to the orthogonal direct sum decomposition \(H = E_+(H) \oplus E_-(H)\) as
	\[
		\begin{pmatrix} E_+ N_\ver \pi_D(p\vartheta) N_\ver E_+ & E_+ N_\hor \pi_D(\omega_\hor) N_\hor E_-\\ E_- N_\hor \pi_D(\omega_\hor) N_\hor E_+ & E_- N_\ver \pi_D(p\vartheta) N_\ver E_-\end{pmatrix},
	\]
	so that
	\begin{gather*}
		\norm{\omega}_{\tau;\ver} = \norm*{\begin{pmatrix} E_+ N_\ver \pi_D(p\vartheta) N_\ver E_+ & 0\\ 0 & E_- N_\ver \pi_D(p\vartheta) N_\ver E_-\end{pmatrix}} \leq \norm{\omega}_{\tau;\tot}, \\ 
		\norm{\omega}_{\tau,\hor} = \norm*{\begin{pmatrix} 0 & E_+ N_\hor \pi_D(\omega_\hor) N_\hor E_-\\ E_- N_\hor \pi_D(\omega_\hor) N_\hor E_+ & 0 \end{pmatrix}} \leq \norm{\omega}_{\tau;\tot}. \qedhere
	\end{gather*}
\end{proof}

\begin{proof}[Proof of Proposition \ref{prop:lipschitz}]
	In what follows, we use the notation of Lemma \ref{lem:lipschitz}.
	The only non-trivial points are positive-definiteness of \(\norm{}_{\tau,\tot}\), \eqref{eq:twistedmodule}, and \eqref{eq:twistedleibniz}; note that \(\norm{}_{\tau,\tot}\) is positive-definite by the proof of Lemma \ref{lem:lipschitz}, while \eqref{eq:twistedmodule} implies \eqref{eq:twistedleibniz} by the usual Leibniz rule for \(\du_P\).
	Let \(p_1,p_2 \in P\) and \(\omega \in \Omega^1_P\) be given; set \(\omega_\ver \coloneqq (\id-\Pi)(\omega)\) and \(\omega_\hor \coloneqq \Pi(\omega)\).
	Then, by Lemma \ref{lem:lipschitz},
	\begin{multline*}
		\norm{p_1\omega p_2}_{\tau;\ver} = \norm*{N_\ver \pi_D(p_1 \omega_\ver p_2) N_\ver} \leq \norm{\nu_\ver^{-1}(p)} \norm{N_\ver \pi_D(\omega_\ver)} \norm{\nu_\ver(p)}\\ \leq \norm{\nu_\ver^{-1}(p)} \norm{\omega}_{\tau,\tot} \norm{\nu_\ver(p)},
	\end{multline*}
	and similarly \(\norm{p_1 \omega p_2}_{\tau;\hor} \leq \norm{\nu_\hor^{-1}(p)} \norm{\omega}_{\tau,\tot} \norm{\nu_\hor(p)}\), so that, in turn,
	\begin{align*}
		\norm{p_1\omega p_2}_{\tau;\tot} &\leq \norm{p_1\omega p_2}_{\tau;\ver} + \norm{p_1\omega p_2}_{\tau;\hor}\\
		&\leq \norm{\nu_\ver^{-1}(p)} \norm{\omega}_{\tau,\tot} \norm{\nu_\ver(p)} + \norm{\nu_\hor^{-1}(p)} \norm{\omega}_{\tau,\tot} \norm{\nu_\hor(p)}\\
		&\leq \norm{p}_{\tau} \norm{\omega}_{\tau;\tot} \norm{p}_\tau. \qedhere
	\end{align*}
\end{proof}

\begin{example}\label{ex:hopf13}
	Continuing from Examples \ref{ex:hopf10} and \ref{ex:hopf12}, we may apply Proposition \ref{prop:lipschitz} to \((\slashed{S}_q(\SU(2)),\tilde{\pi},\tilde{\slashed{D}}_q,\Gamma_q)\) equipped with its unique vertical twist \((\Lambda_{q^{-1}},\Lambda_q)\) and unique horizontal twist \((\Lambda_{q^{-1/2}},\Lambda_{q^{1/2}})\).
	We claim that the resulting seminorm \(L_\tau\) is equivalent to \(L_{q^2,q}\), where \((L_{t,q})_{t \in (0,\infty)}\) is the family of Lipschitz seminorms on \(\cO_q(\SU(2))\) with which Kaad--Kyed make \(C_q(\SU(2))\) into a compact quantum metric space \cite[Def.\ 4.6 \& Cor.\ 5.24]{KK}.
	First, note that, for all \(p \in \cO_q(\SU(2))\),
	\[
		\norm{p}_{\tau} = \max\Set*{\norm{\Lambda_{q}(p)} + \norm{\Lambda_{q^{1/2}}(p)},\norm{\Lambda_q^{-1}(p)} + \norm{\Lambda_{q^{1/2}}^{-1}(p)}} = \norm{p}_{q^2,q},
	\]
	where \((\norm{}_{t,q})_{t \in (0,\infty)}\) is the family of norms on \(\cO_q(\SU(2))\) of \cite[\S 3.5]{KK}, so that \eqref{eq:twistedleibniz} for \(L_\tau\) is identical to the inequality of \cite[Lemma 4.8]{KK} for \(L_{q^2,q}\).
	Next, using the explicit construction of Example \ref{ex:hopf10} and the proof of Proposition \ref{prop:liftcommrepconstruct}, we may now write \(L_\tau = \norm{\partial_\tot(\cdot)}\), where \(\partial_{\tot} : \cO_q(\SU(2)) \to \bL^{\U}(\slashed{S}_q(\SU(2))\) is given by 
	\begin{align*}
		\partial_\tot &\coloneqq \Lambda_{q^{-1}}\iu{}[\tilde{\slashed{D}}_{q,\ver},\tilde{\pi}(\cdot)]\Lambda_{q^{-1}} + \Lambda_{q^{-1/2}}\iu{}[\tilde{\slashed{D}}_{q,\hor},\tilde{\pi}(\cdot)]\Lambda_{q^{-1/2}}\\
		&= \sigma^2 \otimes \begin{pmatrix} \Lambda_q \circ \partial_{q^2} & 0 \\ 0 & \Lambda_q \circ \partial_{q^2}\end{pmatrix} + \sigma^3 \otimes \begin{pmatrix} 0 & \Lambda_{q^{-1/2}} \circ \partial_+ \\ \Lambda_{q^{-1/2}} \circ \partial_- &0\end{pmatrix};
	\end{align*}
	here, we identify \(M_2(\cO_q(\SU(2))) \cong M_2(\bC) \otimes \cO_q(\SU(2))\) with its image in \(\bL(\bC^2 \otimes \cO_q(\SU(2)))\) via left multiplication of \(\cO_q(\SU(2))\) on itself, while, for \(t \in (0,\infty)\), we define \(\partial_t : \cO_q(\SU(2)) \to \cO_q(\SU(2))\) by \(\partial_t \coloneqq \bigoplus_{j \in \bZ} 2\pi\iu{}[j]_t \id_{\cO_q(\SU(2))_j}\).
	At last, we relate \(L_\tau\) to \(L_{q^2,q}\) as follows.
	On the one hand, if \(H\) is a \(\bZ/2\bZ\)-graded pre-Hilbert space with an odd self-adjoint unitary \(c\) and \(S : H \to H\) is an odd bounded operator supercommuting with \(c\), then \(\norm{S} = \norm{S_0}\) for \(S_0 \coloneqq -\iu{} c \circ \rest{S}{H_{\even}} = \iu{} S \circ \rest{c}{H_\even}\).
	On the other, we may construct unitary \(U : \cO_q(\SU(2))^2 \to \slashed{S}_q(\SU(2))_{\even}\) by \(U \coloneqq \left(\left(\begin{smallmatrix}p_1\\p_2\end{smallmatrix}\right) \mapsto \left(\begin{smallmatrix}1\\0\end{smallmatrix}\right) \otimes \left(\begin{smallmatrix}1\\0\end{smallmatrix}\right) \otimes p_1 + \left(\begin{smallmatrix}0\\1\end{smallmatrix}\right) \otimes \left(\begin{smallmatrix}0\\1\end{smallmatrix}\right) \otimes p_2\right)\),
	Applying these considerations to \(c = \sigma^1 \otimes \id \otimes \id\) and \(S = \partial_\tot(p)\) for \(p \in \cO_q(\SU(2))\) shows that \(L_\tau = \norm{\partial^\prime_{\tot}(\cdot)}\), where \(\partial^\prime_\tot : \cO_q(\SU(2)) \to M_2(\cO_q(\SU(2)))\) is given by
	\[
		\partial^\prime_\tot \coloneqq \begin{pmatrix} \Lambda_q \circ \partial_{q^2} & -\Lambda_{q^{-1/2}} \circ \partial_+ \\ \Lambda_{q^{-1/2}} \circ \partial_- & -\Lambda_q \circ \partial_{q^2} \end{pmatrix}.
	\]
	But now, given \(t \in (0,\infty)\), comparison with Kaad--Kyed's notations~\cite[\S\S 3.1, 3.5, 4.1]{KK} shows that \(L_{t,q} = \norm{\partial_{t,q}(\cdot)}\) for \(\partial_{t,q} : \cO_q(\SU(2)) \to M_2(\cO_q(\SU(2)))\) given by
	\[
		\partial_{t,q} = \begin{pmatrix} -\iu{}K_t \Lambda_{t^{1/2}} \circ \partial_{t} & -\Lambda_{q^{-1/2}} \circ \partial_+ \\ -\Lambda_{q^{-1/2}} \circ \partial_- & \iu{} K_t \Lambda_{t^{1/2}} \circ \partial_{t}\end{pmatrix}, \quad K_t \coloneqq \frac{1}{2\pi(1+t^{-1})}.
	\]
	Hence, an elementary comparison of \(\partial_{\tot}^\prime\) with \(\partial_{q^2,q}\) implies that
	\[
		\forall p \in \cO_q(\SU(2)), \quad \frac{1}{1+K_{q^2}^{-1}}L_\tau(p) \leq L_{q^2,q}(p) \leq (1+K_{q^2})L_\tau(p).
	\]
\end{example}

\printbibliography

\end{document}